\documentclass[11pt,a4paper,twoside]{book}

 \usepackage{graphicx}
 \usepackage{fancyhdr}
 \usepackage{psfrag}
 \usepackage[section]{placeins}
 \usepackage[hang]{caption2}
 \usepackage{amsfonts}
 \usepackage{amsmath}
 \usepackage{amssymb}
 \usepackage{epic}
 \usepackage{pifont}
 \usepackage{curves}
 \usepackage{graphics}
 \usepackage{graphicx}
 \usepackage{latexsym}
 \usepackage[all]{xy}

 \newtheorem{example}{Example}
 
 \newtheorem{definition}{Definition}
 \newtheorem{proof}{Proof}
 \newtheorem{tableindoc}{Table}
 \newtheorem{lemma}{Lemma}
 \newtheorem{theorem}{Theorem}
 \newtheorem{corollary}{Corollary}

 \newtheorem{proposal}{Proposal}
 \newtheorem{proposition}{Proposition}

 \newcommand{\M}[1]{\mathbb{#1}}

\begin{document}

\begin{titlepage}
    \vspace*{10ex}
    \huge
    \begin{center}
     \bf{Gauge-Higgs unification with broken \\ flavour symmetry} 
    \end{center}
    \vspace{12ex}
    \Large
    \begin{center}
     Dissertation \\ zur Erlangung des Doktorgrades \\ des Fachbereichs Physik \\ 
     der Universit\"at Hamburg  \end{center}
    \vspace{12ex}
    \begin{center}
     vorgelegt von \\ Michael Olschewsky \\ aus Soltau
    \end{center}
    \vspace{2ex}
    \begin{center}
     Hamburg 2007
    \end{center}
\end{titlepage}

\thispagestyle{empty}
\begin{tabular}[bl!]{ll}
\vspace{16cm}\\
Gutachter der Dissertation:             & Prof.~Dr. Gerhard Mack\\
                                        & Prof.~Dr. Klaus Fredenhagen\\
Gutachter der Disputation:              & Prof.~Dr. Gerhard Mack\\
                                        & Prof.~Dr. Jochen Bartels\\
Datum der Disputation:                  & 18. Mai 2007\\
Vorsitzender des Pr\"ufungsausschusses:  & Prof.~Dr. Jochen Bartels\\
Vorsitzender des Promotionsausschusses: & Prof.~Dr. G\"unter Huber\\
Departmentleiter:                       & Prof.~Dr. Robert Klanner\\
Dekan der Fakult\"at f\"ur Mathematik,  &  \\
 Informatik und Naturwissenschaften:    &  Prof.~Dr. Arno Fr\"uhwald
\end{tabular}

\clearpage

~
\vspace{2cm} 

\begin{center} {\bf{Abstract}}  \end{center}

 We study a five-dimensional Gauge-Higgs unification model on the orbifold
 $S^{1}/\mathbb{Z}_{2}$ based on the extended standard model (SM) gauge group 
 $SU(2)_{L} \times U(1)_{Y} \times SO(3)_{F}$. The
 group $SO(3)_{F}$ is treated as a chiral gauged flavour symmetry.  
 Electroweak-, flavour- and Higgs interactions
 are unified in one single gauge group $SU(7)$. The unified gauge
 group $SU(7)$ is broken down to $SU(2)_{L} \times U(1)_{Y} \times SO(3)_{F}$ by orbifolding and
 imposing Dirichlet and Neumann boundary conditions. The compactification
 scale of the theory is $\mathcal{O}(1)$ TeV.
 Furthermore, the orbifold $S^{1}/\mathbb{Z}_{2}$ is put
 on a lattice. This setting gives a well-defined staring point for renormalisation group (RG) 
 transformations. As a result of the RG-flow, the bulk 
 is integrated out and the extra dimension will consist of only two points:
 the orbifold fixed points. The model obtained this way is called an effective bilayered 
 transverse lattice model. Parallel transporters (PT) in the extra dimension become
 nonunitary as a result of the blockspin transformations. In addition, a Higgs potential $V(\Phi)$
 emerges naturally. The PTs can be written as a product
 $e^{A_{y}} e^{\eta} e^{A_{y}}$ of unitary factors $e^{A_{y}}$ and a selfadjoint factor $e^{\eta}$.
 The reduction $\mathbf{48} \to \mathbf{35} + \mathbf{6} + \bar{\mathbf{6}} + \mathbf{1}$
 of the adjoint representation of $SU(7)$ with respect to $SU(6) \supset SU(2)_{L} \times U(1)_{Y} 
 \times SO(3)_{F}$ leads to three $SU(2)_{L}$ Higgs doublets: one for the first, one for the second
 and one for the third generation. Their zero modes serve as a substitute for the SM
 Higgs. When the extended SM gauge group $SU(2)_{L} \times U(1)_{Y} \times SO(3)_{F}$ is
 spontaneously broken down to $U(1)_{em}$, an exponential gauge boson mass splitting occurs naturally.
 At a first step $SU(2)_{L} \times U(1)_{Y} \times SO(3)_{F}$ is broken to 
 $SU(2)_{L} \times U(1)_{Y}$ by VEVs for the
 selfadjoint factor $e^{\eta}$. This breaking leads to masses of flavour changing $SO(3)_{F}$ gauge bosons
 much above the compactification scale. Such a behaviour has no counterpart within the customary 
 approximation scheme of an ordinary orbifold theory.
 This way tree-level flavour-changing-neutral-currents are naturally suppressed.  
 In a second step the electroweak gauge group $SU(2)_{L} \times U(1)_{Y}$ is broken 
 to $U(1)_{em}$ by VEVs for the unitary factors $e^{A_{y}}$ 
 at the electroweak scale. This breaking is equivalent to a Wilson line breaking. 
 Making some simplifying assumptions we also calculate fermion masses and 
 CKM mixing angles. As for the gauge bosons an exponential fermion mass splitting occurs 
 naturally. Fermion masses and mixing angles are determined by the VEVs for $e^{\eta}$ and $e^{A_{y}}$
 of PTs for quarks and leptons. The model predicts a large Higgs
 sector consisting of altogether $30$ Higgs particles.  
 The model in its simplest form also predicts the (too small) weak mixing angle $\theta_{W}=0.125$.

\clearpage

~
\vspace{2cm}

\begin{center}\bf Zusammenfassung\end{center}

Wir untersuchen ein f\"unfdimenisonales Eich-Higgs Vereinigungsmodell auf der 
Orbifold $S^{1}/\mathbb{Z}_{2}$ basierend auf der erweiterten Standardmodell (SM)
Eichgruppe $SU(2)_{L} \times U(1)_{Y} \times SO(3)_{F}$. Die Gruppe $SO(3)_{F}$
wird behandelt als chirale geeichte Flavoursymmetrie. Elektroschwache-, 
Flavour- und Higgswechselwirkungen sind in einer einzigen Eichgruppe $SU(7)$
vereinigt. Die Vereinigungsgruppe $SU(7)$ wird durch Orbifolding und 
Dirichlet- und Neumannrandbedingungen auf
$SU(2)_{L} \times U(1)_{Y} \times SO(3)_{F}$ gebrochen. Die
Kompaktifizierungsskala der Theorie ist $\mathcal{O}(1)$ TeV. Weiterhin 
setzen wir die Orbifold $S^{1}/\mathbb{Z}_{2}$ auf ein Gitter. Dieser
Rahmen gibt einen wohldefinierten Startpunkt f\"ur die Betrachtung von 
Renormierungsgruppentransformationen. Als Ergebnis des 
Renormierungsgruppenflusses wird der Bulk ausintegriert und die Extradimension 
besteht aus nur zwei Punkten: Die Fixpunkte der Orbifold. Wir nennen das 
auf diese Weise erhaltene Modell ein effektives, transverses Zweischichtmodell.
Als ein Ergebnis Blockspintransformationen werden Paralleltransporter (PT) in der Extradimension 
nichtunit\"ar. Zus\"atzlich entsteht ein 
Higgspotential auf nat\"urliche Art und Weise. Die PT k\"onnen
geschrieben werden als ein Produkt $e^{A_{y}} e^{\eta} e^{A_{y}}$ von
unit\"aren Faktoren  $e^{A_{y}}$ und einem selbstadjungierten Faktor
 $e^{\eta}$.
Die Reduktion $\mathbf{48} \to \mathbf{36} + \mathbf{6} + \bar{\mathbf{6}} + \mathbf{1}$ der adjungierten Darstellung von $SU(7)$ bez\"uglich $SU(6)$ f\"uhrt
auf drei $SU(2)_{L}$ Higgsdoublets: Eines f\"ur die erste, eines f\"ur die 
zweite und eines f\"ur die dritte Generation. Ihre Nullmoden dienen als Ersatz
f\"ur das SM Higgs. Wenn die erweiterte SM 
Eichgruppe $SU(2)_{L} \times U(1)_{Y} \times SO(3)_{F}$ spontan zu $U(1)_{em}$
gebrochen wird, entsteht eine exponentielle Aufspaltung der Eichbosonenmassen auf auf nat\"urliche
Art und Weise. Dies f\"uhrt auf Flavoureichbosonen mit Massen weit oberhalb der
Kompaktifizierungsskala. Solch ein Verhalten hat keine Entsprechung
innerhalb der herk\"ommlichen N\"aherungen einer  
Orbifoldtheorie. Flavourver\"anderne neutrale 
Str\"ome sind auf nat\"urliche Art und Weise unterdr\"uckt. Die 
elektroschwache Eichgruppe $SU(2)_{L} \times U(1)_{Y}$ wird durch
Vakuumerwartungswerte f\"ur die unit\"aren Faktoren  
$e^{A_{y}}$ bei der elektroschwachen Brechungsskala auf
$U(1)_{em}$ gebrochen.  Ausserdem berechnen wir unter
vereinfachenden Annahmen Fermionenmassen und die CKM Matrix. Wie f\"ur
Eichbosonen, so ersteht auch f\"ur Fermionen eine exponentielle Massenaufspaltung. Fermionenmassen 
und Mischungswinkel sind festgelegt durch Vakuumerwartungswerte f\"ur
$e^{\eta}$ und $e^{A_{y}}$ von PTn f\"ur Quarks und Leptonen. Das Modell sagt
insgesamt $30$ Higgsteilchen voraus. In seine einfachsten Version sagt das Modell 
den (zu kleinen) schwachen Mischungswinkel $\theta_{W}=0.125$ voraus.

\bibliographystyle{plain}
\tableofcontents

\chapter{Introduction}

  \label{chapterintrodcution}

  During the last ten years much attention has been paid to gauge theories in higher dimensions.
  One of the strongest motivations for extra dimensions is based on the very attractive idea that gauge
  and Higgs fields can be unified in higher dimensions \cite{Manton:1979kb,Hatanaka:1998yp}.
  Gauge bosons and Higgs fields arise from the four-dimensional and extra components of
  higher-dimensional gauge fields, respectively. This scenario is
  called Gauge-Higgs unification \cite{Hall:2001zb,Antoniadis:2001cv,Csaki:2002ur,Burdman:2002se,
  Scrucca:2003ra,Kubo:2001zc,Haba:2004qf}. The gauge group in this class of models must be larger
  than the Standard model (SM) gauge group in order to obtain Higgs fields which transform 
  according to the fundamental representation of $SU(2)_{L}$.
  The larger amount of gauge symmetry can be reduced to the SM one
  by compactifying the extra dimensions on an orbifold. Orbifolding \cite{Hebecker:2001jb,Quiros:2003gg}
  is a technique used to break a gauge group without the use of Higgs fields.
  It has many applications not only in Gauge-Higgs unification models
  but also in GUT breaking \cite{Hebecker:2001wq,Asaka:2001eh}. 

  The SM is extremely successful in reproducing all the available data up to currently accessible
  energies. However, it has serious unsolved problems. One of the biggest problems is the stability 
  of the electroweak scale 
  against quadratically divergent corrections to the Higgs mass. This problem, called
  hierarchy problem, suggests the presence of new physics at the TeV scale \cite{Antoniadis:1990ew,
  Delgado:1998qr,Antoniadis:1998sd}. In Gauge-Higgs unification
  models tree level Higgs masses are forbidden by higher-dimensional gauge invariance. For infinite 
  large extra dimensions the
  masslessness of Higgs fields should hold to any order of perturbation theory. 
  However, for compact extra dimensions radiative corrections
  generate finite mass terms for the Higgs $\sim 1/R$, where $1/R$ is the compactification scale of the 
  theory. In Gauge-Higgs unification models the compactification scale is usually set to $\mathcal{O}(1)$
  TeV. This way Gauge-Higgs unification models on orbifolds give a solution for the 
  hierarchy problem. Electroweak symmetry breaking occurs radiatively in this class of models 
  and is equivalent to a Wilson line symmetry breaking \cite{Witten:1985xc,Ibanez:1987pj} or
  Hosotani breaking \cite{Hosotani:1983xw,Hosotani:1983vn,Hosotani:1988bm}.
  Matter fields can be introduced either as bulk fields \cite{Burdman:2002se}
  in representations of the unified gauge group or as boundary fields \cite{Csaki:2002ur}
  localised at the orbifold fixed points where the
  unified gauge group is broken to its subgroup. 

  Another central problem of the SM is the arbitrariness of the Yukawa couplings and the 
  related problem of the strength of the CKM (and PMNS) matrix elements. 
  This is called the flavour problem: The question why there are three families of quarks and leptons 
  in the SM and how they get their masses and mixing angles. In the literature,
  there are many postulated forms of Yukawa matrices \cite{Fritzsch:1979zq,Fritzsch:1989qm,
  Fritzsch:1999ee}. In order to understand their origin one can try to apply
  a family symmetry $G_{f}$ connecting different generations. Some candidates for a family
  symmetry group are continuous groups like $U(1)$ \cite{Irges:1998ax}, $SU(2)$ \cite{Chen:2000fp, 
  Kuchimanchi:2002yu}, $SU(3)$ \cite{King:2001uz} or $SO(3)$ \cite{Ghosal:1999jb,Wu:1998if,
  King:2005bj} or discrete groups
  like $S_{3}$ \cite{Harrison:2003aw}, $S_{4}$ \cite{Ma:2005pd} or $A_{4}$ \cite{Ma:2002yp}.
  The groups $SU(2)$, $SO(3)$ and $SU(3)$
  as well as $S_{4}$ and $A_{4}$ have the advantage over the groups $U(1)$ and $S_{3}$  that they
  have irreducible three dimensional representations into which the three families of the SM can fit.
  If one adds to the SM gauge group $G_{SM}=SU(3)_{c} \times SU(2)_{L} \times U(1)_{Y}$
  a gauged flavour group $G_{f}$, e.g. $SU(3)$ or $SO(3)$, one is faced with the problem that
  the latter leads to flavour-changing-neutral-currents (FCNC). However, FCNC are highly
  suppressed in the SM due to the GIM-mechanism. There are experimental lower bounds on the masses of such
  flavour gauge bosons of $\mathcal{O}(10^{3})-\mathcal{O}(10^{5})$ TeV, which is much above the
  electroweak breaking scale (but far below the GUT scale). 
  Hence a Gauge-Higgs unification model, which also includes a gauged family symmetry, with a
  compactification scale  $\mathcal{O}(1)$ TeV will in general 
  fail \cite{Martinelli:2005ix}, because it leads to unsuppressed FCNC.
  The reason is that flavour gauge bosons in this scenario will get masses at most of the order of the
  compactification scale, i.e. $\mathcal{O}(1)$ TeV. 

  A possible solution to this problem is to built a Gauge-Higgs unification model
  with the help of nonunitary parallel transporters (PTs). Gauge theories with nonunitary PTs 
  were first examined in \cite{Mack:2005fv,Lehmann:2003jh}. They are based on 
  the idea to abandon unitarity of PTs. In this class of theories PTs 
  are no longer elements of a (unitary \footnote{i.e. a compact gauge group whose finite dimensional
  representations are unitary.}) gauge group $G$ but are rather elements of a
  holonomy group $H$. The holonomy group $H$ is typically noncompact and larger than the unitary gauge
  group one has started with. Nonunitary PTs occur naturally in effective theories
  as a result of the renormalisation group (RG) flow \cite{Lehmann:2003jh}. One starts in a fundamental
  theory with conventional (unitary)  PTs. Blockspin transformations will 
  in general lead to nonunitary PTs. The most exciting property of gauge theories with 
  nonunitary PTs is that an exponential mass hierarchy appears naturally when
  the local gauge symmetry is spontaneously broken by a Higgs mechanism. In \cite{Lehmann:2003jn}
  it has been shown that an exponential flavour mass splitting for quarks
  can be obtained this way.
  
  In this thesis we show that also exponential (flavour) gauge bosons masses can be obtained when the  
  PTs in the extra dimension become nonunitary. This opens up the possibility of
  suppressing tree-level FCNC by large flavour gauge boson masses. We will present a Gauge-Higgs unification
  model, which includes a gauged flavour symmetry, with nonunitary PTs 
  in the extra dimension. It will be consistent with existing experimental constraints on 
  FCNC. The compactification scale of the theory is $\mathcal{O}(1)$ TeV.
 
  The thesis is organised as follows. In chapter \ref{chapterorbifold} we review orbifolds 
  \cite{Hebecker:2001jb,Quiros:2003gg} in one 
  extra dimension. For this analysis, we will refer to the space group
  $\mathbb{D}_{\infty}$ \cite{Nilse:2006jv}. In comparison with the more ad hoc definitions in 
  the literature, the definition of orbifolds in terms of space groups is attractive. The reason is that
  all properties of the orbifold, in particular the orbifold space-time and the various relations the
  projection matrices and twist matrices have to fulfil, can be derived directly from the defining space 
  group. Furthermore, we review the issue of gauge symmetry breaking \cite{Hebecker:2003jt} through orbifolding
  and consider also familiar orbifold constructions in orbifold GUTs. 
  We will work out the Fourier mode expansions and zero modes on 
  the orbifold $S^{1}/\mathbb{Z}_{2}$, which will be useful for the topics discussed 
  in chapter \ref{chaptereffectivetheorie}. 
  In addition, we review continuous Wilson line breaking, also know as Hosotani breaking 
  \cite{Hosotani:1988bm,Haba:2002py}. 

  In chapter \ref{chaptereffectivetheorie} we describe how an effective 
  transverse lattice model can be obtained from an ordinary $S^{1}/\mathbb{Z}_{2}$ orbifold
  model. We start with the five-dimensional space-time $M^{4} \times S^{1}/\mathbb{Z}_{2}$ where $M^{4}$
  is the four-dimensional Minkowski space-time and $S^{1}/\mathbb{Z}_{2}$ is the orbifold.
  Furthermore we put the orbifold $S^{1}/\mathbb{Z}_{2}$ on a lattice.  
  Hence the four-dimensional Minkowski space-time will remain
  continuous and only the extra dimension is latticized. Such a scenario is known as a transverse
  lattice and it occurs naturally in deconstruction theories 
  \cite{Hill:2000mu,Arkani-Hamed:2001ca,Arkani-Hamed:2001nc}. This setting gives a
  well-defined starting point for RG transformations. Starting with this latticized extra
  dimension one can calculate the RG-flow. The endpoint of the RG flow will be 
  an extra dimension, which consists of only two points: the two orbifold fixed points. The
  bulk is completely integrated out. We call the model obtained this way 
  an effective bilayered transverse lattice model 
  (eBTLM). The PTs $\Phi$ in the extra dimension from one orbifold fixed point to
  the other will be nonunitary as a result of the blockspin transformations. They can be interpreted as 
  Higgs fields. When $\Phi$ becomes nonunitary, a Higgs potential $V(\Phi)$ naturally emerges. 
  We will discuss in detail the physical interpretation of an eBTLM. It will turn out that 
  for trivial minimum of the Higgs potential and trivial orbifold projection an eBTLM equals
  an ordinary $S^{1}/\mathbb{Z}_{2}$ orbifold model with trivial orbifold projection,
  if one truncates the Fourier mode expansion for all fields in the $S^{1}/\mathbb{Z}_{2}$
  orbifold model at the first Kaluza-Klein mode. In order to handle also non-trivial minima of the
  Higgs potential and non-trivial orbifold projections we formulate orbifold 
  conditions for nonunitary PTs $\Phi$ and consider spontaneous symmetry breaking.
  As an application, we analyse in detail an eBTLM
  based on the (flavour) gauge group $SU(2)$.  
  The most exciting result is that exponential gauge boson
  masses can occur for some of the first excited {\em{and}} the zero mode gauge bosons, 
  when the gauge group $SU(2)$ is broken spontaneously. 
  This behaviour has no counterpart within the customary approximation scheme
  of an ordinary orbifold theory. 

  In chapter \ref{su7model} we present a realistic Gauge-Higgs unification model, which
  includes a chiral gauged $SO(3)_{F}$ flavour symmetry. This model is based on the gauge group
  $SU(7)$. The gauge group $SU(7)$ unifies electroweak-, flavour- and Higgs interactions. 
  Colour will be ignored.   As an intermediate step the model also unifies weak- and flavour
  interactions in the gauge group $SU(6)_{L} \subset SU(7)$. Zero modes of the 
  extra-dimensional component of the five-dimensional gauge fields, 
  transforming according to the fundamental representation of 
  $SU(2)_{L}$ and carrying the hypercharge $\frac{1}{2}$, will serve as a substitute for the
  SM Higgs. The theory will include three $SU(2)_{L}$ Higgs doublets, one for the first, one for the second
  and one for the third generation. They generate
  the unitary part of the nonunitary bulk parallel transporter $\Phi$. 
  We break $SU(7)$ again down to $SU(6)_{L}\times U(1)_{Y}$ by orbifolding.  
  The gauge symmetry breaking $SU(6)_{L}\times U(1)_{Y} \to SU(2)_{L} \times U(1)_{Y} \times SO(3)_{F}$
  can be achieved by demanding additional Dirichlet- and Neumann boundary conditions 
  for the $SU(6) \times U(1)_{Y}$ gauge fields. When spontaneous symmetry
  breaking occurs, the $SO(3)_{F}$ flavour symmetry is broken by vacuum expectation values (VEVs)
  for the selfadjoint part of $\Phi$.
  This way the flavour gauge bosons can receive very large masses in comparison to the compactification
  scale $1/R=\mathcal{O}(1)$ TeV. Hence tree-level FCNC are naturally suppressed due to the large flavour
  gauge boson
  masses. The electroweak gauge symmetry $SU(2)_{L} \times U(1)_{Y}$ is broken to $U(1)_{em}$ by VEVs for
  the three $SU(2)_{L}$ Higgs doublets. We calculate all gauge boson masses in the model in terms of the 
  minimum $\Phi_{min}$ of the Higgs potential $V(\Phi)$. The model will also make a prediction for
  the weak mixing angle $\theta_{W}$. 

  In chapter \ref{su7fermionmasses} we will calculate the fermion masses and the CKM mixing matrix in
  the $SU(7)$ model under some simplifying assumptions.
  We assume that nonunitary parallel transporters for gauge fields, quarks and leptons
  are different. 

  In chapter \ref{SummaryOutlook} will draw our conclusions and discuss possible extensions
  of the $SU(7)$ model.

\chapter[Orbifolds in one extra dimension]{Orbifolds in one extra dimension, Fourier mode expansion and the Hosotani mechanism}

 \label{chapterorbifold}

  In this chapter we review orbifolds \cite{Hebecker:2001jb,Hebecker:2003jt,Quiros:2003gg,Scrucca:2004jn} and
  gauge symmetry breaking through orbifolding in one extra dimension. In contrast to
  the literature, we will define orbifolds in terms of one-dimensional space groups \cite{Nilse:2006jv}.
  The definition of orbifolds in terms of space groups is attractive since all properties of the 
  orbifold can be derived directly from the defining space group. 
  Furthermore we will work out the Fourier mode expansions and zero modes on 
  the orbifold $S^{1}/\mathbb{Z}_{2}$ which will be useful for the topics discussed 
  in chapter \ref{chaptereffectivetheorie}. 
  In addition, we review continuous Wilson line breaking also know as Hosotani breaking 
  \cite{Hosotani:1988bm,Haba:2002py,Hosotani:2004wv,Hosotani:2005fk}.
  In the following section we sketch the basic ideas \cite{Hebecker:2001jb} of orbifolding.

 \section{The meaning of orbifolding}

   We consider a quantum field theory (QFT) with gauge group $G$
   in $D=d+4$ dimensions, where $d$ denote the number of extra dimensions. 
   The QFT is defined on $M=M^{4} \times C$, where $M^{4}$ is 
   the four-dimensional Minkowski spacetime and $C$ is a smooth manifold. 
   Let 
   \begin{equation}
    x^{M}=(x^{\mu},y^{m}) \quad  
   \mu=0,\dots,3 \quad   m=1,\dots,d
   \end{equation}
   denote the coordinates of the $D$-dimensional space, where $x_{\mu}$ and $y^{m}$ are the 
   coordinates on $M^{4}$ and $C$, respectively. 

   We suppose that both the manifold $C$ and the QFT possess a symmetry
   under a discrete group $\mathcal{K}$, i.e. 
   \begin{enumerate}
    \item
     $\mathcal{K}$ acts on the manifold $C$ as  
     \begin{equation}
      \label{actionkonc}
      \mathcal{K} : y \to \tau_{k}(y) 
     \end{equation}
     where $y=(y^{m})$ and $\tau_{k}$ constitute a representation of $\mathcal{K}$ on $C$.
    \item
     $\mathcal{K}$ acts on the field space as
     \begin{equation}
      \label{actionkonphi}
      \mathcal{K} : \Phi_{(i)} \to P_{k\; (ij)} \Phi_{(j)} \; ,
     \end{equation}
     where $\Phi$ is a vector containing all fields of the theory and $P_{k}$ is a matrix 
     representation of $\mathcal{K}$ on the field space.   
   \end{enumerate}
 
   With the symmetry group $\mathcal{K}$ at hand we can now construct the space $C/\mathcal{K}$ by 
   identifying points $y$ and 
   $\tau_{k}(y)$ that belong to the same orbit
   \begin{equation}
    y \equiv \tau_{k}(y) \; .
   \end{equation} 
   According to the action of $\mathcal{K}$ on $C$ there are two possibilities
   \begin{enumerate}
    \item
     $\mathcal{K}$ acts freely on $C$, i.e.
     \begin{equation}
      \tau_{k}(y) \neq y \quad \forall y \in C \; , \forall k \in \mathcal{K} \; , k \neq 1 \; .
     \end{equation}
     This means that non-trivial elements of $\mathcal{K}$ move all points of $C$. The space 
     $C/\mathcal{K}$ is then again a smooth manifold.  
    \item
     $\mathcal{K}$ acts non freely on $C$, i.e. the action of $\mathcal{K}$ on $C$ has
     fixed points
     \begin{equation}
      \tau_{k}(y) = y \quad \text{for some } y \in C \quad k \neq 1 \; .
     \end{equation}
     The resulting space $C/\mathcal{K}$ is not a smooth manifold but it has singularities at the
     fixed points. Such a space is known as an {\em{orbifold}}.
   \end{enumerate}

   We set $C=\M{R}^{d}$ and consider the quotient space $\M{R}^{d}/\mathcal{K}$. Note that for
   $\mathcal{K}$ we cannot choose any arbitrary discrete group. Instead of that $\mathcal{K}$ is restricted
   to be a $d$-dimensional 
   space group. A $d$-dimensional space group is defined as a discrete group of isometries of $\M{R}^{d}$. 
   \begin{definition}[Orbifold]
    \label{definitionorbifold}
    Let $\mathcal{K}$ be a space group in $d$-dimensions acting {\bf{non freely}} on $\M{R}^{d}$. We define 
    an orbifold in $d$ extra dimensions to be the quotient space
    \begin{equation}
     \label{quotientspace}
     \M{R}^{d}/\mathcal{K} \; .
    \end{equation}  
   \end{definition}
   Remarks: 
   i) Space groups are also known as crystallographic groups and their classification is 
   known for dimensions $d \le 6$. \\
   ii) Since orbifolds are defined as quotient spaces $\M{R}^{d}/\mathcal{K}$, their classification 
   follows directly from the classification of the space groups $\mathcal{K}$. \\
    
   Recall that $\mathcal{K}$ is assumed to be a symmetry of both $\M{R}^{d}$ and the QFT. We declare
   that only field configurations invariant under the actions 
   (\ref{actionkonc}) and (\ref{actionkonphi}) are physical. This means that we demand
   \begin{equation}
    \label{orbifoldconditionsgeneral}
    \Phi_{(i)}(x^{\mu},\tau_{k}(y))=P_{k\; (ij)} \Phi_{(j)}(x^{\mu},y) \; .
   \end{equation}
   In general the action of $\mathcal{K}$ on the fields can make use of all symmetries of the QFT.
   This means that $P_{k}$ can involve 
   gauge transformations, discrete parity transformations and in the supersymmetric case,
   $R$-symmetry transformations \cite{Hebecker:2003jt}. In this thesis we consider the case where 
   $P_{k}$ involves gauge
   transformations and restrict ourselves to orbifolds in one extra dimension, i.e. we take $d=1$
   in (\ref{quotientspace}).

  \section{One-dimensional orbifolds}

   Let us first consider all possible space groups in one dimension and as a start do not care whether they
   act freely on $\M{R}$ or not. 
   The real line $\M{R}$ has two possible isometries, the 
   translation $t$ and the $\pi$-rotation $r$. The one-dimensional
   space groups are therefore \cite{Nilse:2006jv}
   \begin{eqnarray}
     \label{onedimspacegroups}
     \M{Z}          & = & \langle t \rangle   \; , \\
     \M{D}_{\infty} & = & \langle t,r \mid r^{2}=1, (tr)^{2}=1 \rangle 
                            \supseteq \M{Z},\M{Z}_{2} \nonumber  \; ,
   \end{eqnarray}
   where $\M{Z}_{2}=\langle r \mid r^{2}=1 \rangle$.
   The space groups (\ref{onedimspacegroups}) are defined in a purely algebraic way, i.e. initially
   we do not specify a particular  representation of them. Instead of that we define a set of 
   generators and list the relations among them. This
   way the space groups (\ref{onedimspacegroups}) are uniquely defined. Take
   for example the space group $\M{D}_{\infty}$. It is generated by a 
   translation $t$ and a $\pi$-rotation $r$. The generators $r$ and
   $t$ fulfil the relations $r^{2}=1$ and $(tr)^{2}=1$.
   It is important and we will make use of this fact later that the 
   choice of the generators in (\ref{onedimspacegroups}) is not unique \cite{Nilse:2006jv}. For instance the
   space group $\M{D}_{\infty}$ can be defined equally in terms of 
   two $\pi$-rotations
   \begin{equation}
    \label{dinfty2}
    \M{D}_{\infty}=\langle r,r^{\prime} \mid r^{2}=r^{\prime 2}=1 \rangle \; ,
   \end{equation}
   with $r^{\prime}=tr$. Note that $r r^{\prime} \neq r^{\prime} r$.  \\
   Remark: $\M{D}_{\infty}$ may have representations $P,P^{\prime}$ of $r,r^{\prime}$ 
   on the field space, which are not faithful. For instance, one can have representations
   $P,P^{\prime}$ fulfilling $P P^{\prime}=P^{\prime} P$. In fact we will consider this
   possibility later in section \ref{sectioncontinuousdiscretewislonlineshosotani}.

   For each $\mathcal{K}$ let $\mathcal{K}^{\prime}$ be the largest subgroup of $\mathcal{K}$ that 
   does not include translations. Thus we can rewrite 
   (\ref{quotientspace}) for $d=1$ as
   \begin{equation}
    \M{R}/\mathcal{K}=S^{1}/\mathcal{K}^{\prime} \; ,
   \end{equation}
   where $S^{1}$ is the circle. The circle $S^{1}$ is the quotient space
   $\M{R}/\M{Z}$ and it is constructed by identifying the points
   \begin{equation}
    y \to y +2 \pi R \label{periodicity} \; ,
   \end{equation}
   on $\M{R}$. Here $y$ denotes the coordinate on $\M{R}$ and $R$ is the 
   compactification radius, i.e. the radius of $S^{1}$. In (\ref{periodicity}) we have given a particular
   representation of $t$ on $\M{R}$. Since in one dimension there exist only two space
   groups, namely $\M{Z}$ and $\M{D}_{\infty}$, we arrive at the two one-dimensional
   compact spaces
   \begin{equation}
    S^{1}=\mathbb{R}/\mathbb{Z} \quad , \quad 
    S^{1}/\mathbb{Z}_{2}=\mathbb{R}/\mathbb{D}_{\infty} \; .
   \end{equation}
   We will see later that only $S^{1}/\mathbb{Z}_{2}$ has fixed points and is therefore 
   the only one-dimensional orbifold.

  \subsection{Gauge symmetry breaking through orbifolding}

   We consider now the case where $\mathcal{K}$ acts on the space of gauge fields.  
   If $\mathcal{K}$ is a symmetry of the gauge action $P_{k}$ will act as a gauge transformation.
   To be more precise, 
   consider a five-dimensional gauge field $A_{M}=A_{M}^{A} T^{A}$
   where $T^{A}$ are the generators of $G$, $M \in (\mu,y)$ 
   and $A=1,\dots,\dim(G)$. Let the generators $T^{A}$ be 
   normalised such that $\text{tr} \left( T^{A} T^{B} \right)=\frac{1}{2} \delta_{AB}$. The
   five-dimensional Yang-Mills action reads 
   \begin{equation}
    \label{5dactionorbi}
    S_{5D}=\int d^{4} x dy \; \text{tr} \left( -\frac{1}{2} F_{MN} F^{MN} \right) \; ,
   \end{equation}
   where $F_{MN}=F_{MN}^{A} T^{A}$, $F_{MN}^{A}=\partial_{M} A_{N}^{A} - \partial_{N} A_{M}^{A}
   + g_{5} f^{ABC} A_{M}^{B} A_{N}^{C}$, $M,N \in (\mu,y)$, $\left[ T^{A}, T^{B} \right]=i f^{ABC} T^{C}$ and
   $g_{5}$ denote the five-dimensional gauge coupling constant. 
   The $T^{A}$ are considered here as a matrix representation of the generators of $G$.
   Under a gauge transformation $\Omega(x^{\mu},y) \in G$ on the covering space $\mathbb{R}$
   the five-dimensional gauge field $A_{M}(x^{\mu},y)$ transforms as
   \begin{equation}
    A_{M}(x^{\mu},y) \to A^{\prime}_{M}(x^{\mu},y)=\Omega(x^{\mu},y) A_{M}(x^{\mu},y)
                                                   \Omega(x^{\mu},y)^{-1}
                    -\frac{i}{g} \Omega(x^{\mu},y) \partial_{M} \Omega(x^{\mu},y)^{-1} \; .
   \end{equation}
   We represent $r$ and $t$ on $\mathbb{R}$ by
   \begin{gather}
    y \to -y   \; , \\
    y \to y+2 \pi R  \; ,
   \end{gather}
   respectively, and on $\mathfrak{g}=\text{Lie} \; G$ by
   \begin{gather}
    A_{M}(x^{\mu},y) \to P \; A_{M}(x^{\mu},y) \; P^{-1}  \; ,  \\
    A_{M}(x^{\mu},y) \to T \; A_{M}(x^{\mu},y) \; T^{-1}  \; ,
   \end{gather}
   respectively. Note that we have restricted here to the case where the action of $t$ and $r$ on
   $\mathfrak{g}=\text{Lie} \; G$ can be written as an inner automorphism. 
   According to (\ref{orbifoldconditionsgeneral}) we demand 
   \footnote{Note that
    the minus sign in (\ref{orbifoldboundaryconditions2}) is needed in order to maintain 
    the gauge covariance for $F_{\mu y}$, i.e.
    \begin{eqnarray*}
     F_{\mu y}(x^{\mu},-y) & = & \partial_{\mu} A_{y}(x^{\mu},-y) - \partial_{y} A_{\mu}(x^{\mu},-y)
     -i g_{5} \left[ A_{\mu}(x^{\mu},-y), A_{y}(x^{\mu},-y) \right] \\
     & = &   - \partial_{\mu} (P A_{y}(x^{\mu},y) P^{-1}) - \partial_{y} (P A_{\mu}(x^{\mu},y) P^{-1}) 
             + i g_{5} P \left[A_{\mu}(x^{\mu},y) , A_{y}(x^{\mu},y) \right] P^{-1}  \\
     & = &  - P \left(\partial_{\mu} A_{y}(x^{\mu},y) - \partial_{y} A_{\mu}(x^{\mu},y) 
             - i g_{5} \left[A_{\mu}(x^{\mu},y) , A_{y}(x^{\mu},y) \right] \right) P^{-1}  \\
     & = & - P F_{\mu y}(x^{\mu},y) P^{-1}  \; .
    \end{eqnarray*}  
    If we instead of (\ref{orbifoldboundaryconditions2}) demand 
    \begin{equation*}
     A_{y}(x^{\mu},-y)=P \; A_{y}(x^{\mu},y) \; P^{-1} \; ,
    \end{equation*}
    we get
    \begin{eqnarray*}
     F_{\mu y}(x^{\mu},-y) & = & \partial_{\mu} A_{y}(x^{\mu},-y) - \partial_{y} A_{\mu}(x^{\mu},-y)
     -i g_{5} \left[ A_{\mu}(x^{\mu},-y), A_{y}(x^{\mu},-y) \right] \\
     & = &    \partial_{\mu} (P A_{y}(x^{\mu},y) P^{-1}) - \partial_{y} (P A_{\mu}(x^{\mu},y) P^{-1}) 
             + i g_{5} P \left[A_{\mu}(x^{\mu},y) , A_{y}(x^{\mu},y) \right] P^{-1}  \\
     & = &  P \left(\partial_{\mu} A_{y}(x^{\mu},y) + \partial_{y} A_{\mu}(x^{\mu},y) 
             + i g_{5} \left[A_{\mu}(x^{\mu},y) , A_{y}(x^{\mu},y) \right] \right) P^{-1}  \; .
    \end{eqnarray*}}
   \begin{gather}
    \label{orbifoldboundaryconditions1}
    A_{\mu}(x^{\mu},-y)=P \; A_{\mu}(x^{\mu},y) \; P^{-1}  \\
    \label{orbifoldboundaryconditions2}
    A_{y}(x^{\mu},-y)=-P \; A_{y}(x^{\mu},y) \; P^{-1}    \\
    \label{orbifoldboundaryconditions3}
    A_{M}(x^{\mu},y+2 \pi R)=T \; A_{M}(x^{\mu},y) \; T^{-1}  \; .
   \end{gather}
   It follows that
   \begin{gather}
    F_{\mu \nu}(x^{\mu},-y)=P \; F_{\mu \nu}(x^{\mu},y) \; P^{-1} \\
    \label{fmuycondition}
    F_{\mu y}(x^{\mu},-y)=-P \; F_{\mu y}(x^{\mu},y) \; P^{-1} \\
    F_{MN}(x^{\mu},y+2 \pi R)=T \; F_{MN}(x^{\mu},y) \; T^{-1} \; . 
   \end{gather}
   Thus (\ref{5dactionorbi}) is invariant under the action of $\M{D}_{\infty}$.
   The conditions (\ref{orbifoldboundaryconditions1}), (\ref{orbifoldboundaryconditions2}) are known 
   as boundary conditions and the condition (\ref{orbifoldboundaryconditions3}) 
   is known as periodicity condition.
   
   Let us discuss the issue of gauge symmetry breaking due to the boundary condition
   (\ref{orbifoldboundaryconditions1}) and the periodicity condition (\ref{orbifoldboundaryconditions3})
   for the four-dimensional components $A_{\mu}(x^{\mu},y)$ of the five-dimensional gauge field
   $A_{M}(x^{\mu},y)$.
   First (\ref{orbifoldboundaryconditions1}) and (\ref{orbifoldboundaryconditions3}) can
   alternatively be understood in terms of local gauge symmetry breaking at the
   various fixed points of the orbifold. This reinterpretation comes out if one 
   takes into account that a generic fixed point $y_{i} \in \mathbb{R}$ is left fixed 
   by an element $k^{\prime} \in \mathcal{K}^{\prime}$ only modulo a suitable translation
   in the covering space $\mathbb{R}$, i.e.
   \begin{equation}
    \label{definitionfixedpoints}
    y_{i}=k^{\prime}(y_{i})+n_{i} \cdot 2 \pi R \; ,
   \end{equation}
   where $n_{i} \in \mathbb{N}$ depend on the particular fixed point $y_{i}$.
   Thus we conclude that the effective orbifold projection $P_{i}$, assigned to
   the fixed point $y_{i}$, is given by
   \begin{equation}
    P_{i}=T^{n_{i}} P \; .
   \end{equation}
   The boundary condition for the four-dimensional gauge fields at a given fixed point $y_{i}$ then reads
   \begin{equation}
     A_{\mu}(x^{\mu},y_{i}-y)=P_{i} \; 
     A_{\mu}(x^{\mu},y_{i}+y) \; P_{i}^{-1}  \; .
   \end{equation}
   This formula shows explicitly that 
   \begin{itemize}
    \item the gauge group $G$ is broken {\em{locally}} at
          the orbifold fixed point $y_{i}$ to the centraliser of $P_{i}$ in $G$
          \begin{equation}
           \label{localunbrokengaugegroup}
            H_{i}=\{ g \in G \mid P_{i}g=gP_{i} \}
          \end{equation}
    \item away from the fixed points, i.e. in the bulk, the gauge group $G$ remains unbroken.
   \end{itemize}
   The {\em{globally}} unbroken gauge group $H$, i.e. the gauge group of the low energy four-dimensional
   effective theory, is given by the intersection 
   \begin{equation}
    \label{globalunbrokengaugegroupasintersection}
    H=\cap_{i} H_{i} \; .
   \end{equation}
   It is remarkable that this reinterpretation follows directly from the fact
   that the definition of the space group generators $t,r$ in
   (\ref{onedimspacegroups}) is not unique. In fact one can always redefine the generators
   $t$ and $r$ such that to every fixed point $y_{i}$ of the orbifold one can assign one
   generator of the space group. In the next section we will discuss this topic for the orbifold
   $S^{1}/\mathbb{Z}_{2}$.

 \section{The orbifold $S^{1}/\mathbb{Z}_{2}$}

  \label{sectionorbifolds1z2}
         
  The orbifold $S^{1}/\mathbb{Z}_{2}=\mathbb{R}/\mathbb{D}_{\infty}$ is 
  the  quotient space of the real line modulo $\mathbb{D}_{\infty}$. Recall
  that  $\mathbb{D}_{\infty}$ is defined as 
  \begin{equation}
   \label{definitionrt}
   \mathbb{D}_{\infty}=\langle t,r \mid r^{2}=1, (tr)^{2}=1 \rangle \; .
  \end{equation}
  This space group has two generators, the translation $t$ and the reflection $r$.
  We represent $t$ on $\mathbb{R}$ by
  \begin{equation}
   \label{s1condition}
   y \to y+2 \pi R \; .
  \end{equation}
  Thus we arrive as an intermediate step at the circle 
  $S^{1}=\mathbb{R}/\mathbb{Z}$. Figure \ref{figurecircle} shows the representation
  of $t$ on $\mathbb{R}$ and the resulting space $S^{1}$.
  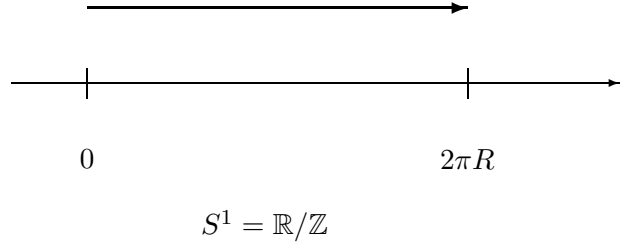
\begin{figure}[h]
   \begin{equation*}  
    \begin{picture} (8,4)
    \thicklines
    \put(1,3) {\vector (1,0) {5}}
    \thinlines
    \put (0,2) {\vector (1,0) {8}}
    \multiput(1,1.8)(5,0) {2} {\line (0,1){0.4}}
    \put(2.5,0) {$S^{1}=\mathbb{R}/\mathbb{Z}$}
    \multiputlist(1,1)(5,0){$0$,$2\pi R$}
    \end{picture}
   \end{equation*}
   \caption{Representation of $t$ on $\mathbb{R}$ (thick black arrow) and the resulting space $S^{1}.$}
   \label{figurecircle}
  \end{figure}
  Note that $t$ acts freely on $\mathbb{R}$ and thus $S^{1}$ possesses no fixed points. 
  Consequently $S^{1}$ is not an orbifold. 
  In order to arrive at the orbifold $S^{1}/\mathbb{Z}_{2}$ we represent $r$ on $\mathbb{R}$ by 
  \begin{equation}
   \label{z2condition}
   y \to -y   \; ,
  \end{equation}
  i.e. we divide the circle $S^{1}$ by a $\mathbb{Z}_{2}$ transformation. Figure \ref{figorbifold}
  shows the representation of $t$ and $r$ 
  on $\mathbb{R}$ and the resulting space $S^{1}/\mathbb{Z}_{2}$.  
  \begin{figure}[h]
   \begin{equation*}
    \begin{picture} (8,4)
     \thicklines
     \put(1,3) {\vector (1,0) {5}}
     \thinlines
     \put (0,2) {\line (1,0) {3.5}}
     \put (1,2) {\circle*{0,15}}
     \dottedline[\circle*{0.05}]{0.1}(3.5,2) (6,2)
     \put (6,2) {\vector (1,0) {2}}
     \multiput(1,1.8)(5,0) {2} {\line (0,1){0.4}}
     \multiputlist(1,1)(2.5,0){$0$,$\pi R$,$2\pi R$}
     \put(2.2,0) {$S^{1}/\mathbb{Z}_{2}=\mathbb{R}/\mathbb{D}_{\infty}$}
    \end{picture}
   \end{equation*}
   \caption{Representation of $t$ (thick black arrow) and $r$ (black dot at $y=0$)  
            on $\mathbb{R}$ and the resulting space $S^{1}/\mathbb{Z}_{2}$.}
   \label{figorbifold}
  \end{figure}
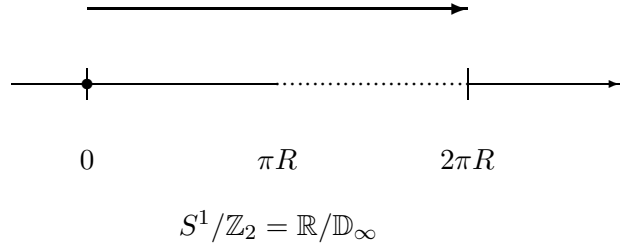
  Due to the definition of $\mathbb{D}_{\infty}$ (\ref{definitionrt}), the following relations hold
  \begin{equation}
   \label{dinftyrelations}
   r^{2}=(tr)^{2}=1 \quad , \quad t=(tr)r \quad , \quad trt=r \; .
  \end{equation}
  The gauge fields have to fulfil the boundary conditions
  \begin{gather}
    \label{pboundarycondition}
    A_{\mu}(x^{\mu},-y)=P \; A_{\mu}(x^{\mu},y) \; P^{-1}    \\
    A_{y}(x^{\mu},-y)=-P \; A_{y}(x^{\mu},y) \; P^{-1}   \nonumber  
  \end{gather}
  and the periodicity condition 
  \begin{equation}
   \label{tboundarycondition}
   A_{M}(x^{\mu},y+2 \pi R)=T \; A_{M}(x^{\mu},y) \; T^{-1} \; .
  \end{equation}

  The orbifold  $S^{1}/\mathbb{Z}_{2}$ has two fixed points $y_{1}=0$ 
  and $y_{2}=\pi R$, where $y_{1}=0$ is invariant under the group element $r$ 
  \begin{equation}
   y_{1}=0 \stackrel{r}{\rightarrow} 0=y_{1}
  \end{equation}
  and $y_{2}=\pi R$ is invariant under the group element $tr$
  \begin{equation}
    y_{2}=\pi R \stackrel{r}{\rightarrow} -\pi R \stackrel{t}{\rightarrow}
     \pi R=y_{2} \; .
  \end{equation}
  This means that in (\ref{definitionfixedpoints}) we have $n_{1}=0$ and $n_{2}=1$. 
  The corresponding effective projections are therefore 
  \begin{equation}
   \label{s1z2z2primetwists}
   P_{1}=P \quad , \quad P_{2}=TP \; .
  \end{equation} 
  Consequently we can rewrite the orbifold boundary conditions (\ref{pboundarycondition}) and
  the periodicity condition (\ref{tboundarycondition}) as
  \begin{gather}
    \label{boundaryconditionss1z2twisted}
    A_{\mu}(x^{\mu},-y)=P_{1} \; A_{\mu}(x^{\mu},y) \; P_{1}^{-1}   \\
    A_{y}(x^{\mu},-y)=-P_{1} \; A_{y}(x^{\mu},y) \; P_{1}^{-1} \; , \nonumber \\
    \nonumber  \\
    A_{\mu}(x^{\mu},\pi R-y)=P_{2} \; A_{\mu}(x^{\mu},\pi R+y) \; 
                             P_{2}^{-1}   \nonumber \\
    A_{y}(x^{\mu},\pi R-y)=-P_{2} \; A_{y}(x^{\mu},\pi R+y) \;
                            P_{2}^{-1} \; . \nonumber 
  \end{gather}   
  Due to (\ref{dinftyrelations}), the projection matrices $P_{1}$ and $P_{2}$ fulfil
  \begin{equation}
   \label{twistrelations}
   P_{1}^{2}=P_{2}^{2}=1 \quad , \quad T=P_{2}P_{1} \quad , 
   \quad TP_{1}T=P_{1} \; .
  \end{equation}
  The resulting physical space  $S^{1}/\mathbb{Z}_{2}$ is the interval $[0,\pi R]$. 

  At $y=0$, the gauge group $G$ is broken to the centraliser of $P_{1}$ in $G$
  \begin{equation}
    \label{unbrokenh1}
    H_{1}=\{ g \in G \mid P_{1}g=gP_{1} \}
  \end{equation}
  and  at $y=\pi R$ is broken to the centraliser of $P_{2}$ in $G$
  \begin{equation}
     \label{unbrokenh2}
    H_{2}=\{ g \in G \mid P_{2}g=gP_{2} \} \; .
  \end{equation}
  The low energy four-dimensional gauge group is given by the intersection
  \begin{equation}
   H=H_{1} \cap H_{2}=\{ g \in G \mid P_{1}g=gP_{1} \; \wedge P_{2}g=gP_{2}  \} \; .
  \end{equation}  
  It is remarkable that in general
  \begin{equation}
   \left[ P_{1},P_{2} \right] \neq 0
  \end{equation}
  This allows to reduce the $rank$ of $\mathfrak{g}$, i.e. $rank \; \mathfrak{h} < rank \; \mathfrak{g}$,
  where $\mathfrak{h}=Lie H$.

  In the last section we have argued that this reinterpretation  
  follows directly from the fact that the definition of the space group generators
  in $\mathbb{D}_{\infty}$ is not unique. In fact we can rewrite
  \begin{equation}
   \label{definitionrrprime}
   \M{D}_{\infty}=\langle r,r^{\prime} \mid r^{2}=r^{\prime 2}=1 \rangle \; ,
  \end{equation}  
  with $r^{\prime}=tr$. Remember that $r r^{\prime} \neq r^{\prime} r$. 
  In definition (\ref{definitionrrprime}) the two space group
  generators $r$ and $r^{\prime}$ are {\em{directly assigned}} to the fixed points
  $y_{1}=0$ and $y_{2}=\pi R$, respectively.  Figure \ref{figorbifold2}
  shows the representation of $r$ and $r^{\prime}$
  on $\mathbb{R}$, and the resulting space $S^{1}/\mathbb{Z}_{2}$.
  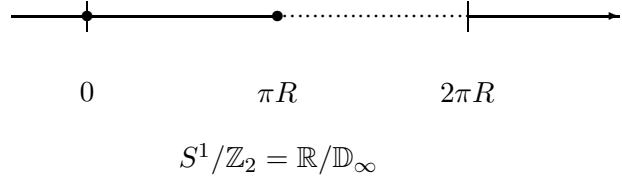
\begin{figure}[h]
   \begin{equation*}
    \begin{picture} (8,4)
    \thinlines
    \put (0,2) {\line (1,0) {3.5}}
    \multiput (1,2) (2.5,0){2} {\circle*{0,15}}
    \dottedline[\circle*{0.05}]{0.1}(3.5,2) (6,2)
    \put (6,2) {\vector (1,0) {2}}
    \multiput(1,1.8)(5,0) {2} {\line (0,1){0.4}}
    \multiputlist(1,1)(2.5,0){$0$,$\pi R$,$2\pi R$}
    \put(2.2,0) {$S^{1}/\mathbb{Z}_{2}=\mathbb{R}/\mathbb{D}_{\infty}$}
    \end{picture}
   \end{equation*}
   \label{figorbifold2}
   \caption{Representation of $r$ (black dot at $y=0$) and $r^{\prime}$ (black dot at $y=\pi R$)
            on $\mathbb{R}$ and the resulting space $S^{1}/\mathbb{Z}_{2}$.}
  \end{figure}
  The orbifold defined by (\ref{definitionrt}) and leading to the boundary
  condition (\ref{pboundarycondition}) and the periodicity condition
  (\ref{tboundarycondition}) is known as $S^{1}/\mathbb{Z}_{2}$ with twisted
  boundary conditions.

 \subsection{The orbifold $S^{1}/\mathbb{Z}_{2} \times \mathbb{Z}^{\prime}_{2}$}

  \label{sectionorbifolds1z2z2prime}

  In the literature, especially in orbifold GUTs, one is often 
  faced with the orbifold $S^{1}/\mathbb{Z}_{2} \times \mathbb{Z}^{\prime}_{2}$ \cite{Hebecker:2001wq}.
  It is constructed as follows. The starting point is a circle 
  $S^{1}$ of radius $R^{\prime}$. We divide $S^{1}$ by two 
  $\mathbb{Z}_{2}$ transformations 
  \begin{equation}
   \label{z2z2prime}
   \mathbb{Z}_{2}: \quad y \to -y  \quad , \quad
   \mathbb{Z}^{\prime}_{2}: \quad y^{\prime} \to -y^{\prime} \; , 
  \end{equation}
  where $y^{\prime}=y-\pi R^{\prime}/2$. In this case the resulting physical space is the 
  interval $[0,\pi R^{\prime}/2]$. 

  The gauge fields have to fulfil the boundary conditions
  \begin{gather}
    \label{boundaryconditionss1z2z2prime}
    A_{\mu}(x^{\mu},-y)=P \; A_{\mu}(x^{\mu},y) \; P^{-1}  \\
    A_{y}(x^{\mu},-y)=-P \; A_{y}(x^{\mu},y) \; P^{-1}  \; , \nonumber \\
    \nonumber  \\
    A_{\mu}(x^{\mu},-y^{\prime})=P^{\prime} A_{\mu}(x^{\mu},y^{\prime}) 
                             P^{\prime \; -1}  \nonumber \\
    A_{y}(x^{\mu},-y^{\prime})=-P^{\prime} A_{y}(x^{\mu},y^{\prime}) 
                            P^{\prime \; -1} \nonumber \; .
  \end{gather}   
  The projection matrices $P$ and $P^{\prime}$ fulfil
  \begin{equation}
   \label{twistrelations}
   P^{2}=P^{\prime \; 2}=1 \; .
  \end{equation} 
  If one compares the resulting physical space $S^{1}/\mathbb{Z}_{2} \times \mathbb{Z}^{\prime}_{2}
  =[0,\pi R^{\prime}/2]$ generated by the two reflections 
  (\ref{z2z2prime}) with the resulting physical space $S^{1}/\mathbb{Z}_{2}=[0,\pi R]$ 
  generated by the translation (\ref{s1condition}) 
  and the reflection (\ref{z2condition}) and the boundary conditions 
  (\ref{boundaryconditionss1z2z2prime}) and 
  (\ref{boundaryconditionss1z2twisted}), one observes that the orbifold
  $S^{1}/\mathbb{Z}_{2} \times \mathbb{Z}^{\prime}_{2}$ is equivalent to 
  the orbifold $S^{1}/\mathbb{Z}_{2}$ with twisted boundary conditions
  if we take
  \begin{equation}
   R^{\prime}=2 R \; .
  \end{equation}  
  Therefore we will not distinguish between the orbifold
  $S^{1}/\mathbb{Z}_{2} \times \mathbb{Z}^{\prime}_{2}$ and the orbifold
  $S^{1}/\mathbb{Z}_{2}$ with twisted boundary conditions. 

  In general, $[P,P^{\prime}] \neq 0$. However in orbifold GUTs it is assumed \cite{Hebecker:2001wq} that 
  projection matrices $P$ and $P^{\prime}$ commute
  \begin{equation}
   \label{commutingpprime}
   [P,P^{\prime}]=0 \; .
  \end{equation}
  This means that is this case the representation $P,P^{\prime}$ of $\mathbb{D}_{\infty}$ on the field space 
  is not faithful. Due to (\ref{commutingpprime}) the $rank$ of $\mathfrak{g}$ is not reduced.

 \section{Continuous versus discrete Wilson line breaking}
  
  \label{sectioncontinuousdiscretewislonlineshosotani}

  In this section we give an interpretation of the twist matrix $T$ in terms of Wilson lines 
  \cite{Scrucca:2004jn,Hall:2001tn}.
  Remember that the five-dimensional gauge field $A_{M}(x)$ has to fulfil the periodicity condition 
  \begin{equation}
   \label{tboundarycondition2}
   A_{M}(x^{\mu},y+2 \pi R)=T \; A_{M}(x^{\mu},y) \; T^{-1} \; .
  \end{equation}
  The twist matrix $T$ can always be interpreted as a Wilson line $W$
  \begin{equation}
   \label{wilsonline}
   W=\exp \left(2 \pi i g R  \left< A_{y} \right> \right)  \; ,
  \end{equation}
  where $\left< A_{y} \right>$ is a constant VEV for $A_{y}(x^{\mu},y)$. 
  However $W$ and therefore $\left< A_{y} \right>$ must be compatible with
  the boundary condition for $A_{y}(x^{\mu},y)$ (\ref{pboundarycondition}).
  This means that according to (\ref{dinftyrelations}) the orbifold projection $P$ and the Wilson
  line $W$ has to fulfil the consistency condition
  \begin{equation}
   \label{consistencycondition}
   (WP)^{2}=1 \; .
  \end{equation}
  In general $W$, and therefore $\left< A_{y}\right>$, need not to commute with $P$.
  To be more precise, suppose that boundary conditions for $A_{y}$ are given
  \begin{equation}
    \label{pboundaryconditionwilson1}
    A_{y}(x^{\mu},-y)=-P \; A_{y}(x^{\mu},y) \; P^{-1}  \; .
  \end{equation}

  Three possibilities \cite{Scrucca:2004jn} can occur
  \begin{enumerate}
   \item Let $\{ T^{a} \}$ denote the set of generators of $G$, which 
         fulfil simultaneously
         \begin{equation}
          \label{generatorsofhglobal}
          [P,T^{a}]=0 \quad , \quad [WP,T^{a}]=0 \; .
         \end{equation}
         Note, that these generators create the four-dimensional
         unbroken gauge group $H$ due to (\ref{unbrokenh1}) and
         (\ref{unbrokenh2}). The relations (\ref{generatorsofhglobal}) imply $[W,T^{a}]=0$.
         Thus (\ref{wilsonline}) can be written as
         \begin{equation}
          \label{discretewilsonline}
          W=\exp \left(2 \pi i g R \sum_{a} \left< A^{a}_{y} \right> T^{a} \right) 
         \end{equation}
         i.e. $A_{y}=\sum_{a} A^{a}_{y}T^{a}$. The Wilson line (\ref{discretewilsonline})
         commutes with every $T^{b} \in \{ T^{a} \}$
         \begin{equation}
          [\exp \left(2 \pi i g R \sum_{a} \left< A^{a}_{y} \right> T^{a} \right),T^{b}]=0 \; .
         \end{equation}
         Due to the
         minus sign in the boundary conditions (\ref{pboundaryconditionwilson1})
         $A_{y}$ is {\em{odd}} under $P$.
         Since $[P,T^{a}]=0$, $W$ also commutes with $P$
         \begin{equation}
          \label{norankreduction}
          [P,W]=0 \; .
         \end{equation}
         Together with (\ref{consistencycondition}) this yields 
         \begin{equation}
          \label{conditiondiscretewilsonline}
          W^{2}=1 \; .
         \end{equation}
         Thus $\left< A^{a}_{y} \right>$ in (\ref{discretewilsonline}) can take only special values
         compatible with (\ref{conditiondiscretewilsonline}). Therefore the Wilson line constructed
         from $A_{y}=A^{a}_{y}T^{a}$ is called a {\em{discrete Wilson line}}. Note, that because
         of (\ref{norankreduction}) a discrete Wilson line symmetry breaking preserve the
         $rank$, i.e. $rank \; \mathfrak{h}=rank \; \mathfrak{g}$ \footnote{Recall that
         $\mathfrak{h}=Lie H$ and $\mathfrak{g}=Lie G$.}.
   \item Let $\{T^{\hat{a}}\}$ denote the set of generators of $G$ which fulfil simultaneously 
         \begin{equation}
          \label{definitiongeneratorsthat}
          \{P,T^{\hat{a}}\}=0 \quad , \quad \{WP,T^{\hat{a}}\}=0 \; .
         \end{equation}
         This implies $[W,T^{\hat{a}}]=0$. 
         Thus (\ref{wilsonline}) can be written as 
         \begin{equation}
          \label{wilonlinecontious}
          W=\exp \left(2 \pi i g R  \sum_{\hat{a}} \langle A^{\hat{a}}_{y} \rangle T^{\hat{a}} \right) 
         \end{equation}
         i.e. $A_{y}=\sum_{\hat{a}} A^{\hat{a}}_{y} T^{\hat{a}}$.
         The Wilson line (\ref{wilonlinecontious}) commutes 
         \footnote{Let $W$ be given by (\ref{wilonlinecontious}). According to 
          (\ref{definitiongeneratorsthat}) we have for any $T^{\hat{b}} \in \{T^{\hat{a}}\}$:
          $WP T^{\hat{b}}=-W T^{\hat{b}} P \stackrel{!}{=} T^{\hat{b}} WP$. Thus $[W,T^{\hat{b}}]=0$ 
          for every $T^{\hat{b}}$.}  
         with every $T^{\hat{b}} \in \{ T^{\hat{a}} \}$
         \begin{equation}
          [\exp \left(2 \pi i g R  \sum_{\hat{a}} \langle A^{\hat{a}}_{y} \rangle T^{\hat{a}} \right),
          T^{\hat{b}}]=0 \; .
         \end{equation}
         Due to the minus sign in the boundary conditions (\ref{pboundaryconditionwilson1})
         $A_{y}$ is {\em{even}} under $P$.
         Since $\{ P,T^{\hat{a}} \}=0$, $W$ does not commute with $P$
         \begin{equation}
          \label{rankreduction}
          [P,W] \neq 0 \; .
         \end{equation}
         In this case the VEV for $A_{y}$ can be an arbitrary constant. Therefore, we
         call the Wilson line constructed from $A_{y}=\sum_{\hat{a}} A^{\hat{a}}_{y}T^{\hat{a}}$
         a {\em{continuous Wilson line}}. Due to (\ref{rankreduction}),
         a continuous Wilson line induces a spontaneous $rank$ reducing gauge symmetry
         breaking, i.e. $rank \; \mathfrak{h} < rank \; \mathfrak{g}$.
   \item The remaining generators of $G$, which are even under one effective projection and odd under
         the other, can never give rise to a consistent Wilson line $W$.
  \end{enumerate}
  Remarks: i) \;
  Following the line of thinking of section \ref{sectionorbifolds1z2z2prime},
  the orbifold $S^{1}/\mathbb{Z}_{2}$ with continuous Wilson line breaking 
  is equivalent to the orbifold 
  $S^{1}/\mathbb{Z}_{2} \times \mathbb{Z}^{\prime}_{2}$ if we allow the
  orbifold projection $P^{\prime}$ (\ref{twistrelations}) to depend on a
  continuous parameter. In this case rank reduction is also possible on the 
  orbifold $S^{1}/\mathbb{Z}_{2} \times \mathbb{Z}^{\prime}_{2}$.

  \begin{example}
   Let $G=SU(3)$ be the bulk gauge group and let $T^{A}=\lambda^{A}$ be the Gell-Mann matrices
   generating $SU(3)$. We break $G=SU(3)$ down to
   $H_{1}=SU(2) \times U(1)$ at the orbifold fixed point $y_{1}=0$ by choosing
   \begin{equation}
    \label{su3twist}
    P_{1}=\exp(\pi i \lambda_{3})=
         \left(\begin{array}{ccc} -1 & 0 & 0 \\ 0 & -1 & 0 \\ 0 & 0 & 1 
         \end{array} \right)\; ,
   \end{equation}
   where $H_{1}$ is generated by $\{T^{a}\}, \; a=1,2,3,8$ and the coset
   $G/H_{1}$ is generated by $\{T^{\hat{a}}\}, \; \hat{a}=4,5,6,7$. Note that $P_{1} \in G$. 

   Let us first consider a discrete Wilson line, for example
   \begin{equation}
    \label{wilsonlineexmaple}
    W=\exp(\pi i (\lambda_{3}+\sqrt{3}\lambda_{8})/2 )
     =\left( \begin{array}{ccc} -1 & 0 & 0 \\ 0 & 1 & 0 \\ 0 & 0 & -1 
      \end{array} \right) \; .
   \end{equation} 
   This Wilson line leads to the breaking $H_{1} \to H=U(1) \times U(1)$ 
   generated by $T^{3},T^{8}$. Alternatively, we can directly assign the 
   projection $P_{2}$ to the fixed point $y_{2}=\pi R$
   \begin{equation}
    P_{2}=WP_{1}=\exp(\pi i (2\lambda_{3}+\sqrt{3}\lambda_{8})/2) 
                \left( \begin{array}{ccc} 1 & 0 & 0 \\ 0 & -1 & 0 
                \\ 0 & 0 & -1 \end{array} \right) \; .
   \end{equation}
   Thus $G$ is broken at the orbifold fixed point $y_{2}=\pi R$ down to $H_{2}=SU(2) \times U(1)$
   generated by $\{T^{a^{\prime}}\}, \; a=3,6,7,8$. In fact, $H=H_{1} \cap H_{2}$ is generated by
   $T^{3},T^{8}$.
   For $T^{a} \in \{ T^{3},T^{8} \}$ the following relations hold 
   \begin{equation}
    [P_{1},T^{a}]=[WP_{1},T^{a}]=[W,T^{a}]=0 \; .
   \end{equation}
   The projection matrices fulfil 
   \begin{equation}
    P^{2}_{1}=P^{2}_{2}=1 \; .
   \end{equation}
   In particular we have
   \begin{equation}
    W^{2}=1 \; .
   \end{equation}
   Let us construct the Wilson line $W$ explicitly. Since $A_{y}=\sum_{a} A^{a}_{y} T^{a}=A^{3}_{y}T^{3}+
   A^{8}_{y}T^{8}$ we have
   \begin{equation}
    \label{discretewilsonlineexample}
    W=\exp(2\pi i g R \left< A^{3}_{y} \right> T^{3}+ 2\pi i g R 
     \left< A^{8}_{y} \right> T^{8}) \; .
   \end{equation}
   We can built four different discrete Wilson lines
   \begin{equation}
    \left< A^{3}_{y} \right>=0 \; , \; 
    \left< A^{8}_{y} \right>=0 \; \to \;
     W=\left( \begin{array}{ccc} 1 & 0 & 0 \\ 0 & 1 & 0 \\ 0 & 0 & 1 
     \end{array} \right) \; ,
   \end{equation}
   \begin{equation}
    \left< A^{3}_{y} \right>=\frac{1}{2gR} \; , \;
    \left< A^{8}_{y} \right>=0 \; \to \; 
    W=\left( \begin{array}{ccc} -1 & 0 & 0 \\ 0 & -1 & 0 \\ 0 & 0 & 1 
     \end{array} \right) \; ,
   \end{equation}
   \begin{equation}
    \label{discretewilsonline1}
    \left< A^{3}_{y} \right>=-\frac{1}{2gR} \; , \;
    \left< A^{8}_{y} \right>=\frac{\sqrt{3}}{4gR} \; \to \; 
    W=\left( \begin{array}{ccc} 1 & 0 & 0 \\ 0 & -1 & 0 \\ 0 & 0 & -1 
     \end{array} \right) \; ,
   \end{equation}
   \begin{equation}
    \label{discretewilsonline2}
    \left< A^{3}_{y} \right>=\frac{1}{4gR} \; , \; 
    \left< A^{8}_{y} \right>=\frac{\sqrt{3}}{4gR} \; \to \;
    W=\left( \begin{array}{ccc} -1 & 0 & 0 \\ 0 & 1 & 0 \\ 0 & 0 & -1 
     \end{array} \right)
    \footnote{We stress that is possible to rewrite (\ref{discretewilsonlineexample}) is terms of
    equivalent sets of generators of $SU(3)$ such that one VEV equals $\frac{1}{2gR}$ while the other
    equals $0$, i.e.
    \begin{itemize}
     \item
      (\ref{discretewilsonline1}) can be rewritten as 
      \begin{equation*}
       W=\exp(2\pi i g R 
       \left< A^{\eta}_{y} \right> \eta + 2\pi i g R 
       \left< A^{\eta^{\prime}}_{y} \right> \eta^{\prime})=\exp( \pi i \; \eta)
      \end{equation*}
      where $\eta=\frac{1}{2} \left(\sqrt{3} \lambda_{8}
      - \lambda_{3} \right)$, $\eta^{\prime}=\frac{1}{2} \left(\lambda_{8}
      - \sqrt{3} \lambda_{3} \right)$ and $\left< A^{\eta}_{y} \right>=\frac{1}{2gR}$ and
      $\left< A^{\eta^{\prime}}_{y} \right>=0$.
     \item (\ref{discretewilsonline2}) can be rewritten as 
      \begin{equation*}
       W=\exp(2\pi i g R 
       \left< A^{\rho}_{y} \right> \rho + 2\pi i g R 
       \left< A^{\rho^{\prime}}_{y} \right> \rho^{\prime})=\exp( \pi i \; \rho)
      \end{equation*}
      where $\rho=\frac{1}{2} \left( \lambda_{3}
      + \sqrt{3} \lambda_{8} \right)$, $\rho^{\prime}=\frac{1}{2} \left(\sqrt{3} \lambda_{3} + \lambda_{8}
      \right)$ and $\left< A^{\rho}_{y} \right>=\frac{1}{2gR}$ and
      $\left< A^{\rho^{\prime}}_{y} \right>=0$.
     \end{itemize}
     Note that with the definitions for $\eta$, $\eta^{\prime}$, $\rho$ and $\rho^{\prime}$ above the
     three equivalent sets of generators of $SU(3)$ read: $\{ \lambda_{1},\lambda_{2},\lambda_{3}
     ,\lambda_{4},\lambda_{5},\lambda_{5},\lambda_{7},\lambda_{8}\}$, $\{\lambda_{1},\lambda_{2},
     \rho,\lambda_{4},\lambda_{5},\lambda_{5},\lambda_{7},\rho^{\prime}\}$ and $\{\lambda_{1},
     \lambda_{2},\lambda_{4},\lambda_{5},\eta,\lambda_{6},\lambda_{7},\eta^{\prime}\}$} \; . 
   \end{equation}
   We see that the Wilson line (\ref{wilsonlineexmaple}) is obtained from the
   choice $\left< A^{3}_{5} \right>=\frac{1}{4gR}$ and $\left< A^{8}_{5} 
   \right>=\frac{\sqrt{3}}{4gR}$ in (\ref{discretewilsonlineexample}). Note that the 
   orbifold projection $P_{1}$ (\ref{su3twist}) commutes with every discrete Wilson line $W$. 
   This shows explicitly that discrete Wilson line breaking is rank preserving.

   Next we choose a continuous Wilson line, e.g.
   \begin{equation}
    W(\alpha)=\exp(2 \pi i \; \alpha \lambda_{7})
     =\left(\begin{array}{ccc} 1 & 0 & 0 \\ 0 & \cos\; 2 \pi \alpha & \sin\; 2 \pi \alpha \\
            0 & -\sin\; 2 \pi \alpha & \cos\; 2 \pi \alpha 
     \end{array} \right) \; .
   \end{equation} 
   Let the parameter $\alpha$ be limited to $0 < \alpha < 1$ but otherwise 
   arbitrary. This Wilson line leads to the breaking $H_{1} \to H=U(1)$.
   Alternatively, we again directly assign the 
   projection $P_{2}$ to the fixed point $y_{2}=\pi R$
   \begin{equation}
    P_{2}(\alpha)=W(\alpha)P_{1}=\left(\begin{array}{ccc} -1 & 0 & 0 \\ 
                 0 & -\cos\; 2 \pi \alpha & \sin\; 2 \pi \alpha \\
                 0 & \sin\; 2 \pi \alpha & \cos\; 2 \pi \alpha 
                 \end{array} \right) \; .
   \end{equation}
   $P_{2}$ now depends on $\alpha$.
   For $T^{\hat{a}} \in \{T^{\hat{7}}\}$ the following relations hold 
   \begin{equation}
    \{P_{1},T^{\hat{a}}\}=\{WP_{1},T^{\hat{a}}\}=[W,T^{\hat{a}}]=0 \; .
   \end{equation}
   Therefore $A_{y}=\sum_{\hat{a}} A^{\hat{a}}_{y} T^{\hat{a}}
   =A^{\hat{7}}_{y} T^{\hat{7}}$. 
   The projection matrices $P_{1}$ and $P_{2}$ fulfil 
   \begin{equation}
    P^{2}_{1}=P^{2}_{2}=1 \; .
   \end{equation}
   But now we obviously have
   \begin{equation}
    W^{2} \neq 1 \; .
   \end{equation}  
   Let us again explicitly construct the Wilson line $W$. Since
   $A_{y}=A^{\hat{7}}_{y} T^{\hat{7}}$ we have
   \begin{equation}
    W=\exp(2\pi i g R < A^{\hat{7}}_{y} > T^{\hat{7}})
     =\exp(2 \pi i \;  \alpha \lambda_{7}) \; ,
   \end{equation}
   where
   \begin{equation}
    \label{higgsvev}
    < A^{\hat{7}}_{y} >=\frac{\alpha}{g R} \; .
   \end{equation}
   For $0 < \alpha \ll 1$ the VEV for $A_{y}$ (\ref{higgsvev}) can be much smaller
   than the compactification scale $1/R$.
  \end{example}

 \section{Fourier expansion and zero modes on $S^{1}/\mathbb{Z}_{2}$}

  In this section, we discuss the Fourier mode expansions 
  of $A_{\mu}(x^{\mu},y)$ and $A_{y}(x^{\mu},y)$ on the orbifold $S^{1}/\mathbb{Z}_{2}$.
  Recall that the gauge fields have to fulfil the boundary conditions
  \begin{gather}
    \label{boundaryconditionsfromtwistp}
    A_{\mu}(x^{\mu},-y)=P \; A_{\mu}(x^{\mu},y) \; P^{-1} \\
    \label{boundaryconditionsfromtwistpsc}
    A_{y}(x^{\mu},-y)=-P \; A_{y}(x^{\mu},y) \; P^{-1}  \; . 
  \end{gather} 
  and the periodicity condition 
  \begin{equation}
   \label{periodictyforsu7}
   A_{M}(x^{\mu},y+2 \pi R)=W \; A_{M}(x^{\mu},y) \; W^{-1} \; ,
  \end{equation}
  where $W$ is the corresponding Wilson line. In general, three cases can arise
  \begin{enumerate}
   \item
    $W=1$. This means that we admit the trivial periodicity condition, i.e.
    \begin{equation}
     \label{trivialperidocity}
     A_{M}(x^{\mu},y+2 \pi R)=A_{M}(x^{\mu},y) \; .
    \end{equation}  
    The boundary condition (\ref{boundaryconditionsfromtwistp}) 
    breaks the bulk gauge group $G$ down to its subgroup
    $H_{y_{1}}$ 
    \begin{equation}
     \label{gaugegroup4dlowenergy}
     H_{y_{1}}=\{ g \in G \mid Pg=gP \}
    \end{equation}
    at the fixed point $y_{1}=0$. Let $\{T^{a}\}$ denote the set of generators
    creating $H_{y_{1}}$ and let $\{T^{\hat{a}}\}$ denote the set of generators creating the coset $G/H$.
    In following we call $T^{a}$ the unbroken generators and $T^{\hat{a}}$
    the broken generators, respectively.

    According to (\ref{boundaryconditionsfromtwistp}) and (\ref{boundaryconditionsfromtwistpsc}) 
    unbroken gauge $A^{a}_{\mu}(x^{\mu},y)$ and the scalar fields $A^{\hat{a}}_{y}(x^{\mu},y)$
    \footnote{Note that from a four-dimensional point of view 
    the gauge fields $A_{y}(x^{\mu},y)$ are seen as scalar fields. Therefore we will also call 
    $A_{y}(x^{\mu},y)$ scalar fields.}     
    are even functions, i.e.
    \begin{gather}
     A^{a}_{\mu}(x^{\mu},-y)=A^{a}_{\mu}(x^{\mu},y) \\
     A^{\hat{a}}_{y}(x^{\mu},-y)=A^{\hat{a}}_{y}(x^{\mu},y) \nonumber \; .
    \end{gather}
    Thus we can Fourier expand
    \begin{eqnarray}
     \label{fourierexpandunbrokengaugeandscalarfields}
     && A^{a}_{\mu}(x^{\mu},y)=\frac{1}{\sqrt{2 \pi R}}
     A_{\mu}^{a}{}^{(0)}(x^{\mu})+\frac{1}{\sqrt{\pi R}}\sum_{n=1}^{\infty} 
     A_{\mu}^{a}{}^{(n)}(x^{\mu}) \cos(\frac{n y}{R}) \\
     && A^{\hat{a}}_{y}(x^{\mu},y)=\frac{1}{\sqrt{2 \pi R}}
     A_{y}^{\hat{a}}{}^{(0)}(x^{\mu})+ \frac{1}{\sqrt{\pi R}}
     \sum_{n=1}^{\infty} 
     A_{y}^{\hat{a}}{}^{(n)}(x^{\mu}) \cos(\frac{n y}{R})  \nonumber \; .
    \end{eqnarray} 
    Since $\cos(\frac{n y}{R})$ is $2 \pi R$-periodic,
    $A^{a}_{\mu}(x^{\mu},y)$ and $A^{\hat{a}}_{y}(x^{\mu},y)$ fulfil also
    the periodicity condition (\ref{trivialperidocity}). 
    Note that for the scalar fields $A_{y}(x^{\mu},y)$ 
    the situation is opposite ($a$ and $\hat{a}$ are interchanged)
    due to the relative minus sign in the boundary conditions 
    (\ref{boundaryconditionsfromtwistp}) and (\ref{boundaryconditionsfromtwistpsc}), respectively. 
    
    The Fourier coefficients 
    $A_{\mu}^{a}{}^{(n)}(x^{\mu})$ and $A_{y}^{\hat{a}}{}^{(n)}(x^{\mu})$ are
    given by 
    \begin{gather}
     \label{nmodesintegration}
     A_{\mu}^{a}{}^{(n)}(x^{\mu})=\frac{1}{\sqrt{\pi R}} \int_{0}^{\pi R} 
     A^{a}_{\mu}(x^{\mu},y) \cos(\frac{n y}{R}) \;  dy  \\
     A_{y}^{\hat{a}}{}^{(n)}(x^{\mu})=\frac{1}{\sqrt{\pi R}} \int_{0}^{\pi R} 
     A^{\hat{a}}_{y}(x^{\mu},y) \cos(\frac{n y}{R}) \;  dy \nonumber \; .
    \end{gather}
    The zero modes read
    \begin{gather}
     A_{\mu}^{a}{}^{(0)}(x^{\mu})=\frac{1}{\sqrt{2 \pi R}} \int_{0}^{\pi R} 
     A^{a}_{\mu}(x^{\mu},y) \;  dy  \\
     A_{y}^{\hat{a}}{}^{(0)}(x^{\mu})=\frac{1}{\sqrt{2 \pi R}} 
     \int_{0}^{\pi R} A^{\hat{a}}_{y}(x^{\mu},y) \;  dy \nonumber \; .
    \end{gather}

    On the other hand, according to (\ref{boundaryconditionsfromtwistp}) and 
    (\ref{boundaryconditionsfromtwistpsc}), broken gauge $A^{\hat{a}}_{\mu}(x^{\mu},y)$ and
    the scalar fields $A^{a}_{y}(x^{\mu},y)$ are odd functions, i.e.
    \begin{gather}
     A^{\hat{a}}_{\mu}(x^{\mu},-y)=-A^{\hat{a}}_{\mu}(x^{\mu},y) \\
     A^{a}_{y}(x^{\mu},-y)=-A^{a}_{y}(x^{\mu},y) \nonumber \; .
    \end{gather}
    Thus we can Fourier expand 
    \begin{eqnarray}
     \label{fourierexpandbrokengaugeandscalarfields}
     && A^{\hat{a}}_{\mu}(x^{\mu},y)=\frac{1}{\sqrt{\pi R}}\sum_{n=1}^{\infty} 
     A_{\mu}^{\hat{a}}{}^{(n)}(x^{\mu}) \sin(\frac{n y}{R}) \\
     && A^{a}_{y}(x^{\mu},y)=\frac{1}{\sqrt{\pi R}}\sum_{n=1}^{\infty} 
     A_{y}^{a}{}^{(n)}(x^{\mu}) \sin(\frac{n y}{R})  \nonumber \; .
    \end{eqnarray} 
    Again since $\sin(\frac{n y}{R})$ is $2 \pi R$-periodic,
    $A^{\hat{a}}_{\mu}(x^{\mu},y)$ and $A^{a}_{y}(x^{\mu},y)$ fulfil also
    the periodicity condition (\ref{trivialperidocity}).

    The Fourier coefficients $A_{\mu}^{\hat{a}}{}^{(n)}(x^{\mu})$ and 
    $A_{y}^{a}{}^{(n)}(x^{\mu})$ are given by 
    \begin{gather}
     A_{\mu}^{\hat{a}}{}^{(n)}(x^{\mu})=\frac{1}{\sqrt{\pi R}} 
     \int_{0}^{\pi R} 
     A^{\hat{a}}_{\mu}(x^{\mu},y) \sin(\frac{n y}{R}) \;  dy  \\
     A_{y}^{a}{}^{(n)}(x^{\mu})=\frac{1}{\sqrt{\pi R}} \int_{0}^{\pi R} 
     A^{a}_{y}(x^{\mu},y) \sin(\frac{n y}{R}) \;  dy \nonumber \; .
    \end{gather}
    Note that in contrast to (\ref{fourierexpandunbrokengaugeandscalarfields}) in 
    (\ref{fourierexpandbrokengaugeandscalarfields})
    no zero modes occurs. Since only (\ref{fourierexpandunbrokengaugeandscalarfields}) contains zero modes
    the gauge group of the low energy four-dimensional effective theory is given by 
    (\ref{gaugegroup4dlowenergy}).

   \item
    $W$ is a discrete Wilson line. This means that $W^{2}=1$. We know that
    a discrete Wilson line $W$ commutes with the orbifold projection $P$
    \begin{equation}
     [P,W]=0  \; .
    \end{equation}
    Therefore $W$ and $P$ have a common set of 
    eigenfunctions. In order to find their eigenfunctions
    we first look at the periodicity condition 
    \begin{equation}
     A_{M}(x^{\mu},y+2 \pi R)=W \; A_{M}(x^{\mu},y) \; W^{-1} \; .
    \end{equation}
    Due to $W^{2}=1$ the five-dimensional gauge field $A_{M}(x^{\mu},y)$ splits into 
    an even part
    \begin{equation}
     A^{a}_{M}(x^{\mu},y+2 \pi R)T^{a}=\; + \; A^{a}_{M}(x^{\mu},y)T^{a} \; ,
    \end{equation}
    where the $\{T^{a}\}$ satisfy $[W,T^{a}]=0$, and an odd part
    \begin{equation}
     A^{\hat{a}}_{M}(x^{\mu},y+2 \pi R)T^{\hat{a}}
     =\; - \; A^{\hat{a}}_{M}(x^{\mu},y)T^{\hat{a}} 
    \end{equation}
    where the $\{T^{\hat{a}}\}$ satisfy $\{W,T^{\hat{a}}\}=0$. Taking further into
    account that the orbifold projection $P$ acts at $y_{1}=0$ according to
    \begin{gather}
     \label{boundaryconditionsamumy}
     A_{\mu}(x^{\mu},-y)=P \; A_{\mu}(x^{\mu},y) \; P^{-1} \\
     A_{y}(x^{\mu},-y)=-P \; A_{y}(x^{\mu},y) \; P^{-1}  \nonumber \; , 
    \end{gather}
    we can Fourier expand \cite{Haba:2004qf}
    \begin{eqnarray}
     && A_{\mu}^{(+,+)}(x^{\mu},y)=\frac{1}{\sqrt{2 \pi R}} 
     A_{\mu}^{(+,+)(0)}(x^{\mu})+\frac{1}{\sqrt{\pi R}} \sum_{n=1}^{\infty} 
     A_{\mu}^{(+,+)(n)}(x^{\mu}) \cos(\frac{n y}{R}) \; , \nonumber \\
     && A_{\mu}^{(+,-)}(x^{\mu},y)=\frac{1}{\sqrt{\pi R}}\sum_{n=0}^{\infty} 
     A_{\mu}^{(+,-)(n)}(x^{\mu}) \sin(\frac{n y}{R}) \; , \nonumber \\
     && A_{\mu}^{(-,+)}(x^{\mu},y)=\frac{1}{\sqrt{\pi R}}\sum_{n=0}^{\infty} 
     A_{\mu}^{(-,+)(n)}(x^{\mu}) \cos(\frac{(n+1/2) y}{R}) \; , \nonumber \\
     && A_{\mu}^{(-,-)}(x^{\mu},y)=\frac{1}{\sqrt{\pi R}}\sum_{n=0}^{\infty} 
     A_{\mu}^{(-,-)(n)}(x^{\mu}) \sin(\frac{(n+1/2) y}{R})  \; .
     \label{fourierexpannsionwp}
    \end{eqnarray}
    The superscript $(\pm,\pm)$ denotes the eigenvalue of $W$ and $P$, respectively.    
    This means that $A_{\mu}^{(+,+)}(x^{\mu},y)=A_{\mu}^{++}(x^{\mu},y)T^{++}$ 
    commutes with $W$ and $P$
    \begin{equation}
     [W,T^{++}]=[P,T^{++}]=0 \; ,
    \end{equation} 
    $A_{\mu}^{(+,-)}(x^{\mu},y)=A_{\mu}^{+-}(x^{\mu},y)T^{+-}$ 
    commutes with $W$ and anticommutes with $P$
    \begin{equation}
     [W,T^{+-}]=\{P,T^{+-}\}=0 \; .
    \end{equation}  
    $A_{\mu}^{(-,+)}(x^{\mu},y)=A_{\mu}^{-+}(x^{\mu},y)T^{-+}$ 
    anticommutes with $W$ and commutes with $P$
    \begin{equation}
     \{W,T^{-+}\}=[P,T^{-+}]=0  \; ,
    \end{equation} 
    and $A_{\mu}^{(-,-)}(x^{\mu},y)=A_{\mu}^{--}(x^{\mu},y)T^{--}$ 
    anticommutes with $W$ and $P$
    \begin{equation}
     \{W,T^{--}\}=\{P,T^{--}\}=0 \; .
    \end{equation}
    The expansion for $A_{y}$ is done in the same way but, due to (\ref{boundaryconditionsamumy}),
    with the opposite eigenvalue for $W$ and $P$.

    The Fourier coefficient for $A_{\mu}^{++(n)}(x^{\mu})$ is given by
    \begin{equation}
     A_{\mu}^{++(n)}(x^{\mu})=\frac{1}{\sqrt{\pi R}} \int_{0}^{\pi R} 
     A^{++}_{\mu}(x^{\mu},y) \cos(\frac{n y}{R}) \;  dy \; .
    \end{equation}
    All other Fourier coefficients can be obtained in an analogous manner. 
    The zero modes are contained in $A_{\mu}^{(+,+)(0)}(x^{\mu},y)$ and read
    \begin{equation}
     \label{4dunbrokengaugefield}
     A_{\mu}^{(+,+)(0)}(x^{\mu})= A_{\mu}^{++(0)}(x^{\mu})T^{++} \; .
    \end{equation} 
    The Fourier coefficient $A_{\mu}^{++(0)}(x^{\mu})$ is given by
    \begin{equation}
     A_{\mu}^{++(0)}(x^{\mu})=\frac{1}{\sqrt{2\pi R}} \int_{0}^{\pi R} 
     A_{\mu}^{++}(x^{\mu},y) \;  dy \; . 
    \end{equation}
    The low energy four-dimensional unbroken gauge group $H$ is created by the generators
    $\{T^{++}\}$. Note that the generators $\{T^{++}\}$ commute with $W$ and $P$, i.e.
    \begin{equation}
     [W,T^{++}]=[P,T^{++}]=0 \; .
    \end{equation}      
    
    If we switch to the reinterpretation of $S^{1}/\mathbb{Z}_{2}$ in terms
    of the two effective orbifold projections  
    $P_{1}=P$ and $P_{2}=WP$ \footnote{Note that due to $[P,W]=0$ we have $[P,WP]=0$ and thus also
    $P_{1}=P$ and $P_{2}=WP$ have a common set of eigenfunctions}
    the orbifold boundary conditions
    read (\ref{boundaryconditionss1z2twisted}) 
    \begin{gather}
     A_{\mu}(x^{\mu},-y)=P_{1} \; A_{\mu}(x^{\mu},y) \; P_{1}^{-1}  \\
     A_{y}(x^{\mu},-y)=-P_{1} \; A_{y}(x^{\mu},y) \; P_{1}^{-1}  \; , \nonumber \\
     \nonumber  \\
     A_{\mu}(x^{\mu},\pi R-y)=P_{2} \; A_{\mu}(x^{\mu},\pi R+y) \;
                              P_{2}^{-1} \nonumber \\
     A_{y}(x^{\mu},\pi R-y)=-P_{2} \; A_{y}(x^{\mu},\pi R+y) \; 
                             P_{2}^{-1} \nonumber \; .
    \end{gather}  
    These boundary conditions lead to the Fourier expansion
    \begin{eqnarray}
     && A_{\mu}^{(+,+)}(x^{\mu},y)=\frac{1}{\sqrt{2 \pi R}} 
     A_{\mu}^{(+,+)(0)}(x^{\mu})+\frac{1}{\sqrt{\pi R}} \sum_{n=1}^{\infty} 
     A_{\mu}^{(+,+)(n)}(x^{\mu}) \cos(\frac{n y}{R}) \; , \nonumber \\
     && A_{\mu}^{(+,-)}(x^{\mu},y)=\frac{1}{\sqrt{\pi R}}\sum_{n=0}^{\infty} 
     A_{\mu}^{(+,-)(n)}(x^{\mu}) \cos(\frac{(n+1/2) y}{R}) \; , \nonumber \\
     && A_{\mu}^{(-,+)}(x^{\mu},y)=\frac{1}{\sqrt{\pi R}}\sum_{n=0}^{\infty} 
     A_{\mu}^{(-,+)(n)}(x^{\mu}) \sin(\frac{(n+1/2) y}{R}) \; , \nonumber \\
     && A_{\mu}^{(-,-)}(x^{\mu},y)=\frac{1}{\sqrt{\pi R}}\sum_{n=0}^{\infty} 
     A_{\mu}^{(-,-)(n)}(x^{\mu}) \sin(\frac{(n+1) y}{R})  \; .
     \label{fouriermodeexpansionz2twistedbc}
    \end{eqnarray}   
    The subscript $(\pm,\pm)$ denote the eigenvalues of $(P_{1},P_{2})$. 
    The connection to (\ref{fourierexpannsionwp}) becomes apparent if
    we looks at the different eigenvalues
    \begin{equation}
     \begin{array}{|c|c|c|c|}  \hline
       W & P & P_{1} & P_{2} \\ \hline
       + & + & + & + \\ \hline
       + & - & - & - \\ \hline
       - & + & + & - \\ \hline
       - & - & - & + \\ \hline
     \end{array} 
    \end{equation}
    Recall again that $P_{1}=P$ and $P_{2}=W P$.
    The gauge group $G$ is broken to its subgroup $H_{1}$ generated by
    $\{T^{++},T^{+-}\}$ at $y_{1}=0$ and to its subgroup $H_{2}$ generated by
    $\{T^{++},T^{-+}\}$  at $y_{2}=\pi R$. 
    The  low energy four-dimensional unbroken gauge group $H$ is thus generated by $\{T^{++}\}$ and we have
    \begin{equation}
     H=H_{1} \cap H_{2} \; .
    \end{equation}

  \subsection{Fourier expansion and zero modes on 
              $S^{1}/\mathbb{Z}_{2} \times \mathbb{Z}^{\prime}_{2}$} 

   Since the orbifold $S^{1}/\mathbb{Z}_{2} \times \mathbb{Z}^{\prime}_{2}$ 
   is familiar in orbifold GUTs, we shortly discuss their Fourier 
   mode expansion. We recall that the gauge fields have to fulfil the boundary conditions 
   (\ref{boundaryconditionss1z2z2prime})
   \begin{gather}
    A_{\mu}(x^{\mu},-y)=P \; A_{\mu}(x^{\mu},y) \; P^{-1}  \\
    A_{y}(x^{\mu},-y)=-P \; A_{y}(x^{\mu},y) \; P^{-1} \; , \nonumber \\
    \nonumber  \\
    A_{\mu}(x^{\mu},-y^{\prime})=P^{\prime} \; A_{\mu}(x^{\mu},y^{\prime}) \;
                             P^{\prime \; -1} \nonumber \\
    A_{y}(x^{\mu},-y^{\prime})=-P^{\prime} \; A_{y}(x^{\mu},y^{\prime}) \;
                            P^{\prime \; -1} \nonumber \; .
   \end{gather}   
   Remember that $y^{\prime}=y-\pi R^{\prime}/2$ and 
   $\mathbb{Z}_{2}: \;y \to -y$, $\mathbb{Z}^{\prime}_{2}: \; y^{\prime} 
   \to -y^{\prime}$. 
   The orbifold $S^{1}/\mathbb{Z}_{2} \times \mathbb{Z}^{\prime}_{2}$
   has the two fixed points $y=0$ and $y=\pi R^{\prime}/2$. Fourier expanding yields
   \begin{eqnarray}
     && A_{\mu}^{(+,+)}(x^{\mu},y)=\frac{1}{\sqrt{\pi R^{\prime}}} 
     A_{\mu}^{(+,+)(0)}(x^{\mu})+\frac{1}{\sqrt{\pi R^{\prime}/2}} \sum_{n=1}^{\infty} 
     A_{\mu}^{(+,+)(n)}(x^{\mu})  \cos(\frac{2 n y}{R^{\prime}}) \; , \nonumber \\
     && A_{\mu}^{(+,-)}(x^{\mu},y)=\frac{1}{\sqrt{\pi R^{\prime}/2}}\sum_{n=0}^{\infty} 
     A_{\mu}^{(+,-)(n)}(x^{\mu}) \cos(\frac{(2n+1) y}{R^{\prime}}) \; , \nonumber \\
     && A_{\mu}^{(-,+)}(x^{\mu},y)=\frac{1}{\sqrt{\pi R^{\prime}/2}}\sum_{n=0}^{\infty} 
     A_{\mu}^{(-,+)(n)}(x^{\mu}) \sin(\frac{(2n+1) y}{R^{\prime}}) \; , \nonumber \\
     && A_{\mu}^{(-,-)}(x^{\mu},y)=\frac{1}{\sqrt{\pi R^{\prime}/2}}\sum_{n=0}^{\infty} 
     A_{\mu}^{(-,-)(n)}(x^{\mu}) \sin(\frac{(2n+2) y}{R^{\prime}}) \; .
     \label{fouriermodeexpansionz2z2prime}
   \end{eqnarray}   
   The expansion for $A_{y}$ is again done in the same way but with the opposite eigenvalue for $W$ and $P$.
   We already know that the orbifold
   $S^{1}/\mathbb{Z}_{2} \times \mathbb{Z}^{\prime}_{2}$ is equivalent to 
   the orbifold $S^{1}/\mathbb{Z}_{2}$ with twisted boundary conditions.
   To show this equivalence also for the Fourier mode expansion we have to keep in mind that 
   \begin{equation}
    \label{relationradii}
    R^{\prime}=2 R \; .
   \end{equation} 
   Indeed, if we insert (\ref{relationradii}) in (\ref{fouriermodeexpansionz2z2prime}) we
   recover (\ref{fouriermodeexpansionz2twistedbc}). The Fourier mode 
   expansion (\ref{fouriermodeexpansionz2z2prime}) well known in
   orbifold GUTs \cite{Hebecker:2001wq,Hall:2001pg}.

  \item
   $W$ is a continuous Wilson line. In this case the Wilson line
   $W$ and the projection $P$ do not commute
   \begin{equation}
    [P,W] \neq 0 \; .
   \end{equation}
   Therefore $P$ and $W$ do not have a common set of eigenfunctions. 
  
  \end{enumerate}

  \section[The Hosotani mechanism on the orbifold $S^{1}/\mathbb{Z}_{2}$]{Continuous Wilson line
           breaking and the Hosotani mechanism on the orbifold $S^{1}/\mathbb{Z}_{2}$}

   \label{sectionhosotanimechanism}

   In section \ref{sectioncontinuousdiscretewislonlineshosotani} we have seen that 
   a continuous Wilson line is given by
   \begin{equation}
    \label{hosotanicontwilsonline}
    W=\exp \left(2 \pi i g R  \sum_{\hat{a}} \langle A^{\hat{a}}_{y} \rangle T^{\hat{a}} \right) 
   \end{equation}
   where $T^{\hat{a}} \in H_{W}$ and
   \begin{equation}
    \label{hwhosotani}
    H_{W}=\{T^{\hat{a}} \in G \mid \{T^{\hat{a}},P_{1}\}=\{T^{\hat{a}},P_{2}\}=0 \} \; .
   \end{equation}
   We consider $x$- and $y$-independent modes of  
   \begin{equation}
    A_{y}=\sum_{\hat{a}} A_{y}^{\hat{a}} T^{\hat{a}} \quad , \; T^{\hat{a}} \in H_{W} \; .
   \end{equation}
   They correspond to Wilson line phases via \cite{Haba:2002py}
   \begin{equation}
    \theta_{\hat{a}}:=g \pi R A^{\hat{a}}_{y}  \; .
   \end{equation}
   Wilson line phases are part of the Hosotani mechanism 
   \cite{Hosotani:1988bm}. The Hosotani mechanism is used in a series of papers
   \cite{Haba:2002py,Hosotani:2003ay,Hosotani:2004wv,Hosotani:2005fk,Scrucca:2003ra,Kubo:2001zc}.
   Here we describe the main ingredients \cite{Haba:2002py,Hosotani:2003ay} of the Hosotani mechanism:
   \begin{itemize}
    \item
     Wilson line phases $\theta_{\hat{a}}$ along noncontractible
     loops become physical degrees of freedom which cannot be gauged away
     once boundary conditions are given. They yield vanishing field strengths
     such that they appear as degenerate vacua at the classical level.
    \item
     The degeneracy of the classical vacuum is in general lifted by quantum
     effects. Let $V_{eff}=V_{eff}(\theta_{\hat{a}})$ be the effective potential
     for the Wilson line phases $\theta_{\hat{a}}$. Then the true physical vacuum
     is given by those configurations of the Wilson line phases $\theta_{\hat{a}}$
     which minimise $V_{eff}$.
    \item
     Suppose that the effective potential $V_{eff}$ is minimised at nontrivial
     configurations of the Wilson line phases. Then the gauge symmetry
     is spontaneously broken by radiative corrections. This
     part of the mechanism is called {\em{Wilson line symmetry breaking}}.
     Gauge fields in four dimensions whose gauge symmetry is spontaneously broken get masses from
     nonvanishing VEVs for the Wilson line phases. In addition, some matter 
     fields also acquire masses.
    \item
     All zero modes of the extra-dimensional component of the higher dimensional gauge field 
     become massive. Their masses are given by the second derivatives of 
     $V_{eff}$ up to numerical constants.
    \item
     The physical symmetry of the theory is determined by orbifold boundary conditions
     and the VEVs of the Wilson line phases. 
   \end{itemize}

\chapter{Effective Theories and nonunitary parallel transporters}
 
 \label{chaptereffectivetheorie}

  In this chapter we describe how an effective bilayered transverse lattice model can be
  obtained from an ordinary $S^{1}/\mathbb{Z}_{2}$ orbifold model via 
  renormalisation group (RG) transformations. We start with a five-dimensional space-time 
  $M^{4} \times S^{1}/\mathbb{Z}_{2}$,
  which is the product of the four-dimensional Minkowski space-time $M^{4}$ and the orbifold 
  $S^{1}/\mathbb{Z}_{2}$.
  Recall that the orbifold $S^{1}/\mathbb{Z}_{2}$ is obtained by dividing the circle $S^{1}$ with radius $R$
  by a $\mathbb{Z}_{2}$ transformation. The resulting space is the interval $\left[ 0, \pi R \right]$.
  Let $G$ be the bulk gauge group. In order to obtain a well-defined starting point for RG transformations
  we put the orbifold $S^{1}/\mathbb{Z}_{2}$ on a lattice. Thus the four-dimensional Minkowski space-time 
  $M^{4}$ remains continuous while the extra dimension is latticized. Such a scenario is known as a 
  transverse lattice and it occurs naturally in deconstruction models 
  \cite{Hill:2000mu,Cheng:2001vd,Cheng:2001nh}. 
  Starting with this latticized extra dimension we calculate the RG-flow. 
  The endpoint of the RG-flow will be an extra dimension which consists of only two points: the two orbifold
  fixed points $y=0$ and $y=\pi R$. The bulk is completely integrated out.
  The effective theory obtained this way will be called an {\em{effective bilayered 
  transverse lattice model}} (eBTLM). We call the four-dimensional boundary at the fixed point $y=0$ 
  the $L$-boundary and the four-dimensional boundary at the fixed point $y=\pi R$ the $R$-boundary.
  PTs $\Phi$ in the extra dimension from 
  the $L$- to the $R$-boundary (and vice versa) become nonunitary as a result of the blockspin 
  transformation. They 
  take their values in a Lie group $H$ which is typically noncompact and larger than the unitary gauge 
  group $G$ we have started with. We always consider the case where $G$ is the maximal compact subgroup
  of $H$. It will turn out that these nonunitary PTs $\Phi$ can be interpreted as Higgs fields. 
  In this chapter we will also formulate orbifold conditions for nonunitary PTs $\Phi$.
  As an application, we analyse in detail an eBTLM based on the gauge group $SU(2)$.
 
  .

 \section{$S^{1}/\mathbb{Z}_{2}$ orbifold model on a lattice}

  We consider a one-dimensional lattice $\Gamma$ with lattice spacing $a$. The points of $\Gamma$ have the 
  coordinate $y=a \; n_{y}$, where $n_{y}=-N_{y}+1,\dots,N_{y}$, $N_{y} \in \mathbb{N}_{\ast}$.
  If we identify the points $y=-a N_{y}$ and $y=a N_{y}$, $\Gamma$ will possesses the 
  translation invariance 
  \begin{equation}
   t: \; n_{y} \to n_{y}+2 N_{y} \quad \Longleftrightarrow \quad y \to y+2 \pi R   \; .
  \end{equation}
  Thus the physical extension of $\Gamma$ is $2 \pi R=2 N_{y} a$. We define the reflection $r$ on
  $\Gamma$ by
  \begin{equation}
   \label{latticereflection}
   r : \;  n_{y} \to -n_{y} \quad \Longleftrightarrow \quad y \to -y \; .
  \end{equation} 
  Figure \ref{latticerefelction} shows the representation of the orbifold reflection $r$ on $\Gamma$. 
  \begin{figure}[h]
   \vspace{1cm}
   \begin{equation*}
   \begin{picture} (12,10)
    \thinlines
    \put(6,5){\bigcircle{8}}
    \multiput(2,5)(8,0){2}{\circle*{0.2}}
    \matrixput(9.47,7)(-6.93,0){2}(0,-4){2}{\circle*{0.2}}
    \matrixput(8,8.47)(-4,0){2}(0,-6.93){2}{\circle*{0.2}}
    \multiput(2.54,5)(6.93,0){2}{\vector(0,1){1.9}}
    \multiput(2.54,5)(6.93,0){2}{\vector(0,-1){1.9}}
    \multiput(4,5)(4,0){2}{\vector(0,1){3,36}}
    \multiput(4,5)(4,0){2}{\vector(0,-1){3,36}}
    \multiput(5.8,9.8)(0,-9.6){2}{\dots}
    \multiputlist(6,10.5)(0,-11){$n_{y}$,$-n_{y}$}
    \multiputlist(1,5)(10.5,0){$n_{y}=0$,$n_{y}=N_{y}$}
    \multiputlist(1.7,7.5)(8.6,0){1,$N_{y}-1$}
    \multiputlist(1.7,2.5)(8.6,0){-1,$-N_{y}+1$}
    \multiputlist(3.5,9.3)(5,0){2,$N_{y}-2$}
    \multiputlist(3.5,0.7)(5,0){-2,$-N_{y}+2$}
   \end{picture}
   \end{equation*}
   \label{latticerefelction}
   \caption{Representation of the $S^{1}/\mathbb{Z}_{2}$ orbifold reflection $r$ on the lattice $\Gamma$.}
  \end{figure}
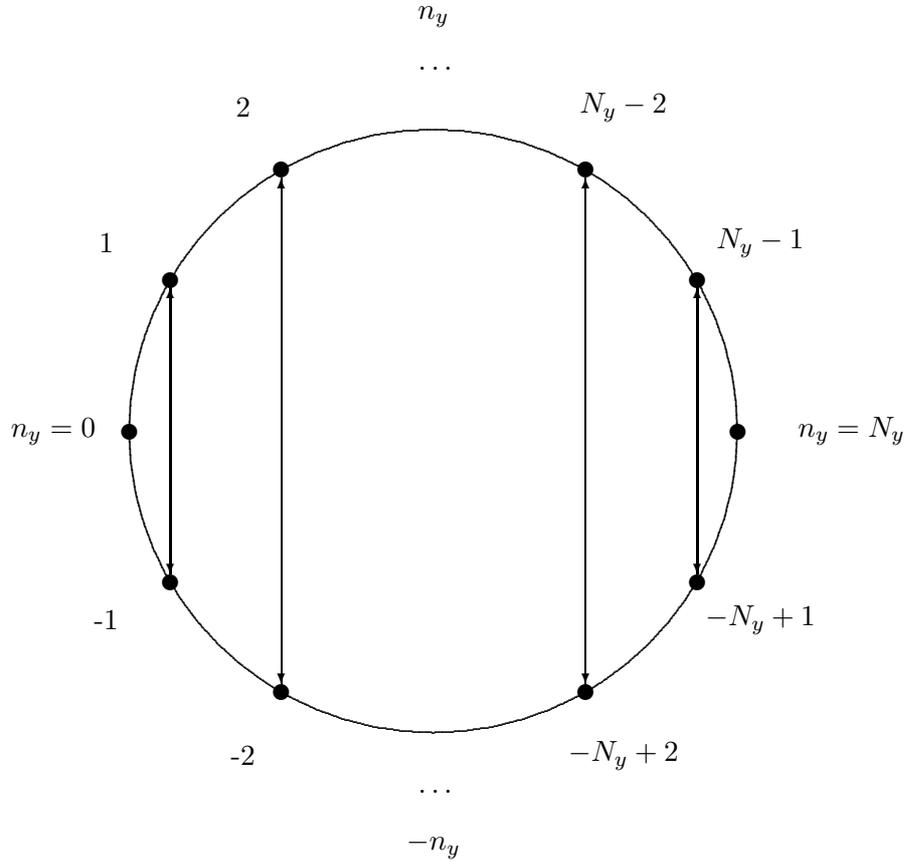
  The orbifold has two fixed points $n_{y}=0$, invariant under $r$, and $n_{y}=N_{y}$, invariant
  under $tr$. In terms of the coordinate $y$ they read $y=0$ and $y=\pi R$. After the identification
  (\ref{latticereflection}) the resulting space is the latticized interval
  $[0,N_{y}]$.
  \begin{figure}[h]  
   \begin{equation*}
    \begin{picture} (10,10)
    \thicklines
    \multiput(2,1)(0,8){2}{\line(1,0){5}}
    \thinlines
    \multiput(3,2)(0,1){3}{\line(1,0){3}}
    \multiput(3,6)(0,1){3}{\line(1,0){3}}
    \put(4.5,5){\vdots}
    \put(2.5,5){\vector(0,1){4}}
    \put(2.5,5){\vector(0,-1){4}}
    \put(3.5,2.5){\vector(0,1){0.5}}
    \put(3.5,2.5){\vector(0,-1){0.5}}
    \put(1.5,5){$\pi R$}
    \put(3,2.5){$a$}
    \multiputlist(9,1)(0,8){$n_{y}=0 \; (y=0)$,$n_{y}=N_{y} \; (y=\pi R)$} 
   \end{picture}
   \end{equation*}
   \label{latticizidinterval}
   \caption{The latticized interval $[0,N_{y}]$. The two orbifold fixed points 
            are $n_{y}=0 \; (y=0)$ and $n_{y}=N_{y} \; (y=\pi R)$.}
  \end{figure}
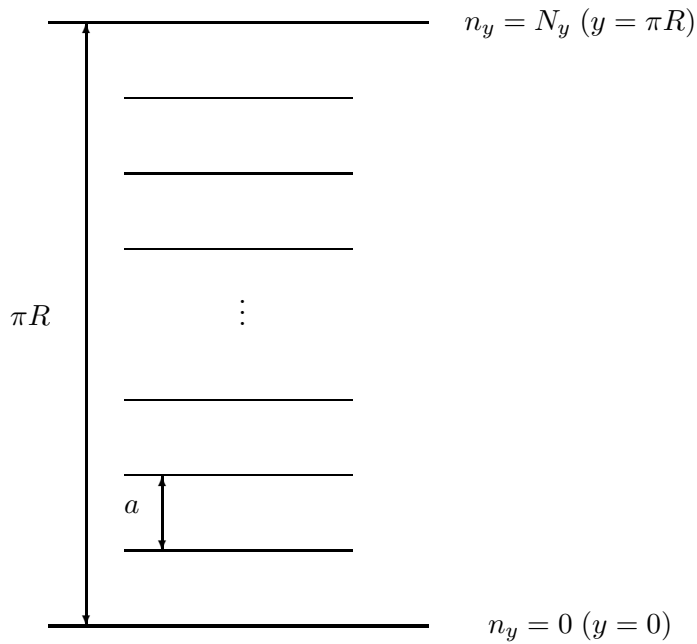

 \section[From latticized interval to bilayered transverse lattice]
         {From an orbifold model on the latticized interval to an 
          effective bilayered transverse lattice model
          via renormalisation group transformations}

  In this section we sketch the basic ideas which lead to an effective bilayered transverse lattice model.
  We start with the latticized interval $\Delta=[0,N_{y}]$, where $N_{y} \gg 1$, and take it as the
  fundamental lattice. The bilayered transverse lattice model is treated as an effective theory,
  which leads at a coarser scale to the same expectation values as the $S^{1}/\mathbb{Z}_{2}$
  orbifold theory on $\Delta$, however, with less degrees of freedom. The transition from a theory on the 
  latticized interval $\Delta$ to a theory on a bilayered transverse lattice is given by RG
  transformations. Let $\phi$ be an unitary PT in the
  fundamental theory and let $\phi^{\prime}$ a PT in the effective
  theory. The transition from $\phi$ to $\phi^{\prime}$ is given by a blockspin operator $\mathcal{C}$
  \begin{equation}
   \phi^{\prime}=\mathcal{C} \phi \; .
  \end{equation}
  To be more precise, let $\phi^{\prime}$ be the PT from a point $x \in \Delta^{\prime}$ to a point 
  $y \in \Delta^{\prime}$ along a path $C:x \to y$, where $\Delta^{\prime}$ is a coarser latticized 
  interval than $\Delta$. The blockspin is given by
  \begin{equation}
   \label{blockspin}
   \phi^{\prime}= \sum_{C:x \to y} \rho(C) \phi(C) \; ,
  \end{equation}
  where the sum goes over all paths $C:x \to y$ in $\Delta$ and $\rho(C)$ is
  a weight factor. Thus $\phi^{\prime}$ will be in the linear span of the unitary bulk gauge
  group $G$. The blockspin $\phi^{\prime}$ polar decomposes into a unitary part $U$ and
  a selfadjoint part $S$
  \begin{equation}
   \phi^{\prime}=U \; S \; .
  \end{equation}
  If it is possible to integrate out the selfadjoint part $S$ we will recover a local effective theory
  with unitary PTs $\phi^{\prime}$. We assume that this procedure will fail 
  after $n$ steps in the sense that the theory would acquire bad locality properties \cite{Lehmann:2003jh}
  and the effective theory needs nonunitary PTs for its locality.

  Let us consider the family of latticized intervals $\{ \Delta^{0},\Delta^{1},\dots,\Delta^{m} \}$
  \footnote{Note that in the context of RG transformations one usually considers a family of hypercubic
  lattices $\{ \Lambda^{0},\dots,\Lambda^{i}, \Lambda^{i+1},\dots \}$.},
  where $\Delta^{0}=\Delta$ is the fundamental latticized interval $[0,N_{y}]$ and $\Delta^{m}$, $m > n$
  is the bilayered transverse lattice. Obviously $\Delta^{m}$ is the coarsest 
  latticized interval as it consists of only two points. Hence an eBTLM 
  can always be interpreted as the endpoint of a RG-flow. 
  The PTs $\Phi$ \footnote{In the following we write $\Phi$ instead of $\phi^{\prime}$ for 
  nonunitary PTs} in the extra dimension are nonunitary as a consequence of the blockspin
  transformation (\ref{blockspin}). They can be interpreted as Higgs fields.
  When $\Phi$ becomes nonunitary a Higgs potential $V(\Phi)$ naturally
  emerges \cite{Pruestel2003}. The nonunitary PTs $\Phi$ take their values in a Lie group $H$ which is
  typically noncompact and larger 
  than the unitary bulk gauge group $G$. We call $H$ the holonomy group.  We always
  consider the case where $G$ is the maximal compact subgroup of $H$. As already mentioned above,
  the extra dimension consists of only two points which are the orbifold fixed points 
  $n_{y}=0 \; (y=0)$ and $n_{y}=N_{y} \; (y=\pi R)$. 
  We call the four-dimensional boundary at the fixed point $y=0$ the $L$-boundary and the four-dimensional
  boundary at the fixed point $y=\pi R$ the $R$-boundary.
  $G_{R}$ denotes the gauge group of the $R$-boundary, and $G_{L}$ denotes the gauge group 
  of the $L$-boundary.
  In principle, an orbifold breaking can lead to different gauge groups $G_{L}$ and $G_{R}$ at the
  boundaries $R$ and $L$. In the following however we will restrict ourselves to the case where
  $G_{L}=G_{R}=G_{0}$. The gauge group $G_{0}$ is the subgroup of $G$ left unbroken by the orbifold 
  projection $P$ i.e. the centraliser of $P$ in $G$. We call $G_{0}$ the orbifold unbroken gauge
  group.
  \begin{figure}[h]
   \begin{equation*}
    \begin{picture} (18,4.5)
     \thicklines
     \multiput(2,1.5) (0,2.5) {2} {\line (1,0) {9}}
     \thinlines
     \multiputlist(12,1.5)(0,2.5){$L \; (y=0)$,$R \; (y=\pi R)$}
     \multiputlist(6.5,1)(0,3.5){$G_{L}=G_{0}$,$G_{R}=G_{0}$}
     \multiputlist(6.5,2.6)(2.5,0){$G \subset H$}
     \put(2.5,2.5) {$\Phi$}
     \put(10.3,2.5) {$\Phi^{\ast}$}
     \matrixput (3,1.5)(7,0){2}(0,2,5){2}{\circle*{0.2}}
     \thicklines
     \dottedline[\circle*{0.05}]{0.1}(3,1.5)(3,4)
     \dottedline[\circle*{0.05}]{0.1}(10,1.5)(10,4)
     \thinlines
     \put(2.82,3.65){\Pisymbol{pzd}{115}}
     \put(9.82,1.58){\Pisymbol{pzd}{116}}
    \end{picture}
   \end{equation*}
   \label{btlmfigure}   
   \caption{Effective bilayered transverse lattice model (eBTLM).}
  \end{figure}
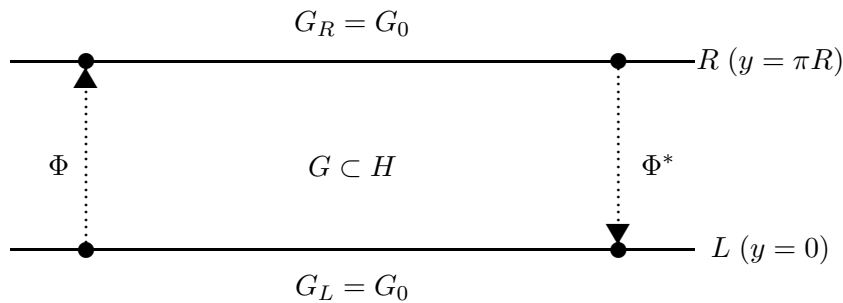  

  In the simplest approximation, the four-dimensional effective Lagrangian of an eBTLM reads
  \begin{equation}
    \label{bilayeredaction}
    \mathcal{L}_{4D}=-\frac{1}{4} \sum_{i=L,R} F_{i \mu \nu}^{a} F^{i \mu \nu a}
    + \text{tr} \left[ \left(D_{\mu}\Phi \right)^{\dagger} \left( D_{\mu}\Phi \right)
      \right] + V(\Phi) \; ,
  \end{equation}
  where the covariant derivative is given by
  \begin{equation}
   D_{\mu}\Phi=\partial_{\mu}\Phi
               +i g \left( A^{R}_{\mu}\Phi-\Phi A^{L}_{\mu} \right) \; .
  \end{equation}
  We discuss the terms in (\ref{bilayeredaction}):
  \begin{itemize}
   \item The term $F_{i \mu \nu}^{a} F^{i \mu \nu a}$ is a Yang-Mills term for the boundary
         gauge fields $A_{\mu}^{R}$ and $A_{\mu}^{L}$. Let $\mathfrak{g}_{0}=\text{Lie} \; G_{0}$. Then
         $A_{\mu}^{R}, A_{\mu}^{L} \in \mathfrak{g}_{0}$.
   \item The term tr$ \left[ \left(D_{\mu}\Phi \right)^{\dagger} \left( D_{\mu}\Phi \right) \right]$ 
         is the kinetic term for $\Phi$. It will lead to a mass term for the boundary
         gauge fields $A_{\mu}^{R}$ and $A_{\mu}^{L}$. 
   \item The term $V(\Phi)$ is the Higgs potential. If $V(\Phi)$ takes its minimum at
         non-trivial $\Phi_{min}$ the orbifold unbroken gauge group $G_{0}$ is spontaneously broken. 
  \end{itemize}
  We will see that under certain circumstances the effective four-dimensional Lagrangian
  (\ref{bilayeredaction}) equals the effective
  four-dimensional  Lagrangian of a corresponding $S^{1}/\mathbb{Z}_{2}$ continuum
  \footnote{By continuum we mean that $S^{1}/\mathbb{Z}_{2}$ is treated as usual as the quotient
  space $\mathbb{R}/\mathbb{D}_{\infty}$ and not as the latticized interval $[0,N_{y}]$.} 
  orbifold model.

 \section[EBTLM, ordinary Higgs mechanism and renormalisability]
 {Effective bilayered transverse lattice model, ordinary Higgs mechanism and renormalisability}

  \label{sectioneBTLM}

  In this section, we start to work out the correspondence between an eBTLM
  and a $S^{1}/\mathbb{Z}_{2}$ continuum orbifold model by investigating two simple examples. In both
  examples the orbifold projection $P$ is chosen to be trivial. Thus the bulk gauge group $G$ remains 
  unbroken. In the first example we consider Abelian gauge theory, i.e. we set $G=U(1)$, while 
  in the second example we consider non-Abelian gauge theory, e.g. we set $G=SU(N)$. It will turn out 
  that the Lagrangian of an eBTLM equals the effective four-dimensional 
  Lagrangian of an $S^{1}/\mathbb{Z}_{2}$ continuum orbifold model if we truncate the Kaluza-Klein (KK) 
  expansion for all fields in the $S^{1}/\mathbb{Z}_{2}$ continuum orbifold model at the first excited 
  KK mode. In addition, we need to make certain assumptions about the minimum 
  of the Higgs potential $V(\Phi)$ in the eBTLM. We will also demonstrate the close analogy between an  
  $S^{1}/\mathbb{Z}_{2}$ continuum orbifold model with truncated KK-mode expansion and the {\em{ordinary}}
  Higgs mechanism of four-dimensional gauge theories \cite{Dienes:1998vg}.
  Since the ordinary Higgs mechanism of four-dimensional gauge theories preserves 
  renormalisability, we conclude that also the truncated $S^{1}/\mathbb{Z}_{2}$ continuum orbifold model
  is renormalisable and therewith the corresponding eBTLM. As already mentioned in the introduction
  of this chapter, for {\em{trivial orbifold projection $P$}} and {\em{trivial minimum of the
  Higgs potential $V(\Phi)$}} there is a close analogy between an eBTLM and deconstruction models
  \cite{Hill:2000mu,Cheng:2001vd,Cheng:2001nh}.

  \subsection{Abelian gauge theory}

   In the first example we consider Abelian gauge theory, i.e. we start with the bulk gauge group $G=U(1)$.
   Figure \ref{figureabelianexample} summarises the setup.
   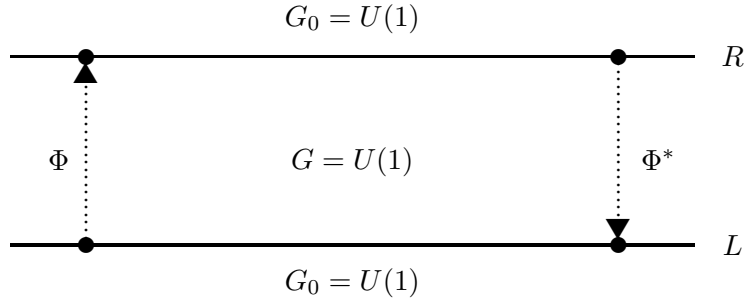
\begin{figure}[h]
    \begin{equation*}
     \begin{picture} (18,4.5)
      \thicklines
      \multiput(2,1.5) (0,2.5) {2} {\line (1,0) {9}}
      \thinlines
      \multiputlist(11.5,1.5)(0,2.5){$L$,$R$}
      \multiputlist(6.5,4.5)(7,0){$G_{0}=U(1)$}
      \multiputlist(6.5,2.6)(2.5,0){$G=U(1)$}
      \multiputlist(6.5,1)(7,0){$G_{0}=U(1)$}
      \put(2.5,2.5) {$\Phi$}
      \put(10.3,2.5) {$\Phi^{\ast}$}
      \matrixput (3,1.5)(7,0){2}(0,2,5){2}{\circle*{0.2}}
      \thicklines
      \dottedline[\circle*{0.05}]{0.1}(3,1.5)(3,4)
      \dottedline[\circle*{0.05}]{0.1}(10,1.5)(10,4)
      \thinlines
      \put(2.82,3.65){\Pisymbol{pzd}{115}}
      \put(9.82,1.58){\Pisymbol{pzd}{116}}
     \end{picture}
    \end{equation*}
    \caption{Effective bilayered transverse lattice model for the bulk gauge group $G=U(1)$ and
             trivial orbifold projection $P=1$.}
    \label{figureabelianexample}
   \end{figure}
   The effective four-dimensional Lagrangian (\ref{bilayeredaction}) reads 
   \begin{equation}
     \label{Lagrangianuone}
     \mathcal{L}_{4D}=-\frac{1}{4} F^{L}_{\mu \nu} F^{L \mu \nu}-\frac{1}{4} F^{R}_{\mu \nu} F^{R \mu \nu}
     +\left( D_{\mu}\Phi \right)^{\dagger}\left( D_{\mu}\Phi \right)+V(\Phi)
   \end{equation}  
   where
   \begin{gather}
    F^{L}_{\mu\nu}=\partial_{\mu}A^{L}_{\nu}-\partial_{\nu}A^{L}_{\mu} \quad ,
    \quad F^{R}_{\mu\nu}=\partial_{\mu}A^{R}_{\nu}-\partial_{\nu}A^{R}_{\mu}  \; , \\
    \label{covariantderivaticephiu1}
    D_{\mu}\Phi=\partial_{\mu}\Phi
               +i g \left( A^{R}_{\mu}\Phi-\Phi A^{L}_{\mu} \right) \; .
   \end{gather}
   The fields $A^{L}_{\mu}$ and $A^{R}_{\mu}$ are $U(1)$ gauge 
   fields on the $L$- and $R$-boundary respectively.

   Suppose that the Higgs potential $V(\Phi)$ takes its minimum at $\Phi_{min}$ where
   \begin{equation}
    \label{minimumphiabelian}
    \Phi_{min}=\frac{1}{2} \; \rho_{min} \; , \quad \rho_{min} \in \mathbb{R}_{\ast}^{+} 
   \end{equation} 
   and $\mathbb{R}_{\ast}^{+}=\mathbb{R}^{+} / \{0\}$.
   Inserting (\ref{minimumphiabelian}) in (\ref{covariantderivaticephiu1}) yields for the kinetic term
   \begin{equation}
    \left(  D_{\mu}\Phi_{min} \right)^{\dagger} \left(
    D_{\mu}\Phi_{min} \right)=\frac{1}{4} \; g^{2} \rho_{min}^{2}
    \left(A^{R}_{\mu}-A^{L}_{\mu}\right)^{2} \; .
   \end{equation}
   This is a mass term for the $U(1)$ gauge fields $A^{R}_{\mu}$ and 
   $A^{L}_{\mu}$. Defining $A:=(A^{R}_{\mu},A^{L}_{\mu})$, we can rewrite 
   \begin{equation}
    \mathcal{L}_{mass}=\frac{1}{4} \; g^{2} \rho_{min}^{2}\left(A^{R}_{\mu}-A^{L}_{\mu}\right)^{2}
    =A \; M \; A^{t} \; ,
   \end{equation}
   where $M$ is the mass-squared matrix
   \begin{equation}
    M=\frac{1}{4} \; g^{2} \rho_{min}^{2} 
    \left( \begin{array}{cc} 1 & -1 \\ -1 & 1 \end{array} \right) \; ,
   \end{equation}
   and $A^{t}$ is the transpose of $A$. 
   The gauge fields $A^{R}_{\mu}$ and $A^{L}_{\mu}$ can be expressed as
   real linear combinations of their mass eigenstates $A^{(0)}_{\mu}$ and
   $A^{(1)}_{\mu}$
   \begin{gather}
    \label{masseigenstates}
    A^{R}_{\mu}=\frac{1}{\sqrt{2}} \left( A^{(0)}_{\mu}+A^{(1)}_{\mu} \right) \; ,\\
    A^{L}_{\mu}=\frac{1}{\sqrt{2}} \left( A^{(0)}_{\mu}-A^{(1)}_{\mu} \right) \; .
    \nonumber
   \end{gather}
   In this new basis the mass-squared matrix $M$ is diagonal. We obtain
   \begin{eqnarray}
    \mathcal{L}_{mass} & = & A \; M \; A^{t}=\frac{1}{4} \; g^{2} \rho_{min}^{2}
    \left(A^{R}_{\mu}-A^{L}_{\mu}\right)^{2}  \\
    & = & \frac{1}{8} \; g^{2} \rho_{min}^{2}\left( A^{(0)}_{\mu}+A^{(1)}_{\mu}
          - A^{(0)}_{\mu}-A^{(1)}_{\mu}
          \right)^{2} \nonumber \\
    & = & \frac{1}{2} \; g^{2} \rho_{min}^{2} \left( A^{(1)}_{\mu} \right)^{2}  \nonumber 
   \end{eqnarray}
   which leads to the mass 
   \begin{equation}
    \label{transversebilayermass}
    m=g \rho_{min} 
   \end{equation}
   for the gauge field $A^{(1)}_{\mu}$, while the gauge field $A^{(0)}_{\mu}$ remains massless.
   In the basis of mass eigenstates (\ref{masseigenstates}) the Lagrangian
   (\ref{Lagrangianuone}) reads
   \begin{equation}
    \label{4defflagu1btlm}
    \mathcal{L}_{4D}=-\frac{1}{4} \; \left( \partial_{\mu} A_{\nu}^{(0)} 
                     -\partial_{\nu} A_{\mu}^{(0)} \right)^{2} 
                     -\frac{1}{4} \; \left( \partial_{\mu} A_{\nu}^{(1)} 
                     -\partial_{\nu} A_{\mu}^{(1)} \right)^{2} + \frac{1}{2} \;
                      m^{2} \left( A_{\mu}^{(1)} \right)^{2} \; .
   \end{equation}    
   This Lagrangian describes two Abelian gauge fields $A^{(0)}_{\mu}$ and 
   $A^{(1)}_{\mu}$, where $A^{(0)}_{\mu}$ is a massless field and the field 
   $A^{(1)}_{\mu}$ is massive with mass $m=g \rho_{min}$. 

   We compare this result to an $S^{1}/\mathbb{Z}_{2}$ continuum
   orbifold model. Let $G=U(1)$ be the bulk gauge group. The five-dimensional
   Lagrangian  \footnote{Recall that $M,N \in (\mu,y)$, where $\mu=0,1,2,3$.} reads
   \begin{equation} 
    \label{5dlagrangianorbifoldabelain}
    \mathcal{L}_{5D}=-\frac{1}{4} F_{MN} F^{MN} \; ,
   \end{equation}
   where
   \begin{equation}
    F_{MN}=\partial_{M}A_{N}-\partial_{N}A_{M} \; .
   \end{equation}
   The boundary conditions for the gauge fields read
   \begin{gather}
    A_{\mu}(x^{\mu},-y)=P \; A_{\mu}(x^{\mu},y) \; P^{-1}  \\
    A_{y}(x^{\mu},-y)=-P \; A_{y}(x^{\mu},y) \; P^{-1}  \; .
   \end{gather}   
   We take the trivial orbifold projection
   \begin{equation}
    P=1 
   \end{equation}
   and Fourier expand
   \begin{gather}
    A_{\mu}(x^{\mu},y)=\frac{1}{\sqrt{2 \pi R}}
    A_{\mu}^{(0)}(x^{\mu})+\frac{1}{\sqrt{\pi R}}\sum_{n=1}^{\infty} 
    A_{\mu}^{(n)}(x^{\mu}) \cos(\frac{n y}{R}) \\ 
    A_{y}(x^{\mu},y)=\frac{1}{\sqrt{\pi R}}
    \sum_{n=1}^{\infty} A_{y}^{(n)}(x^{\mu}) \sin(\frac{n y}{R})  \; .
   \end{gather}  
   Truncating this expansion at $n=1$ yields
   \begin{gather}
    \label{kkmodeexpansionabelian}
    A_{\mu}(x^{\mu},y)=\frac{1}{\sqrt{2 \pi R}} A_{\mu}^{(0)}(x^{\mu})
      +\frac{1}{\sqrt{\pi R}} A_{\mu}^{(1)}(x^{\mu}) \cos(\frac{y}{R}) \\
    A_{y}(x^{\mu},y)=\frac{1}{\sqrt{\pi R}}A_{y}^{(1)}(x^{\mu}) 
      \sin(\frac{y}{R}) \nonumber \; .
   \end{gather}
   The field strength $F_{MN}$ consists of two parts
   \begin{eqnarray}
    \label{u1fmunu}
    F_{\mu\nu} & = & \partial_{\mu}A_{\nu}-\partial_{\nu}A_{\mu} \\ & = &
        \frac{1}{\sqrt{2 \pi R}} \left[ \left( \partial_{\mu} 
        A_{\nu}^{(0)} - \partial_{\nu} A_{\mu}^{(0)} \right)
        + \left( \partial_{\mu} A_{\nu}^{(1)} - \partial_{\nu} 
        A_{\mu}^{(1)} \right) \cdot \sqrt{2} \cos(\frac{y}{R}) 
        \right] \; , \nonumber \\  
    F_{\mu y} & = & \partial_{\mu}A_{y}-\partial_{y}A_{\mu}=
             \frac{1}{\sqrt{\pi R}} \left[ \partial_{\mu} A_{y}^{(1)} 
             \cdot \sin(\frac{y}{R})+ A_{\mu}^{(1)} \frac{1}{R} 
             \sin(\frac{y}{R}) \right] \; , \nonumber 
   \end{eqnarray}
   where we have inserted the truncated KK-mode expansion (\ref{kkmodeexpansionabelian}). 
   We insert $F_{\mu\nu}$ and $F_{\mu y}$ into the five-dimensional Lagrangian
   (\ref{5dlagrangianorbifoldabelain}) and integrate over the circle $S^{1}$. This yields 
   \begin{eqnarray}
    \label{4defflagu1orbi}
    \mathcal{L}_{4D} & = & \int_{0}^{2\pi R} \{ -\frac{1}{4} F_{MN} F^{MN} \} \; dy=
                           \int_{0}^{2\pi R} \{ -\frac{1}{4} F_{\mu\nu} F^{\mu\nu}
                           - \frac{1}{2} F_{\mu y} F^{\mu y} \} \; dy \\
     & = & -\frac{1}{4} \left(\partial_{\mu} A_{\nu}^{(0)} - 
     \partial_{\nu} A_{\mu}^{(0)} \right)^{2} -\frac{1}{4}
     \left(\partial_{\mu} A_{\nu}^{(1)} - 
     \partial_{\nu} A_{\mu}^{(1)}\right)^{2}    
     + \frac{1}{2} \left( \partial_{\mu} A_{y}^{(1)}+ \frac{1}{R} A_{\mu}^{(1)} \right)^{2}  \; . \nonumber
   \end{eqnarray}
   This Lagrangian describes two Abelian gauge fields $A_{\mu}^{(0)}$ and
   $A_{\mu}^{(1)}$, where the field $A_{\mu}^{(0)}$ is massless and the field $A_{\mu}^{(1)}$ is 
   massive. We define 
   \begin{equation}
    B_{\mu}^{(1)}=A_{\mu}^{(1)}+ R \; \partial_{\mu} A_{y}^{(1)} \; ,
   \end{equation}
   and express the Lagrangian (\ref{4defflagu1orbi}) in terms of $A_{\mu}^{(0)}$ and $B_{\mu}^{(1)}$
   \begin{equation}
    \mathcal{L}_{4D}= -\frac{1}{4} \left(\partial_{\mu} A_{\nu}^{(0)} - 
     \partial_{\nu} A_{\mu}^{(0)} \right)^{2} -\frac{1}{4}
     \left(\partial_{\mu} B_{\nu}^{(1)} - \partial_{\nu} B_{\mu}^{(1)}\right)^{2} + \frac{1}{2} \;
     \frac{1}{R^{2}}\left(  B_{\mu}^{(1)} \right)^{2} \; .
   \end{equation}
   We observe that the field $B_{\mu}^{(1)}$ has mass $1/R$.

   We compare this result to the ordinary $U(1)$ Abelian Higgs model \cite{Peskin:1995ev}. Let
   \begin{equation}
    \mathcal{L}=-\frac{1}{4} F_{\mu\nu} F^{\mu\nu} + \mid D_{\mu} \phi \mid^{2} - V(\phi) \; ,
   \end{equation}
   with $D_{\mu}=\partial_{\mu}+i e A_{\mu}$, be the Lagrangian of a complex
   scalar field $\phi$ coupled both to itself and an electromagnetic field. The potential $V(\Phi)$
   is chosen to be of the form
   \begin{equation}
    V(\phi)=-\mu^{2} \left(\phi^{\ast} \phi \right) +
    \lambda^{2} \left(\phi^{\ast} \phi \right)^{2}
   \end{equation}
   where $\mu^{2} >0$. With the minimum of $V(\phi)$ at
   \begin{equation}
    \label{vevabelianordinaryabelianhiggs}
    \langle \phi \rangle=\frac{1}{\sqrt{2}} \phi_{0}  
   \end{equation}
   where $\phi_{0}=\mu/\lambda$. We expand the complex field $\phi(x)$ around the minimum $\phi_{0}$ as
   \begin{equation}
    \phi(x)=\frac{1}{\sqrt{2}}\left(\phi_{0}+\phi_{1}+ i \phi_{2} \right) \; .
   \end{equation}
   We insert this expansion into the kinetic term $ \mid D_{\mu} \phi \mid^{2}$. Thus we obtain
   \begin{equation}
    \label{ordinaryhiggskineticexpansion}
    \mid D_{\mu} \phi \mid^{2}=\frac{1}{2} \left( \partial_{\mu} \phi_{1} \right)^{2}
                              +\frac{1}{2} \left( \partial_{\mu} \phi_{2} \right)^{2}
                              +e \phi_{0} \; A_{\mu} \partial^{\mu} \phi_{2} + \frac{1}{2} e^{2} \phi_{0}^{2}
                               A_{\mu} A^{\mu} + \dots \; ,
   \end{equation}
   where we have omitted terms cubic and quartic in the fields $A_{\mu}$, $\phi_{1}$ and $\phi_{2}$.
   We compare this result with
   \begin{equation}
    \label{masstermfirstkkmode}
    \frac{1}{2} \left( \partial_{\mu} A_{y}^{(1)}+ \frac{1}{R} A_{\mu}^{(1)} \right)^{2}
   \end{equation} 
   from (\ref{4defflagu1orbi}).
   This apparently coincides with the Abelian Higgs model if we identify \cite{Dienes:1998vg}
   \begin{equation}
    e \; \phi_{0} \quad \Longleftrightarrow \quad \frac{1}{R} \; .
   \end{equation}
   In addition, the comparison of (\ref{ordinaryhiggskineticexpansion}) with (\ref{masstermfirstkkmode})
   shows that the first excited KK-mode gauge field $A_{y}^{(1)}$ plays the role of the Goldstone boson
   $\phi_{2}$ \cite{Dienes:1998vg}. It is therefore natural to go to unitary gauge, i.e. we set
   \begin{equation}
    \label{axialgauge}
    A_{y}^{(1)}=0 \; .
   \end{equation} 
   In the context of gauge theories in extra dimensions this gauge is known as axial gauge
   and we will from now on call (\ref{axialgauge}) axial gauge.
   In axial gauge the Lagrangian (\ref{4defflagu1orbi}) reads
   \begin{equation}
    \label{4defflagabelianorbiaxial}
    \mathcal{L}_{4D}= -\frac{1}{4} \left(\partial_{\mu} A_{\nu}^{(0)} - 
     \partial_{\nu} A_{\mu}^{(0)} \right)^{2} -\frac{1}{4}
     \left(\partial_{\mu} A_{\nu}^{(1)} - \partial_{\nu} A_{\mu}^{(1)}\right)^{2} + \frac{1}{2} \;
     \frac{1}{R^{2}}\left(  A_{\mu}^{(1)} \right)^{2} \; ,
   \end{equation}
   since $B_{\mu}^{(1)}=A_{\mu}^{(1)}$ for $A_{y}^{(1)}=0$.
   Due to the close analogy of the 
   $S^{1}/\mathbb{Z}_{2}$ continuum orbifold model with truncated KK-mode expansion and the ordinary Higgs
   mechanism, we conclude that the truncated $S^{1}/\mathbb{Z}_{2}$ continuum orbifold model is
   renormalisable.

   We compare the Lagrangian (\ref{4defflagu1orbi}) of the truncated  $S^{1}/\mathbb{Z}_{2}$ continuum 
   orbifold model  in axial
   gauge (\ref{4defflagabelianorbiaxial}) with the Lagrangian of the corresponding eBTLM 
   (\ref{4defflagu1btlm}). If we require  
   \begin{equation}
    g \rho_{min}=\frac{1}{R} \; ,
   \end{equation}  
   both Lagrangian's equal and hence both theories describe the same physics. Thus we conclude that 
   also the eBTLM (\ref{4defflagu1btlm}) is {\em{renormalisable}}.

 \subsection{Non-Abelian gauge theory}

  \label{sectionenonabBTLM}
  
  In the second example we consider non-Abelian gauge theory, i.e. we start for example
  with the bulk gauge group $G=SU(N)$. 
  Figure \ref{figurenonabelianexample} summarises the setup.
  \begin{figure}[h]
   \begin{equation*}
    \begin{picture} (18,4.5)
     \thicklines
     \multiput(2,1.5) (0,2.5) {2} {\line (1,0) {9}}
     \thinlines
     \multiputlist(11.5,1.5)(0,2.5){$L$,$R$}
     \multiputlist(6.5,4.5)(7,0){$G_{0}=SU(N)$}
     \multiputlist(6.5,2.6)(2.5,0){$G=SU(N)$}
     \multiputlist(6.5,1)(7,0){$G_{0}=SU(N)$}
     \put(2.5,2.5) {$\Phi$}
     \put(10.3,2.5) {$\Phi^{\ast}$}
     \matrixput (3,1.5)(7,0){2}(0,2,5){2}{\circle*{0.2}}
     \thicklines
     \dottedline[\circle*{0.05}]{0.1}(3,1.5)(3,4)
     \dottedline[\circle*{0.05}]{0.1}(10,1.5)(10,4)
     \thinlines
     \put(2.82,3.65){\Pisymbol{pzd}{115}}
     \put(9.82,1.58){\Pisymbol{pzd}{116}}
    \end{picture}
   \end{equation*}
   \caption{Effective bilayered transverse lattice model for the bulk gauge group $G=SU(N)$ and
            trivial orbifold projection and $P=\text{diag}(1,\dots,1)$.}
   \label{figurenonabelianexample}
  \end{figure}
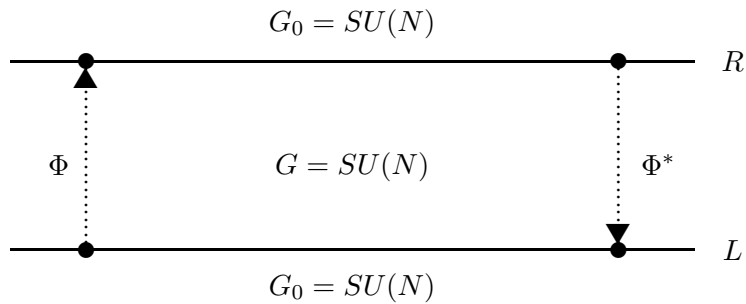
  The effective four-dimensional Lagrangian reads
  \begin{equation}
    \label{Lagrangiansun}
    \mathcal{L}_{4D}=-\frac{1}{4}F^{i}_{L \mu \nu} F^{L i \mu \nu}
                -\frac{1}{4}F^{i}_{R \mu \nu} F^{R i \mu \nu}
                +\text{tr}\left[ \left( D_{\mu}\Phi \right)^{\dagger}
                 \left(D_{\mu}\Phi \right) \right]+V(\Phi) \; ,
  \end{equation}  
  where
  \begin{gather}
    F^{L i}_{\mu\nu}=\partial_{\mu}A^{Li}_{\nu}-\partial_{\nu}A^{Li}_{\mu} 
                 + g f^{ijk} A^{Lj}_{\nu} A^{Lk}_{\mu} \; ,
    \; F^{R}_{\mu\nu}=\partial_{\mu}A^{Ri}_{\nu}-\partial_{\nu}A^{Ri}_{\mu}
                 + g f^{ijk} A^{Lj}_{\nu} A^{Lk}_{\mu} \; , \\
    \label{covariantderivaticephisun}
    D_{\mu}\Phi=\partial_{\mu}\Phi+i g \left( A^{R}_{\mu}\Phi
                -\Phi A^{L}_{\mu} \right)=\partial_{\mu}\Phi+i g \left( A^{Ri}_{\mu} L_{i} \;  \Phi
                -\Phi \; A^{Li}_{\mu}  L_{i} \right) \; .
  \end{gather} 
  The $L_{i}$ denote the generators of $G$ normalised as $\text{tr}( L_{i} L_{j} )=\frac{1}{2} \delta_{ij}$.
  Let $V(\Phi)$ take its minimum at $\Phi_{min}$ with
  \begin{equation}
   \label{mimimumphinonabelian}
   \Phi_{min}=\rho_{min} \; \frac{1}{\sqrt{2}} \mathbf{1}_{N} \; ,
  \end{equation}
  where $\mathbf{1}_{N}$ is the $N \times N$ unit matrix and $\rho_{min} \in \mathbb{R}_{\ast}^{+}$. 
  Inserting (\ref{mimimumphinonabelian}) in (\ref{covariantderivaticephisun}) yields for the kinetic term
  \begin{equation}
   \text{tr} \left[ \left( D_{\mu}\Phi_{min} \right)^{\dagger} 
   \left(D_{\mu}\Phi_{min} \right) \right]=\frac{1}{4} \; g^{2} \rho_{min}^{2}
   \left(A^{R i}_{\mu}-A^{L i}_{\mu}\right)^{2} \; .
  \end{equation}
  Defining $A:=(A^{R i}_{\mu},A^{L i}_{\mu})$, we can rewrite 
  \begin{equation}
   \mathcal{L}_{mass}=\frac{1}{4} \; g^{2} \rho_{min}^{2}\left(A^{R i}_{\mu}
   -A^{L i}_{\mu}\right)^{2}=A \; M \; A^{t} \; ,
  \end{equation}
  where $M$ is the mass-squared matrix
  \begin{equation}
    M=\frac{1}{4} \; g^{2} \rho_{min}^{2} 
    \left( \begin{array}{cc} 1 & -1 \\ -1 & 1 \end{array} \right) \; ,
  \end{equation}
  and $A^{t}$ is the transpose of $A$. 
  The gauge fields $A^{R i}_{\mu}$ and $A^{L i}_{\mu}$ can be expressed as
  real linear combinations of their mass eigenstates $A^{i(0)}_{\mu}$ and
  $A^{i(1)}_{\mu}$
  \begin{gather}
   \label{masseigenstatessun}
   A^{R i}_{\mu}=\frac{1}{\sqrt{2}} \left( A^{i(0)}_{\mu}+A^{i(1)}_{\mu} \right) \; ,\\
   A^{L i}_{\mu}=\frac{1}{\sqrt{2}} \left( A^{i(0)}_{\mu}-A^{i(1)}_{\mu} \right) \; .
   \nonumber
  \end{gather}
  In this new basis, the mass squared matrix $M$ is diagonal. We obtain
  \begin{eqnarray}
   \label{nonableinamassterm1}
   \mathcal{L}_{mass} & = & \frac{1}{4} \;g^{2} \rho_{min}^{2}
   \left( A^{R i}_{\mu}-A^{L i}_{\mu} \right)^{2} \\
   & = & \frac{1}{8} \; g^{2} \rho_{min}^{2} \left(  
         A^{i(0)}_{\mu}+A^{i(1)}_{\mu} - A^{i(0)}_{\mu}-A^{i(1)}_{\mu}  
         \right)^{2}  \nonumber \\
   & = & \frac{1}{2} \; g^{2} \rho_{min}^{2} \left(A^{i(1)}_{\mu}\right)^{2}  \; . \nonumber
  \end{eqnarray}
  This leads to the common mass term
  \begin{equation}
   \label{transversnonabelianmass2}
   m=g \rho_{min} 
  \end{equation}
  for all gauge fields $A^{i(1)}_{\mu}$. The gauge fields $A^{i(0)}_{\mu}$ remain
  massless. 

  We calculate the Lagrangian (\ref{Lagrangiansun}) in the basis of mass eigenstates 
  (\ref{masseigenstatessun}). The Yang-Mills term reads 
  \begin{eqnarray}
    \label{yangmillsnonablian}
    \mathcal{L}_{YM} & = & -\frac{1}{4} \left( \partial_{\mu}A^{L i}_{\nu}
                           -\partial_{\nu}A^{L i}_{\mu} 
                           +g f^{ijk} A^{L j}_{\mu} A^{L k}_{\nu} \right)^{2}-\frac{1}{4}
                           \left( \partial_{\mu}A^{R l}_{\nu}
                           -\partial_{\nu}A^{R l}_{\mu}+g f^{lmn} 
                           A^{R m}_{\mu} A^{R n}_{\nu} \right)^{2} \nonumber \\
                     & = & -\frac{1}{8} \left( \partial_{\mu} A^{i(0)}_{\nu}-
                           \partial_{\nu} A^{i(0)}_{\mu}- \left(
                           \partial_{\mu} A^{i(1)}_{\nu}-
                           \partial_{\nu} A^{i(1)}_{\mu} \right) \right. \\
                     & + & \left. \frac{g}{\sqrt{2}} f^{ijk} \left(
                           A^{j(0)}_{\mu}A^{k(0)}_{\nu}-
                           A^{j(0)}_{\mu}A^{k(1)}_{\nu}-
                           A^{j(1)}_{\mu}A^{k(0)}_{\nu}+
                           A^{j(1)}_{\mu}A^{k(1)}_{\nu}  
                           \right)\right)^{2}  \nonumber \\
                     & - & \frac{1}{8} \left( \partial_{\mu} A^{l(0)}_{\nu}-
                           \partial_{\nu} A^{l(0)}_{\mu}+ \left(
                           \partial_{\mu} A^{l(1)}_{\nu}-
                           \partial_{\nu} A^{l(1)}_{\mu} \right) \right. \nonumber \\
                     & + & \left. \frac{g}{\sqrt{2}} f^{lmn} \left(
                           A^{m(0)}_{\mu}A^{n(0)}_{\nu}+
                           A^{m(0)}_{\mu}A^{n(1)}_{\nu}+
                           A^{m(1)}_{\mu}A^{n(0)}_{\nu}+ 
                           A^{m(1)}_{\mu}A^{n(1)}_{\nu}
                           \right)\right)^{2} \; . \nonumber
  \end{eqnarray}  
  We isolate the zero-mode. This yields
  \begin{equation}
    \label{4dggbtlm}
    \mathcal{L}_{0}=-\frac{1}{4} \left( \partial_{\mu} A^{i(0)}_{\nu}-
                           \partial_{\nu} A^{i(0)}_{\mu}
                  +\frac{g}{\sqrt{2}} f^{ijk} A^{j(0)}_{\mu} A^{k(0)}_{\nu} \right)^{2} \; .
  \end{equation} 
  If we define the coupling constant
  \begin{equation}
   \tilde{g}=\frac{g}{\sqrt{2}} \; ,
  \end{equation}
  the zero mode has the canonical four-dimensional kinetic term with field strength
  \begin{equation}
   F_{\mu\nu}^{i(0)}=\partial_{\mu} A^{i(0)}_{\nu}-\partial_{\nu} A^{i(0)}_{\mu}
                    +\tilde{g} f^{ijk} A^{j(0)}_{\mu} A^{k(0)}_{\nu} \; .
  \end{equation}
  In contrast, for the first KK-mode we obtain 
  \begin{equation}
   \mathcal{L}_{1}=-\frac{1}{4} \left( \partial_{\mu} A^{i(1)}_{\nu}-
                           \partial_{\nu} A^{i(1)}_{\mu} \right)^{2} \; .
  \end{equation} 
  Note that a term 
  $-\frac{1}{4} \left( \partial_{\mu} A^{i(1)}_{\nu}-\partial_{\nu} A^{i(1)}_{\mu}
                  +g f^{ijk} A^{j(1)}_{\mu} A^{k(1)}_{\nu} \right)^{2}$ for
  the first excited KK-mode does not occur due to the relative minus sign for 
  $\partial_{\mu} A^{i(1)}_{\nu}-
  \partial_{\nu} A^{i(1)}_{\mu}$ and $\partial_{\mu} A^{l(1)}_{\nu}-\partial_{\nu} A^{l(1)}_{\mu}$
  in (\ref{yangmillsnonablian}). However there will the term  
  \begin{equation}
   \label{firstkkmodeinteraction}
   -\frac{1}{4} \tilde{g}^{2} f^{ijk} f^{imn}  A^{j(1)}_{\mu}A^{k(1)}_{\nu}
                          A^{m(1)}_{\mu}A^{n(1)}_{\nu}  \; ,
  \end{equation}
  which describes the self-interaction of the first excited KK-modes. 
  In addition, we obtain the following interaction terms among the zero mode and the first excited mode
  linear in $\tilde{g}$ 
  \begin{eqnarray}
   \label{linearf}
   \mathcal{L}_{\tilde{g}} & = & -\frac{1}{4} \; \tilde{g} \; \left( \left(\partial_{\mu} 
                      A^{i(0)}_{\nu}-\partial_{\nu} A^{i(0)}_{\mu} \right)
                      f^{ijk} \; A^{j(1)}_{\mu}A^{k(1)}_{\nu} \right .\\
                      & + &   \left. \left(\partial_{\mu} 
                      A^{i(1)}_{\nu}-\partial_{\nu} A^{i(1)}_{\mu} \right)
                      f^{ijk} \; A^{j(0)}_{\mu}A^{k(1)}_{\nu} 
                        +  \left(\partial_{\mu} 
                      A^{i(1)}_{\nu}-\partial_{\nu} A^{i(1)}_{\mu} \right)
                      f^{ijk} \; A^{j(1)}_{\mu}A^{k(0)}_{\nu} \right)
                      \nonumber 
  \end{eqnarray}
  and quadratic in $\tilde{g}$ 
  \begin{eqnarray}
   \label{quadf}
   \mathcal{L}_{\tilde{g}^{2}} & = & -\frac{1}{4} \; \tilde{g}^{2} \; f^{ijk} f^{imn}  \; \left(  
                          A^{j(0)}_{\mu}A^{k(1)}_{\nu}
                          A^{m(0)}_{\mu}A^{n(1)}_{\nu}   
                         +A^{j(1)}_{\mu}A^{k(0)}_{\nu}
                          A^{m(1)}_{\mu}A^{n(0)}_{\nu} \right)  \\
                   & - &  \frac{1}{2} \; \tilde{g}^{2} \; f^{ijk} f^{imn}  \; \left( 
                          A^{j(0)}_{\mu}A^{k(0)}_{\nu}
                          A^{m(1)}_{\mu}A^{n(1)}_{\nu}      
                          +A^{j(0)}_{\mu}A^{k(1)}_{\nu}
                          A^{m(1)}_{\mu}A^{n(0)}_{\nu} \right)  \nonumber \; .
  \end{eqnarray}
 
  We compare this result to an $S^{1}/\mathbb{Z}_{2}$ continuum
  orbifold model. Let $G=SU(N)$ be the bulk gauge group. The five-dimensional
  Lagrangian reads
  \begin{equation}
   \label{5dlagrangianorbifoldnonabelain}
   \mathcal{L}_{5D}=-\frac{1}{4}  F^{a}_{MN} F^{aMN} \; ,
  \end{equation}
  where
  \begin{equation}
   F^{a}_{MN}=\partial_{M}A^{a}_{N}-\partial_{N}A^{a}_{M}
              +g_{5}f^{abc}A^{b}_{M}A^{c}_{N}  \; .
  \end{equation}
  In this equation, $g_{5}$ is the five-dimensional gauge coupling constant.
  The boundary conditions for the gauge fields read
  \begin{gather}
   A_{\mu}(x^{\mu},-y)=P \; A_{\mu}(x^{\mu},y) \; P^{-1}  \\
   A_{y}(x^{\mu},-y)=-P \; A_{y}(x^{\mu},y) \; P^{-1}  \; .
  \end{gather}   
  As in the Abelian case, we take the trivial orbifold projection
  \begin{equation}
   P=\text{diag}(1,\dots,1)  \; .
  \end{equation}
  The Fourier expansion up to the first KK-mode reads 
  \begin{gather}
   A^{a}_{\mu}(x^{\mu},y)=\frac{1}{\sqrt{2 \pi R}} A_{\mu}^{a(0)}(x^{\mu})
     +\frac{1}{\sqrt{\pi R}} A_{\mu}^{a(1)}(x^{\mu}) \cos(\frac{y}{R}) \\
   A^{a}_{y}(x^{\mu},y)=\frac{1}{\sqrt{\pi R}}A_{y}^{a(1)}(x^{\mu}) 
     \sin(\frac{y}{R}) \; .
  \end{gather}
  The field strength $F_{MN}$ consists of two parts
  \begin{eqnarray}
   F^{a}_{\mu\nu}& = & \partial_{\mu}A^{a}_{\nu}-\partial_{\nu}A^{a}_{\mu}
                       +g_{5} f^{abc}A^{b}_{\mu}A^{c}_{\nu} \label{fmununonabelorbi} \\
       & = & \frac{1}{\sqrt{2 \pi R}} \left[ \left( \partial_{\mu} 
             A_{\nu}^{a(0)} - \partial_{\nu} A_{\mu}^{a(0)} \right)
             +\left( \partial_{\mu} A_{\nu}^{a(1)} - \partial_{\nu} 
             A_{\mu}^{a(1)} \right) \cdot \sqrt{2} \cos(\frac{y}{R})\right]
             \nonumber \\
       & + & \frac{g_{5}}{2 \pi R} f^{abc} \left[A_{\mu}^{b(0)}(x^{\mu})
             +A_{\mu}^{b(1)}(x^{\mu}) \sqrt{2} \cos(\frac{y}{R}) \right] \nonumber \\
       &   & \left[A_{\nu}^{c(0)}(x^{\mu})
             +A_{\nu}^{c(1)}(x^{\mu}) \sqrt{2} \cos(\frac{y}{R}) \right] \; ,
             \nonumber\\  
   F^{a}_{\mu y} & = & \partial_{\mu}A^{a}_{y}-\partial_{y}A^{a}_{\mu}+
                       g_{5}f^{abc} A^{b}_{\mu}A^{c}_{y} \label{fmu5nonabelorbi} \\
            & = & \frac{1}{\sqrt{\pi R}} \partial_{\mu} A_{y}^{a(1)} 
                  \cdot \sin(\frac{y}{R})-
                  \frac{1}{\sqrt{\pi R}} A_{\mu}^{a(1)} \frac{1}{R} 
                  \sin(\frac{y}{R}) \nonumber \\
            & + & \frac{g_{5}}{2 \pi R} f^{abc} 
                  \left[A_{\mu}^{b(0)}(x^{\mu})+A_{\mu}^{b(1)}(x^{\mu}) 
                  \sqrt{2} \cos(\frac{y}{R}) \right]
                  \left[A_{y}^{a(1)}(x^{\mu}) 
                  \sqrt{2}\sin(\frac{y}{R}) \right] \nonumber \; .
  \end{eqnarray}
  The second term in (\ref{fmu5nonabelorbi}) will lead to mass terms for the gauge fields 
  $A_{\mu}^{a(1)}$. As in the Abelian case, we compare this result to the ordinary non-Abelian Higgs
  model. The result is analogous to the Abelian case. In particular, the first excited KK-mode
  gauge fields $A_{y}^{a(1)}$ 
  play again the role of Goldstone bosons. Therefore we go to axial gauge 
  \begin{equation} 
   \label{axialgaugenonnabelian}
   A_{y}^{a(1)}=0 \; .
  \end{equation}
  Due to the close analogy of the non-Abelian $S^{1}/\mathbb{Z}_{2}$ continuum orbifold model 
  with truncated KK-mode expansion and the non-Abelian ordinary Higgs
  mechanism, we can conclude that also the non-Abelian truncated $S^{1}/\mathbb{Z}_{2}$ continuum orbifold 
  model is renormalisable. In axial gauge  
  (\ref{fmu5nonabelorbi}) becomes
  \begin{equation}
   \label{fmu5nonabelorbiaxial} 
   F^{a}_{\mu y}=-\frac{1}{\sqrt{\pi R}} A_{\mu}^{a(1)} \frac{1}{R} 
             \sin(\frac{y}{R}) \; .
  \end{equation}

  We insert (\ref{fmununonabelorbi}) and (\ref{fmu5nonabelorbiaxial})
  into the five-dimensional Lagrangian
  (\ref{5dlagrangianorbifoldnonabelain}) and integrate over the circle $S^{1}$. This yields  
  \begin{eqnarray} 
   \label{4dlagninabeorbi}
   \mathcal{L}_{4D} & = & \int_{0}^{2\pi R} \{ -\frac{1}{4} F^{a}_{MN} F^{a MN} \}=
                       \int_{0}^{2\pi R} \{ -\frac{1}{4} F^{a}_{\mu\nu} F^{a\mu\nu}
                       -\frac{1}{2} F^{a}_{\mu y} F^{a\mu y}  \} \; dy  \\
                   & = & -\frac{1}{4} \left( \partial_{\mu} A^{a(0)}_{\nu}-
                       \partial_{\nu} A^{a(0)}_{\mu}
                       +\frac{g_{5}}{\sqrt{2 \pi R}} f^{abc} A^{b(0)}_{\mu} A^{c(0)}_{\nu} \right)^{2} 
                       -\frac{1}{4} \left( \partial_{\mu} A^{a(1)}_{\nu}-
                         \partial_{\nu} A^{a(1)}_{\mu}\right)^{2}  \nonumber \\
                   & + & \frac{1}{2} \; \frac{1}{R^{2}} \left( A_{\mu}^{a(1)} \right)^{2} +
                         \mathcal{L}^{\prime}_{g_{5}}+\mathcal{L}^{\prime}_{g_{5}^{2}}\nonumber \; ,
  \end{eqnarray}
  where $\mathcal{L}_{g_{5}}^{\prime}$ and $\mathcal{L}^{\prime}_{g_{5}^{2}}$ are interaction terms.
  Note that the zero mode has the canonical field strength
  \begin{equation}
   \label{4dggorbi}
   F^{a(0)}_{\mu\nu} =\partial_{\mu}A^{a(0)}_{\nu}-\partial_{\nu}A^{a(0)}_{\mu}
                       +g_{4} f^{abc}A^{b(0)}_{\mu}A^{c(0)}_{\nu} 
  \end{equation}
  if we identify
  \begin{equation}
   \label{relation4deffgaugecoupling5dgaugecoupling}
   g_{4}=\frac{g_{5}}{\sqrt{2 \pi R}} \; ,
  \end{equation}
  where $g_{4}$ is the four-dimensional effective gauge coupling constant. This relation is well-known from
  higher-dimensional gauge theories \cite{Sundrum:2005jf}. Note that while $g_{4}$ is dimensionless,
  $g_{5}$ has mass dimension $-1/2$. The comparison of (\ref{4dggorbi}) and (\ref{4dggbtlm})
  yields the relation 
  \begin{equation}
   g_{4}=\tilde{g}=\frac{g}{\sqrt{2}} \quad \Longrightarrow \quad g=\frac{g_{5}}{\sqrt{\pi R}} \; .
  \end{equation}

  Finally we compare the Lagrangian (\ref{4dlagninabeorbi}) of the non-Abelian truncated
  $S^{1}/\mathbb{Z}_{2}$ continuum orbifold
  model in axial gauge (\ref{axialgaugenonnabelian}) with the Lagrangian 
  of the corresponding eBTLM. First, using
  \begin{gather}
   \int_{0}^{2\pi R} \cos(\frac{y}{R}) \; dy=\int_{0}^{2\pi R} \sin(\frac{y}{R}) \; dy=
   \int_{0}^{2\pi R} \cos^{3}(\frac{y}{R}) \; dy=0 \; ,\\  
   \int_{0}^{2\pi R} \cos^{2}(\frac{y}{R}) \; dy=\pi R \quad , \quad 
   \int_{0}^{2\pi R} \cos^{4}(\frac{y}{R}) \; dy=\frac{3}{4} \pi R 
  \end{gather}
  and inserting (\ref{relation4deffgaugecoupling5dgaugecoupling}) 
  an elementary but lengthy calculation shows that the interaction terms
  $\mathcal{L}^{\prime}_{g_{5}}$ and $\mathcal{L}^{\prime}_{g_{5}^{2}}$ in (\ref{4dlagninabeorbi})
  equal (\ref{linearf}) and (\ref{quadf}) (including the term (\ref{firstkkmodeinteraction})), respectively.
  Second, if we require as in the Abelian case
  \begin{equation}
   \label{kkmoderequirement}
   g \rho_{min}=\frac{1}{R} \; ,
  \end{equation}  
  both Lagrangian's equal and hence both theories describe the same physics. Therefore we conclude
  that also the eBTLM is {\em{renormalisable}}.

 \section[Orbifold conditions for nonunitary parallel transporters]
         {Orbifold conditions for nonunitary parallel transporters in the effective 
          bilayered tranverse lattice model}

    In the last section we have restricted to the case where the orbifold projection $P$ 
    is trivial. In addition, we have made certain assumptions about the minimum $\Phi_{min}$ of the 
    Higgs potential $V(\Phi)$, see (\ref{minimumphiabelian}) and (\ref{mimimumphinonabelian}).
    As a result, the gauge group $G$ remained unbroken and the zero mode gauge fields remained massless. 
    
    In this section we determine orbifold conditions for nonunitary parallel transporters $\Phi$.
    As a result, we can also handle non-trivial minima of the Higgs potential $V(\Phi)$  
    and non-trivial orbifold projections $P$. 
    At first, we recall some standard facts about Lie algebras, which can be found  
    in \cite{Knapp}.

    \begin{theorem}[Cartan decomposition 1]\cite{Knapp}
     Let $H$ be a real semi-simple Lie group with Lie algebra $\mathfrak{h}$. Then $\mathfrak{h}$
     has a Cartan involution $\theta$. A Cartan involution $\theta$ of $\mathfrak{h}$ leads to
     an eigenspace decomposition
     \begin{equation}
      \mathfrak{h}=\mathfrak{g} \oplus \mathfrak{p}
     \end{equation}
     of $\mathfrak{h}$ such that  
     \begin{equation}
      [\mathfrak{g},\mathfrak{g}] \subseteq \mathfrak{g}, \quad
      [\mathfrak{g},\mathfrak{p}] \subseteq \mathfrak{p}, \quad
      [\mathfrak{p},\mathfrak{p}] \subseteq \mathfrak{g}
     \end{equation}
     and $\mathfrak{g},\mathfrak{p}$ are $+1$ and $-1$ eigenspaces of $\theta$, i.e.
     \begin{itemize}
      \item 
       $\theta \; X=X \quad \text{for} \; X \in \mathfrak{g}$ 
      \item
       $\theta \; X=-X \quad \text{for} \; X \in \mathfrak{p}$ \; .
     \end{itemize}
     Let $\kappa$ be the Killing form of $\mathfrak{h}$. Then $\mathfrak{g}$ and $\mathfrak{p}$ are 
     orthogonal under $\kappa$ and $\kappa$ is negative definite on $\mathfrak{g}$ and positive
     definite on $\mathfrak{p}$. 
    \end{theorem}
    Remark: \; If $\mathfrak{h}=\mathfrak{g} \oplus \mathfrak{p}$ is a Cartan decomposition of
    $\mathfrak{h}$ then $\mathfrak{g} \otimes i \mathfrak{p}$ is a compact real form of its
    complexification $(\mathfrak{h})^{\mathbb{C}}$.  

    \begin{theorem}[Cartan decomposition 2]\cite{Knapp}
     Let $H$ be a real semi-simple Lie group with Lie algebra $\mathfrak{h}$, let $\theta$ be a
     Cartan decomposition of its Lie algebra $\mathfrak{h}$ and let 
     $\mathfrak{h}=\mathfrak{g} \oplus \mathfrak{p}$ be the corresponding Cartan decomposition. Suppose
     $H$ has a finite centre, then $G$ is the maximal compact subgroup of $H$ and $G$ has
     Lie algebra $\mathfrak{g}$. The elements of $H$ can be written as
     \begin{equation}
      \label{Cartandecomposition2}
      h=g \exp{X} \; , \quad g \in G \; , X \in \mathfrak{p}
     \end{equation}
     This decomposition is called the global Cartan decomposition. 
    \end{theorem}
    Remark: \; The global Cartan decomposition generalises the polar decomposition of matrices.
    
    \begin{definition}
     \label{definitionmaximalabelianliealgebrainp}
     Let $\mathfrak{a} \subset \mathfrak{p}$ be a maximal Abelian Lie algebra in $\mathfrak{p}$ and let $A$ be
     the corresponding subgroup of $H$. We call $A$ a maximal noncompact Abelian subgroup of $H$.
    \end{definition}
    Remark: \; $A$ is not unique and we will make use of this fact later. \\
    \\
    Let us recall the definition of the Weyl group $W(G,A)$ of the pair $(G,A)$. We use the notations above. 
    Let $W^{\ast}$ be the normaliser of $\mathfrak{a}$ in $G$, i.e. 
    \begin{equation}
     W^{\ast}=\{ g \in G \mid Ad(g) \mathfrak{a} \subset \mathfrak{a} \} \; ,
    \end{equation} 
    where $Ad(g) \mathfrak{a} \subset \mathfrak{a}$ means that for all $x \in \mathfrak{a}$
    we have $g x g^{-1} \in \mathfrak{a}$ 
    \footnote{For matrixgroups the adjoint action $Ad$ can be written as 
    $Ad(g)x=g x g^{-1}$}.
    Let $W$ be the centraliser of $a$ in $G$, i.e.
    \begin{equation}
     W=\{ g \in G \mid Ad(g)x = x \; \text{for all} \; x \in \mathfrak{a} \} \; .
    \end{equation} 
    Their quotient group is the Weyl group 
    \begin{equation}
     W(G,A)=W^{\ast}/W  \; .
    \end{equation}   
    \begin{theorem}[Uniqueness of KAK-decomposition]
     \label{theoremKAKdecomposition}
     Let $H$ be a reductive Lie group, $G$ the maximal compact subgroup of $H$, $A$ a maximal
     noncompact Abelian subgroup of $H$ and $\mathfrak{a}$ the corresponding Abelian subspace in 
     $\mathfrak{p}$. Then every element $h \in H$ admits a decomposition
     \begin{equation}
      \label{KAKdecomposition}
      h=k_{1} a k_{2}^{-1}
     \end{equation}
     where $k_{1},k_{2} \in G$ and $a \in A$. In this decomposition
     \begin{itemize}
      \item
       $a$ is unique up to conjugation with elements of $W(G,A)$ 
      \item 
       Given $a \in A$, let $W_{a}=\{ g \in G \mid g a g^{-1}=a\}$. Then
       $k_{1}$ and $k_{2}$ are unique up to right multiplication by an element of $W_{a}$, i.e.
       \begin{equation}
        k_{1} \to k_{1}^{\prime}=k_{1} k \quad , \quad 
        k_{2} \to k_{2}^{\prime}=k_{2} k 
       \end{equation}
       where $k \in W_{a}$.
     \end{itemize}
    \end{theorem}
    \begin{proof}
     The existence of the decomposition can be found in \cite{Knapp}. It is based on the global Cartan
     decomposition $H=G \exp{\mathfrak{p}}$ and the equality 
     $\mathfrak{p}=\cup_{g \in G} Ad(g) \mathfrak{a}$. 
     Let us come to the uniqueness of (\ref{KAKdecomposition}). First, the proof for
     nonuniqueness of $a$ can be found in \cite{Knapp}. Given now $a \in A$. Suppose
     \begin{equation}
      \label{nonuniqeness}
      k_{1} a k_{2}^{-1}=k_{1}^{\prime} a k_{2}^{\prime -1} \; .
     \end{equation}
     If $\tilde{k}_{1}=k_{1}^{\prime -1} k_{1}$ and $\tilde{k}_{2}=k_{2}^{-1} k_{2}^{\prime}$, then 
     $\tilde{k}_{1} a \tilde{k}_{2}=a$
     and therefore $(\tilde{k}_{1}\tilde{k}_{2}) (\tilde{k}_{2}^{-1} a \tilde{k}_{2} )=a$. 
     By the uniqueness of the global Cartan decomposition (\ref{Cartandecomposition2}) 
     it follows that $\tilde{k}_{1}\tilde{k}_{2}=1$ and $\tilde{k}_{2}^{-1} a \tilde{k}_{2}=a$.
     Thus $\tilde{k}_{2} \in W_{a}$. Then 
     \begin{equation}
      \tilde{k}_{2}=k_{2}^{-1} k_{2}^{\prime} \Longrightarrow  k_{2}^{\prime}=k_{2} \tilde{k}_{2}=
      k_{2} k \; ,
     \end{equation}
     with $k=\tilde{k}_{2} \in W_{a}$. In addition
     \begin{equation}
      \tilde{k}_{1}=k_{1}^{\prime -1} k_{1} \Longrightarrow k_{1}^{\prime}=k_{1} \tilde{k}_{1}^{-1}
      =k_{1} \tilde{k}_{2}=k_{1} k \; ,
     \end{equation}
     because $\tilde{k}_{1}\tilde{k}_{2}=1$. 
    \end{proof}
    
    \begin{definition}
     We call $a \in A$ generic if $W_{a}=W$. 
    \end{definition} 

    \begin{corollary}
     \label{corallarykakdecomposition}
     i) (generic case) \; Given $a \in A$ in (\ref{KAKdecomposition}). If $a$ is generic, it follows that
     \begin{equation}
      k_{1}^{\prime}=k_{1} w \; , \quad k_{2}^{\prime}=k_{2} w
     \end{equation}
     in (\ref{nonuniqeness}) where $w \in W$. \\
     ii) (unitary case) \; Let $a=1$ in (\ref{KAKdecomposition}). Then
     \begin{equation}
      k_{1} k_{2}^{-1}=k_{1}^{\prime} k_{2}^{\prime -1}
     \end{equation}
     allows
     \begin{equation}
      k_{1}^{\prime}=k_{1} g \; , \quad k_{2}^{\prime}=k_{2} g
     \end{equation}
     where $g \in G$ that is $W_{a}=G$ for $a=1$.
    \end{corollary}
    \begin{proof}
     i) We know that $\tilde{k}_{2}^{-1} a \tilde{k}_{2}=a$ where 
     $\tilde{k}_{2}=W_{a}$. Now let $a$ be generic. Thus $\tilde{k}_{2} \in W$. \\
     ii) is obvious since in this case $W_{a}=G$. \\
    \end{proof}

  \subsection{Orbifold conditions for nonunitary parallel transporters}
   
    We now determine orbifold conditions for nonunitary PTs.  
    Let $G$ be the unitary gauge group of the bulk and let $\mathfrak{g}$ be its Lie algebra. Any orbifold
    projection 
    $P \in G$ can be written as an exponential of some Lie algebra element and is therefore contained in some
    $U(1)$ subgroup of $G$. If we start with this $U(1)$ subgroup, we can
    construct a maximal torus $T \subset G$. By $\mathfrak{t}$ we denote the Lie algebra of $T$.
    Let $\mathfrak{t}=\text{Lie} \; T$. Let $\{ H_{i} \}_{i=1}^{r}$, with
    $r=rank \; \mathfrak{g}$, denote the generators
    of $\mathfrak{t}$. Since $P \in T$ by construction we can always write
    \begin{equation}
     P=\exp(-2 \pi i \; \vec{V} \cdot \vec{H}) \; ,
    \end{equation}
    where $\vec{V}$ is a shift vector 
    \footnote{Every possible orbifold projection $P$ can be specified by a corresponding
    shift vector $\vec{V}$. Shift vectors are listed in the literature 
    for many gauge groups, see e.g. \cite{Bouwknegt:1988hn}.}
    and $\vec{H}=(H_{1},H_{2},\dots,H_{r})$.
    The shift vector $\vec{V}$ is an element of the weight space of $\mathfrak{g}$.
    We consider the case where $\mathfrak{g}$ can be obtained from a complex Lie algebra $\mathfrak{h}$, i.e.
    \begin{equation}
     \mathfrak{h}=\mathfrak{g} \oplus i \mathfrak{g} \; .
    \end{equation}
    An important example is $\mathfrak{h}=\mathfrak{sl}(N,\mathbb{C})$. In this case 
    $\mathfrak{g}=\mathfrak{su}(N)$. 

    Let $\mathfrak{a}$ be a maximal Abelian Lie algebra in $i \mathfrak{g}$ and let $A$ be
    its corresponding subgroup in $G$ according to Definition \ref{definitionmaximalabelianliealgebrainp}. 
    The KAK-decomposition (\ref{KAKdecomposition}) holds for any choice of $\mathfrak{a}$. It is natural to
    make the special choice
    \begin{equation}
     \mathfrak{a}=i \mathfrak{t} \; .
    \end{equation}
    Let $\eta \in \mathfrak{a}$. Then we have by construction
    \begin{equation}
     P \eta P^{-1}=\eta  \; .
    \end{equation}
    Thus $P \in W_{e^{\eta}}$.
    \begin{example}
     Let $G=SU(N)$ and $P^{2}=1$. Without loss of generality we can write $P$ as 
     \begin{equation} 
      P=\text{diag}(\underbrace{1,\dots,1}_{N-m},\underbrace{-1,\dots,-1}_{m})
     \end{equation}
     for $1 \le m \le N$ and $m$ is restricted to be an even integer.
     Let $\mathfrak{h}=\mathfrak{sl}(N,\mathbb{C})$. The Cartan decomposition of
     $\mathfrak{sl}(N,\mathbb{C})$ reads
     $\mathfrak{sl}(N,\mathbb{C})=\mathfrak{su}(N) + i \mathfrak{su}(N)$. We are free to choose 
     \begin{equation}
      \mathfrak{a}=\{ \eta=\text{diag}(a_{1},\dots,a_{N}) \} \subset i \mathfrak{su}(N) \; ,
     \end{equation}
     where $\sum a_{i}=0$, $a_{i} \in \mathbb{R}$. It follows that
     \begin{equation}  
      P \eta P^{-1}=\eta  \; .
     \end{equation}
     Thus $P \in W_{e^{\eta}}$.
    \end{example}

    Let $V_{L}(V_{R})$ be the fibre over the $L(R)$-boundary. The parallel transporter $\Phi$
    is a map $\Phi: \; V_{R} \to V_{L}$. In addition, the parallel transporter $\Phi^{\ast}$
    in the backwards direction is a map $\Phi^{\ast}: \; V_{L} \to V_{R}$. If we identify $V_{L}$ 
    and $V_{R}$ via a map $i: V_{L} \to V_{R}$ \cite{Mack:2005fv} $(i^{-1}: V_{R} \to V_{L})$,
    there remains the freedom that $\Phi \in H$ transforms under a unitary gauge transformation 
    according to 
    \begin{equation}
     \Phi \mapsto S(x) \Phi S(x)^{-1} \; ,
    \end{equation}
    where $S(x) \in G$. Hence we can require for $\Phi,\Phi^{\ast} \in H$ and $P \in T \subset G$
    the orbifold condition
    \begin{equation}
     \label{boundaryconditionspt}
     \Phi=P \; \Phi^{\ast} P^{-1} \; ,
    \end{equation}
    where $P^{\ast}=P^{-1}$.  
    According to (\ref{KAKdecomposition}), we can write $\Phi \in H$ as
    \begin{equation}
     \label{kakdecompositionincalcu}
     \Phi=U_{L} \; e^{\eta} \; U_{R}^{\ast} \; , 
    \end{equation}
    where $\eta \in \mathfrak{a}=i \mathfrak{t}$, $\mathfrak{t}=\text{Lie} \; T$ and $U_{L},U_{R} \in G$.
    We insert (\ref{kakdecompositionincalcu}) in (\ref{boundaryconditionspt}). Consequently we
    get for the right-hand side of (\ref{boundaryconditionspt}) 
    \begin{equation}
     \label{righthandsidepphiastp}
     P \Phi^{\ast} P^{-1}=P U_{R} e^{\eta^{\ast}} U_{L}^{\ast} P^{-1}=
     P U_{R} e^{\eta} U_{L}^{\ast} P^{-1}  \; .
    \end{equation}
    The second equation holds since $e^{\eta}$ is selfadjoint. 
    According to Theorem \ref{theoremKAKdecomposition} however, the
    decomposition (\ref{kakdecompositionincalcu}) is not unique. 
    The comparison of (\ref{kakdecompositionincalcu})
    with (\ref{righthandsidepphiastp}) tells us that there is a $K \in W_{e^{\eta}}$ such that
    \begin{gather}
     \label{nonuniquenessofphiorbifoldeq1}
     U_{L}=P U_{R} K  \\
     \label{nonuniquenessofphiorbifoldeq2}
     U_{R}^{\ast}=K^{\ast} U_{L}^{\ast} P^{-1} \; .
    \end{gather}
    Let us consider (\ref{nonuniquenessofphiorbifoldeq2}). We obtain 
    \begin{equation}
     \label{nonuniquenessofphiorbifoldeq3}
     U_{R}^{\ast}=K^{\ast} U_{L}^{\ast} P^{-1} \Longrightarrow U_{R}=P U_{L} K 
     \Longrightarrow U_{L}=P^{-1} U_{R} K^{-1} \; .
    \end{equation}
    We now restrict to involutive $P$. On the orbifold $S^{1}/\mathbb{Z}_{2}$ this 
    restriction is empty because $P$ fulfils already $P^{2}=1$. In addition, because in one extra
    dimension there exits only one orbifold, namely $S^{1}/\mathbb{Z}_{2}$, the following is true for all
    orbifold models in one extra dimension. Since $P^{2}=1$ it follows that $P=P^{-1}$. Then
    (\ref{nonuniquenessofphiorbifoldeq3}) and (\ref{nonuniquenessofphiorbifoldeq1}) are compatible if 
    \begin{equation}
     K^{2}=1  \; .
    \end{equation}
    This result shows that $K$ can be interpreted as a projection.
    Recapitulating, we can write $\Phi \in H$ as
    \begin{equation}
     \label{expressionphi2}
     \Phi=U_{L} \; e^{\eta} \; K U_{L}^{\ast} P^{-1} \; ,
    \end{equation} 
    where $K,P \in W_{e^{\eta}}$ and $K^{2}=P^{2}=1$. 
    This result motivates the following 
    \begin{definition}[Sharpened orbifold condition for nonunitary $\Phi$]
     Let $H$ be a reductive Lie group, $G$ the maximal compact subgroup of $H$, $A$ a maximal
     noncompact Abelian subgroup of $H$ and $\mathfrak{a}$ the corresponding maximal Abelian Lie Algebra in 
     $\mathfrak{p}$. $\Phi \in H$ can be decomposed according to Theorem \ref{theoremKAKdecomposition} as
     \begin{equation}
      \Phi=U_{L} \; e^{\eta} \; U_{R}^{\ast} \; ,
     \end{equation}
     where $\eta \in \mathfrak{a}$, $U_{L},U_{R} \in G$ and $\mathfrak{a}$ is not unique.
     Given $P \in G$. Then we demand that  
     \begin{enumerate}
       \item \begin{equation}
              \label{choiceofa}
              P \eta P^{-1}=\eta 
             \end{equation} 
       \item $\eta, U_{L}, U_{R}$ satisfy the condition
             \begin{equation} 
              U_{R}=P U_{L} P^{-1} 
             \end{equation}
     \end{enumerate}  
     for a suitable choice of $\mathfrak{a}$.
    \end{definition}
    Remark: \; For a complex Lie group $H$ one can choose a maximal Abelian Lie Algebra $\mathfrak{a}$ 
    such that (\ref{choiceofa}) is automatically fulfilled.  

    \begin{corollary}
     If $\Phi$ fulfils the sharpened orbifold condition and $P$ is involutive, then $\Phi$ also
     satisfies $\Phi=P \Phi^{\ast} P^{-1}$.
    \end{corollary}
    \begin{proof}
     If $\Phi$ fulfils the sharpened orbifold condition, we can decompose 
     \begin{equation}
      \Phi=U_{L} \; e^{\eta} \; U_{R}^{\ast} \; ,
     \end{equation}
     where $\eta \in \mathfrak{a}$, $U_{L},U_{R} \in G$, $P \in G$ with $P \eta P^{-1}=\eta$
     and $U_{R}=P U_{L} P^{-1}$.
     We get
     \begin{equation}
      P \Phi^{\ast} P^{-1}=P \; U_{R} \; e^{\eta} \; U_{L}^{\ast} \; P^{-1}
      \stackrel{!}{=} P \; U_{R} \; P^{-1} \; e^{\eta} \; P \; U_{L}^{\ast} \; P^{-1} \; ,
     \end{equation}
     where in the second step we have used that $P \eta P^{-1}=\eta$.
     Since $U_{R}=P U_{L} P^{-1}$ and $P^{2}=1$, we further obtain 
     \begin{equation} 
      P \; U_{R} \; P^{-1} \; e^{\eta} \; P \; U_{L}^{\ast} \; P^{-1}=U_{L} \; e^{\eta} \; U_{R}^{\ast}
      =\Phi \; .
     \end{equation} 
    \end{proof} 

    Let $\Phi \in H$ fulfil the sharpened orbifold condition.
    Then $\Phi$ can be decomposed as
    \begin{equation}
     \label{expressionphigen}
     \Phi=U_{L} e^{\eta} U_{R}^{\ast}=U_{L} \; e^{\eta} \; P U_{L}^{\ast} P^{-1} \; ,
    \end{equation}
    where $P \eta P^{-1}=\eta$. $P$ acts on $G$ through an 
    automorphism on its Lie algebra $\mathfrak{g}$. Let $G_{0}$ be the centraliser of $P$ in $G$.
    Since $P$ is an involutive automorphism $\mathfrak{g}$ splits as
    \begin{equation}
     \label{orbifoldprojectiondecompostion}
     \mathfrak{g}=\mathfrak{g}_{0} \oplus \mathfrak{g}_{1} \; ,
    \end{equation}
    where $[\mathfrak{g}_{0},\mathfrak{g}_{0}] \subseteq \mathfrak{g}_{0},
    [\mathfrak{g}_{0},\mathfrak{g}_{1}] \subseteq \mathfrak{g}_{1},
    [\mathfrak{g}_{1},\mathfrak{g}_{1}] \subseteq \mathfrak{g}_{0}$ and 
    $\mathfrak{g}_{0}=\text{Lie} \; G_{0}$.
    $G_{0}$ is called the orbifold unbroken gauge group.  
    $\mathfrak{g}_{1}$ is the orthogonal complement of $\mathfrak{g}_{0}$ and may be viewed as the
    tangent vector to the coset space $G/G_{0}$.  
    Let $g_{0} \in G_{0}$ and $A_{y} \in \mathfrak{g}_{1}$. 
    Then $g_{0}$ and $A_{y}$ fulfil
    \begin{gather}
     P g_{0} P^{-1}=g_{0} \; , \\
     P A_{y} P^{-1}=-A_{y} \; .
    \end{gather}   
    We can decompose $U_{L}$ (at least in a small neighbourhood of the identity) as
    \begin{equation}
     U_{L}= g_{0} \; e^{A_{y}}
    \end{equation}
    according to the action of $P$ on $G$.
    We insert this decomposition into (\ref{expressionphigen}) and obtain 
    \begin{eqnarray}
     \Phi & = & g_{0} e^{A_{y}}  \; e^{\eta} \; P ( g_{0} e^{A_{y}}  )^{\ast} P^{-1} \\
          & = & g_{0} e^{A_{y}}  \; e^{\eta} \; P e^{A_{y}^{\ast}}  \underbrace{P^{-1}P}_{=1}  
                g_{0}^{\ast} P^{-1} \\
          & = & g_{0} e^{A_{y}} \; e^{\eta} \; P e^{A_{y}^{\ast}} P^{-1} \; g_{0}^{\ast} \\
          & = & g_{0} e^{A_{y}} g_{0}^{-1} \; g_{0} e^{\eta} g_{0}^{-1} \; g_{0} e^{A_{y}}  g_{0}^{-1} \\
          & = & e^{A^{\prime}_{y}} \; e^{\eta^{\prime}} \; e^{A^{\prime}_{y}} \; ,
    \end{eqnarray}
    where 
    \begin{gather}
     A^{\prime}_{y}=g_{0} A_{y} g_{0}^{-1} \in \mathfrak{g}_{1} \\
     \eta^{\prime}=g_{0} \eta g_{0}^{-1} \in  \mathfrak{a}^{\prime}=Ad(g_{0}) \mathfrak{a} \; .
    \end{gather}
    We summarise this result in the following 
    \begin{theorem}
     \label{theoremorbifoldconditionsfornupt}
     Suppose that $\Phi \in H$ fulfils the sharpened orbifold condition.
     Then $\Phi$ can be decomposed as
     \begin{equation}  
      \label{kakdecompositionsharpoc}
      \Phi=U_{L} e^{\eta} U_{R}^{\ast}=U_{L} \; e^{\eta} \; P U_{L}^{\ast} P^{-1} \; ,
     \end{equation}
     where $\eta \in \mathfrak{a}$, $U_{L},U_{R} \in G$, $\mathfrak{a}$ such that $P \eta P^{-1}=\eta$
     and $U_{L},U_{R}$ such that $U_{R}=P U_{L} P^{-1}$. Let $P \in G$ be involutive
     and let $\mathfrak{g}=Lie \; G$ split as
     \begin{equation}
      \mathfrak{g}=\mathfrak{g}_{0} \oplus \mathfrak{g}_{1}
     \end{equation} 
     according to the action of $P$ on $\mathfrak{g}$. Then $\Phi$ can be written as
     \begin{equation}
      \Phi = e^{A_{y}} \; e^{\eta} \; e^{A_{y}} \; ,
     \end{equation}
     and 
     \begin{eqnarray} 
      P A_{y} P^{-1} & = & -A_{y} \; , \\
      P \eta P^{-1}  & = & \eta  \; ,
     \end{eqnarray}
     where $A_{y} \in \mathfrak{g}_{1}$ and $\eta \in \mathfrak{a}$. Since $\Phi$ in $H$ fulfils the
     sharpened orbifold condition $P$ has the property
     \begin{equation}
      P \in W_{e^{\eta}} \; .
     \end{equation}
    \end{theorem}
    Remarks: \; i) \; For unitary $\Phi$, that is $e^{\eta}=1$, we have $W_{e^{\eta}}=G$ and 
    therefore $P \in G$. $\Phi$ can be written as 
    \begin{equation} 
      \Phi = e^{2 A_{y}} \; .
    \end{equation} 
    Thus if $e^{\eta}=1$ we recover the conventional orbifold case. \\
    ii) \; If $a=e^{\eta}$ is generic then Corollary \ref{corallarykakdecomposition} yield
    $W_{e^{\eta}}=W$. Thus $P \in W$. \\

    \begin{example}
     \label{examplesl3c}
     We consider the Lie algebra $\mathfrak{h}=sl(3,\mathbb{C})$. It possesses the Cartan 
     decomposition $sl(3,\mathbb{C})=su(3) \oplus i su(3)$. Let $\mathfrak{a}=\{ \eta=diag(a_{1},a_{2}
     ,a_{3}) \}$ be a maximal Abelian Lie algebra of $\mathfrak{p}$, where $\sum a_{i}=0, \; a_{i} \in 
     \mathbb{R}$. In addition, let $\lambda_{i},i=1,\dots,8$ denote the Gell-Mann matrices 
     generating $SU(3)$. Let $\eta \in \mathfrak{a}$ be generic, i.e. $\eta=(a_{1},a_{2},a_{3})$ where the 
     $a_{i}$ are all distinct. Then $W_{e^{\eta}}=W$ is the torus $T$ consisting of all diagonal matrices
     in $SU(3)$. Since $P \in T$, the orbifold projection $P$ has to be a diagonal matrix.
     For example we can choose $P \in T \subset G$ as 
     \begin{equation}
      P=\exp(i \pi \sqrt{3} \lambda_{8})=diag(-1,-1,1) \; .
     \end{equation}
     Note that this choice for $P$ breaks the unitary gauge group $G=SU(3)$ down to 
     $G_{0}=SU(2) \times U(1)$. 
    \end{example}

   \section{Spontaneous symmetry breaking}

    In this section, we discuss the topic of spontaneous symmetry breaking in detail.
    First we introduce some terminology.
    \begin{definition}
     \label{definitiongeta}
     Let $G_{0 \eta}$ be the centraliser of $\eta \in \mathfrak{a}$ in $G_{0}$, i.e.
     \begin{equation}
      G_{0 \eta}=\{ g \in G_{0} \mid Ad(g) \eta=\eta\} \; .
     \end{equation}
     $G_{0 \eta}$ is called the unbroken subgroup of $G_{0}$ with respect to $\eta$.
    \end{definition}
    In the context of orbifold- and spontaneous symmetry breaking one is usually faced with the 
    situation where first the bulk gauge group $G$ is
    broken by orbifolding to $G_{0}$ at high energies and second $G_{0}$ is broken further spontaneously
    to $G_{0 \eta}$. We schematically write
    \begin{equation}
     G \stackrel{P}{\longrightarrow} G_{0} \stackrel{\eta}{\longrightarrow} G_{0 \eta} \; .
    \end{equation}

    Let us consider the case where $P=1$. For $P=1$
    the unitary gauge group $G$ remains unbroken and $\mathfrak{g}_{0}$ the Lie algebra of $G_{0}$
    equals $\mathfrak{g}$. Suppose $\Phi$ fulfils the sharpened orbifold
    condition. Thus we can write $\Phi$ according to
    Theorem \ref{theoremorbifoldconditionsfornupt} as
    \begin{equation}
     \Phi=U_{L} e^{\eta} U_{L}^{\ast} \; ,
    \end{equation}
    where $\eta \in \mathfrak{a}$, $U_{L},U_{L}^{\ast} \in G$ and $\mathfrak{a}$ such that $P \eta P=\eta$.
    This follows directly from the decomposition (\ref{kakdecompositionsharpoc})  
    with $P=1$. Let us consider the Higgs potential 
    \begin{equation}
     \label{phiunbroken1}
     V(\Phi)=V(U_{L} e^{\eta} U_{L}^{\ast}) \; .
    \end{equation}  
    Since $\Phi$ transforms under a unitary gauge transformation as
    \begin{equation}
     \Phi \mapsto S(x) \Phi S(x)^{-1} \; ,
    \end{equation}
    where $S(x) \in G$, the unitary factors $U_{L},U_{L}^{\ast} \in G$ in (\ref{phiunbroken1})
    can be transformed away.  
    As a consequence the Higgs potential $V(\Phi)$ can be written as a function which depends only
    on $\eta$, and we have
    \begin{corollary}
     \label{coralarytrivialphiggspotential}
     Suppose that the Higgs potential $V(\Phi)$ is $G$-invariant, i.e.
     \begin{equation}
      V(S(x) \Phi S(x)^{-1})=V(\Phi)
     \end{equation}
     for all $S(x) \in G, \Phi \in H$. Then there exists a function $\mathcal{V}$ on
     $\mathfrak{a}$ such that 
     \begin{equation}
      V(\Phi)=V(g e^{\eta} g^{\ast})=\mathcal{V}(\eta) \quad \text{for all} \; g \in G, \eta \in \mathfrak{a}
     \end{equation}
     and
     \begin{equation} 
      \mathcal{V}(\eta)=\mathcal{V}(w(\eta)) \quad \text{for all} \; w \in W(G,A) \; .
     \end{equation}
    \end{corollary}
    Thus $W(G,A)$ is a discrete group of symmetries of the Higgs potential $V(\Phi)$.

    \begin{definition}
     Let $\mathcal{W}_{\eta} \subset W(g,A)$ be the subgroup of elements which leave $\eta$ invariant, i.e.
     \begin{equation}
      \mathcal{W}_{\eta}=\{ \omega \in W(H,A) \mid \omega(\eta)=\eta \} \; .
     \end{equation} 
    \end{definition}
    We note that $\mathcal{W}_{\eta}=(G_{0\eta} \cap W^{\ast})/W$.

    \begin{definition}
     We call $\eta \in \mathfrak{a}$ generic if $G_{0\eta}=W$. 
    \end{definition}
 
    \begin{corollary}
     If $\eta \in \mathfrak{a}$ is generic, then $\mathcal{W}_{\eta}$ is trivial.
    \end{corollary}
    The proof is obvious.
    
    \begin{example}
     \label{examplesu3forsu7}
     We consider the Lie algebra $\mathfrak{h}=sl(3,\mathbb{C})$. It possesses the Cartan 
     decomposition $sl(3,\mathbb{C})=su(3) \oplus i su(3)$. Let $\mathfrak{a}=\{ \eta=diag(a_{1},a_{2}
     ,a_{3}) \}$ be a maximal Abelian Lie algebra of $\mathfrak{p}$, where $\sum a_{i}=0, \; a_{i} \in 
     \mathbb{R}$. In addition, let $\lambda_{i},i=1,\dots,8$ denote the Gell-Mann matrices 
     generating $SU(3)$. \\
     First, let $\eta \in \mathfrak{a}$ be generic, i.e. $\eta=(a_{1},a_{2},a_{3})$ where the $a_{i}$ are
     all distinct.  Then $G_{0\eta}=T=W$ is a Cartan subgroup of $SU(3)$.
     Since $G_{0\eta}=W$, it follows that $\mathcal{W}_{\eta}$ is trivial. \\
     Second, let $\eta \in \mathfrak{a}$ be nongeneric, e.g. $\eta=(a_{1},a_{2},a_{3})$ where 
     $a_{1}=a_{2}$. In this case the unbroken subgroup $G_{0\eta}$ of $G=SU(3)$ is generated 
     by $\{ \lambda_{1},\lambda_{2},\lambda_{3},\lambda_{8} \}$ and consequently 
     $G_{0\eta}=SU(2) \times U(1)$. The Weyl group
     $\mathcal{W}_{\eta}$ is $\mathcal{S}_{2}$, that is the permutation group of the two elements
     $(a_{1},a_{2})$.
    \end{example}
    
    Next we consider the case where $P \neq 1$. Then $\mathfrak{g}=\text{Lie} \; G$ splits as
    \begin{equation}
      \mathfrak{g}=\mathfrak{g}_{0} \oplus \mathfrak{g}_{1}
    \end{equation} 
    according to the action of $P$ on $\mathfrak{g}$. Hence the orbifold unbroken gauge group
    $G_{0}$ has Lie algebra $\mathfrak{g}_{0}$. 
    Suppose $\Phi \in H$ fulfils the sharpened orbifold condition.
    According to Theorem \ref{theoremorbifoldconditionsfornupt}, we can write
    \begin{equation}
     \Phi = e^{A_{y}} \; e^{\eta} \; e^{A_{y}} \; ,
    \end{equation}
    and
    \begin{eqnarray} 
     P A_{y} P^{-1} & = & -A_{y} \; , \\
     P \eta P^{-1}  & = & \eta \; ,
    \end{eqnarray}
    where $A_{y} \in \mathfrak{g}_{1}$ and $\eta \in \mathfrak{a}$. 
    We consider the Higgs potential 
    \begin{equation}
     \label{higgspotnontrivailp}
     V(\Phi)=V(e^{A_{y}} \; e^{\eta} \; e^{A_{y}}) \; .
    \end{equation} 
    Since $G$ is broken to $G_{0}$, $\Phi$ transforms under a unitary gauge transformation as
    \begin{equation}
     \Phi \mapsto S_{0}(x) \Phi S_{0}(x)^{-1} \; ,
    \end{equation}
    where $S_{0}(x) \in G_{0}$. In contrast to the case where $P=1$, the unitary factors $e^{A_{y}}$ in
    (\ref{higgspotnontrivailp}) cannot be gauged away due to the lack of gauge invariance. 
    Therefore the Higgs potential $V(\Phi)$ depends also $A_{y}$ and we have 
    \begin{theorem}
     \label{theoremnuhiggspotential}
     Suppose that $\Phi \in H$ can be written as
     \begin{equation}
      \Phi = e^{A_{y}} \; e^{\eta} \; e^{A_{y}} \; ,
     \end{equation}
     and
     \begin{eqnarray} 
      P A_{y} P^{-1} & = & -A_{y}  \; ,\\
      P \eta P^{-1}  & = & \eta \; ,
     \end{eqnarray}
     where $A_{y} \in \mathfrak{g}_{1}$, $\eta \in \mathfrak{a}$ for a suitable choice of $\mathfrak{a}$. The 
     action of $P \in W_{e^{\eta}}$ leads to a split $
     \mathfrak{g}=\mathfrak{g}_{0} \oplus \mathfrak{g}_{1}$,
     where the orbifold unbroken gauge group $G_{0}$ has Lie algebra $\mathfrak{g}_{0}$. Then
     \begin{enumerate}
      \item 
       the Higgs potential $V(\Phi)$ is $G_{0}$-invariant 
       \begin{equation}
        \label{theoremreducedgaugeinvariance}
        V(S_{0}(x) \Phi S_{0}(x)^{-1})=V(\Phi)
       \end{equation}
       for all $S_{0}(x) \in G_{0}, \Phi \in H$.
      \item
       there exists a function $\mathcal{V}$ on 
       $\mathfrak{a} \times  \mathfrak{g}_{1}$ such that 
       \begin{equation}
        \label{theoremhiggspotentialetapluswilsonlines}
        V(\Phi)=V(e^{A_{y}} \; e^{\eta} \; e^{A_{y}})
        =\mathcal{V}(\eta,A_{y}) \; . 
       \end{equation}
       Due to (\ref{theoremreducedgaugeinvariance}) we have
       \begin{equation}
        \mathcal{V}(\eta^{\prime},A^{\prime}_{y})=\mathcal{V}(\eta,A_{y}) \; ,
       \end{equation}
       when $\eta^{\prime} \in \mathfrak{a}^{\prime}=Ad(g_{0}) \mathfrak{a}$,
       $A_{y}^{\prime}=g_{0} A_{y} g_{0}^{-1} \in \mathfrak{g}_{1}$ for some $g_{0} \in G_{0}$.
      \item
       in (\ref{theoremhiggspotentialetapluswilsonlines}) $A_{y}$ cannot be gauged away because
       $S_{0}(x)$ in (\ref{theoremreducedgaugeinvariance}) is restricted to $G_{0} \subset G$.
     \end{enumerate}
    \end{theorem}

 \section{The customary approximation scheme of a truncated $S^{1}/\mathbb{Z}_{2}$ orbifold model}

  In section \ref{sectioneBTLM} we have obtained
  that for a non-Abelian gauge theory with gauge group $G=SU(N)$ and trivial orbifold 
  projection $P$
  \begin{enumerate}
   \item an eBTLM with minimum of the Higgs potential $V(\Phi)$ of the form 
         \begin{equation}
          \label{miniumphiebtlm}
          \Phi_{min}=\rho_{min} \; \frac{1}{\sqrt{2}} \mathbf{1}_{N} \; ,
         \end{equation} 
         see (\ref{mimimumphinonabelian}), 
         leads to the common mass term 
         \begin{equation}
          m=g \rho_{min} 
         \end{equation}
         for all first excited KK mode gauge fields.
   \item such an eBTLM with $\Phi_{min}$ given by (\ref{miniumphiebtlm}) equals a $S^{1}/\mathbb{Z}_{2}$ 
         continuum orbifold model with trivial orbifold projection and a Fourier
         mode expansion truncated for all fields at the first Kaluza-Klein mode, if we require  
         \begin{equation}
          g \rho_{min}=\frac{1}{R} \; ,
         \end{equation}
         see (\ref{kkmoderequirement}).
  \end{enumerate}
  Note that in this equation $g$ is the dimensionless four-dimensional effective gauge coupling constant of 
  the eBTLM and thus $\rho_{min}$ has mass dimension $1$. 

  We will now derive
  (\ref{miniumphiebtlm}). In the following calculations we restrict ourselves for simplicity to
  the bulk gauge group $G=SU(2)$.
  However, the results of this section can be generalised to $G=SU(N)$ in a straightforward
  way.  
  For $G=SU(2)$, we assume that the holonomy group is given by
  $H=\mathbb{R}_{\ast}^{+} \; SL(2{,}\mathbb{C})$, with $\mathbb{R}_{\ast}^{+}=\mathbb{R}^{+} / \{0\}$.
  Let $\Phi \in H$ fulfil the sharpened orbifold condition. Then $\Phi \in H$ can be written
  according to Theorem \ref{theoremorbifoldconditionsfornupt} as
  \begin{equation}
   \label{expansionphiunbrokensu2}
   \Phi=\rho \; U_{L} e^{\eta} U_{R}^{\ast}=\rho \; U_{L} \; e^{\eta} \; P U_{L}^{\ast} P^{-1} \; ,
  \end{equation}
  where $\eta \in \mathfrak{a}$, $U_{L},U_{R} \in SU(2)$, $\rho \in \mathbb{R}_{\ast}^{+}$,
  $\mathfrak{a}$ such that $P \eta P^{-1}=\eta$
  and $U_{L},U_{R}$ such that $U_{R}=P U_{L} P^{-1}$. We have to make a choice for
  $\mathfrak{a}$. Since $H$ is complex we can always choose a maximal Abelian Lie algebra 
  $\mathfrak{a} \subset i \mathfrak{su}(2)$, such that $P \eta P^{-1}=\eta$ is automatically fulfilled.
  Without loss of generality suppose $P$ is diagonal. Then
  \begin{equation}
   \label{maximalablianunbrokensu2}
   \mathfrak{a}=\{ \eta=diag(a_{1},a_{2}) \} \quad , \quad a_{1}=-a_{2} \; , \quad a_{i} \in 
     \mathbb{R}
  \end{equation}
  is a maximal Abelian Lie algebra of $i \mathfrak{su(2)}$ and $P \eta P^{-1}=\eta$ for all $\eta 
  \in \mathfrak{a}$. We first focus on the case where $P$ is trivial, i.e. $P=\text{diag}(1,1)$.
  Hence (\ref{expansionphiunbrokensu2}) reads
  \begin{equation}
   \Phi=\rho \; U_{L} \; e^{\eta} \; U_{L}^{\ast} \; .
  \end{equation}
  Let us consider the Higgs potential 
  \begin{equation}
   \label{Higgspotentialunbrokensu2}
   V(\Phi)=V(\rho \; U_{L} \; e^{\eta} \; U_{L}^{\ast}) \; .
  \end{equation}
  $\Phi$ transforms under a unitary gauge transformation according to
  $\Phi \to S(x) \Phi S(x)^{-1}$, where $S(x) \in SU(2)$. Consequently the unitary factor $U_{L}$ in 
  (\ref{Higgspotentialunbrokensu2}) can be transformed away. Thus, according to 
  Corollary \ref{coralarytrivialphiggspotential}, the Higgs potential $V(\phi)$ depends only on
  $\eta$ and $\rho$
  \begin{equation}
   \label{fieldphiunbrokensu2}
   V(\rho \; U_{L} \; e^{\eta} \; U_{L}^{\ast})=V(\rho \; e^{\eta})=\mathcal{V}(\rho,\eta) \; .
  \end{equation}
  Let $V(\Phi)$ assume its minimum at $\Phi_{min}$. According to 
  (\ref{fieldphiunbrokensu2}) we can parametrise any $\Phi_{min}$ as   
  \begin{equation}
   \label{minimumofphi1}
   \Phi_{min}=\rho_{min} \frac{1}{\sqrt{2}} \left( \begin{array}{cc} e^{a_{1}} & 0 \\ 
                      0 & e^{a_{2}} \end{array} \right)
   \; , \; a_{1}=-a_{2} \; , \; a_{i} \in \mathbb{R} \; .
  \end{equation}
  Note that the $a_{i}$ in (\ref{minimumofphi1}) are dimensionless parameters.
  Since $P$ is trivial, the bulk gauge group $G=SU(2)$ remains unbroken at both boundaries. 
  Figure \ref{figureunbrokensu2} summarises the setting.
  \begin{figure}
   \begin{equation*}
   \begin{picture} (18,4.5)
   \thicklines
   \multiput(2,1.5) (0,2.5) {2} {\line (1,0) {9}}
   \thinlines
   \multiputlist(11.5,1.5)(0,2.5){$L$,$R$}
   \multiputlist(6.5,4.5)(7,0){$G_{0}=SU(2)$}
   \multiputlist(6.5,2.6)(2.5,0){$G=SU(2) \subset H=\mathbb{R}_{\ast}^{+} \; SL(2{,}\mathbb{C})$}
   \multiputlist(6.5,1)(7,0){$G_{0}=SU(2)$}
   \put(2.5,2.5) {$\Phi$}
   \put(10.3,2.5) {$\Phi^{\ast}$}
   \matrixput (3,1.5)(7,0){2}(0,2,5){2}{\circle*{0.2}}
   \thicklines
   \dottedline[\circle*{0.05}]{0.1}(3,1.5)(3,4)
   \dottedline[\circle*{0.05}]{0.1}(10,1.5)(10,4)
   \thinlines
   \put(2.82,3.65){\Pisymbol{pzd}{115}}
   \put(9.82,1.58){\Pisymbol{pzd}{116}}
   \end{picture}
   \end{equation*}
   \caption{Effective bilayered transverse lattice model for bulk gauge group $G=SU(2)$ and
            trivial orbifold projection $P=\text{diag}(1,1)$.} 
   \label{figureunbrokensu2}
  \end{figure}
  In order to arrive at (\ref{miniumphiebtlm}) for $N=2$, we set
  $a_{1}=-a_{2}=0$ and (\ref{minimumofphi1}) becomes
  \begin{equation}
    \label{trivialmimphical}
    \Phi_{min}=\rho_{min} \frac{1}{\sqrt{2}} \left( \begin{array}{cc} 1 & 0 \\ 
                       0 & 1 \end{array} \right) \; .
  \end{equation}
  We know from section \ref{sectioneBTLM} that (\ref{trivialmimphical}) leads to the mass term
  \begin{equation}
   \label{masstermsu2trivialphi}
   \mathcal{L}_{mass}=tr \left[ \left( D_{\mu}\Phi_{min} \right)^{\dagger} 
   \left( D_{\mu}\Phi_{min} \right) \right]
   =\frac{1}{2} \;  g^{2} \rho_{min}^{2}  \left( A^{i(1)}_{\mu} \right)^{2} \; ,
  \end{equation}
  with $i=1,2,3$, i.e. only the first excited KK mode gauge fields $A^{i(1)}_{\mu}$ 
  become massive with common mass  
  $m=g \rho_{min}$. All zero mode gauge fields $A^{i(0)}_{\mu}$ remain massless. 
  Therefore we make the following 
  \begin{definition}
   \label{defintiontrivailminphi} 
   We call
   \begin{equation} 
    \label{trivialminumumsu2}
    \Phi_{min}=\rho_{min} \frac{1}{\sqrt{2}} \left( \begin{array}{cc} 1 & 0 \\ 
                       0 & 1 \end{array} \right)
   \end{equation}
   the trivial minimum of the Higgs potential $V(\Phi)$ for $G=SU(2)$.
  \end{definition}
  From section \ref{sectionenonabBTLM} we also know that
  a truncated $S^{1}/\mathbb{Z}_{2}$ orbifold model with bulk gauge group $G=SU(2)$ and trivial 
  orbifold projection $P=\text{diag}(1,1)$ leads to the mass term (\ref{4dlagninabeorbi})
  \begin{equation}
   \label{orbiunbrokensu21}
   \mathcal{L}_{mass}=\frac{1}{2} \frac{1}{R^{2}} \left( A_{\mu}^{i(1)} \right)^{2} \; ,
  \end{equation}
  with $i=1,2,3$. If we insert $g \rho_{min}=1/R$ in
  (\ref{masstermsu2trivialphi}), (\ref{orbiunbrokensu21})
  and (\ref{masstermsu2trivialphi}) coincide. In fact we know from the discussion in section
  \ref{sectionenonabBTLM} that also the effective 
  four-dimensional Lagrangian of the truncated $S^{1}/\mathbb{Z}_{2}$ continuum orbifold model equals the 
  effective four-dimensional Lagrangian of the corresponding eBTLM.
  Therefore we conclude 
  \begin{proposition}
   An $S^{1}/\mathbb{Z}_{2}$ continuum orbifold model with bulk gauge group $G=SU(2)$, trivial 
   orbifold projection $P=\text{diag}(1,1)$ and a Fourier mode expansion for all gauge fields
   truncated at the first excited Kaluza-Klein mode in axial gauge equals an effective 
   bilayered transverse lattice model with bulk gauge group $G=SU(2)$, trivial 
   orbifold projection $P=\text{diag}(1,1)$ and trivial minimum of the Higgs potential $V(\Phi)$.
  \end{proposition}
  Remark: i) \; In general the minimum $\Phi_{min}$ of the Higgs potential is given by
  \begin{equation}
   \label{minimumofphi}
   \Phi_{min}=\rho_{min} \frac{1}{\sqrt{2}} \left( \begin{array}{cc} e^{a_{1}} & 0 \\ 
                      0 & e^{a_{2}} \end{array} \right) \; ,
  \end{equation} 
  where $a_{1}=-a_{2} \neq 0$. Therefore an eBTLM
  with nonunitary parallel transporter $\Phi$ is richer in its physical content than a truncated
  $S^{1}/\mathbb{Z}_{2}$ continuum orbifold model.

 \section[Exponential gauge boson masses]
         {Beyond the customary approximation scheme of a truncated $S^{1}/\mathbb{Z}_{2}$ orbifold
          model: Exponential gauge boson masses}

  \label{sectionunbrokensu2}
  
  In this section we  consider the case where 
  the minimum $\Phi_{min}$ of the Higgs potential $V(\Phi)$ is non-trivial, i.e. 
  \begin{equation}
   \Phi_{min}=\rho_{min} \frac{1}{\sqrt{2}} \left( \begin{array}{cc} e^{a_{1}} & 0 \\ 
                      0 & e^{a_{2}} \end{array} \right) \; , a_{1}=-a_{2} \neq 0 \; .
  \end{equation}
  We calculate the mass term for the $SU(2)$ gauge bosons by computing  
  the kinetic term 
  \begin{equation}
   \mathcal{L}_{mass}=tr \left[ \left(D_{\mu}\Phi_{min} \right)^{\dagger} 
                         \left( D_{\mu}\Phi_{min} \right) \right] \; .
  \end{equation}

  We start with the covariant derivative
  \begin{equation}
   \label{covariantderivaticesu2}
   D_{\mu}\Phi=\partial_{\mu}\Phi+i g \left( A^{R}_{\mu}\Phi
                -\Phi A^{L}_{\mu} \right) \; ,
  \end{equation} 
  where 
  \begin{equation}
   A^{R}_{\mu}= A^{R i}_{\mu} t_{i}  \quad , \quad
   A^{L}_{\mu}= A^{L i}_{\mu} t_{i}  \; ,
  \end{equation}
  $t_{i}=\frac{1}{2} \tau_{i}$ and $\tau_{i}$ denote the Pauli matrices. 
  Note that $\text{tr} \left( t_{i} t_{j} \right)=\frac{1}{2} \delta_{ij}$.
  We transform
  \begin{gather}
   \label{masseigenstatesunbrokensu2}
   A^{R i}_{\mu}=\frac{1}{\sqrt{2}} \left( A^{i(0)}_{\mu}+A^{i(1)}_{\mu} \right) \\
   A^{L i}_{\mu}=\frac{1}{\sqrt{2}} \left( A^{i(0)}_{\mu}-A^{i(1)}_{\mu} \right) \; .
   \nonumber
  \end{gather}
  Recall that $A^{i(0)}_{\mu}$ and $A^{i(1)}_{\mu}$ denote mass eigenstates.
  The covariant derivative (\ref{covariantderivaticesu2})
  reads in terms of these mass eigenstates 
  \begin{equation}
   \label{covariantderivaticesu2transform}
   D_{\mu}\Phi=\partial_{\mu}\Phi+i \frac{g}{\sqrt{2}}  A^{i(0)}_{\mu} 
               \left[t_{i}, \Phi \right]+i \frac{g}{\sqrt{2}}  A^{i(1)} \{ t_{i}, \Phi \} \; ,
   \end{equation}
  where $[,]$ and $\{ , \}$ denote the commutator and  anticommutator, respectively.
  In order to calculate $\left[t_{i}, \Phi_{min} \right]$ and $\{ t_{i}, \Phi_{min} \}$, respectively,
  it is convenient to add the generator $t_{0}=\frac{1}{2} \mathbf{1}_{2}$ to the generators of $SU(2)$.  
  The set $\{t_{i}\},i=0,\dots,3$ is a basis of the Lie algebra
  $u(2)$ of $U(2)$. We can expand every diagonal
  $2 \times 2$ matrix $\phi$ is terms of $t_{3}$ and $t_{0}$ as
  \begin{equation}
   \phi=\phi_{0} t_{0} + \phi_{3} t_{3}=\frac{1}{2} \; \left(
        \begin{array}{cc} \phi_{0}+\phi_{3} & 0 \\ 0 & \phi_{0}-\phi_{3} 
        \end{array} \right) \; .
  \end{equation}
  Using this expansion we rewrite $\Phi_{min}$ as  
  \begin{equation}
   \Phi_{min}=\rho_{min} \frac{1}{\sqrt{2}} \left( \begin{array}{cc} e^{a_{1}} & 0 \\ 
                      0 & e^{a_{2}} \end{array} \right)
             =:\frac{1}{2} \; \left( \begin{array}{cc} \phi_{0}+\phi_{3} & 0 \\ 
               0 & \phi_{0}-\phi_{3} \end{array} \right) \; ,
  \end{equation}
  where
  \begin{gather}
   \label{parametrizationofphisu2}
   \phi_{0}=\rho_{min} \frac{1}{\sqrt{2}} \left( e^{a_{1}} + e^{a_{2}} \right) \\
   \phi_{3}=\rho_{min} \frac{1}{\sqrt{2}} \left( e^{a_{1}} - e^{a_{2}} \right) \; .
  \end{gather}
  For the commutators and anticommutators we obtain
  \begin{gather}
   \left[ t_{i}, \Phi_{min} \right]
   =\left[ t_{i}, \phi_{0} t_{0} + \phi_{3} t_{3} \right]
   =\phi_{0} \underbrace{\left[ t_{i}, t_{0} \right]}_{=0}
   +\phi_{3} \underbrace{\left[ t_{i}, t_{3} \right]}_{=i 
    \epsilon_{i3k}t_{k}} \; , \\
   \{ t_{i}, \Phi_{min} \}
   =\{ t_{i}, \phi_{0} t_{0} + \phi_{3} t_{3} \}
   =\phi_{0} \underbrace{\{ t_{i}, t_{0} \}}_{=t_{i}}
   +\phi_{3} \underbrace{\{ t_{i}, t_{3} \}}_{\delta_{i3}t_{0}} \; .
  \end{gather}
  Inserting $\left[ t_{i}, \Phi_{min} \right]$ and $ \{ t_{i}, \Phi_{min} \}$ 
  in (\ref{covariantderivaticesu2transform}) yields 
  \begin{equation}
   D_{\mu}\Phi_{min}=-\frac{g}{\sqrt{2}} \; A^{i(0)}_{\mu} \; \phi_{3} \epsilon_{i3k} t_{k}
               +i \frac{g}{\sqrt{2}}  \;  A^{i(1)}_{\mu} \; \left(
               \phi_{0} t_{i} + \phi_{3} t_{0} \delta_{i3} \right) \; .
  \end{equation}
  Taking the adjoint $(D_{\mu}\Phi_{min})^{\dagger}=- \frac{g}{\sqrt{2}} \; A^{i(0)}_{\mu} \; 
  \phi_{3} \epsilon_{i3k} t_{k} - i \frac{g}{\sqrt{2}} \;  A^{i(1)}_{\mu}  \left(
  \phi_{0} t_{i} + \phi_{3} t_{0} \delta_{i3} \right)$ and multiplying $(D_{\mu}\Phi_{min})^{\dagger}$
  by $D_{\mu}\Phi_{min}$ we obtain
  \begin{eqnarray}
   && \left( D_{\mu}\Phi_{min} \right)^{\dagger} \left( D_{\mu}\Phi_{min} \right) \nonumber \\
   & = &
   \label{zeromodetermsu2}
   \frac{1}{2} g^{2} \; A^{i(0)}_{\mu}  A^{\tilde{i}(0)}_{\mu} \;
   \phi_{3}^{2} \epsilon_{i3k} \epsilon_{\tilde{i}3\tilde{k}}  t_{k}
   t_{\tilde{k}} \\
   & + &
   \label{firstexictedmodetermsu2}
   \frac{1}{2} g^{2} \; A^{i(1)}_{\mu}  A^{\tilde{i}(1)}_{\mu} \;
   \left( \phi_{0}^{2} t_{i} t_{\tilde{i}} + \frac{1}{2} \phi_{0} \phi_{3} 
   \left( t_{i} \delta_{i3} \delta_{\tilde{i}3} \right) + 
   \phi_{3}^{2} t_{0}^{2} \delta_{i3} \delta_{\tilde{i}3} \right) \\
   & + &
   \label{zerofirstmodemixedtermsu2}
   i \; \frac{1}{2} g^{2} \; A^{i(0)}_{\mu}  A^{\tilde{i}(1)}_{\mu} \;
   \phi_{3}  \phi_{0} \; \epsilon_{i3k} \left[ t_{k}, t_{\tilde{i}} \right] 
  \end{eqnarray}
  with $i,\tilde{i}=1,2,3$. As $\text{tr} \; t_{i}=0$ for $i=1,2,3$, the mixed term
  (\ref{zerofirstmodemixedtermsu2}) vanish after taking the trace.

  Let us focus on the mass term (\ref{zeromodetermsu2}) for the zero mode
  $A^{i(0)}_{\mu}$ 
  \begin{equation}
    \frac{1}{2} g^{2} \; A^{i(0)}_{\mu}  A^{\tilde{i}(0)}_{\mu} \;
    \phi_{3}^{2} \epsilon_{i3k} \epsilon_{\tilde{i}3\tilde{k}}  t_{k}
   t_{\tilde{k}} \; .
  \end{equation}
  Since 
  $\text{tr} \left( t_{i} t_{j} \right)=\frac{1}{2} \delta_{ij}$
  we get after taking the trace
  \begin{equation}
   \frac{1}{4} g^{2} \; A^{i(0)}_{\mu}  A^{\tilde{i}(0)}_{\mu} \;  \phi_{3}^{2} \;
   \left( \epsilon_{i3k} \epsilon_{\tilde{i}3k} \right) \; . 
  \end{equation}
  For $i=3$ this term vanishes and thus the corresponding gauge field $A^{3(0)}_{\mu}$ 
  remains massless. With $\epsilon_{132}^{2}=\epsilon_{231}^{2}=1$ we obtain for $i=1,2$ 
  \begin{equation}
   \frac{1}{4} g^{2} \; \left( A^{i(0)}_{\mu}  \right)^{2} \; \phi_{3}^{2} \; .
  \end{equation}

  Next we consider the mass term (\ref{firstexictedmodetermsu2})
  for the first excited mode $A^{i(1)}_{\mu}$ 
  \begin{equation}
   \frac{1}{2} g^{2} \; A^{i(1)}_{\mu}  A^{\tilde{i}(1)}_{\mu} \;
   \left( \phi_{0}^{2} t_{i} t_{\tilde{i}} + \frac{1}{2} \phi_{0} \phi_{3} 
   \left( t_{i} \delta_{i3} \delta_{\tilde{i}3} \right) + 
   \phi_{3}^{2} t_{0}^{2} \delta_{i3} \delta_{\tilde{i}3} \right) \; .
  \end{equation}
  The second term vanishes after taking the trace. With $\text{tr} \; t_{0}^{2}=\frac{1}{2}$ and 
  $\text{tr} \left( t_{i} t_{j} \right)=\frac{1}{2} \delta_{ij}$ for $i,j \in \{1,2,3\}$ we obtain
  \begin{equation}
    \label{firstexictedmodemasstermssu2}
    \frac{1}{4} g^{2} \; \left( A^{j(1)}_{\mu} \right)^{2} \;
    \phi_{0}^{2}
   +\frac{1}{4} g^{2} \; \left( A^{3(1)}_{\mu} \right)^{2} \;
    \left( \phi_{0}^{2} + \phi_{3}^{2} \right) \; ,
  \end{equation}
  where $j=1,2$. Recapitulating we have obtained
  \begin{eqnarray}
   \label{masstermsu2p0p3}
   tr\left[ \left( D_{\mu}\Phi_{min} \right)^{\dagger} 
     \left( D_{\mu}\Phi_{min} \right) \right] & = &
     \frac{1}{4} g^{2} \; \left( A^{i(0)}_{\mu} \right)^{2} \phi_{3}^{2} +  
     \frac{1}{4} g^{2} \; \left( A^{i(1)}_{\mu} \right)^{2} 
     \phi_{0}^{2}  \nonumber \\
   & + &
    \frac{1}{4} g^{2} \; \left( A^{3(1)}_{\mu} \right)^{2} 
    \left( \phi_{0}^{2} + \phi_{3}^{2} \right)
   \nonumber
  \end{eqnarray}
  with $i=1,2$. Inserting $\phi_{0}$  and $\phi_{3}$ 
  (\ref{parametrizationofphisu2}) into (\ref{masstermsu2p0p3}) we 
  get the final result
  \begin{eqnarray}
   \label{masstermsu2finalresult}
   tr\left[ \left( D_{\mu}\Phi_{min} \right)^{\dagger} 
   \left( D_{\mu}\Phi_{min} \right) \right] & = &
   \frac{1}{8} g^{2} \rho^{2}_{min} \left( e^{a1}-e^{a2} \right)^{2} \; 
   \left( A^{i(0)}_{\mu} \right)^{2} \nonumber \\ & + &   
   \frac{1}{8} g^{2} \rho^{2}_{min} \left( e^{a1}+e^{a2} \right)^{2} \; 
   \left( A^{i(1)}_{\mu} \right)^{2} \nonumber \\
   & + &
   \frac{1}{4} g^{2} \rho^{2}_{min} \left( e^{2a1}+e^{2a2} \right) \; 
   \left( A^{3(1)}_{\mu} \right)^{2}
  \end{eqnarray}
  with $i=1,2$. Table \ref{tableunbrokensu2} summarises the result
  \begin{tableindoc}
   \label{tableunbrokensu2}
   \begin{equation}
    \label{masstermsu2finalresultabular}
    \begin{array}{|c|c|} \hline
     \text{Field} & \text{Mass squared} \\ \hline
      A^{i(0)}_{\mu} \quad i=1,2 
      & \frac{1}{8} g^{2} \rho^{2}_{min} \left( e^{a_{1}}-e^{a_{2}} \right)^{2} \\ \hline
      A^{3(0)}_{\mu} & 0 \\ \hline
      A^{i(1)}_{\mu} \quad i=1,2 
      & \frac{1}{8} g^{2} \rho^{2}_{min} \left( e^{a_{1}}+e^{a_{2}} \right)^{2} \\ \hline
      A^{3(1)}_{\mu} & \frac{1}{4} g^{2} \rho^{2}_{min} 
      \left( e^{2a_{1}}+e^{2a_{2}} \right) \\
      \hline
    \end{array} 
   \end{equation}
  \end{tableindoc}
  {\em{Discussion}}:   
  i) \; We observe that only  the zero mode gauge field $A^{3(0)}_{\mu}$ remains massless.
  This follows from the
  fact that $t_{3}$ commutes with $\Phi_{min}$
  \begin{equation}
   \left[ t_{3}, \Phi_{min} \right]=0 \; .
  \end{equation}
  Thus the $U(1)$ subgroup of $SU(2)$ generated by $t_{3}$ remains always unbroken. Note that
  also $\left[ P, t_{3} \right]=0$. We have the
  spontaneous symmetry breaking scheme
  \begin{equation}  
   SU(2) \stackrel{\langle \eta \rangle}{\longrightarrow} U(1)
  \end{equation}
  for $a_{1}=-a_{2} \neq 0$. \\
  ii) \; For $a_{1}=-a_{2}=0$ in (\ref{masstermsu2finalresultabular}) 
  we recover (\ref{masstermsu2trivialphi}). \\
  iii) \; For (\ref{masstermsu2finalresultabular}) there are two cases of special interest:
  \begin{enumerate}
   \item Limit of small $a_{1}$, i.e. $0 < a_{1} \ll 1$:. In this case it is possible to find a corresponding
         orbifold model as an approximation. We will discuss this case in detail in the next section.
   \item Limit of large $a_{1}$, i.e. $a_{1} \gg 1$: In this case gauge boson masses can be
         very large in comparison to the compactification scale $g \rho_{min}=1/R$. This behaviour has
         no counterpart within the customary approximation scheme of an orbifold model.
         We will discuss this case in detail in section \ref{sectionlargegaugebosonmasses}.
  \end{enumerate}

 \section[Linear approximation and corresponding truncated orbifold model]{Linear approximation and
          corresponding truncated $S^{1}/\mathbb{Z}_{2}$ orbifold model
          with an additional scalar field in the adjoint representation of $SU(2)$}

    Let $0 < a_{1} \ll 1$  in (\ref{masstermsu2finalresultabular}).
    Thus we can approximate
    \begin{equation}
     e^{a_{i}} \approx 1+a_{i} 
    \end{equation}
    for $i=1,2$.
    In this approximation we obtain
    \begin{gather} 
     \label{limitofsmallaiapprox}
     e^{a_{1}}-e^{a_{2}}=a_{1}-a{_2}=2 a_{1} \; , \\
     e^{a_{1}}+e^{a_{2}}=a_{1}+a_{2}+2=2 \nonumber \; ,
    \end{gather}
    where we have used that $a_{1}=-a_{2}$.
    Inserting (\ref{limitofsmallaiapprox}) in (\ref{masstermsu2finalresultabular}) we obtain
    \begin{tableindoc}
     \begin{equation}
      \label{masstermsu2finalresultabularsmalla1}
      \begin{array}{|c|c|} \hline
       \text{Field} & \text{Mass squared} \\ \hline
        A^{i(0)}_{\mu} \quad i=1,2 
        & \frac{1}{2} g^{2} \rho^{2}_{min} \; a_{1}^{2} \\ \hline
        A^{3(0)}_{\mu} & 0 \\ \hline
        A^{i(1)}_{\mu} \quad i=1,2 
        & \frac{1}{2} g^{2} \rho^{2}_{min} \\ \hline
        A^{3(1)}_{\mu} & \frac{1}{2} g^{2} \rho^{2}_{min} \\ \hline
       \end{array} 
     \end{equation}
    \end{tableindoc}
    {\em{Discussion}}: i) \; The zero mode gauge fields $A^{i(0)}_{\mu}$, for $i=1,2$, 
    get small masses in comparison to the compactification scale $g \rho_{min}=1/R$. \\
    ii) \; The first excited KK-mode gauge fields $A^{i(1)}_{\mu}$ get the common mass term 
    $g \rho_{min}=1/R$. This result is just what one would expect from the customary
    approximation scheme of a truncated $S^{1}/\mathbb{Z}_{2}$ continuum orbifold model.

    (\ref{masstermsu2finalresultabularsmalla1}) suggests that there exits a corresponding
    $S^{1}/\mathbb{Z}_{2}$ orbifold model which at least approximately describes an eBTLM in the limit
    of small $a_{2}$. In fact, let us
    consider a $S^{1}/\mathbb{Z}_{2}$ continuum orbifold model with bulk gauge group 
    $G=SU(2)$. In addition, we introduce a bulk scalar field $\phi(x^{\mu},y)$ transforming according to
    the adjoint representation of $SU(2)$. The five-dimensional Lagrangian reads
    \begin{equation}
     \label{5dlagrangianunbrokensu2orbi}
     \mathcal{L}_{5D}=-\frac{1}{4}  F^{a}_{MN} F^{aMN} + \mid D_{M} \phi^{a} \mid^{2} \; ,
    \end{equation}
    where
    \begin{equation}
     F^{a}_{MN}=\partial_{M}A^{a}_{N}-\partial_{N}A^{a}_{M}
              +g_{5}f^{abc}A^{b}_{M}A^{c}_{N}
    \end{equation}
    and 
    \begin{equation}
     \label{covariantphiorbi}
     D_{M} \phi^{a}=\partial_{M} \phi^{a}+g_{5}f^{abc}A^{b}_{M}\phi^{c} \; .
    \end{equation}
    The boundary conditions read
    \begin{gather}
     A_{\mu}(x^{\mu},-y)=P \; A_{\mu}(x^{\mu},y) \; P^{-1}  \\
     A_{y}(x^{\mu},-y)=-P \; A_{y}(x^{\mu},y) \; P^{-1}  \\
     \phi(x^{\mu},-y)=P \; \phi(x^{\mu},y) \; P^{-1}  \; .
    \end{gather}   
    We choose the trivial orbifold projection
    \begin{equation}
     P=\text{diag}(1,1)  \; .
    \end{equation}
    Thus $G=SU(2)$ remains unbroken. The Fourier mode expansion up to the first KK-mode  reads 
    \begin{gather}
     A^{a}_{\mu}(x^{\mu},y)=\frac{1}{\sqrt{2 \pi R}} A_{\mu}^{a(0)}(x^{\mu})
      +\frac{1}{\sqrt{\pi R}} A_{\mu}^{a(1)}(x^{\mu}) \cos(\frac{y}{R}) \\
     A^{a}_{y}(x^{\mu},y)=\frac{1}{\sqrt{\pi R}}A_{y}^{a(1)}(x^{\mu}) 
      \sin(\frac{y}{R}) \\
     \label{kktermphi}
     \phi^{a}(x^{\mu},y)=\frac{1}{\sqrt{2 \pi R}} \phi^{a(0)}(x^{\mu})
      +\frac{1}{\sqrt{\pi R}} \phi^{a(1)}(x^{\mu}) \cos(\frac{y}{R}) \; .
    \end{gather}
    We insert the KK-mode expansion for $\phi$ (\ref{kktermphi}) in the covariant derivative 
    for $\phi$ (\ref{covariantphiorbi}). This yields
    \begin{eqnarray}
     \label{additionalscalarfieldcalculation1}
     D_{M} \phi^{a}& = & \frac{1}{\sqrt{2 \pi R}} \partial_{\mu} \phi^{a(0)}+
                  \frac{1}{\sqrt{\pi R}} \partial_{\mu} \phi^{a(1)} 
                  \cdot \cos(\frac{y}{R})+\frac{g_{5}}{2 \pi R} f^{abc} 
                  \left[A_{\mu}^{b(0)}(x^{\mu})\right. \nonumber\\
            & + & \left.  A_{\mu}^{b(1)}(x^{\mu}) 
                  \sqrt{2} \cos(\frac{y}{R}) \right]  
                  \left[\phi^{c(0)}(x^{\mu})+\phi^{c(1)}(x^{\mu}) 
                  \sqrt{2} \cos(\frac{y}{R}) \right] \nonumber \\
            & + & \frac{1}{\sqrt{\pi R}} \phi^{a(1)} \frac{1}{R} \cdot \sin(\frac{y}{R}) 
                  +\frac{g_{5}}{2 \pi R} f^{abc} 
                  \left[A_{y}^{b(1)}(x^{\mu}) 
                  \sqrt{2} \cos(\frac{y}{R}) \right] \nonumber \\
            &   & \left[\phi^{c(0)}(x^{\mu}) +\phi^{c(1)}(x^{\mu}) 
                  \sqrt{2} \cos(\frac{y}{R}) \right] 
    \end{eqnarray}
    We assume that $\phi$ gets a VEV in its diagonal direction
    \begin{equation}
     \phi \; \longrightarrow \;  \langle \; \phi^{3(0)} \; \rangle \; .
    \end{equation}
    Inserting this VEV in (\ref{additionalscalarfieldcalculation1}) and imposing axial gauge
    we obtain
    \begin{equation}
     \label{vevcovarintphiorbi}
     D_{M} \phi^{a} = \frac{g_{5}}{2 \pi R} f^{ab3} 
                  \left[A_{\mu}^{b(0)}(x^{\mu}) \langle \; \phi^{3(0)} \; \rangle 
                  +A_{\mu}^{b(1)}(x^{\mu}) \langle \; \phi^{3(0)} \; \rangle
                  \sqrt{2} \cos(\frac{y}{R}) \right]  \; .
    \end{equation}
    We already see that this term will vanish for $b=3$. The field $A_{\mu}^{3(0)}$ will therefore remain
    massless. 
    We insert (\ref{vevcovarintphiorbi}) into the five-dimensional Lagrangian 
    (\ref{5dlagrangianunbrokensu2orbi}) and integrate over the circle $S^{1}$. This yields
    \begin{eqnarray}
     \mathcal{L}_{mass}^{\phi} & = & \int_{0}^{2 \pi R} \mid D_{M} \phi^{a} \mid^{2}  \\
                               & = & g_{4}^{2} \left( A_{\mu}^{b(0)}(x^{\mu}) \right)^{2} 
                         \langle \; \phi^{3(0)} \; \rangle^{2} 
                         +2 g_{4}^{2} \left( A_{\mu}^{b(1)}(x^{\mu}) \right)^{2}
                         \langle \; \phi^{3(0)} \; \rangle^{2}  \nonumber
    \end{eqnarray}
    where $b=1,2$ and we have inserted  (\ref{relation4deffgaugecoupling5dgaugecoupling})
    \begin{equation}
     g_{4}=\frac{g_{5}}{\sqrt{2 \pi R}}  \; .
    \end{equation}
    The Yang-Mills term in (\ref{5dlagrangianunbrokensu2orbi})
    yield a mass term for the gauge fields
    $A_{\mu}^{a(1)}$ as usual (\ref{4dlagninabeorbi})
    \begin{equation}
     \mathcal{L}^{ym}_{mass} = \int_{0}^{2\pi R} \{ -\frac{1}{2} F^{a}_{\mu y} F^{a \mu y} \} \; dy  
                   = \frac{1}{2} \; \frac{1}{R^{2}} \left( A_{\mu}^{a(1)} \right)^{2} 
    \end{equation}
    with $a=1,2,3$. Recapitulating we have obtained the following mass terms for
    the truncated $S^{1}/\mathbb{Z}_{2}$
    orbifold model
    \begin{eqnarray}
     \label{masstermobrismallai}
      \mathcal{L}^{orbifold}_{mass} & = &  g_{4}^{2} \left( A_{\mu}^{b(0)} \right)^{2} 
                              \langle \; \phi^{3(0)} \; \rangle^{2} 
                            + 2 g_{4}^{2} \left( A_{\mu}^{b(1)} \right)^{2} 
                              \langle \; \phi^{3(0)} \; \rangle^{2} \\ 
                           & + & \frac{1}{2} \; \frac{1}{R^{2}} \left( A_{\mu}^{b(1)} \right)^{2}
                             + \frac{1}{2} \; \frac{1}{R^{2}} \left( A_{\mu}^{3(1)} \right)^{2} \nonumber
    \end{eqnarray}
    where $b=1,2$. 

    From (\ref{masstermsu2finalresultabularsmalla1}) we read off the mass terms in the corresponding eBTLM
    \begin{equation}
     \mathcal{L}_{mass}^{eBTLM} 
       =  \frac{1}{2} \; g^{2} \rho_{min}^{2}  a^{2}_{1} \; \left( A^{j(0)}_{\mu} \right)^{2} 
       +  \frac{1}{2} \; g^{2} \rho_{min}^{2} \; \left( A^{b(1)}_{\mu} \right)^{2}+
          \frac{1}{2} \; g^{2} \rho_{min}^{2} \; \left( A^{3(1)}_{\mu} \right)^{2} \; .
    \end{equation}
    Inserting  the identification $g \rho_{min}=1/R$ we obtain
    \begin{equation}
     \label{masstermbtlmsmallai}
     \mathcal{L}_{mass}^{eBTLM} 
       =   \frac{1}{2} \; \frac{ a^{2}_{1}}{R^{2}} \; \left( A^{j(0)}_{\mu} \right)^{2} 
       +   \frac{1}{2} \; \frac{1}{R^{2}} \; \left( A^{b(1)}_{\mu} \right)^{2}+
           \frac{1}{2} \; \frac{1}{R^{2}} \; \left( A^{3(1)}_{\mu} \right)^{2} \; .
    \end{equation}
    The comparison of (\ref{masstermbtlmsmallai}) with (\ref{masstermobrismallai}) yields
    \begin{itemize}
     \item The mass term for the gauge field $A_{\mu}^{3(1)}$ coincides in both models . 
     \item Since the zero KK modes of all fields are expected to be much lighter than their first KK 
           excitation, we assume  
           \begin{equation}
            \label{assumptinvevphiorbi1}
            \langle \; \phi^{3(0)} \; \rangle \ll \frac{1}{g_{4} R}  \; .
           \end{equation}
           Thus the masses for the gauge fields $A_{\mu}^{1,2(1)}$ are approximately equal 
           in both models. 
     \item Setting
           \begin{equation}
            \langle \; \phi^{3(0)} \; \rangle = \frac{a_{1}}{g_{4} R} \; ,
           \end{equation}  
           both models yield the same mass terms for the gauge fields $A^{b(0)}_{\mu}$ with $b=1,2$. 
           Note that $0 < a_{1} \ll 1$, which is compatible with the 
           assumption (\ref{assumptinvevphiorbi1}). Since $a_{1}$ and $g_{4}$ are dimensionless and
           $1/R$ has mass dimension $1$, the VEV $\langle \; \phi^{3(0)} \; \rangle$ has mass dimension
           $1$.
    \end{itemize}
    \begin{proposition}
     An $S^{1}/\mathbb{Z}_{2}$ continuum orbifold model with bulk gauge group $G=SU(2)$, an
     additional scalar field $\phi$ transforming according to the adjoint representation of $G=SU(2)$,
     trivial orbifold projection $P=\text{diag}(1,1)$ and a Fourier mode expansion for all fields
     truncated at the first excited Kaluza-Klein mode in axial gauge gives an approximation to an effective 
     bilayered transverse lattice model with bulk gauge group $G=SU(2)$, trivial 
     orbifold projection $P=\text{diag}(1,1)$ and minimum of the Higgs potential at
     \begin{equation}
      \label{minimumofphi}
      \Phi_{min}=\rho_{min} \frac{1}{\sqrt{2}} \left( \begin{array}{cc} e^{a_{1}} & 0 \\ 
                      0 & e^{a_{2}} \end{array} \right) \; , \; a_{2}=-a_{1} \; , 
     \end{equation} 
     in the limit of small $a_{1}$ ($0 < a_{1} \ll 1$), if the scalar field $\phi$ 
     gets the VEV
     \begin{equation}
      \phi \; \longrightarrow \;  \langle \; \phi^{3(0)} \; \rangle=\frac{a_{1}}{g_{4} R} \; .
     \end{equation}
    \end{proposition}

 \section{Large gauge boson masses from spontaneous symmetry breaking}

    \label{sectionlargegaugebosonmasses}

    Let $a_{2} \gg 1$ in (\ref{masstermsu2finalresultabular}).
    Since $e^{a_{2}}=e^{-a_{1}}$ and $a_{2} \gg 1$, it follows $e^{a_{1}} \approx 0$ 
    and we obtain the following mass squared terms
    \begin{tableindoc}
     \begin{equation}
      \label{masstermsu2finalresultabularlargeai}
       \begin{array}{|c|c|} \hline
        \text{Field} & \text{Mass squared} \\ \hline
        A^{i(0)}_{\mu} \quad i=1,2 
        & \frac{1}{8} \; g^{2} \rho^{2}_{min} \; e^{2a_{1}} \\ \hline
        A^{3(0)}_{\mu} & 0 \\ \hline
        A^{i(1)}_{\mu} \quad i=1,2 
        & \frac{1}{8} \; g^{2} \rho^{2}_{min} \; e^{2a_{1}} \\ \hline
        A^{3(1)}_{\mu} & \frac{1}{4} \; g^{2} \rho^{2}_{min} \; e^{2a_{1}} \\ \hline
       \end{array} 
     \end{equation}
    \end{tableindoc}
    We recognise that for $a_{1} \gg 1$ the gauge field masses show 
    an exponential dependence on $a_{1}$ and can therefore be very large. {\em{It is
    remarkable that already the zero mode gauge fields $A^{1,2(0)}_{\mu}$ can have masses
    much above the compactification scale $1/R$.}} 
    This behaviour has no counterpart within the customary approximation scheme of
    an ordinary orbifold model. 

    \begin{proposition}
 
     \label{propositionlargegaugebosonmasses}

     An effective bilayered transverse lattice model with bulk gauge group $G=SU(2)$, trivial orbifold 
     projection $P=\text{diag}(1,1)$ and minimum of the Higgs potential at
     \begin{equation}
      \Phi_{min}=\rho_{min} \frac{1}{\sqrt{2}} \left( \begin{array}{cc} e^{a_{1}} & 0 \\ 
                      0 & e^{a_{2}} \end{array} \right) \; , \; a_{2}=-a_{1} \; , 
     \end{equation} 
     in the limit of large $a_{1}$ ($a_{1} \gg 1$) allows masses for some zero mode
     and first excited KK-mode gauge fields, which are much larger than the compactification scale 
     $g \rho_{min}=1/R$.
    \end{proposition}

 \section[EBTLM and continuous Wilson line breaking]{Effective bilayered transverse
          lattice model and continuous Wilson line breaking}

   In this section we consider the case where the gauge group $G=SU(2)$
   is broken via orbifolding to its subgroup $G_{0}=U(1)$. We embed the orbifold 
   projection $P$ in $G$ by setting
   \begin{equation}
    \label{nontrivialpsuu1}
    P=\exp(2 \pi i \; t_{3})=diag(1,-1)  \; .
   \end{equation}
   This choice for $P$ breaks $G=SU(2)$ down to $G_{0}=U(1)$ where $G_{0}$
   is generated by $t_{3}$. 
   As in section \ref{sectionunbrokensu2} we assume that the holonomy group $H$ is given by
   $\mathbb{R}_{\ast}^{+} \; SL(2{,}\mathbb{C})$, with $\mathbb{R}_{\ast}^{+}=\mathbb{R}^{+} / \{0\}$.
   The action of $P$ on $G$ leads to the split
   \begin{equation}
    \mathfrak{su}(2)=\mathfrak{u}(1) \oplus \mathfrak{su}(2)/\mathfrak{u}(1) \; ,
   \end{equation}
   where $\mathfrak{u}(1)=\text{Lie} \; G_{0}$. Figure \ref{figbrokensu2u1table} summarises the setting.
   \begin{figure}[h]
    \begin{equation*}
     \begin{picture} (18,4.5)
     \thicklines
     \multiput(2,1.5) (0,2.5) {2} {\line (1,0) {9}}
     \thinlines
     \multiputlist(11.5,1.5)(0,2.5){$L$,$R$}
     \multiputlist(6.5,4.5)(7,0){$G_{0}=U(1)$}
     \multiputlist(6.5,2.6)(2.5,0){$G=SU(2) \subset H=\mathbb{R}_{\ast}^{+} \; SL(2{,}\mathbb{C})$}
     \multiputlist(6.5,1)(7,0){$G_{0}=U(1)$}
     \put(2.5,2.5) {$\Phi$}
     \put(10.3,2.5) {$\Phi^{\ast}$}
     \matrixput (3,1.5)(7,0){2}(0,2,5){2}{\circle*{0.2}}
     \thicklines
     \dottedline[\circle*{0.05}]{0.1}(3,1.5)(3,4)
     \dottedline[\circle*{0.05}]{0.1}(10,1.5)(10,4)
     \thinlines
     \put(2.82,3.65){\Pisymbol{pzd}{115}}
     \put(9.82,1.58){\Pisymbol{pzd}{116}}
     \end{picture}
    \end{equation*}
    \caption{Effective bilayered transverse lattice model for bulk gauge group $G=SU(2)$ and
             non-trivial orbifold projection $P=\text{diag}(1,-1)$. The bulk gauge group $G=SU(2)$
             is broken to its subgroup $G_{0}=U(1)$ via orbifolding.} 
    \label{figbrokensu2u1table}
   \end{figure}
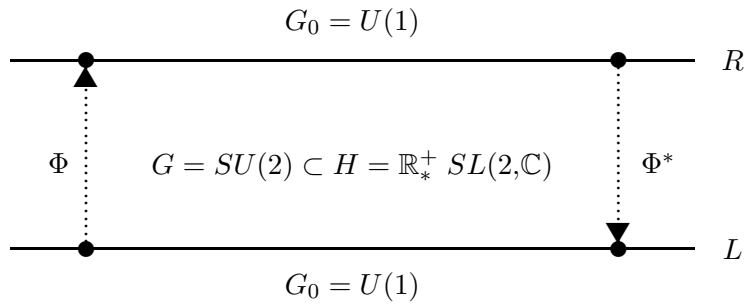 

   Let $\Phi \in H$ fulfil the sharpened orbifold condition. Then $\Phi$ can be written
   according to Theorem \ref{theoremorbifoldconditionsfornupt} as
   \begin{equation}
    \label{decphisuu1translattice}
    \Phi =\rho \; e^{A_{y}} \; e^{\eta} \; e^{A_{y}} \; ,
   \end{equation}
   where
   \begin{eqnarray} 
    P A_{y} P^{-1} & = & -A_{y}  \; ,\\
    P \eta P^{-1}  & = & \eta \; ,
   \end{eqnarray}
   $\eta \in \mathfrak{a}$, for an appropriate choice of $\mathfrak{a}$, and 
   $A_{y} \in \mathfrak{su}(2)/\mathfrak{u}(1)$. As in section
   \ref{sectionunbrokensu2} we choose
   \begin{equation}
    \mathfrak{a}=\{ \eta=diag(a_{1},a_{2}) \} \quad , \quad a_{1}=-a_{2} \; , \quad a_{i} \in 
     \mathbb{R} \; .
   \end{equation}
   Thus we have $ P \eta P^{-1}=\eta$ for all $\eta \in \mathfrak{a}$. 

   Let us consider the Higgs potential 
   \begin{equation}
     \label{higgspotentialbroken}
     V(\Phi)=V(\rho \; e^{A_{y}} \; e^{\eta} \; e^{A_{y}}) \; .
   \end{equation}
   According to Theorem \ref{theoremnuhiggspotential}, $V(\Phi)$ is invariant under unitary gauge 
   transformations 
   \begin{equation}
    V(S_{0}(x) \Phi S_{0}(x)^{-1})=V(\Phi) 
   \end{equation} 
   where $S_{0}(x) \in G_{0}$. Thus 
   $e^{A_{y}}$ in (\ref{higgspotentialbroken}) cannot be gauged away. Consequently
   the Higgs potential $V(\Phi)$ depends on $\rho$, $\eta$ and $A_{y}$  
   \begin{equation}
    V(\Phi)=V(\rho \; e^{A_{y}} \; e^{\eta} \; e^{A_{y}})=\mathcal{V}(\rho,\eta,A_{y}) \; .
   \end{equation}

   The unitary factor $e^{A_{y}}$ in (\ref{decphisuu1translattice}) can be written as
   \begin{equation}
    \label{aywilsonline}
    \exp\left(2 \pi i \; g R \; \mathcal{A}_{y}^{(0)} \right) \; ,
   \end{equation}
   where $\mathcal{A}_{y}^{(0)} \in \mathfrak{su}(2)/\mathfrak{u}(1)$ is the zero of the extra-dimensional 
   component of the five-dimensional gauge field.
   We can expand
   \begin{equation} 
    \label{ayexand}
    \mathcal{A}_{y}^{(0)}=\mathcal{A}_{y}^{1(0)} t_{1} + \mathcal{A}_{y}^{2(0)} t_{2} \; .
   \end{equation}
   Let us consider the case where $\mathcal{A}_{y}$ assume a VEV. Without loss of generality
   we lay this VEV in the $t_{1}$-direction, i.e.
   \begin{equation}
    \label{vevfora2}
    \mathcal{A}_{y} \to \langle \mathcal{A}_{y}^{1(0)} \rangle t_{1} \; .
   \end{equation}
   Inserting (\ref{vevfora2}) in (\ref{aywilsonline}) we get
   \begin{equation} 
    \label{btlmconwilson}
    W=\exp(2 \pi i \; g R \; \langle \mathcal{A}_{y}^{1(0)} \rangle t_{1}) \; .
   \end{equation}     
   This is a Wilson line, compare with  
   (\ref{wilsonline}), and since $\left[ P,t_{1} \right] \neq 0$ it does
   not commute with $P$. Thus the VEV for $\langle \mathcal{A}_{y}^{1(0)} \rangle$ can be an arbitrary
   constant and thus $W$ is a continuous Wilson line, compare with (\ref{wilonlinecontious}).

   However $\mathcal{A}_{y}^{(0)}$ is not in its canonical
   four-dimensional form. Therefore we make the following   
   \begin{definition} 
    \label{defintionextradimvecpot}
    The unitary factor in the decomposition (\ref{decphisuu1translattice}) is given by
    \begin{equation} 
     \label{extradimvecpot}   
     \exp(A_{y})=\exp(i \; g_{4} R \; \mathcal{A}_{y}^{(0)}) \; .
    \end{equation}
   \end{definition}
   Remarks: 
   i) In (\ref{extradimvecpot}) $\mathcal{A}_{y}^{(0)}$ is the zero of the extra-dimensional 
   component of the five-dimensional gauge field in its canonical four-dimensional form and can be 
   interpreted as a usual four-dimensional Higgs field, and
   $g_{4}$ is the four-dimensional effective gauge coupling
   constant. This definition is convenient because the kinetic term  $\mathcal{L}_{mass}= 
   tr \left[ \left(D_{\mu}\Phi_{min} \right)^{\dagger}\left( D_{\mu}\Phi_{min} \right) \right]$ 
   should involve only rescaled four-dimensional terms, compare with section \ref{sectionenonabBTLM}. \\
   ii) In contrast to (\ref{aywilsonline}) we have rescaled $\mathcal{A}_{y}^{(0)}$
   \footnote{In an orbifold theory $\mathcal{A}_{y}$ and its zero mode $\mathcal{A}_{y}^{(0)}$ is related by
   $\mathcal{A}_{y}=\frac{1}{\sqrt{2 \pi R}} \mathcal{A}_{y}^{(0)}$ (\ref{fmexexdimvp}).}
   by a factor $2 \pi$. \\
   iii) $e^{A_{y}}$ and $g_{4}$ are dimensionless, $R$ has mass dimension $-1$ and 
   $\mathcal{A}_{y}^{(0)}$ has mass dimension $1$.

   Within Definition \ref{defintionextradimvecpot} the Wilson line (\ref{btlmconwilson}) becomes
   \begin{equation}
    \label{rescaledwilson}
    W=\exp( i \; g_{4} R \; \langle \mathcal{A}_{y}^{1(0)} \rangle t_{1}) \; .
   \end{equation}      
   It is convenient to rewrite the VEV $\langle \mathcal{A}_{y}^{1(0)} \rangle$ as
   \begin{equation}
    \label{vevforay}
    \langle \mathcal{A}_{y}^{1(0)} \rangle=\frac{\alpha_{1}}{g_{4} R} \; ,
   \end{equation}
   where $0 < \alpha_{1} < 1$ is a dimensionless parameter. 
   Inserting (\ref{vevforay}) in (\ref{rescaledwilson}) the Wilson line becomes
   \begin{equation} 
    W=\exp(i \; g_{4} R \; \langle \mathcal{A}_{y}^{1(0)} \rangle t_{1})
     =\exp(i \; \alpha_{1} t_{1}) \; .
   \end{equation}
   A VEV for $\mathcal{A}_{y}^{(0)}$ is usually much smaller than the compactification scale $1/R$. 
   Thus $0 < \alpha_{1} \ll  1$ in (\ref{vevforay}) and we can approximate
   \begin{equation}
    W=\exp(i \; \alpha_{1} t_{1}) 
    \approx 1+i \; \alpha_{1} t_{1}
    =\begin{pmatrix} 1 & i \frac{\alpha_{1}}{2}   \\ 
                       i \frac{\alpha_{1}}{2} & 1 \end{pmatrix} \; .
   \end{equation}

   According to (\ref{higgspotentialbroken})
   we can parametrise the minimum $\Phi_{min}$ of the Higgs potential as
   \begin{eqnarray}
    \label{phimingeneralsuu1}
    \Phi_{min} & = & \rho_{min}  \frac{1}{\sqrt{2}} \; 
                      \begin{pmatrix} 1 & i \frac{\alpha_{1}}{2} \\ 
                      i \frac{\alpha_{1}}{2} & 1 \end{pmatrix} \; 
                      \begin{pmatrix} e^{a_{1}} & 0  \\ 0 & e^{a_{2}} \end{pmatrix} \; 
                      \begin{pmatrix} 1 & i \frac{\alpha_{1}}{2}  \\
                      i \frac{\alpha_{1}}{2} & 1 \end{pmatrix} \nonumber  \\
               & = &  \rho_{min} \frac{1}{2} \; 
                      \begin{pmatrix} e^{a_{1}} & \left( e^{a_{1}} + e^{a_{2}} \right) 
                      i \frac{\alpha_{1}}{2} \\  
                      -\left( e^{a_{1}} + e^{a_{2}} \right) i \frac{\alpha_{1}}{2} &
                      e^{a_{2}} \end{pmatrix} \; + \mathcal{O}\left( \alpha_{1}^{2} \right) \; .
                      \nonumber \\
   \end{eqnarray}
   In the following since $0 < \alpha_{1} \ll  1$ we neglect terms of 
   $\mathcal{O}\left( \alpha_{1}^{2} \right)$.
   We calculate the mass terms for the 
   gauge fields $A^{3(0)}_{\mu}$ and $A^{3(1)}_{\mu}$. 
   The covariant derivative reads (\ref{covariantderivaticesu2transform}) 
   \begin{equation}
    \label{covariantderivativebrokensu2}
    D_{\mu}\Phi=\partial_{\mu}\Phi+i \frac{g}{\sqrt{2}}  A^{3(0)}_{\mu} 
               \left[t_{3}, \Phi \right]+i \frac{g}{\sqrt{2}}  A^{3(1)} \{ t_{3}, \Phi \}
   \end{equation}
   where $[,]$ and $\{ , \}$ denote the commutator and anticommutator, respectively.
   We restrict ourselves to the case where $a_{1}=-a_{2}=0$. Then (\ref{phimingeneralsuu1}) becomes 
   \begin{equation}
    \Phi_{min}=\rho_{min} \frac{1}{\sqrt{2}} \; 
               \left( \begin{array}{cc} 1 & i  \alpha_{1} \\
               i  \alpha_{1} & 1 \end{array} \right) \; .
   \end{equation}
   We can expand $\Phi_{min}$ in terms of $t_{0}$ and $t_{1}$ as
   \begin{eqnarray}
    \Phi_{min} & = & \phi_{0} t_{0} + \phi_{1} t_{1} \\
               & = & \frac{1}{2} \left( \begin{array}{cc} \phi_{0} & \phi_{1} \\ 
                      \phi_{1} & \phi_{0} \end{array} \right)
               := \rho_{min} \frac{1}{\sqrt{2}} \; 
               \left( \begin{array}{cc} 1 &  i \alpha_{1} \\
               i \alpha_{1}  & 1 \end{array} \right) \; .
   \end{eqnarray}
   where  
   \begin{gather}
    \label{sutou1phinullphi+-generalphi}
    \phi_{0}=\sqrt{2} \rho_{min}  \; , \\
    \phi_{1}=i \phi_{1}^{\prime} \; , \quad \phi_{1}^{\prime}=\sqrt{2} \rho_{min} \; \alpha_{1}
    \nonumber \; . 
   \end{gather}
   For the commutators $\left[ t_{3}, \Phi_{min} \right]$ and anticommutators $\{ t_{3}, \Phi_{min} \}$
   we obtain
   \begin{gather}
    \left[ t_{3}, \Phi_{min} \right]
    =\phi_{0} \underbrace{\left[ t_{3}, t_{0} \right]}_{=0}
     +\phi_{1} \underbrace{\left[ t_{3}, t_{1} \right]}_{=i t_{2}} \\
    \{ t_{3}, \Phi_{min} \}
    =\phi_{0} \underbrace{\{ t_{3}, t_{0} \}}_{=t_{3}}
     +\phi_{1} \underbrace{\{ t_{3}, t_{1} \}}_{=0} \; .
   \end{gather}
   Inserting $\left[ t_{3}, \Phi_{min} \right]$ and $\{ t_{3}, \Phi_{min} \}$ in
   (\ref{covariantderivativebrokensu2}) we find
   \begin{equation}
    D_{\mu}\Phi_{min}=- i \frac{g}{\sqrt{2}} \; \phi_{1}^{\prime} t_{2} \;  A^{3(0)}_{\mu} 
                      + i \frac{g}{\sqrt{2}} \; \phi_{0} t_{3} \; A^{3(1)}_{\mu} \; .
   \end{equation}
   Taking the adjoint $(D_{\mu}\Phi_{min})^{\dagger}=i \frac{g}{\sqrt{2}} \;  \phi_{1}^{\prime} t_{2}
   \; A^{3(0)}_{\mu} - i \frac{g}{\sqrt{2}} \; \phi_{0} t_{3} \; A^{3(1)}_{\mu}$, multiplying 
   $(D_{\mu}\Phi_{min})^{\dagger}$ with $D_{\mu}\Phi_{min}$ we obtain
   \begin{eqnarray}
    && 
    \label{productsutou1unitarycase1}
    \left( D_{\mu}\Phi_{min} \right)^{\dagger} \left( D_{\mu}\Phi_{min} \right)  \\
    & = & \frac{1}{2} \;  g^{2} \; \phi_{1}^{\prime 2} t_{1}^{2} \; \left( A^{3(0)}_{\mu} \right)^{2} 
          +\frac{1}{2} \; g^{2} \; \phi_{0}^{2} t_{3}^{2}  
          \; \left( A^{3(1)}_{\mu} \right)^{2}             
      -   \frac{1}{2} \; g^{2} \; \phi_{0} \phi_{1}^{\prime} \; \left[ t_{2}, t_{3} \right]  
          \;  A^{3(0)}_{\mu}  A^{3(1)}_{\mu} \nonumber
   \end{eqnarray}
   First we observe that the mixed term vanishes after taking the trace. 
   The mass term for the zero mode $A^{3(0)}_{\mu}$ becomes after taking the trace 
   \begin{equation}
    \frac{1}{4} g^{2} \; \phi_{1}^{\prime 2}  
    \; \left( A^{3(0)}_{\mu} \right)^{2} \; ,
   \end{equation}
   and the mass term for the first excited mode $A^{3(1)}_{\mu}$ in (\ref{productsutou1unitarycase1}) 
   becomes after taking the trace 
   \begin{equation}
    \frac{1}{4} \; g^{2} \phi_{0}^{2} 
    \; \left( A^{3(1)}_{\mu} \right)^{2} \; .
   \end{equation}  
   Recapitulating we have obtained
   \begin{equation}
    \text{tr} \left[\left( D_{\mu}\Phi_{min} \right)^{\dagger} 
        \left( D_{\mu}\Phi_{min} \right)\right]  \\
     =  \frac{1}{4} \; g^{2} \phi_{1}^{\prime 2} \; \left( A^{3(0)}_{\mu} \right)^{2} 
        + \frac{1}{4} \; g^{2} \phi_{0}^{2} 
          \; \left( A^{3(1)}_{\mu} \right)^{2} \; .
   \end{equation}
   Inserting $\phi_{0}$ and $\phi_{1}$ 
   (\ref{sutou1phinullphi+-generalphi}) we finally arrive at
   \begin{equation} 
      \label{masstermsu2tou1casea0a1equalzero}
      \text{tr} \left[\left( D_{\mu}\Phi_{min} \right)^{\dagger} 
          \left( D_{\mu}\Phi_{min} \right) \right] \\  
      =   \frac{1}{2} \; g^{2} \rho_{min}^{2} \alpha_{1}^{2} \; 
          \left( A^{3(0)}_{\mu} \right)^{2} + \frac{1}{2} \; g^{2} \rho_{min}^{2}  \; 
          \left( A^{3(1)}_{\mu} \right)^{2} \; .
   \end{equation}
   Table \ref{tablesu2u1a1a2ez} summarises the result
   \begin{tableindoc}
    \label{tablesu2u1a1a2ez}
    \begin{equation}
     \begin{array}{|c|c|} \hline
      \text{Field} & \text{Mass squared} \\ \hline
      A^{3(0)}_{\mu} &  g^{2} \rho_{min}^{2} \; \alpha_{1}^{2}  \\ \hline
      A^{3(1)}_{\mu} &  g^{2} \rho_{min}^{2} \\ \hline
     \end{array} 
    \end{equation} 
   \end{tableindoc}   
   {\em{Discussion}}: i) 
   The fields $A_{\mu}^{1,2(0)}$ and $A_{\mu}^{1,2(1)}$ are integrated out due to the choice 
   of the orbifold projection (\ref{nontrivialpsuu1}). \\
   ii) The mass for the zero mode gauge boson $A^{3(0)}_{\mu}$ is
       \begin{equation}
        m=\frac{\alpha_{1}}{R} \; ,
       \end{equation}
       where we have inserted $g \rho_{min}=1/R$ in (\ref{masstermsu2tou1casea0a1equalzero}). 
       For $0 < \alpha_{1} \ll 1$ the mass of  $A^{3(0)}_{\mu}$ is much lower than the compactification
       scale $1/R$. \\
   iii) For $\langle \mathcal{A}_{y}^{1(0)} \rangle \neq 0$, the orbifold unbroken gauge group $G_{0}=U(1)$
           is completely broken and we have the breaking scheme
           \begin{equation}
            SU(2) \stackrel{P}{\longrightarrow} U(1) 
            \stackrel{\langle \mathcal{A}_{y}^{1(0)} \rangle}{\longrightarrow} \emptyset \; .
           \end{equation}
           Thus the $rank$ of the gauge group $G=SU(2)$ is reduced. This follows from the fact that  
           \begin{equation} 
            W=\exp(i \; g_{4} R \; \langle \mathcal{A}_{y}^{1(0)} \rangle t_{1}) 
           \end{equation}
           is a continuous Wilson line. \\
   iv) The first excited KK-mode gauge bosons $A^{3(1)}_{\mu}$ acquire the mass  
       $g \rho_{min}=1/R$. This result is just what one would expect from the customary
       approximation scheme of a truncated $S^{1}/\mathbb{Z}_{2}$ continuum orbifold model.   

   We compare this result to a
   $S^{1}/\mathbb{Z}_{2}$ continuum orbifold model with bulk gauge group $G=SU(2)$ and non-trivial
   orbifold projection $P$.
   The five-dimensional Lagrangian reads
   \begin{equation}
    \label{lagrangiansuu1orbi}
    \mathcal{L}_{5D}=-\frac{1}{4}  F^{a}_{MN} F^{aMN} \; ,
   \end{equation}
   where
   \begin{equation}
    F^{a}_{MN}=\partial_{M}A^{a}_{N}-\partial_{N}A^{a}_{M}
              +g_{5}f^{abc}A^{b}_{M}A^{c}_{N} \; .
   \end{equation}
   The boundary conditions read
   \begin{gather}
    A_{\mu}(x^{\mu},-y)=P \; A_{\mu}(x^{\mu},y) \; P^{-1}  \\
    A_{y}(x^{\mu},-y)=-P \; A_{y}(x^{\mu},y) \; P^{-1}  \; .
   \end{gather}   
   We break $G=SU(2)$ down to $U(1)$ by choosing 
   \begin{equation}
    P=\text{diag}(1,-1)  \; ,
   \end{equation}
   compare with (\ref{nontrivialpsuu1}).
   The Fourier mode expansion up to the first Kaluza-Klein mode reads 
   \begin{gather}
    A^{3}_{\mu}(x^{\mu},y)=\frac{1}{\sqrt{2 \pi R}} A_{\mu}^{3(0)}(x^{\mu})
      +\frac{1}{\sqrt{\pi R}} A_{\mu}^{3(1)}(x^{\mu}) \cos(\frac{y}{R}) \; , \\
    A^{1,2}_{\mu}(x^{\mu},y)=\frac{1}{\sqrt{\pi R}} A_{\mu}^{1,2(1)}(x^{\mu}) 
      \sin(\frac{y}{R}) \; , \\
    \label{fmexexdimvp}
    A^{1,2}_{y}(x^{\mu},y)=\frac{1}{\sqrt{2 \pi R}} A_{y}^{1,2(0)}(x^{\mu})
      +\frac{1}{\sqrt{\pi R}} A_{y}^{1,2(1)}(x^{\mu}) \cos(\frac{y}{R}) \; ,\\
    A^{3}_{y}(x^{\mu},y)=\frac{1}{\sqrt{\pi R}} A_{y}^{3(1)}(x^{\mu}) 
      \sin(\frac{y}{R}) \; .
   \end{gather}
   We calculate $F^{a}_{\mu y}$ in axial gauge, i.e. we set $A_{y}^{a(1)}=0$ for $a=1,2,3$. The result is 
   \begin{eqnarray}
    F^{a}_{\mu y} & = & \partial_{\mu}A^{a}_{y}-\partial_{y}A^{a}_{\mu}+
                        g_{5}f^{abc} A^{b}_{\mu}A^{c}_{y} \label{fmu5su2u1orbi} \\ 
                  & = & \frac{1}{\sqrt{2 \pi R}} \partial_{\mu} A_{y}^{1,2(0)}(x^{\mu}) 
                    +   \frac{1}{\sqrt{\pi R}} A_{\mu}^{3(1)} \frac{1}{R} 
                        \sin(\frac{y}{R})-\frac{1}{\sqrt{\pi R}} A_{\mu}^{1,2(1)} \frac{1}{R} 
                        \cos(\frac{y}{R}) \nonumber \\ 
                  & + & \frac{g_{5}}{2 \pi R} f_{a3c} 
                        \left[ A_{\mu}^{3(0)}(x^{\mu})+A_{\mu}^{3(1)}(x^{\mu}) 
                        \sqrt{2} \cos(\frac{y}{R}) \right] \left[ A_{y}^{c(0)}(x^{\mu}) \right] \nonumber \\
                  & + & \frac{g_{5}}{\sqrt{2} \pi R} \left( 
                        \left[ A_{\mu}^{1(1)}(x^{\mu}) \sin(\frac{y}{R}) \right] 
                        \left[ A_{y}^{2(0)}(x^{\mu}) \right]
                        - \left[ A_{\mu}^{2(1)}(x^{\mu}) \sin(\frac{y}{R}) \right] 
                        \left[ A_{y}^{1(0)}(x^{\mu}) \right] \right)   \nonumber
   \end{eqnarray}
   with $c=1,2$.
   We assume that $A_{y}$ gets a VEV in its $t_{1}$ direction
   \begin{equation}
    A_{y} \; \longrightarrow \; \langle \; A_{y}^{1(0)} \; \rangle \; .
   \end{equation}
   Inserting this VEV in (\ref{fmu5su2u1orbi}), we obtain \footnote{Note that $\left[ t_{i},t_{j} \right]=
   i f_{ijk} t_{k}$ with $f_{ijk}=\epsilon_{ijk}$}
   \begin{eqnarray}
    F^{a}_{\mu y} & = & \frac{1}{\sqrt{\pi R}} A_{\mu}^{3(1)} \frac{1}{R} 
                        \sin(\frac{y}{R})-\frac{1}{\sqrt{\pi R}} A_{\mu}^{1,2(1)} \frac{1}{R} 
                        \cos(\frac{y}{R}) \\ 
                  & + & \frac{g_{5}}{2 \pi R}  
                        \left[ A_{\mu}^{3(0)}(x^{\mu})\langle \; A_{y}^{1(0)} \; \rangle 
                        +A_{\mu}^{3(1)}(x^{\mu}) \sqrt{2} \cos(\frac{y}{R}) 
                        \langle \; A_{y}^{1(0)} \; \rangle\right] \nonumber \\
                  & - & \frac{g_{5}}{\sqrt{2} \pi R}  
                        \left[ A_{\mu}^{1(1)}(x^{\mu}) \sin(\frac{y}{R}) \right] 
                        \langle \; A_{y}^{1(0)} \; \rangle  \nonumber \; .
   \end{eqnarray}
   Inserting $F^{a}_{\mu y}$ in the five-dimensional Lagrangian
   (\ref{lagrangiansuu1orbi}) and integrating over the circle $S^{1}$ we find
   \begin{eqnarray}
    \label{su2tou1orbifoldmassterm}
    \mathcal{L}_{mass}^{orbifold} & = & \int_{0}^{2\pi R} 
                              \{ -\frac{1}{2} F^{a}_{\mu y} F^{a\mu y} \} \; dy  \\ 
                           & = & \frac{1}{2} \; g_{4}^{2} \left( A_{\mu}^{3(0)}(x^{\mu}) \right)^{2} 
                                 \langle \; A_{y}^{1(0)} \; \rangle^{2} 
                                 +\frac{1}{2} \; g_{4}^{2} \left( A_{\mu}^{3(1)}(x^{\mu}) \right)^{2}
                                 \langle \; A_{y}^{1(0)}  \; \rangle^{2}   \nonumber \\
                           & + & \frac{1}{2} \; g_{4}^{2} \left( A_{\mu}^{1(1)}(x^{\mu}) \right)^{2}
                                 \langle \; A_{y}^{1(0)}  \; \rangle^{2} 
                                 +\frac{1}{2} \; \frac{1}{R^{2}} \left( A_{\mu}^{a(1)} \right)^{2} \nonumber
   \end{eqnarray}
   for $a=1,2,3$.
   We compare this result to (\ref{masstermsu2tou1casea0a1equalzero})
   \begin{eqnarray}
    \label{masstermbtlmsuu1}
    \mathcal{L}_{mass}^{eBTLM} 
    & = &  \frac{1}{2} \; g^{2} \rho_{min}^{2} \alpha_{1}^{2}  \; 
           \left( A^{3(0)}_{\mu} \right)^{2} + \frac{1}{2} \; g^{2} \rho_{min}^{2}  \; 
           \left( A^{3(1)}_{\mu} \right)^{2}  \nonumber \\
    & = &  \frac{1}{2} \; \frac{ \alpha_{1}^{2}}{R^{2}}   \; 
           \left( A^{3(0)}_{\mu} \right)^{2} + \frac{1}{2} \; \frac{1}{R^{2}} \; 
           \left( A^{3(1)}_{\mu} \right)^{2} 
   \end{eqnarray}
   where we have inserted $g \rho_{min}=1/R$ in the second step. The comparison
   of (\ref{su2tou1orbifoldmassterm}) with (\ref{masstermbtlmsuu1}) yield
   \begin{itemize}
    \item The zero KK modes of all fields are expected to be much lighter than their first KK 
          excitation. Therefore we can assume  
          \begin{equation}
           \label{assumptinvevayorbi}
           \langle \; A^{1(0)}_{y} \; \rangle \ll \frac{1}{g_{4} R}  \; .
          \end{equation}
          Then the mass of the gauge field $A_{\mu}^{3(1)}$ is approximately equal 
          in both models. 
    \item For
          \begin{equation}
           \label{orbifoldvev}
           \langle \; A^{1(0)}_{y} \; \rangle = \frac{\alpha_{1}}{g_{4} R} \; ,
          \end{equation}  
          both models yield the same mass term for $A^{3(0)}_{\mu}$, i.e.
          \begin{equation}
           m=\frac{\alpha_{1}}{R} \; .
          \end{equation}
          This is a consequence of the fact that we have rescaled the extra-dimensional vector potential
          in the eBTLM such that the VEV for $A^{2(0)}_{y}$ is given by (\ref{orbifoldvev}).
          Note that $0 < \alpha_{2} \ll  1$ which is compatible with the 
          assumption (\ref{assumptinvevayorbi}).
    \end{itemize}
    \begin{proposition}
     An $S^{1}/\mathbb{Z}_{2}$ continuum orbifold model with bulk gauge group $G=SU(2)$, non-trivial
     orbifold projection $P=\text{diag}(1,-1)$ and a Fourier mode expansion for all fields truncated
     at the first excited Kaluza-Klein mode in axial gauge 
     gives an approximation to an effective bilayered transverse
     lattice model with bulk gauge group $G=SU(2)$, non-trivial
     orbifold projection $P=\text{diag}(1,-1)$ and the minimum of the Higgs potential at
     \begin{equation}
      \Phi_{min}=\rho_{min} \frac{1}{\sqrt{2}} \; 
               \left( \begin{array}{cc} 1 & i \alpha_{1} \\
                i \alpha_{1} & 1 \end{array} \right) 
     \end{equation}
     with $0 < \alpha_{1} \ll 1$.
    \end{proposition}

\chapter{$SU(7)$ unified model}
   
  \label{su7model}
  
  \section{Introduction: Why $SU(7)$ ?}

    In this chapter we will present a realistic five-dimensional Gauge-Higgs unification model based on the 
    unified gauge group $SU(7)$. The gauge group $SU(7)$ unifies electroweak-, flavour- and Higgs 
    interactions in one single gauge group. Colour will be ignored. In the following we will outline 
    the basic considerations that will lead to the unified gauge group $SU(7)$.
    
    Let us start with the electroweak gauge group of the SM
    \begin{equation}   
     SU(2)_{L} \times U(1)_{Y} \; .
    \end{equation}
    In order to proceed we need to add a suitable flavour gauge group. In chapter \ref{chapterintrodcution}
    we have given an overview of flavour groups discussed in the literature. In particular, there are 
    the four continuous flavour groups: $SU(2)_{F}$, $SU(3)_{F}$, $SO(3)_{F}$ and $U(1)_{F}$. 
    Note that a possible flavour gauge group should remain unbroken by the orbifold projection $P$.
    The considerations of the last chapter suggest that the flavour gauge group in our model should 
    be $SO(3)_{F}$. There are three reasons that motivates this choice:
    \begin{enumerate}
     \item We want to explain naturally why there are three generations in the SM. The flavour gauge
           groups $SU(2)_{F}$, $SU(3)_{F}$ and $SO(3)_{F}$ possess all an irreducible three dimensional 
           representation in which the three generations of the SM can fit. For this reason we exclude
           $U(1)_{F}$. 
     \item Masses for all flavour gauge fields must be very large $\mathcal{O}(10^{3})-\mathcal{O}(10^{5})$
           TeV in comparison to the electroweak breaking scale $\mathcal{O}(246)$ GeV in order
           to suppress tree-level FCNC. Thus for
           a compactification scale $1/R$ of the theory of $\mathcal{O}(1)$ TeV such flavour gauge fields 
           must receive masses from VEVs for the selfadjoint part of $\Phi$. In section 
           \ref{sectionlargegaugebosonmasses}, see Proposition \ref{propositionlargegaugebosonmasses}, 
           we have obtained that in the limit of 
           large $a_{1}$ gauge field masses for some {\em{zero}} and first excited KK mode gauge fields
           are much larger than the compactification scale $1/R$.
           This behaviour is what we need here. Note that
           a Wilson line breaking of the flavour gauge group would lead to flavour gauge field masses 
           below the compactification scale $1/R$. 

           The three-dimensional representation of $SU(2)_{F}$
           is not faithful (faithful representations of $SU(2)_{F}$ have even dimension). However the
           three-dimensional representation of $SU(2)_{F}$ is a 
           faithful representation of $SO(3)_{F}$. The three generations of the SM can fit is this
           three-dimensional representation of $SO(3)_{F}$. If we embed $SO(3)_{F}$ into $SU(7)$ in an
           appropriate way it is possible for a suitable minimum of the Higgs potential that all
           $SO(3)_{F}$ gauge fields can receive masses from VEVs for the selfadjoint part of $\Phi$ much
           above the compactification scale $1/R$. The three generations of the SM can also fit is the
           three-dimensional representation of $SU(3)_{F}$. In this chapter we discuss also
           the embedding of $SU(3)_{F}$ into $SU(7)$. Note that we have the embedding scheme:
           $SO(3)_{F} \subset SU(3)_{F} \subset SU(7)$. However for the embedding of $SU(3)_{F}$ into
           $SU(7)$ there remains at least an $U(1)_{F} \times U(1)_{F}$ left unbroken
           by VEVs for the selfadjoint part of $\Phi$.  Note again that
           we do not want to break the flavour gauge group by orbifolding. Thus in this setting we exclude 
           $SU(3)_{F}$.
     \item The three-dimensional representation of $SO(3)_{F}$ is anomaly-free while 
           the  three-dimensional representation of $SU(3)_{F}$ is not anomaly-free. 
           This is an additional reason why we exclude $SU(3)_{F}$. We will discuss the issue
           of anomaly cancellation in the $SU(7)$ model in detail in section \ref{anomalycancellation}.     
    \end{enumerate}

    If we add the flavour gauge group $SO(3)_{F}$ to the electroweak gauge group of the SM we arrive at 
    \begin{equation}   
     SU(2)_{L} \times U(1)_{Y} \times SO(3)_{F} \; .
    \end{equation} 
    Since VEVs for the selfadjoint
    part of $\Phi$ break only the flavour gauge group, i.e.
    \begin{equation}   
     SU(2)_{L} \times U(1)_{Y} \times SO(3)_{F} \stackrel{\langle \eta \rangle}{\rightarrow}  
     SU(2)_{L} \times U(1)_{Y} \; ,
    \end{equation} 
    there remains an unbroken electroweak gauge group. The question is now:
    \begin{itemize}
     \item 
      How can we include electroweak symmetry breaking in the model?
    \end{itemize}
    Note that we do not want to introduce extra Higgs fields in the model (besides the nonunitary parallel
    transporters $\Phi$). The answer to this question is the following:
    First we embed the flavour gauge group $SO(3)_{F}$ into $SU(3)_{F}$. Consequently we
    arrive at $SU(2)_{L} \times U(1)_{Y} \times SU(3)_{F}$. The purpose is
    to unify weak- and flavour interactions in one single gauge group $SU(6)_{L}$ \cite{Ponce:1989wr}.
    The embedding of $SU(2)_{L} \times SU(3)_{F}$ in $SU(6)_{L}$ is a special maximal one.
    The flavour gauge group $SU(3)_{F}$ itself appears 
    only at an intermediate step towards the unified gauge group $SU(6)_{L}$ and not as an
    unbroken symmetry for the reasons mentioned above. 
    Second, we embed $SU(6)_{L} \times U(1)_{Y}$ into $SU(7)$ in an appropriate way
    \begin{equation}
     SU(6)_{L} \times U(1)_{Y} \subset SU(7) \; ,
    \end{equation}
    and thus arrive at the unified gauge group $SU(7)$. 
 
    Starting with the unified bulk gauge group $G=SU(7)$ on the five-dimensional space-time 
    $M^{4} \times S^{1}/\mathbb{Z}_{2}$, we put $S^{1}/\mathbb{Z}_{2}$ on a lattice,
    calculate the RG-flow and consequently arrive at an eBTLM with unitary bulk gauge group
    $G=SU(7)$ and holonomy group $H=\mathbb{R}^{+}_{\ast} \; 
    SL(7,\mathbb{C})$ \footnote{Note that the linear span of $SU(7)$ reads
    $\mathbb{R}^{+}_{\ast} \; SL(7,\mathbb{C})$ and thus $H$ is unique.},
    with $\mathbb{R}^{+}_{\ast}=\mathbb{R}^{+}/\{0\}$.
    This procedure was explained in detail in the last chapter. The model contains nonunitary
    PTs $\Phi \in H$ in the extra dimension. 
   
    The main idea is now to choose a non-trivial orbifold projection $P$. Via
    a non-trivial orbifold projection $P$ the unified gauge group $SU(7)$ is broken down to
    $SU(6)_{L} \times U(1)_{Y}$ at the orbifold fixed points, i.e.
    \begin{equation}
     SU(7) \stackrel{P}{\longrightarrow} SU(6)_{L} \times U(1)_{Y} \; .
    \end{equation}
    At first view, this seems to be curious: First embedding $SU(6)_{L} \times U(1)_{Y}$ 
    into $SU(7)$ and one step later breaking $SU(7)$ again down to $SU(6)_{L} \times U(1)_{Y}$. However, 
    there is a bonus. If $\Phi$ fulfils the sharpened orbifold condition we can write $\Phi \in H$ as
    \begin{equation}
     \Phi=\rho \; e^{A_{y}} e^{\eta} e^{A_{y}}
    \end{equation}
    where $\eta \in \mathfrak{a}$ for an appropriate choice of $\mathfrak{a}$, $A_{y} \in
    \mathfrak{su}(6) \oplus \mathfrak{u}(1)$ and $\rho \in \mathbb{R}^{+}_{\ast}$.
    The Higgs potential $V(\Phi)$ will therefore also depend
    on $A_{y}$:
    \begin{equation}
     V(\Phi)=V(\rho \; e^{A_{y}} e^{\eta} e^{A_{y}})=\mathcal{V}(\rho,\eta,A_{y}) \; .
    \end{equation}
    Note that $A_{y}$ in $V(\Phi)$ cannot be gauged away because $G=SU(7)$ is broken to $G_{0}=
    SU(6)_{L} \times U(1)_{Y}$. The gauge group
    $G_{0}=SU(6)_{L} \times U(1)_{Y}$ is broken further to $SU(2)_{L} \times U(1)_{Y} \times
    SO(3)_{F}$ by imposing Dirichlet and Neumann boundary conditions. 

    We will show that zero modes $A_{y}^{(0)}$ of $A_{y}$ have the
    required properties to serve as a substitute for the SM Higgs. In particular they are $SU(2)_{L}$ 
    doublets and carry hypercharge $1/2$. In contrast to the SM, the model includes three Higgs doublets,
    one for the first, one for the second and one for the third generation.

    If the zero modes $A_{y}^{(0)}$ acquire VEVs in their $SU(2)_{L}$ down component, 
    the electroweak gauge group
    $SU(2)_{L} \times U(1)_{Y}$ is broken down to $U(1)_{em}$:
    \begin{equation}
     SU(2)_{L} \times U(1)_{Y} \stackrel{\langle A_{y}^{(0)} \rangle}{\longrightarrow} U(1)_{em}
    \end{equation}
    This breaking is equivalent to Wilson
    line breaking or Hosotani breaking. Recapitulating we have the following spontaneous symmetry breaking
    pattern:
    \begin{equation}   
     SU(2)_{L} \times U(1)_{Y} \times SO(3)_{F} 
      \stackrel{\langle \eta \rangle}{\longrightarrow} SU(2)_{L} \times U(1)_{Y}  
     \stackrel{\langle A_{y}^{(0)} \rangle}{\longrightarrow} U(1)_{em} \; ,
    \end{equation}
    where the breaking of $SO(3)_{F}$ takes place at energies much above the compactification scale 
    $1/R=\mathcal{O}(1)$ TeV. This way tree-level FCNC are naturally suppressed. The electroweak gauge bosons
    $W^{\pm},Z$ receive masses only from VEVs for $A_{y}^{(0)}$. Their masses will therefore be
    $\mathcal{O}(246)$ GeV.

  \section{Family unification in $SU(6)_{L} \times U(1)_{Y}$}

    As already mentioned in the introduction $SU(6)_{L}$ 
    unifies \cite{Ponce:1989wr} the weak gauge group $SU(2)_{L}$ 
    of the SM with the flavour gauge group $SU(3)_{F}$. Note that
    $SU(2)_{L} \times SU(3)_{F}$ is a special maximal subgroup of $SU(6)_{L}$. Since    
    $SU(6)_{L} \times U(1)_{Y}$ is broken to $SU(2)_{L} \times U(1)_{Y} \times SO(3)_{F}$ by orbifold
    and additional boundary conditions its full meaning is of no importance for us. 
    However, because $SU(2)_{L}$ and $SO(3)_{F} \subset SU(3)_{F}$ are subgroups of $SU(6)_{L}$ 
    we will discuss a model \cite{Gaitan-Lozano:1995sm,Cotti:1998de}
    based on the gauge group $SU(6)_{L} \times U(1)_{Y}$ shortly on its own.
    The gauge group $SU(6)_{L}$ has $35$ generators $L_{i}$ which in the 
    $SU(2)_{L} \times SU(3)_{F}$ basis can be written as
    \begin{itemize}
     \item   \begin{equation}
              \label{su2su6gen}
               L_{i}=\frac{1}{2\sqrt{3}} \; \sigma_{i} \; \otimes \; \mathbf{1}_{3} \; ,
             \end{equation}
             where $\sigma_{i},i=1,2,3$ are the Pauli matrices and $\frac{1}{2} \sigma_{i} \;  \otimes \; 
             \mathbf{1}_{3}$ are the generators of $SU(2)_{L}$. The symbol $\mathbf{1}_{3}$ stand 
             for the $3 \times 3$ unit matrix.
     \item   \begin{equation}
              \label{su3su6gen}
               L_{i^{\prime}}=\frac{1}{2\sqrt{2}} \; \mathbf{1}_{2} \; \otimes \;  \lambda_{j} \; ,
             \end{equation}
             where $i^{\prime}=4,\dots,11$ \footnote{
             $L_{4}=\frac{1}{2\sqrt{2}} \; \mathbf{1}_{2} \; \otimes \;  \lambda_{1}, \dots,
             L_{11}=\frac{1}{2\sqrt{2}} \; \mathbf{1}_{2} \; \otimes \;  \lambda_{8}$}, $j=1,\dots,8$, 
             $\lambda_{j}$ are the Gell-Mann matrices and $\frac{1}{2}\mathbf{1}_{2} \; \otimes \lambda_{j}$
             are the generators of $SU(3)_{F}$. The symbol $\mathbf{1}_{2}$ stand 
             for the $2 \times 2$ unit matrix.
     \item   \begin{equation}
              \label{su6coseysusu3gen}
              L_{i^{\prime\prime}}=\frac{1}{2\sqrt{2}} \; \sigma_{i} \; \otimes \; \lambda_{j} \; ,
             \end{equation} 
             where $i^{\prime\prime}=12,\dots,35$ 
             \footnote{
              $L_{12}=\frac{1}{2\sqrt{2}} \; \sigma_{1} \; \otimes \; \lambda_{1}, \dots,
              L_{19}=\frac{1}{2\sqrt{2}} \; \sigma_{1} \; \otimes \; \lambda_{8},
              L_{20}=\frac{1}{2\sqrt{2}} \; \sigma_{2} \; \otimes \; \lambda_{1}, \dots,
              L_{27}=\frac{1}{2\sqrt{2}} \; \sigma_{2} \; \otimes \; \lambda_{8},
              L_{28}=\frac{1}{2\sqrt{2}} \; \sigma_{3} \; \otimes \; \lambda_{1}, \dots,
              L_{35}=\frac{1}{2\sqrt{2}} \; \sigma_{3} \; \otimes \; \lambda_{8}$},
              $i=1,2,3$ and $j=1,\dots,8$.
    \end{itemize}
    Note that all generators of $SU(6)_{L}$ are equally normalised as
    \begin{equation}
     \text{tr} \left( L_{i} L_{j} \right)=\frac{1}{2} \delta_{ij} \; .
    \end{equation} 
    The gauge group $SU(6)_{L} \times U(1)_{Y}$ gives rise to $36$ gauge bosons: $35$ are linked to the 
    generators of $SU(6)_{L}$ and
    one is linked to the generator of $U(1)_{Y}$. Besides the standard model gauge bosons there are
    $32$ extra gauge bosons which can be divided into four groups
    \begin{itemize}
     \item
      12 charged gauge bosons associated to the generators 
      $\frac{1}{2\sqrt{2}} \sigma_{i} \otimes \lambda_{j}$
      where $i=1,2$ and $j=1,2,4,5,6,7$. These gauge bosons perform transitions among families. They
      couple to {\em{family changing charged currents}} (FCCC). \\ Example: In the $SU(6)_{L} \times U(1)_{Y}$
      one introduces left-handed quarks in the fundamental representation $\mathbf{6}$ of $SU(6)_{L}$, i.e.
      $q_{L}=(u,c,t,d,s,b)_{L}$. We pick as an example the generator
      $L_{15}=\frac{1}{2\sqrt{2}} \sigma_{1} \otimes \lambda_{4}$. Ignoring the normalisation of $L_{15}$,
      the corresponding family changing charged 
      current reads
      \begin{equation}
       \bar{q}_{L} \; \gamma_{\mu} \; \left( \sigma_{1} \otimes \lambda_{4} \right) 
       \; q_{L}=\bar{u} \gamma_{\mu} b + \bar{t} \gamma_{\mu} d 
       + \bar{d} \gamma_{\mu} t + \bar{b} \gamma_{\mu} u \; .
      \end{equation}
      This means that the corresponding gauge bosons perform the transitions 
      $u \leftrightarrow b$ and $t \leftrightarrow d$.
     \item
      4 charged gauge bosons associated to the generators
      $\frac{1}{2\sqrt{2}} \sigma_{i} \otimes \lambda_{j}$
      where $i=1,2$ and $j=3,8$. These gauge bosons make no transitions among families but their 
      couplings are {\em{family dependent}}. They couple to {\em{non-universal family diagonal charged 
      currents}} (NUFDCC). \\ Example: Using the notations above, we
      pick as an example $L_{14}=\frac{1}{2\sqrt{2}} \sigma_{1} \otimes \lambda_{3}$. The corresponding
      gauge boson perform the transitions $u \leftrightarrow d$ and $s \leftrightarrow c$. 
     \item
      12 neutral gauge bosons associated to the generators 
      $\frac{1}{2\sqrt{2}} \sigma_{3} \otimes \lambda_{j}$
      and $\mathbf{1}_{2} \otimes \lambda_{j}$ where $j=1,2,4,5,6,7$. These gauge bosons perform
      transitions among families and couple to {\em{flavour changing neutral currents}} (FCNC). \\ 
      Example: Using the notations above, we pick as an example $L_{9}=\frac{1}{2\sqrt{2}}
      \mathbf{1}_{2} \otimes \lambda_{6}$. The corresponding gauge bosons 
      perform the transitions $u \leftrightarrow t$ and $d \leftrightarrow c$.
     \item
      4 neutral gauge bosons associated to the generators $\frac{1}{2\sqrt{2}} \sigma_{3} \otimes \lambda_{j}$
      and $\frac{1}{2\sqrt{2}} \mathbf{1}_{2} \otimes \lambda_{j}$ where $j=3,8$. 
      These gauge bosons make no transition 
      among families but their couplings are {\em{family dependent}}.
      They couple to {\em{non-universal family
      diagonal neutral currents}} (NUFDNC). Example: Using the notations above, we pick as an 
      example $L_{30}=\frac{1}{2\sqrt{2}} \sigma_{3} \otimes \lambda_{3}$. The corresponding gauge bosons 
      perform the transitions $u \leftrightarrow u$, $d \leftrightarrow d$, $c \leftrightarrow c$ and 
      $s \leftrightarrow s$.  
    \end{itemize}
    After the additional symmetry breaking by imposing Dirichlet and Neumann boundary conditions 
    \begin{equation}
     SU(6)_{L} \times U(1)_{Y} \rightarrow SU(2)_{L} \times U(1)_{Y} \times SO(3)_{F}  \; ,
    \end{equation}
    besides the SM gauge group, 
    only the flavour gauge group $SO(3)_{F}$ survives. It is generated by
    \begin{equation} 
     \frac{1}{2} \; \mathbf{1}_{2} \; \otimes \;  \lambda_{j}  \; ,
    \end{equation}
    where $j=2,5,7$. Thus the corresponding gauge bosons lead to FCNC. We note that since the bulk is 
    completely integrated out the $SU(7)$ model leads only to FCNC. FCCC, NUFDCC and NUFDNC are
    absent.

  \section{Embedding of $SU(6) \times U(1)_{Y}$ in $SU(7)$}

    In this section we define the generators of $SU(7)$. The unified gauge group $SU(7)$ is broken
    again down to $SU(6) \times U(1)_{Y}$ via orbifolding. This orbifold breaking can be achieved by choosing
    e.g. the orbifold projection $P=\text{diag}(-1,-1,-1,-1,-1,-1,1)$. If we embed the generators of 
    $SU(6) \times U(1)_{Y}$ in $SU(7)$ as upper $6 \times 6$ matrices and $U(1)_{Y}$  in $SU(7)$ as a
    diagonal $7 \times 7$ matrix
    an orbifold breaking by $P=\text{diag}(-1,-1,-1,-1,-1,-1,1)$  leave the $SU(6) \times U(1)_{Y}$ subgroup
    of $SU(7)$ unbroken. To be more precise, 
    take for example the $SU(6)_{L}$ generator 
    \begin{equation}
     L_{1}=\frac{1}{2\sqrt{3}}\left( \sigma_{1}\times\mathbf{1}_{3}\right)=
           \frac{1}{2\sqrt{3}} \begin{pmatrix} \mathbf{0}_{3} & \mathbf{1}_{3} \\ 
           \mathbf{1}_{3} & \mathbf{0}_{0} \end{pmatrix} \; ,
    \end{equation}
    where $\mathbf{1}_{3}(\mathbf{0}_{3})$ stands for the $3 \times 3$ unit(zero) matrix.
    $L_{1}$ is embedded in $SU(7)$ as 
    \begin{equation}
              \tilde{L}_{1}= \frac{1}{2\sqrt{3}} \begin{pmatrix} 
               \mathbf{0} & \mathbf{0} & \mathbf{0} & \mathbf{1} & \mathbf{0} & \mathbf{0} & 0 \\
               \mathbf{0} & \mathbf{0} & \mathbf{0} & \mathbf{0} & \mathbf{1} & \mathbf{0} & 0 \\
               \mathbf{0} & \mathbf{0} & \mathbf{0} & \mathbf{0} & \mathbf{0} & \mathbf{1} & 0 \\
               \mathbf{1} & \mathbf{0} & \mathbf{0} & \mathbf{0} & \mathbf{0} & \mathbf{0} & 0 \\
               \mathbf{0} & \mathbf{1} & \mathbf{0} & \mathbf{0} & \mathbf{0} & \mathbf{0} & 0 \\
               \mathbf{0} & \mathbf{0} & \mathbf{1} & \mathbf{0} & \mathbf{0} & \mathbf{0} & 0 \\
               0 & 0 & 0 & 0 & 0 & 0 & 0
              \end{pmatrix} \; ,
    \end{equation}
    and $P$ acts on $\tilde{L}_{1}$ as
    \begin{equation} 
     \tilde{L}_{1} \to P \tilde{L}_{1} P^{-1}=\tilde{L}_{1} \; .
    \end{equation}
    The same relation holds for all other generators of the $SU(6)_{L} \times U(1)_{Y}$ subgroup of $SU(7)$.
    In order to simplify notations we drop the tilde and write the generators of the 
    $SU(6)_{L}$ subgroup of $SU(7)$ just as $6 \times 6$ matrices. However by this notation we always 
    mean that they are embedded in $SU(7)$ as described above.
    In the following, we choose for all generators of $SU(7)$ the normalisation
    \begin{equation}
     Tr\left( L_{i} L_{j} \right)=\frac{1}{2} \delta_{ij} \; .
    \end{equation}
    Let again $\sigma_{i}, i=1,\dots,3$ denote the Pauli matrices,
    $\lambda_{j}, j=1,\dots,8$ the Gell-Mann matrices,  
    $\mathbf{1}_{2}$ the $2 \times 2$ unit matrix and $\mathbf{1}_{3}$ the $3 \times 3$ unit matrix,
    respectively.

    The gauge group $SU(7)$ has $48$ generators 
    \begin{itemize}
     \item $36$ generators belonging to the $SU(6)_{L}$ subgroup of $SU(7)$:
           \begin{equation}
            L_{i}=\frac{1}{2\sqrt{3}} \; \sigma_{i} \; \otimes \; 
            \mathbf{1}_{3} \quad , \quad
            L_{i^{\prime}}=\frac{1}{2\sqrt{2}} \; \mathbf{1}_{2} \; \otimes \; \lambda_{j}  \quad , \quad
            L_{i^{\prime\prime}}=\frac{1}{2 \sqrt{2}} \; \sigma_{i} \; \otimes \; \lambda_{j} \; ,
           \end{equation}
           compare (\ref{su2su6gen}), (\ref{su3su6gen}) and (\ref{su6coseysusu3gen}),
           where $i=1,2,3$, $i^{\prime}=4,\dots, 11$, $i^{\prime\prime}=12,\dots,35$ and $j=1,\dots,8$
           \footnote{$L_{4}=\frac{1}{2\sqrt{2}} \; \mathbf{1}_{2} \; \otimes \;  \lambda_{1}, \dots,
            L_{11}=\frac{1}{2\sqrt{2}} \; \mathbf{1}_{2} \; \otimes \;  \lambda_{8},
            L_{12}=\frac{1}{2\sqrt{2}} \; \sigma_{1} \; \otimes \; \lambda_{1}, \dots,
            L_{19}=\frac{1}{2\sqrt{2}} \; \sigma_{1} \; \otimes \; \lambda_{8},
            L_{20}=\frac{1}{2\sqrt{2}} \; \sigma_{2} \; \otimes \; \lambda_{1}, \dots,
            L_{27}=\frac{1}{2\sqrt{2}} \; \sigma_{2} \; \otimes \; \lambda_{8},
            L_{28}=\frac{1}{2\sqrt{2}} \; \sigma_{3} \; \otimes \; \lambda_{1}, \dots,
            L_{35}=\frac{1}{2\sqrt{2}} \; \sigma_{3} \; \otimes \; \lambda_{8}$}. Note that $L_{1},\dots,
            L_{35}$ are embedded in $SU(7)$ as upper $6 \times 6$ matrices as described above.
     \item $1$ generator belonging to the $U(1)_{Y}$ subgroup of $SU(7)$: 
           \begin{equation}
             L_{36}=\frac{1}{2\sqrt{21}} diag(1,1,1,1,1,1,-6) \; .
           \end{equation}
           Note that $L_{36}$ commutes with $L_{1},\dots,L_{35}$.
     \item $12$ generators belonging to the coset $SU(7)/SU(6)_{L}\times U(1)_{Y}$:
           \begin{equation*}
               L_{37}=\frac{1}{2} \; \begin{pmatrix} 
               0 & 0 & 0 & 0 & 0 & 0 & \mathbf{1} \\
               0 & 0 & 0 & 0 & 0 & 0 & \mathbf{0} \\
               0 & 0 & 0 & 0 & 0 & 0 & \mathbf{0} \\
               0 & 0 & 0 & 0 & 0 & 0 & \mathbf{0} \\
               0 & 0 & 0 & 0 & 0 & 0 & \mathbf{0} \\
               0 & 0 & 0 & 0 & 0 & 0 & \mathbf{0} \\
               \mathbf{1} & \mathbf{0} & \mathbf{0} & \mathbf{0} & \mathbf{0} & \mathbf{0} & \mathbf{0}    
              \end{pmatrix}
              \; , \dots \; , \; 
              L_{48}=\frac{1}{2} \; \begin{pmatrix} 
               0 & 0 & 0 & 0 & 0 & 0 & \mathbf{0} \\
               0 & 0 & 0 & 0 & 0 & 0 & \mathbf{0} \\
               0 & 0 & 0 & 0 & 0 & 0 & \mathbf{0} \\
               0 & 0 & 0 & 0 & 0 & 0 & \mathbf{0} \\
               0 & 0 & 0 & 0 & 0 & 0 & \mathbf{0} \\
               0 & 0 & 0 & 0 & 0 & 0 & \mathbf{-i} \\
               \mathbf{0} & \mathbf{0} & \mathbf{0} & \mathbf{0} & \mathbf{0} & \mathbf{i} & \mathbf{0}    
              \end{pmatrix}
           \end{equation*}  
    \end{itemize}

    Next we identify the generators of weak gauge group $SU(2)_{L}$, the hypercharge $U(1)_{Y}$ and 
    flavour gauge group $SO(3)_{F}$ as follows
    \begin{itemize}
     \item  $SU(2)_{L}:$ The weak generators of the SM are identified with
      \begin{equation}
       \label{smweakgaugegroup}
         T_{i}=\sqrt{3}  L_{i}=\frac{1}{2} \; \sigma_{i} \; \otimes \; \mathbf{1}_{3}  \; ,
      \end{equation} 
      where $i=1,2,3$.
     \item $U(1)_{Y}:$ The hypercharge generator of the SM is identified with
      \begin{equation}
       \label{smhyperchargegroup}
        Y=\sqrt{21} L_{36}=\frac{1}{2} \text{diag}(1,1,1,1,1,1,-6) \; .
      \end{equation}
     \item $ SO(3)_{F}:$ The generators of the flavour gauge group $SO(3)_{F}$ will be identified with
      \begin{equation}
       \label{so3flavourgaugegroup}         
       H_{j}=\sqrt{2} L_{i^{\prime}}=\frac{1}{2} 
                      \; \mathbf{1}_{2} \; \otimes \; \lambda_{j} \; ,
      \end{equation}
      where $i^{\prime}=5,8,10$, $j=2,5,7$.
    \end{itemize} 
    Since the hypercharge operator is normalised as
    \begin{equation}
     Y=\sqrt{21} L_{36}=\frac{1}{2} \text{diag}(1,1,1,1,1,1,-6) \; ,
    \end{equation}
    we can define the electric charge operator as usual
    \begin{equation}
     \label{electricchargeoperator}
     Q=T_{3}+Y \; .
    \end{equation}

  \section{Matter fields in the $SU(7)$ model}

    \label{su7matterfields}
 
    In this section we come to the fermionic content of the $SU(7)$ model. After symmetry breaking by
    orbifolding and imposing Dirichlet and Neumann boundary conditions the orbifold fixed points possess the 
    gauge symmetry
    \begin{equation}
     \label{matterfieldsunbrokensymmetry} 
     SU(2)_{L} \times U(1)_{Y} \times SO(3)_{F} \; .
    \end{equation}
    If we put matter fields at the orbifold fixed points, i.e. as brane fields, they have to transform
    according to the unbroken gauge group (\ref{matterfieldsunbrokensymmetry}) only and not according to
    the unified gauge group $SU(7)$. Therefore we can introduce the SM matter at the orbifold
    fixed points without any difficulty. In the following, by $(\mathbf{X},Y,\mathbf{Z})$ we denote the
    irreducible representations of $SU(2)_{L} \times U(1)_{Y} \times SO(3)_{F}$ where $Y$ denotes the
    hypercharge.   
    \begin{itemize}
     \item Left-handed quarks localised on the $L$-boundary \\
      \begin{equation}
       \label{lefthandedquarks}
       q_{L}=\left( u \atop d \right)_{L}=\left( \begin{matrix} u \\ c \\ t \\ d \\ s \\ b 
       \end{matrix} \right)_{L} \; : \; (\mathbf{2},\frac{1}{6},\mathbf{3}) \; .
      \end{equation}
     \item Left-handed leptons are localised on the $L$-boundary \\
      \begin{equation}
       l_{L}=\left( \nu_{e} \atop e \right)_{L}=\left( \begin{matrix} 
       \nu_{e} \\ \nu_{\mu} \\ \nu_{\tau} \\ e \\ \mu \\ \tau \end{matrix} \right)_{L} 
       \; : \; (\mathbf{2},-\frac{1}{2},\mathbf{3}) \; .
      \end{equation}
     \item Right-handed quarks localised on the $R$-boundary \\
      \begin{eqnarray}
       && u_{R}=\left( u_{R},c_{R},t_{R} \right) \; : \;  
          (\mathbf{1},\frac{2}{3},\mathbf{1}) \; , \\
       && d_{R}=\left( d_{R},s_{R},b_{R} \right) \; : \; 
          (\mathbf{1},-\frac{1}{3},\mathbf{1}) \; .
      \end{eqnarray}
      We put $u_{R}$ and $d_{R}$ together in a vector of two components 
      \begin{equation}
       \label{righthandedquarks}
       q_{R}=\left( u_{R} \atop d_{R} \right)=\begin{pmatrix} u_{R} \\ c_{R} \\ t_{R} \\ d_{R} \\ 
              s_{R} \\ b_{R} \end{pmatrix} \; .
      \end{equation}
     \item Right-handed leptons localised on the $R$-boundary \\
       \begin{eqnarray}
       \label{righthandednu}
       && \nu_{R}=\left( \nu_{R1},\nu_{R2},\nu_{R3} \right) 
          \; : \; (\mathbf{1},0,\mathbf{1}) \; , \\
       && e_{R}=\left( e_{R},\mu_{R},\tau_{R} \right) \; : \; (\mathbf{1},-1,\mathbf{1}) \; .
      \end{eqnarray}
      We put $\nu_{R}$ and $e_{R}$ together in a vector of two components 
      \begin{equation}
       l_{R}=\left( \nu_{R} \atop e_{R} \right)=\begin{pmatrix} \nu_{R1} \\ \nu_{R2} \\ \nu_{R3} \\ e_{R} \\ 
              \mu_{R} \\ \tau_{R} \end{pmatrix} \; .
      \end{equation}
    \end{itemize}
    Figure \ref{figuresu7model} summarises the assignment of matter fields. \\
    Remarks: i) \;Note that we have introduced also right-handed neutrinos (\ref{righthandednu}) in the 
    model. The reason is that we want to be able to give neutrinos a mass. This topic will be discussed 
    in detail in the next subsection. \\
    ii) \; Since left-handed matter transforms according to the $\mathbf{3}$ representation 
    of $SO(3)_{F}$ while right-handed matter transforms according to the $\mathbf{1}$ representation
    of $SO(3)_{F}$, the $SU(7)$ model is a model with a chiral gauged flavour symmetry. The reason why
    we need right-handed matter to transform according to the $\mathbf{1}$ representation 
    of $SO(3)_{F}$ will also be discussed in the next subsection. 
    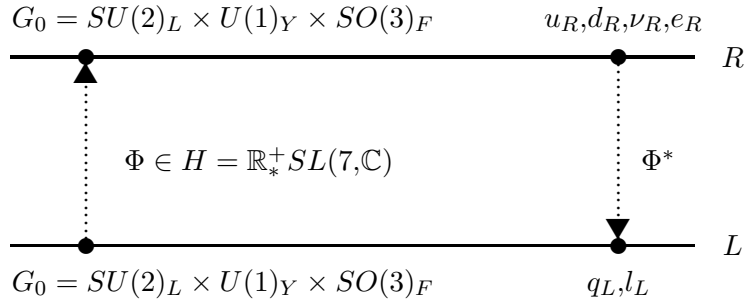
\begin{figure}[h]
     \begin{equation*}
      \begin{picture} (18,4.5)
       \thicklines
       \multiput(2,1.5) (0,2.5) {2} {\line (1,0) {9}}
       \thinlines
       \multiputlist(11.5,1.5)(0,2.5){$L$,$R$}
       \multiputlist(4.8,4.5)(5.2,0){$G_{0}=SU(2)_{L} \times U(1)_{Y} \times SO(3)_{F}$,
                                     $u_{R}{,}d_{R}{,}\nu_{R}{,}e_{R}$}
       \multiputlist(4.8,1)(5.2,0){$G_{0}=SU(2)_{L} \times U(1)_{Y} \times SO(3)_{F}$,$q_{L}{,}l_{L}$}
       \put(3.5,2.5) {$\Phi \in H=\mathbb{R}^{+}_{\ast} SL(7{,}\mathbb{C})$}
       \put(10.3,2.5) {$\Phi^{\ast}$}
       \matrixput (3,1.5)(7,0){2}(0,2,5){2}{\circle*{0.2}}
       \thicklines
       \dottedline[\circle*{0.05}]{0.1}(3,1.5)(3,4)
       \dottedline[\circle*{0.05}]{0.1}(10,1.5)(10,4)
       \thinlines
       \put(2.82,3.65){\Pisymbol{pzd}{115}}
       \put(9.82,1.58){\Pisymbol{pzd}{116}}
      \end{picture}
     \end{equation*}
     \caption{Assignment of SM matter fields in the eBTLM with bulk gauge group $SU(7)$.
              The unified gauge group $SU(7)$ is broken to $SU(2)_{L} \times U(1)_{Y} \times SO(3)_{F}$
              by orbifolding and imposing Dirichlet and Neumann boundary conditions. The
              holonomy group reads $H=\mathbb{R}^{+}_{\ast} SL(7{,}\mathbb{C})$. Since left- and right-handed
              matter transforms different under $SO(3)_{F}$ the $SU(7)$ model is a model with a chiral
              $SO(3)_{F}$ gauged flavour symmetry.}
     \label{figuresu7model}
    \end{figure}

  \subsection{Neutrino masses and the see-saw mechanism}

    In 1998 the Super-Kamiokande experiment \cite{Fukuda:1998mi} showed that muon neutrinos undergo 
    flavour oscillations. This implies that also neutrinos like charged fermions are massive.
    In order to give neutrinos a mass we have introduced right-handed neutrinos $\nu_{R}$ 
    (\ref{righthandednu}) in the model.
    However, if one introduces an ordinary Dirac mass term for neutrinos
    \begin{equation}
     \mathcal{L}_{mass}^{(\nu)}=m_{\nu} \bar{\nu}_{L} \nu_{R} \; ,
    \end{equation}
    one needs tiny Yukawa couplings which is a quite unnatural assumption. 
    A possible solution for this problem is the so called see-saw mechanism \cite{Gell-Mann:1980vs}. 
    The see-saw mechanism involves the introduction of an additional 
    Majorana mass term for $\nu_{R}$. In the SM a Majorana mass term for $\nu_{R}$  is possible since
    right-handed neutrinos carry no colour, weak isospin or hypercharge and thus are gauge-singlets.
    Note that a Majorana mass term breaks lepton number symmetry. 
 
    In the $SU(7)$ model there is an additional flavour gauge group $SO(3)_{F}$. 
    However, since we have introduced $\nu_{R}$ in the
    \begin{equation}
     (\mathbf{1},0,\mathbf{1})
    \end{equation}
    of $SU(2)_{L} \times U(1)_{Y} \times SO(3)_{F}$, right-handed neutrinos are in particular
    $SO(3)_{F}$ singlets.
    Therefore a Majorana mass term for $\nu_{R}$ is allowed and we can write
    \begin{equation}
     \mathcal{L}_{mass}^{(\nu)}=m_{dirac} \bar{\nu}_{L} \nu_{R} 
                                + \frac{1}{2} M_{major} \bar{\nu_{R}^{c}} \nu_{R} + h.c. \; ,
    \end{equation}
    where $\nu^{c}_{R}$ is the CP conjugate of $\nu_{R}$ (representing left-handed antineutrinos), $m_{dirac}$
    is the $3 \times 3$ Dirac mass matrix and the $M_{major}$ is the  heavy $3 \times 3$ Majorana mass 
    matrix. It is important that $M_{major}$ is not generated by a
    nonunitary parallel transporter. Therefore Majorana masses can be several orders of magnitude larger 
    than 
    ordinary quark and lepton masses.  An attractive assumption is that $M_{major}$ may be generated 
    somewhere at the GUT scale \cite{Buchmuller:2002xm,Mohapatra:2006gs}. In contrast, Dirac masses for all 
    leptons (see Proposal \ref{proposalyukawa} on page \pageref{proposalyukawa} for details)
    are given through Yukawa interactions by the nonunitary parallel transporter $\Phi^{lepton}$.
    In this setup the see-saw mechanism works as usual.  
    We note that if we had introduced right-handed SM matter $q_{R}$ and $l_{R}$ in the 
    $\mathbf{3}$ of $SO(3)_{F}$ a Majorana mass term for $\nu_{R}$ would not be possible and the
    see-saw mechanism would not work.

  \section{Anomalies in the $SU(7)$ model}

   \label{anomalycancellation}

   In this section, we discuss the topic of anomalies and anomaly cancellation in the $SU(7)$ model.
   Since we have introduced
   (chiral) SM matter at {\em{different}} orbifold fixed points the cancellation of anomalies in the 
   $SU(7)$ model is in contrast to the SM not automatic. Before we come in detail to the $SU(7)$ model
   we first discuss the issue of anomaly cancellation in orbifold models more generally. For simplicity
   we focus on five-dimensional orbifolds.
   In orbifold theories two types of anomalies can arise:
   \begin{itemize}
     \item
      four-dimensional anomalies intrinsic to the orbifold fixed points.
     \item
      five-dimensional anomalies intrinsic to the bulk.
   \end{itemize}
   For the low energy consistency of the theory, it is necessary that both the anomaly at the orbifold 
   fixed points {\em{and}} the anomaly in the bulk cancels. 
   Let us assume that the four-dimensional anomaly in the effective low-energy
   theory cancels. We then may ask: Is the cancellation of the four-dimensional
   anomaly sufficient to cancel also the five-dimensional anomaly? As it has be worked out by Arkani-Hamed
   and others \cite{Arkani-Hamed:2001is,Scrucca:2001eb,Pilo:2002hu} this is indeed the case. 
   More precisely, for a collection of five-dimensional fermions all one has to care about is that
   their {\em{zero modes}} form an anomaly-free representation of the low-energy four-dimensional gauge 
   group. This means that the five-dimensional anomaly is independent of the details of the physics
   in the bulk.

   \subsection{Anomaly cancellation mechanisms in the $SU(7)$ model}

    We now come in detail to anomaly cancellation mechanisms in the $SU(7)$ model.
    In section \ref{su7matterfields}, we have introduced chiral SM matter at the orbifold fixed points.
    For the following discussion let us initially ignore the flavour gauge group $SO(3)_{F}$.
    We consider two different scenarios.
   
    In the first scenario, we put {\em{both}} left- and
    right-handed SM matter on the same orbifold fixed point. Without loss of generality
    let this orbifold fixed point
    be the $L$-boundary. Figure \ref{figureanomalie1} summarises the setting.
    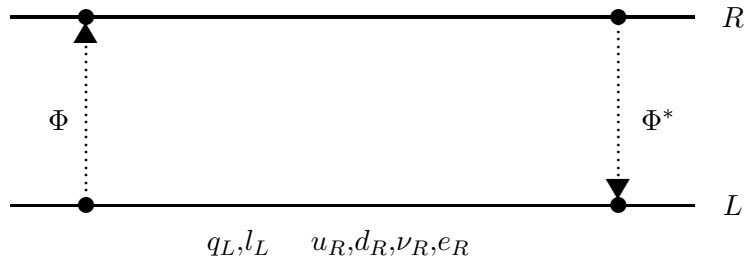
\begin{figure}[h]
     \begin{equation*}
      \begin{picture} (18,4.5)
       \thicklines
       \multiput(2,1.5) (0,2.5) {2} {\line (1,0) {9}}
       \thinlines
       \multiputlist(11.5,1.5)(0,2.5){$L$,$R$}
       \multiputlist(5,1)(2,0){$q_{L}{,}l_{L}$,$u_{R}{,}d_{R}{,}\nu_{R}{,}e_{R}$}
       \put(2.5,2.5) {$\Phi$}
       \put(10.3,2.5) {$\Phi^{\ast}$}
       \matrixput (3,1.5)(7,0){2}(0,2,5){2}{\circle*{0.2}}
       \thicklines
       \dottedline[\circle*{0.05}]{0.1}(3,1.5)(3,4)
       \dottedline[\circle*{0.05}]{0.1}(10,1.5)(10,4)
       \thinlines
       \put(2.82,3.65){\Pisymbol{pzd}{115}}
       \put(9.82,1.58){\Pisymbol{pzd}{116}}
      \end{picture}
     \end{equation*}
     \caption{Anomaly cancellation scenario: Left- and right-handed SM matter on the same orbifold fixed 
              point.}
     \label{figureanomalie1}
    \end{figure}
    We observe that all anomalies cancel {\em{locally}}, in particular at the $L$-boundary, 
    thanks to the usual cancellation of anomalies in the SM. 
    Note that this scenario is {\em{not}} assumed in the $SU(7)$ model.

    In the second scenario, we put left-handed SM matter on the $L$-boundary and right-handed SM matter 
    on the $R$-boundary. This scenario is known as chiral delocalisation \cite{Hill:2006ei}. 
    This is exactly what we have adopted in the $SU(7)$ model.
    This means that since the anomaly of three $U(1)$ gauge bosons is nonzero \cite{Huang:1982ik}
    \begin{equation}
     \label{threeu1anomaly}
     \text{Tr} \left[ Y_{L}^{3} \right] = -\frac{2}{9} \; , \; \text{Tr}
     \left[ Y_{R}^{3} \right] = -\frac{2}{9} \; ,  
    \end{equation}
    this scenario leads to localised SM anomalies at the different orbifold fixed points.
    Figure \ref{figureanomalies2} summarises the setting.
    \begin{figure}[h]
     \begin{equation*}
      \begin{picture} (18,4.5)
       \thicklines
       \multiput(2,1.5) (0,2.5) {2} {\line (1,0) {9}}
       \thinlines
       \multiputlist(11.5,1.5)(0,2.5){$L$,$R$}
       \multiputlist(4,4.5)(4,0){$u_{R}{,}d_{R}{,}\nu_{R}{,}e_{R}$,
                                 $\text{Tr}\left[Y_{R}^{3} \right]=-\frac{2}{9}$}
       \multiputlist(4,1)(4,0){$q_{L}{,}l_{L}$,$\text{Tr}\left[Y_{L}^{3} \right]=-\frac{2}{9}$}
       \put(2.5,2.5) {$\Phi$}
       \put(10.3,2.5) {$\Phi^{\ast}$}
       \matrixput (3,1.5)(7,0){2}(0,2,5){2}{\circle*{0.2}}
       \thicklines
       \dottedline[\circle*{0.05}]{0.1}(3,1.5)(3,4)
       \dottedline[\circle*{0.05}]{0.1}(10,1.5)(10,4)
       \thinlines
       \put(2.82,3.65){\Pisymbol{pzd}{115}}
       \put(9.82,1.58){\Pisymbol{pzd}{116}}
      \end{picture}
     \end{equation*}
     \caption{Anomaly cancellation scenario: Left-handed SM matter on the $L$-boundary, right-handed
              SM matter on the $R$-boundary. The emerging anomalies are inscribed.}  
     \label{figureanomalies2}
    \end{figure}
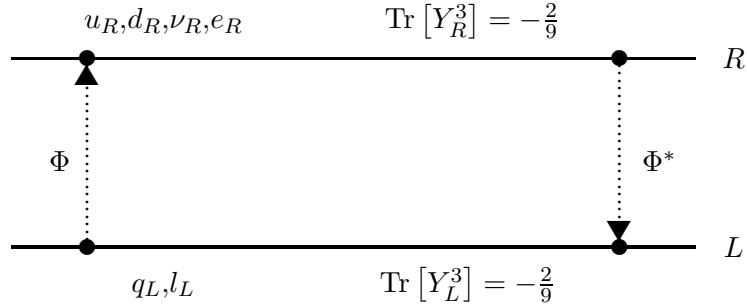
    At first sight, this scenario leads to an inconsistent theory.
    However, one can introduce a bulk Chern-Simons term with a jumping coefficient 
    \cite{Barbieri:2002ic,Pilo:2002hu,Scrucca:2001eb,Arkani-Hamed:2001is,Hill:2006ei}
    in order to {\em{locally}}
    cancel the SM anomalies arising from the three $U(1)$ gauge bosons. Note that this anomaly
    cancellation mechanism works only 
    if the {\em{integrated}} anomaly, i.e. the sum over all local contributions to the anomaly, vanishes.
    Due to (\ref{threeu1anomaly}), this is indeed the case. The work of Arkani-Hamed 
    and others \cite{Arkani-Hamed:2001is} describes a mechanism how the Chern-Simons term can be generated 
    by integrating out massive
    bulk fermions which transforms according to an anomaly-free representation of the low energy
    four-dimensional gauge group.

  \subsection{Contributions to the anomaly from the flavour gauge group $SO(3)_{F}$}

    In this subsection, we include the flavour gauge group $SO(3)_{F}$ in our 
    considerations. In section \ref{su7matterfields}, we have introduced left-handed SM matter 
    in the $\mathbf{3}$ of $SO(3)_{F}$ and right-handed matter SM 
    in the $\mathbf{1}$ of $SO(3)_{F}$. The $\mathbf{3}$ of $SO(3)_{F}$ is formed by the
    generators  
    \begin{equation}
     H_{j}=\frac{1}{2} \; \mathbf{1}_{2} \; \otimes \; \lambda_{j} \; ,
    \end{equation}
    where $j=2,5,7$. Since they fulfil
    \begin{equation}
     \text{Tr} \left[ \{ H_{i} , H_{j} \} H_{k} \right]=0 \; ,
    \end{equation}
    where $i,j,k \in \{ 2,5,7 \}$, the $\mathbf{3}$ of $SO(3)_{F}$
    is an anomaly-free representation of $SO(3)_{F}$. Thus the $SU(7)$ model with
    left- and right-handed matter introduced as in section \ref{su7matterfields} is 
    {\em{free of anomalies}}.   

    Remarks: i) \; Concerning anomalies it doesn't matter that we have introduced left-handed SM matter and 
    right-handed SM matter in different representation of $SO(3)_{F}$. \\ 
    ii)\; If we had used the flavour gauge group $SU(3)_{F}$ instead of $SO(3)_{F}$ the model would contain a
    non-vanishing anomaly due to
    \begin{equation}
     Tr \left[ \{ \lambda_{i} , \lambda_{j} \} \lambda_{k} \right]=4 i d_{ijk} \; ,
    \end{equation}
    where $d_{ijk}$ are the completely symmetric coefficients of $\mathfrak{su}_{3}$.

  \section{Orbifold breaking in the $SU(7)$ model and the electroweak Higgs}

    In this section we describe how a symmetry breaking by a non-trivial orbifold 
    projection $P$ will not only break the unitary gauge group $SU(7)$ down to
    $SU(6)_{L} \times U(1)_{Y}$ but also leads to non-trivial unitary factors  $e^{A_{y}}$ in the 
    decomposition $\Phi =\rho \; e^{A_{y}} \; e^{\eta} \; e^{A_{y}}$. In order to break $SU(7)$ down to
    $SU(6)_{L} \times U(1)_{Y}$ we choose 
    \begin{equation}
     \label{orbifoldprojectionsu7}
     P=\text{diag}(-1,-1,-1,-,1-,1,-1,1) \; .
    \end{equation}
    The branching rule for the adjoint representation $\mathbf{48}$ of $SU(7)$ with respect to 
    $SU(6)_{L}$ reads
    \begin{equation}
     \label{branchingrule}
     \mathbf{48} \rightarrow \mathbf{35} + \mathbf{6} + \mathbf{\bar{6}} 
     + \mathbf{1} \; ,
    \end{equation}
    where $\mathbf{35}$ is the adjoint representation of $SU(6)_{L}$, $\mathbf{1}$ is 
    the trivial representation of $SU(6)_{L}$ and $\mathbf{6}$ ($\mathbf{\bar{6}}$)
    is the fundamental (complex conjugate fundamental) representation of $SU(6)_{L}$. 
    According to the action of $P$ on $\mathfrak{g}=\mathfrak{su}(7)$ we have 
    \begin{equation}
     \mathfrak{g}_{0}=\mathbf{35} + \mathbf{1} \; , 
     \quad  \mathfrak{g}_{1}=\mathbf{6} + \mathbf{\bar{6}} \; .
    \end{equation}
    For $G=SU(7)$ the corresponding holonomy group reads $H=\mathbb{R}^{+}_{\ast} SL(7{,}\mathbb{C})$ 
    \footnote{Note that the linear span of $SU(7)$ reads
    $\mathbb{R}^{+}_{\ast} \; SL(7,\mathbb{C})$ and thus $H$ is unique, compare with
    (\ref{blockspin})}. The Cartan decomposition
    for $\mathfrak{sl}(7{,}\mathbb{C})=\text{Lie} \; SL(7{,}\mathbb{C})$ reads
    \begin{equation}
     \mathfrak{sl}(7{,}\mathbb{C})=\mathfrak{su}(7) + i \mathfrak{su}(7)=\mathfrak{g} + i \mathfrak{g} \; .
    \end{equation}
    We choose an $\mathfrak{a} \subset i \mathfrak{su}(7)$ such that $P \eta P=\eta$
    is automatically fulfilled. In fact for 
    \begin{equation}
     \mathfrak{a}=\{ \eta=\text{diag}(a_{1},a_{2},a_{3},a_{4},a_{5},a_{6},a_{7}) \} 
                    \quad , \quad \sum_{i} a_{i}=0 \; , \quad a_{i} \in \mathbb{R} \; ,
    \end{equation}
    $P \eta P=\eta$ holds for any $\eta \in \mathfrak{a}$.
    Let $\Phi \in \mathbb{R}^{+}_{\ast} SL(7{,}\mathbb{C})$ fulfil the sharpened orbifold condition.
    Then $\Phi$ can be written as
    \begin{equation}
     \label{nonuniptsu7}
     \Phi =\rho \; e^{A_{y}} \; e^{\eta} \; e^{A_{y}} \; ,
    \end{equation}
    where
    \begin{eqnarray} 
     P A_{\mu}^{L(R)} P^{-1} & = & A_{\mu}^{L(R)} \; , \\
     P A_{y} P^{-1} & = & -A_{y}  \; ,\\
    \end{eqnarray}
    for $A_{\mu}^{L(R)} \in \mathbf{35} + \mathbf{1}$, $A_{y} \in \mathbf{6} + \mathbf{\bar{6}}$ and 
    $\rho \in  \mathbb{R}^{+}_{\ast}$. 
    This means that the generators $L_{1},\dots,L_{36}$ of the $SU(6)_{L} \times U(1)_{Y}$ subgroup
    of $SU(7)$  remain unbroken and the gauge fields $A_{\mu}^{L,R}$ can be written as
    \begin{equation}
     \label{su6u1gaugefields}
     A_{\mu}^{L(R)}=\sum_{a=1}^{36} A_{\mu}^{L(R) a} L_{a}  \; .
    \end{equation}

   \subsection{The electroweak Higgs}  

    Let us consider the unitary factor $e^{A_{y}}$ in (\ref{nonuniptsu7}). According to Definition
    \ref{defintionextradimvecpot} (see page \pageref{defintionextradimvecpot}), $e^{A_{y}}$ can be written as
    \begin{equation}   
     \label{contiwilosnsu7}
     \exp(A_{y})=\exp(i \; g_{4} R \; \mathcal{A}_{y}^{(0)}) \; ,
    \end{equation}
    where $\mathcal{A}_{y}^{(0)}$ is the zero mode 
    of the extra-dimensional component of the five-dimensional gauge field in its canonical 
    four-dimensional form.
    With $A_{y}^{(0)} \in \mathbf{6} + \mathbf{\bar{6}}$ we can expand  $\mathcal{A}_{y}^{(0)}$ as   
    \begin{equation}
     \mathcal{A}_{y}^{(0)}=\sum_{\hat{a}=37}^{48} \mathcal{A}_{y}^{\hat{a}(0)} L_{\hat{a}}  \; ,
    \end{equation}
    where $L_{37},\dots,L_{48}$ are the generators of the coset space $SU(7)/SU(6)_{L} \times U(1)_{Y}$.
    We determine the hypercharge of $\mathcal{A}_{y}^{(0)}$. Looking at the generators
    $L_{37},\dots,L_{48}$, e.g. 
    \begin{equation}
     L_{37}=\frac{1}{2} \; \begin{pmatrix} 
               0 & 0 & 0 & 0 & 0 & 0 & \mathbf{1} \\
               0 & 0 & 0 & 0 & 0 & 0 & \mathbf{0} \\
               0 & 0 & 0 & 0 & 0 & 0 & \mathbf{0} \\
               0 & 0 & 0 & 0 & 0 & 0 & \mathbf{0} \\
               0 & 0 & 0 & 0 & 0 & 0 & \mathbf{0} \\
               0 & 0 & 0 & 0 & 0 & 0 & \mathbf{0} \\
               \mathbf{1} & \mathbf{0} & \mathbf{0} & \mathbf{0} & \mathbf{0} & \mathbf{0} & \mathbf{0}    
              \end{pmatrix} \; ,
    \end{equation}
    we see that $\mathcal{A}_{y}^{(0)}$ carries the hypercharge 
    \begin{equation}
     Y=\frac{1}{2} \; .
    \end{equation}      
    Next we consider the fundamental representation $\mathbf{6}$ of $SU(6)$. The fundamental representation
    $\mathbf{6}$ will remain irreducible when restricted to $SU(2)_{L} \times SU(3)_{F} 
    \left( SO(3)_{F} \right)$ 
    \begin{equation}
     \mathbf{6}=\left( \mathbf{2}, \mathbf{3} \right) \; .
    \end{equation}
    This means that $\mathcal{A}_{y}^{(0)}$ is a doublet with respect to $SU(2)_{L}$ and a triplet with 
    respect to $SU(3)_{F} \left( SO(3)_{F} \right)$, respectively. We explicitly write
    \begin{eqnarray}
     \label{mathcalay}
     \mathcal{A}_{y}^{(0)} & = & \sum_{\hat{a}=37}^{48} \mathcal{A}_{y}^{\hat{a}(0)} L_{\hat{a}} \\
     & = & \frac{1}{2} 
      \begin{pmatrix} 
      \begin{matrix} 0 \quad \;  & 0 \quad \; & 0 \quad  \; \\ 0 \quad \; & 0 \quad \; & 0 \quad \; \\ 
                     0 \quad \; \; & 0 \quad \; & 0 \quad \; \end{matrix} & 
      \begin{matrix} 0 \quad \; & 0 \quad \; & 0 \quad \; \\ 0 \quad \; & 0 \quad \; & 0 \quad \; \\ 
                     0 \quad \; & 0 \quad \; & 0 \quad \; \end{matrix} &
      \begin{pmatrix} \mathcal{A}_{y}^{1+(0)} \\ 
                      \mathcal{A}_{y}^{2+(0)} \\  
                      \mathcal{A}_{y}^{3+(0)} \end{pmatrix} \\
      \begin{matrix} 0 \quad \; & 0 \quad \; & 0 \quad \; \\ 0 \quad \; & 0 \quad \; & 0 \quad \; \\ 
                     0 \quad \; & 0 \quad \; & 0 \quad \; \end{matrix} &
      \begin{matrix} 0 \quad \; & 0 \quad \; & 0 \quad \; \\ 0 \quad \; & 0 \quad \; & 0 \quad \; \\ 
                     0 \quad \; & 0 \quad \; & 0 \quad \; \end{matrix} &
      \begin{pmatrix} \mathcal{A}_{y}^{10(0)}  \\  
                      \mathcal{A}_{y}^{20(0)}  \\  
                      \mathcal{A}_{y}^{30(0)} \end{pmatrix} \\
      \begin{pmatrix} \mathcal{A}_{y}^{1+(0) \ast} & 
                      \mathcal{A}_{y}^{2+(0) \ast} &  
                      \mathcal{A}_{y}^{3+(0) \ast} \\ 
                       \end{pmatrix} &
      \begin{pmatrix} \mathcal{A}_{y}^{10(0) \ast} & 
                      \mathcal{A}_{y}^{20(0) \ast} &  
                      \mathcal{A}_{y}^{30(0) \ast} \end{pmatrix} \\  
     \end{pmatrix} \; , \nonumber
    \end{eqnarray}
    where 
    \begin{gather*}
     \mathcal{A}_{y}^{1+(0)}=\mathcal{A}_{y}^{37(0)}-i \mathcal{A}_{y}^{38(0)} \; , \;
     \mathcal{A}_{y}^{2+(0)}=\mathcal{A}_{y}^{39(0)}-i \mathcal{A}_{y}^{40(0)} \; , \;
     \mathcal{A}_{y}^{3+(0)}=\mathcal{A}_{y}^{41(0)}-i \mathcal{A}_{y}^{42(0)} \\
     \mathcal{A}_{y}^{10(0)}=\mathcal{A}_{y}^{43(0)}-i \mathcal{A}_{y}^{44(0)} \; , \;
     \mathcal{A}_{y}^{20(0)}=\mathcal{A}_{y}^{45(0)}-i \mathcal{A}_{y}^{46(0)} \; , \;
     \mathcal{A}_{y}^{30(0)}=\mathcal{A}_{y}^{47(0)}-i \mathcal{A}_{y}^{48(0)} \; .
    \end{gather*}
    Thus  
    $\mathcal{A}_{y}^{1 +(0)}$, $\mathcal{A}_{y}^{2 +(0)}$, $\mathcal{A}_{y}^{3 +(0)}$ form the
    up-component of the doublet while
    $\mathcal{A}_{y}^{1 0(0)}$, $\mathcal{A}_{y}^{2 0(0)}$, $\mathcal{A}_{y}^{3 0(0)}$ 
    form the down-component of the doublet. 
    We see from (\ref{mathcalay}) that the $SU(7)$ model contains three $SU(2)_{L}$ doublets,
    one for each flavour. This comes out if we we consider the mass terms for quarks and leptons,
    respectively. In the next chapter we will discuss the topic of fermion masses in detail.
    However, anticipating a little, the mass term for quarks is given by 
    \begin{equation}
     \bar{q}_{L} \Phi q_{R}=\overline{\left( u \; c \; t \; d \; s \; b \right)}_{L}
     \left( \rho \; e^{A_{y}} \; e^{\eta} \; e^{A_{y}} \right)
     \begin{pmatrix} u_{R} \\ c_{R} \\ t_{R} \\ d_{R} \\ s_{R} \\ b_{R} \end{pmatrix}  \; .
    \end{equation}
    We see that 
    \begin{equation}
     \label{higgsdoubleth1}
     H_{1}=\left( \mathcal{A}_{y}^{1+(0)} \atop \mathcal{A}_{y}^{10(0)} \right)
     =\left( \mathcal{A}_{y}^{37(0)}-i \mathcal{A}_{y}^{38(0)} 
      \atop \mathcal{A}_{y}^{43(0)}-i \mathcal{A}_{y}^{44(0)}
      \right)
    \end{equation}
    is a $SU(2)_{L}$ Higgs doublet coupled to $u,d,\nu_{e},e$,
    \begin{equation}
     \label{higgsdoubleth2}
     H_{2}=\left( \mathcal{A}_{y}^{2+(0)} \atop \mathcal{A}_{y}^{20(0)} \right)
     =\left( \mathcal{A}_{y}^{39(0)}-i \mathcal{A}_{y}^{40(0)} 
      \atop \mathcal{A}_{y}^{45(0)}-i \mathcal{A}_{y}^{46(0)}
      \right)
    \end{equation}
    is a $SU(2)_{L}$ Higgs doublet coupled to $c,s,\nu_{\mu},\mu$ and 
    \begin{equation}
     \label{higgsdoubleth3}
     H_{3}=\left( \mathcal{A}_{y}^{3+(0)} \atop \mathcal{A}_{y}^{30(0)} \right)
     =\left( \mathcal{A}_{y}^{41(0)}-i \mathcal{A}_{y}^{42(0)} 
      \atop \mathcal{A}_{y}^{47(0)}-i \mathcal{A}_{y}^{48(0)}
      \right)
    \end{equation}
    is a $SU(2)_{L}$ Higgs doublet coupled to $t,b,\nu_{\tau},\tau$.
    Thus the model includes three $SU(2)_{L}$ doublets $\{ H_{i} \}$, $i=1,2,3$, one for each flavour.    
    Finally we determine the electric charge of $\mathcal{A}_{y}^{(0)}$ and $\{ H_{i} \}$, respectively.
    For the up-component $\mathcal{A}^{1 +(0)}_{y},\mathcal{A}^{2 +(0)}_{y},\mathcal{A}^{3 +(0)}_{y}$
    of the $\{ H_{i} \}$ we obtain 
    \begin{equation}
     Q=T_{3}+Y=\frac{1}{2}+\frac{1}{2}=1 \; ,
    \end{equation}
    and for the down-component $\mathcal{A}^{1 0(0)}_{y},\mathcal{A}^{2 0(0)}_{y},\mathcal{A}^{3 0(0)}_{y}$
    of the $\{ H_{i} \}$ we obtain  
    \begin{equation}
     Q=T_{3}+Y=\frac{1}{2}-\frac{1}{2}=0 \; .
    \end{equation} 
   
    We summarise: The zero modes of  
    extra-dimensional component of the five-dimensional 
    gauge field $\mathcal{A}_{y}^{(0)}$ have the following properties
    \begin{itemize}
     \item They appear from the four-dimensional point of view as scalar fields.
     \item They are doublets with respect to the weak SM gauge group $SU(2)_{L}$.
     \item They carry hypercharge $\frac{1}{2}$.
     \item Their $SU(2)_{L}$ down-component is electrically neutral and their $SU(2)_{L}$
           up-component has electric charge $+1$.
     \item They include three $SU(2)_{L}$ doublets $\{ H_{i} \}$, $i=1,2,3$, one for first, one for the 
           second and one for the third generation.
    \end{itemize}
    {\bf{Conclusion:}} 
    {\em{ The zero modes of the extra-dimensional component of the five-dimensional 
    gauge field $\mathcal{A}_{y}^{(0)}=
    \sum_{\hat{a}=37}^{48} \mathcal{A}_{y}^{\hat{a}(0)} L_{\hat{a}}$ 
    are a substitute for the SM Higgs. $\mathcal{A}_{y}^{(0)}$ includes three Higgs doublets $\{ H_{i} \}$, 
    $i=1,2,3$, one for each flavour. They generate the unitary factors $e^{A_{y}}$ via
    $A_{y}=i \; g_{4} R \; \mathcal{A}_{y}^{(0)}$ in the decomposition
    $\Phi=\rho \; e^{A_{y}} e^{\eta} e^{A_{y}}$}}. \\

    Making use of the residual $SU(6)_{L} \times U(1)_{Y}$ global symmetry, it is possible to
    transform away the up-components of the $SU(2)_{L}$ doublets
    \begin{gather}
     \mathcal{A}_{y}^{(0)}= \sum_{\hat{a}=37}^{48} \mathcal{A}_{y}^{\hat{a}(0)} L_{\hat{a}}=\frac{1}{2} \;  
      \left( \begin{array} {ccccccc} 0 & 0 & 0 & 0 & 0 & 0 & \mathcal{A}_{y}^{1 +(0)} \\
        0 & 0 & 0 & 0 & 0 & 0 & \mathcal{A}_{y}^{2 +(0)} \\ 0 & 0 & 0 & 0 & 0 & 0 & 
        \mathcal{A}_{y}^{3 +(0)} \\
        0 & 0 & 0 & 0 & 0 & 0 & \mathcal{A}_{y}^{1 0(0)} \\ 0 & 0 & 0 & 0 & 0 & 0 & 
        \mathcal{A}_{y}^{2 0(0)} \\
        0 & 0 & 0 & 0 & 0 & 0 & \mathcal{A}_{y}^{3 0(0)} \\ 
        \mathcal{A}_{y}^{1 + (0) \ast} & \mathcal{A}_{y}^{2 + (0) \ast} & \mathcal{A}_{y}^{3 + (0) \ast} & 
        \mathcal{A}_{y}^{1 0 (0) \ast} &
        \mathcal{A}_{y}^{2 0 (0) \ast} & \mathcal{A}_{y}^{3 0 (0) \ast} & 0 \end{array} \right) \nonumber \\
      \Downarrow \nonumber \\
      \label{a5unitarygauge}
      \mathcal{A}_{y}^{(0)}= \sum_{\hat{a}=43}^{48} \mathcal{A}_{y}^{\hat{a}(0)} L_{\hat{a}}=\frac{1}{2}
      \left( \begin{array} {ccccccc} 0 & 0 & 0 & 0 & 0 & 0 & 0 \\
        0 & 0 & 0 & 0 & 0 & 0 & 0 \\ 0 & 0 & 0 & 0 & 0 & 0 & 0 \\
        0 & 0 & 0 & 0 & 0 & 0 & \mathcal{A}_{y}^{1 0 (0)} \\ 0 & 0 & 0 & 0 & 0 & 0 & 
        \mathcal{A}_{y}^{2 0 (0)} \\
        0 & 0 & 0 & 0 & 0 & 0 & \mathcal{A}_{y}^{3 0 (0)} \\ 
        0 & 0 & 0 & 
        \mathcal{A}_{y}^{1 0 (0) \ast} &
        \mathcal{A}_{y}^{2 0 (0) \ast} & \mathcal{A}_{y}^{3 0 (0) \ast} & 0 \end{array} \right) \; .
    \end{gather}
    This transformation results in a vanishing mass term for the photon and is in analogy to the SM.

  \section[Gauge symmetry breaking by boundary conditions]{Additional gauge symmetry breaking 
           by Dirichlet and Neumann boundary conditions}

    In this section we describe how the symmetry breaking
    \begin{equation}
     SU(6)_{L} \times U(1)_{Y} \to SU(2)_{L} \times U(1)_{Y} \times SO(3)_{F}
    \end{equation}
    can be achieved by imposing Dirichlet and Neumann boundary conditions for the
    gauge fields (\ref{su6u1gaugefields}) which are unaffected by the orbifold projection $P$.
    The reason why we need this additional symmetry breaking is that the orbifold $S^{1}/\mathbb{Z}_{2}$
    possesses only one orbifold projection $P$. Note that the underlying orbifold in the $SU(7)$ model 
    is consider as the orbifold $S^{1}/\mathbb{Z}_{2}$ with twisted boundary conditions, see
    discussion in section \ref{sectionorbifolds1z2}, and thus we
    have one orbifold projection $P$ given by (\ref{orbifoldprojectionsu7}) and one continuous Wilson line 
    $W$ given by (\ref{wilsonsu7}). 
    
    In the next subsection we first describe the issue of gauge symmetry breaking through 
    Dirichlet and Neumann boundary conditions more generally. We define our theory in five dimensions
    between two parallel branes. One brane is located at $y=0$ and the other brane is located at $y=\pi R$.
    The two branes are considered as four-dimensional boundaries. This way we can compare 
    gauge symmetry breaking through Dirichlet and Neumann boundary conditions to gauge symmetry
    breaking on the orbifold $S^{1}/\mathbb{Z}_{2}$. As in the orbifold case $y$ denote the
    coordinate of the extra dimension.

   \subsection{Gauge symmetry breaking by Dirichlet and Neumann boundary conditions and its relation
               to gauge symmetry through orbifolding}

    Let $G$ be the gauge group we want to break with Lie algebra $\mathfrak{g}$. In addition, let
    $G_{0}$ be the subgroup of $G$ we want to obtain as the unbroken gauge group  
    with Lie algebra $\mathfrak{g}_{0}$. 
    By $\{T^{A}\}$ we denote the set of generators creating $G$ and by $\{T^{a}\}$ the set of generators
    creating $G_{0}$. We consider the split 
    \begin{equation}
     \mathfrak{g}=\mathfrak{g}_{0}+\mathfrak{g}_{1} \; ,
    \end{equation}
    where $\mathfrak{g}_{1}$ generate the coset space $G/G_{0}$. By $\{ T^{\hat{a}} \}$ we denote
    the set of generators of the coset space $G/G_{0}$. In order to achieve the symmetry breaking 
    $G \to G_{0}$ we demand 
    \begin{equation}
     A_{\mu}^{\hat{a}}=0 \quad , \quad \partial_{y} A_{\mu}^{a}=0
    \end{equation}
    at both boundaries $y=0$ and $y=\pi R$.
    \begin{itemize}
     \item $A_{\mu}^{\hat{a}}=0$ are Dirichlet boundary conditions for the broken gauge
           fields $A_{\mu}^{\hat{a}}$.
     \item $\partial_{y} A_{\mu}^{a}=0$ are Neumann boundary conditions for the
           unbroken gauge fields $A_{\mu}^{a}$. 
    \end{itemize} 

    We compare this gauge symmetry breaking through boundary conditions to the gauge symmetry breaking on
    the orbifold $S^{1}/\mathbb{Z}_{2}$. Recall that on
    the orbifold $S^{1}/\mathbb{Z}_{2}$ gauge and scalar fields have to fulfil the
    boundary conditions (\ref{boundaryconditionsfromtwistp}), (\ref{boundaryconditionsfromtwistpsc}) 
    \begin{gather}
     \label{boundaryconditionsfromtwistpsu7}
     A_{\mu}(x^{\mu},-y)=P \; A_{\mu}(x^{\mu},y) \; P^{-1}  \\
     \label{boundaryconditionsfromtwistpscsu7}
     A_{y}(x^{\mu},-y)=-P \; A_{y}(x^{\mu},y) \; P^{-1} 
    \end{gather}
    and the periodicity condition (\ref{periodictyforsu7})
    \begin{equation}
     \label{periodboundarsu7}
     A_{M}(x^{\mu},y+2 \pi R)=W \; A_{M}(x^{\mu},y) \; W^{-1} \; ,
    \end{equation}
    In following discussion we admit the trivial periodicity condition, i.e. we set $W=1$ in 
    (\ref{periodboundarsu7}). 
    The boundary condition (\ref{boundaryconditionsfromtwistpsu7}) breaks the bulks gauge group $G$ down to 
    $G_{0}^{\prime}$ 
    \begin{equation} 
     G_{0}^{\prime}=\{ g \in G \mid Pg=gP \}
    \end{equation}
    at $y=0$. 
    Let $\{ T^{a^{\prime}}\}$ denote the set of generators creating $G_{0}^{\prime}$ and let
    $\{ T^{\hat{a}^{\prime}} \}$ denote the set of generators creating the coset space $G/G_{0}^{\prime}$. 
    According to (\ref{boundaryconditionsfromtwistpsu7}) and (\ref{boundaryconditionsfromtwistpscsu7}) 
    unbroken gauge $A^{a^{\prime}}_{\mu}(x^{\mu},y)$ and the scalar fields 
    $A^{\hat{a}^{\prime}}_{y}(x^{\mu},y)$
    are even functions, i.e.
    \begin{gather}
     A^{a^{\prime}}_{\mu}(x^{\mu},-y)=A^{a^{\prime}}_{\mu}(x^{\mu},y) \\
     A^{\hat{a}^{\prime}}_{y}(x^{\mu},-y)=A^{\hat{a}^{\prime}}_{y}(x^{\mu},y)  \nonumber \; .
    \end{gather}
    Thus we can Fourier expand  
    \begin{eqnarray}
     \label{fourierevenfunctionsu7}
     && A^{a^{\prime}}_{\mu}(x^{\mu},y)=\frac{1}{\sqrt{2 \pi R}}
     A_{\mu}^{a^{\prime}}{}^{(0)}(x^{\mu})+\frac{1}{\sqrt{\pi R}}\sum_{n=1}^{\infty} 
     A_{\mu}^{a^{\prime}}{}^{(n)}(x^{\mu}) \cos(\frac{n y}{R}) \\
     && A^{\hat{a}^{\prime}}_{y}(x^{\mu},y)=\frac{1}{\sqrt{2 \pi R}}
     A_{y}^{\hat{a}^{\prime}}{}^{(0)}(x^{\mu})+ \frac{1}{\sqrt{\pi R}}
     \sum_{n=1}^{\infty} 
     A_{y}^{\hat{a}^{\prime}}{}^{(n)}(x^{\mu}) \cos(\frac{n y}{R})  \nonumber \; .
    \end{eqnarray} 
    On the other hand, according to (\ref{boundaryconditionsfromtwistpsu7}) and 
    (\ref{boundaryconditionsfromtwistpscsu7}) broken gauge $A^{\hat{a}^{\prime}}_{\mu}(x^{\mu},y)$ and
    the scalar fields $A^{a^{\prime}}_{y}(x^{\mu},y)$ are odd functions, i.e.
    \begin{gather}
     A^{\hat{a}^{\prime}}_{\mu}(x^{\mu},-y)=-A^{\hat{a}^{\prime}}_{\mu}(x^{\mu},y) \\
     A^{a^{\prime}}_{y}(x^{\mu},-y)=-A^{a^{\prime}}_{y}(x^{\mu},y)   \nonumber \; .
    \end{gather}
    Thus we can Fourier expand 
    \begin{eqnarray}
     \label{fourieroddfunctionsu7}
     && A^{\hat{a}^{\prime}}_{\mu}(x^{\mu},y)=\frac{1}{\sqrt{\pi R}}\sum_{n=1}^{\infty} 
     A_{\mu}^{\hat{a}^{\prime}}{}^{(n)}(x^{\mu}) \sin(\frac{n y}{R}) \\ 
     && A^{a^{\prime}}_{y}(x^{\mu},y)=\frac{1}{\sqrt{\pi R}}\sum_{n=1}^{\infty} 
     A_{y}^{a^{\prime}}{}^{(n)}(x^{\mu}) \sin(\frac{n y}{R})  \nonumber \; .
    \end{eqnarray} 
    The expansions (\ref{fourieroddfunctionsu7}) lead to
    Dirichlet boundary conditions for broken gauge and scalar fields, i.e. 
    \begin{gather}
     A_{\mu}^{\hat{a}^{\prime}}(x^{\mu},0)=A_{\mu}^{\hat{a}^{\prime}}(x^{\mu},\pi R)=0 \\
     \label{bouncongauge}
     A^{a^{\prime}}_{y}(x^{\mu},0)=A^{a^{\prime}}_{y}(x^{\mu},\pi R)=0  \nonumber \; .
    \end{gather}
    The expansions (\ref{fourierevenfunctionsu7}) lead to 
    Neumann boundary conditions for unbroken gauge fields and scalar fields, i.e.
    \begin{gather}
     \label{bouncon2}
     \partial_{y} A_{\mu}^{a^{\prime}}(x^{\mu},0)=\partial_{y} A_{\mu}^{a^{\prime}}(x^{\mu},\pi R)=0  \\ 
     \partial_{y} A_{y}^{\hat{a}^{\prime}}(x^{\mu},0)
     =\partial_{y} A_{y}^{\hat{a}^{\prime}}(x^{\mu},\pi R)=0  \nonumber \; .
    \end{gather}

    The advantage of gauge symmetry breaking by boundary conditions is that unlike in the 
    orbifold case, one can obtain {\em{any}} subgroup $G_{0}$ of $G$. In contrast, in the orbifold case
    only very special subgroups of $G$ compatible with the
    action of $P$ on the Lie algebra of $\mathfrak{g}$ can be obtained.

   \subsection{Gauge symmetry breaking by Dirichlet and Neumann boundary conditions in the $SU(7)$ model}

    Let us return to the $SU(7)$ model. We consider the breaking
    \begin{equation} 
     SU(6)_{L} \times U(1)_{Y} \to SU(2)_{L} \times U(1)_{Y} \times SO(3)_{F} \; .
    \end{equation} 
    This breaking can be achieved by demanding 
    \begin{equation}
     A_{\mu}^{\hat{a}}=0 \quad , \quad  \partial_{y} A_{\mu}^{a}=0 \; ,
    \end{equation}
    for
    \begin{gather}
     A_{\mu}^{a} \in \mathfrak{su}(2) \oplus \mathfrak{u}(1) \oplus \mathfrak{so}(3) \\ 
     A_{\mu}^{\hat{a}} \in \mathfrak{su}(6) \oplus \mathfrak{u}(1)/
     \mathfrak{su}(2) \oplus \mathfrak{u}(1) \oplus \mathfrak{so}(3)  \nonumber
    \end{gather}
    at both boundaries $L$ and $R$.

  \section{Gauge coupling unification and the weak mixing angle}

   The gauge group $SU(7)$ unifies inter alia the weak gauge group $SU(2)_{L}$ and the hypercharge
   gauge group $U(1)_{Y}$ of the SM. This means that in the unified theory
   there exists only one five-dimensional gauge coupling constant which we denote by 
   \begin{equation}
    g_{5}^{SU(7)} \; .
   \end{equation} 
   Therefore it is possible to calculate the effective four-dimensional coupling constants for
   $SU(2)_{L}$ and $U(1)_{Y}$, respectively, and thus the weak mixing angle $\theta_{W}$ of the SM.
   Recall that in the SM the covariant derivative reads \cite{Huang:1982ik}
   \begin{equation}
    D_{\mu}=\partial_{\mu} + i g \; \mathbf{W_{\mu}} \cdot \mathbf{t} + i g^{\prime} \; W^{\mu}_{0} t_{0} \; ,
   \end{equation} 
   where $g$ and $g^{\prime}$ are the four-dimensional coupling constants of 
   $SU(2)_{L}$ and $U(1)_{Y}$, respectively. The weak mixing angle $\theta_{W}$ is given by
   \begin{equation}
    \label{weakmixinganglesm}
    \sin^{2} \theta_{W}=\frac{g^{\prime 2}}{g^{2}+g^{\prime 2}} \; .
   \end{equation}
   In order to compute $\theta_{W}$ in the $SU(7)$ model we have to determine the effective four-dimensional
   coupling constants
   \begin{equation}
    \label{4deffectivegaugecouplingconstantssu2u1}
    g \equiv g_{4}^{SU(2)_{L}} \quad , \quad  g^{\prime} \equiv g_{4}^{U(1)_{Y}} \; .
   \end{equation}
   To calculate (\ref{4deffectivegaugecouplingconstantssu2u1}) we have to take into account the
   normalisation of the generators $T_{3}$ and $Y$, see (\ref{smweakgaugegroup}) and 
   (\ref{smhyperchargegroup}). Thus due to (\ref{smweakgaugegroup}) the five-dimensional 
   gauge coupling constant $g_{5}^{SU(2)_{L}}$ of the $SU(2)_{L}$ subgroup of 
   $SU(7)$ in related to the five-dimensional gauge coupling constant $g_{5}^{SU(7)}$ by 
   \begin{equation}
    \label{5drealtionsu2}
    g_{5}^{SU(2)_{L}}=\frac{g_{5}^{SU(7)}}{\sqrt{3}} 
   \end{equation}
   and due to (\ref{smhyperchargegroup}) the five-dimensional gauge coupling constant $g_{5}^{U(1)_Y}$ 
   of the $U(1)_{Y}$ subgroup of $SU(7)$ is related to the five-dimensional 
   gauge coupling constant $g_{5}^{SU(7)}$  by 
   \begin{equation}
    \label{5drealtionu1}
    g_{5}^{U(1)_{Y}}=\frac{g_{5}^{SU(7)}}{\sqrt{21}} \; .
   \end{equation}
   In addition, due to (\ref{relation4deffgaugecoupling5dgaugecoupling}) an effective four-dimensional gauge 
   coupling constant $g_{4}$ is related to a five-dimensional gauge coupling constant $g_{5}$ via
   \begin{equation}
    g_{4}=\frac{g_{5}}{\sqrt{2 \pi R}} \; ,
   \end{equation}
   where $R$ is the compactification radius. Thus we obtain the following four-dimensional effective
   SM coupling constants 
   \begin{equation}
    \label{smcouplingconstants}
    g=g_{4}^{SU(2)_{L}}=\frac{g_{5}^{SU(7)}}{\sqrt{6 \pi R}}  \quad , \quad
    g^{\prime}=g_{4}^{U(1)_Y}=\frac{g_{5}^{SU(7)}}{\sqrt{42 \pi R}} \; ,
   \end{equation}
   where we have inserted (\ref{5drealtionsu2}) and (\ref{5drealtionu1}).
   Inserting further (\ref{smcouplingconstants}) in (\ref{weakmixinganglesm}) we obtain for
   the weak mixing angle in the $SU(7)$ model: 
   \begin{equation}
    \label{weakmixinganglesu7}
    \sin^{2} \theta_{W}^{SU(7)}=0.125 \; .
   \end{equation}
   We compare this result with the experimental value \cite{Eidelman:2004wy}
   \begin{equation}
    \sin^{2} \theta_{W}^{\text{exp}} \approx 0.23 \; .
   \end{equation}
   We see that the obtained value for $\theta_{W}$ is too small by approximately a factor of two. This
   problem can however be solved by starting with a slightly different unified gauge group. This issue will
   be discussed in the outlook.

  \section{The minimum of the Higgs potential $V(\Phi)$}

    We consider the Higgs potential 
    \begin{equation}
      \label{higgspotentialsu7}
      V(\Phi)=V(\rho \; e^{A_{y}} \; e^{\eta} \; e^{A_{y}}) \; .
    \end{equation}
    According to Theorem \ref{theoremnuhiggspotential} $V(\Phi)$ is invariant under unitary gauge 
    transformations 
    \begin{equation}
     V(S_{0}(x) \Phi S_{0}(x)^{-1})=V(\Phi) 
    \end{equation} 
    where $S_{0}(x) \in SU(2)_{L} \times SO(3)_{F} \times U(1)_{Y}$. Thus $e^{A_{y}}$ in
    (\ref{higgspotentialsu7}) cannot be gauged away and
    the Higgs potential $V(\Phi)$ depends on $\rho$, $\eta$ and $A_{y}$.
    The unitary factor $e^{A_{y}}$ in (\ref{higgspotentialsu7}) is given by
    \begin{equation}
     \label{unitaryfactosu7}
     e^{A_{y}}=e^{i \; g_{4} R \; \mathcal{A}_{y}^{(0)}}
    \end{equation}  
    where 
    \begin{equation}
     \mathcal{A}_{y}^{(0)}= \sum_{\hat{a}=43}^{48} \mathcal{A}_{y}^{\hat{a}(0)} L_{\hat{a}} \; ,
    \end{equation}
    compare with (\ref{a5unitarygauge}).
    Note that $\mathcal{A}_{y}^{(0)}$ denote the neutral components of the three electroweak 
    Higgs doublets (\ref{higgsdoubleth1}), (\ref{higgsdoubleth2}) and (\ref{higgsdoubleth3}).
    We consider now the case where the $\mathcal{A}_{y}^{(0)}$ assume a VEV. Without loss of generality
    we lay it in the $L_{\hat{43}}$- $L_{\hat{45}}$- and $L_{\hat{47}}$-direction, i.e.
    \begin{equation}
     \label{expandvevs}
     \mathcal{A}_{y} \to 
     \langle \mathcal{A}_{y}^{(0)} \rangle=\sum_{\hat{a}=\hat{43},\hat{45},\hat{47}} 
     \langle \mathcal{A}_{y}^{\hat{a}(0)} \rangle L_{\hat{a}} \; .
    \end{equation}
    Inserting (\ref{expandvevs}) in (\ref{unitaryfactosu7}) we get 
    \begin{equation}
     \label{wilsonsu7}
     W=e^{i \;  g_{4} R \; \sum_{\hat{a}} \langle \mathcal{A}_{y}^{\hat{a}(0)} \rangle L_{\hat{a}}} \; .
    \end{equation}
    Since $[P,L_{\hat{a}}] \neq 0$ for $\hat{a}=\hat{43},\hat{45},\hat{47}$ all VEVs 
    $\langle \mathcal{A}_{y}^{\hat{a}(0)} \rangle$
    can be arbitrary constants and thus $W$ is a
    continuous Wilson line. We parametrise the VEVs $\langle \mathcal{A}_{y}^{\hat{a}(0)} \rangle$ as
    \begin{equation}
     \label{vevsthreehiggsdoublets}
     \langle \mathcal{A}_{y}^{\hat{43}(0)} \rangle=\frac{\alpha_{\hat{43}}}{g_{4} R} \quad , \quad
     \langle \mathcal{A}_{y}^{\hat{45}(0)} \rangle=\frac{\alpha_{\hat{45}}}{g_{4} R} \quad , \quad
     \langle \mathcal{A}_{y}^{\hat{47}(0)} \rangle=\frac{\alpha_{\hat{47}}}{g_{4} R} 
    \end{equation}
    where the $\alpha_{\hat{a}}$ are dimensionless parameters.
    Inserting (\ref{vevsthreehiggsdoublets}) in (\ref{wilsonsu7}) we get
    \begin{equation}   
     \label{wilsonlinesu7model}
     W=e^{i \; \sum_{\hat{a}} \alpha_{\hat{a}} L_{\hat{a}}} \; .
    \end{equation}
    The VEVs $\langle \mathcal{A}_{y}^{\hat{a}(0)} \rangle$ 
    are much smaller than the compactification scale $1/R$. Thus
    $0 < \alpha_{\hat{a}} \ll 1$ in (\ref{wilsonlinesu7model}) and we can
    approximate
    \begin{eqnarray}
     \label{a5expansion}
     W & = & e^{i \; \sum_{\hat{a}} \alpha_{\hat{a}} L_{\hat{a}}} \approx \mathbf{1} + \sum_{\hat{a}}
     i \alpha_{\hat{a}} L_{\hat{a}} \\
     & = & 
       \left( \begin{array} {ccccccc} 1 & 0 & 0 & 0 & 0 & 0 & 0 \\
        0 & 1 & 0 & 0 & 0 & 0 & 0 \\ 0 & 0 & 1 & 0 & 0 & 0 & 0 \\
        0 & 0 & 0 & 1 & 0 & 0 & i  \frac{\alpha_{\hat{43}}}{2} \\ 
        0 & 0 & 0 & 0 & 1 & 0 & i  \frac{\alpha_{\hat{45}}}{2} \\
        0 & 0 & 0 & 0 & 0 & 1 & i  \frac{\alpha_{\hat{47}}}{2} \\ 
        0 & 0 & 0 & i \frac{\alpha_{\hat{43}}}{2} &  i \frac{\alpha_{\hat{45}}}{2} 
        & i \frac{\alpha_{\hat{47}}}{2} & 1 \end{array} \right) \; .
   \end{eqnarray}
   According to (\ref{higgspotentialsu7}) we can  parameterise the minimum $\Phi_{min}$ of the Higgs
   potential as
   \begin{equation}
    \label{minimumphi}
    \Phi_{min}=\rho_{min} \frac{1}{\sqrt{2}} \; \begin{pmatrix} e^{a_{1}} & 0 & 0 & 0 & 0 & 0 & 0 \\
                   0 & e^{a_{2}}  & 0 & 0 & 0 & 0 & 0 \\  0 & 0 & e^{a_{3}} & 0 & 0 & 0 & 0 \\
                   0 & 0 & 0 & e^{a_{4}} & 0 & 0 & i \alpha_{\hat{43}}^{\prime} \\
                   0 & 0 & 0 & 0 & e^{a_{5}} & 0 & i \alpha_{\hat{45}}^{\prime} \\
                   0 & 0 & 0 & 0 & 0 & e^{a_{6}} & i \alpha_{\hat{47}}^{\prime} \\
                   0 & 0 & 0 & i \alpha_{\hat{43}}^{\prime} & i \alpha_{\hat{45}}^{\prime} 
                   & i \alpha_{\hat{47}}^{\prime} & e^{a_{7}} \end{pmatrix}
                   + \mathcal{O}(\alpha_{a}^{2})
   \end{equation}
   where $\sum_{i=1}^{7} a_{i}=0$, $a_{i} \in \mathbb{R}$ and
   \begin{equation}
     \label{a5vev}
     \alpha_{\hat{43}}^{\prime}=\frac{\alpha_{\hat{43}}}{2} \left( e^{a_{4}}+ e^{a_{7}} \right) 
     \quad , \quad
     \alpha_{\hat{45}}^{\prime}=\frac{\alpha_{\hat{45}}}{2} \left( e^{a_{5}}+ e^{a_{7}} \right) 
     \quad , \quad
     \alpha_{\hat{47}}^{\prime}=\frac{\alpha_{\hat{47}}}{2} \left( e^{a_{6}}+ e^{a_{7}} \right) \; .
   \end{equation}
   In the following we neglect terms of $\mathcal{O}(\alpha_{a}^{2})$.
   {\em{In order to have a spontaneous symmetry breaking we assume that $V(\Phi)$
   is minimised at non-trivial $\Phi_{min}$, i.e. we assume 
   \begin{equation}
    a_{i} \neq 0 \quad , \quad \alpha_{\hat{a}}^{\prime} \neq 0
   \end{equation} 
   for $i=1,\dots,7$ and $\hat{a}=\hat{43},\hat{45},\hat{47}$ in (\ref{minimumphi}).}}
   For later use we make the following 
   \begin{definition}
    We call a minimum of the Higgs potential $\Phi_{min}$ quasi $\mathcal{S}_{2}$ symmetric if 
    \begin{equation}
     a_{i}=a_{j}
    \end{equation}
    for the pairs $(i,j)=(1,4),(2,5),(3,6)$.
   \end{definition}
   Note that fluctuations of $a_{i}$ and $\alpha_{a}$ around the minimum $\Phi_{min}$ (\ref{minimumphi})
   give rise to altogether $10$ Higgs particles. This topic will be discussed in the outlook.

  \section{Calculation of gauge field masses in the $SU(7)$ model}

    In this section we calculate the masses of all gauge fields in the $SU(7)$ model. Recall that
    $SU(2)_{L} \times U(1)_{Y} \times SO(3)_{F}$ is left unbroken by the orbifolding
    and imposing Dirichlet and Neumann boundary conditions. Consequently we have to compute
    \begin{equation}
     \label{covariantderivativephiminsu7}
     D_{\mu} \Phi_{min}= i \frac{g}{\sqrt{2}} \; A_{\mu}^{i(0)} 
                           \left[ L_{i}, \Phi_{min} \right] + 
                         i \frac{g}{\sqrt{2}} \; A_{\mu}^{i(1)} \{ L_{i}, \Phi_{min} \} \; ,
    \end{equation}
    for the generators $\{ L_{i} \}$ with $i=1,2,3,5,8,10,36$. Note that $T_{j}=\sqrt{3} L_{j}$ 
    (\ref{smweakgaugegroup}), $Y=\sqrt{21} L_{36}$ (\ref{smhyperchargegroup}) and 
    $H_{k}=\sqrt{2} L_{l}$ (\ref{so3flavourgaugegroup}) with $j=1,2,3$, $k=2,5,7$ and $l=5,8,10$ denote
    the generators of the weak gauge group, the hypercharge and the $SO(3)_{F}$
    flavour gauge group, respectively. In the following calculations we can us either the set
    $\{ L_{i} \}$ or the set $\{ T_{j},Y,H_{k} \}$ in 
    (\ref{covariantderivativephiminsu7}). However only the $\{ L_{i} \}$ are normalised as
    $\text{tr} \left( L_{i} L_{j} \right)=\frac{1}{2} \delta_{ij}$ and imply conventional normalisation 
    of the four-dimensional kinetic terms for $A_{\mu}^{i(0)}$ and $A_{\mu}^{i(1)}$ as we have shown 
    explicitly in section \ref{sectionenonabBTLM}.
    Therefore we have to use the set $\{ L_{i} \}$ instead of $\{ T_{j},Y,H_{k} \}$.  
    $\Phi_{min}$ in (\ref{covariantderivativephiminsu7}) is given by equation (\ref{minimumphi}). 
    {\em{We assume now that $\Phi_{min}$ is quasi $\mathcal{S}_{2}$ symmetric}}. The reason will be explained
    in the next subsection. We can expand a quasi $\mathcal{S}_{2}$ symmetric $\Phi_{min}$ as
    \begin{equation}
     \label{paramequasis3phimin}
     \Phi_{min} = \Phi_{min}^{diag} + \Phi_{min}^{offdiag} 
                = \sum_{j=1,2,3,4} \; \phi_{j} \; \tilde{L}_{j} 
                + \sum_{k=\hat{43},\hat{45},\hat{47}} \; \phi_{k} \;  L_{k} \; .
    \end{equation}
    where 
    \begin{gather}
     \label{liquasis2sym}
     \tilde{L}_{1}=\frac{1}{2} \; \text{diag}(1,0,0,1,0,0)=\frac{1}{2} \; \text{diag}(l_{1},l_{1}) \quad , 
        \quad  l_{1}=\text{diag}(1,0,0)  \\            
     \tilde{L}_{2}=\frac{1}{2} \; \text{diag}(0,1,0,0,1,0)=\frac{1}{2} \; \text{diag}(l_{2},l_{2}) \quad , 
        \quad  l_{2}=\text{diag}(0,1,0)  \\            
     \tilde{L}_{3}=\frac{1}{2} \; \text{diag}(0,0,1,0,0,1)=\frac{1}{2} \; \text{diag}(l_{3},l_{3}) \quad , 
        \quad  l_{3}=\text{diag}(0,0,1)  \\                             
     \tilde{L}_{4}=\frac{1}{2} \; \text{diag}(0,0,0,0,0,0,1)  
    \end{gather}
    {\em{Remarks:}} i) \; By writing a matrix 
    just as a $6 \times 6$ matrix we mean that this matrix is embedded in a $7 \times 7$ matrix as an upper
    $6 \times 6$ matrix. This convention is the same convention as we have made for the generators
    $L_{1},\dots,L_{35}$ of $SU(6)_{L} \subset SU(7)$ and we will use this convention in the following
    calculations. \\
    ii) \; $\sum_{j=1,2,3,4} \; \phi_{j} \; \tilde{L}_{j}$ form the diagonal part $\Phi_{min}^{diag}$ of
    $\Phi_{min}$ while $\sum_{k=\hat{43},45,\hat{47}} \; \phi_{k} \;  L_{k}$ form the off-diagonal part 
    $\Phi_{min}^{offdiag}$ of $\Phi_{min}$. \\
    iii) \;  The $\tilde{L}_{i}$, $i=1,2,3,4$, are normalised such that $\text{tr}
     \; (\tilde{L}_{i} \tilde{L}_{j})=\frac{1}{2} \delta_{ij}$. \\
    The $\phi_{j}$ and $\phi_{k}$ in (\ref{paramequasis3phimin}) are given by
    \begin{gather} 
     \label{definitionp0p6p11}
     \phi_{1}=\sqrt{2} \rho_{min} e^{a_{1}}  \; , \;
     \phi_{2}=\sqrt{2} \rho_{min} e^{a_{2}} \; \;
     \phi_{3}=\sqrt{2} \rho_{min} e^{a_{3}}  \; , \; 
     \phi_{4}=\sqrt{2} \rho_{min} e^{a_{7}} \\    
     \phi_{k}=i \phi_{k}^{\prime} \; , \quad \phi_{k}^{\prime}=\sqrt{2} \rho_{min} \alpha_{k}^{\prime}
    \end{gather} 
    for $k=\hat{43},\hat{45},\hat{47}$.
    We then have to calculate the commutators $[ L_{i}, \Phi_{min} ]$ 
    and anticommutators $\{ L_{i}, \Phi_{min} \}$ for $\{L_{i}\}$, $i=1,2,3,5,8,10,35$,  
    insert the results of this computation in (\ref{covariantderivativephiminsu7}) and 
    evaluate the covariant derivative $D_{\mu} \Phi_{min}$. After taking the adjoint 
    $(D_{\mu} \Phi_{min})^{\dagger}$, multiplying the adjoint with
    $(D_{\mu} \Phi_{min})$ and taking the trace we obtain
    \begin{equation} 
     \label{su7generalmass}
     \mathcal{L}_{\text{mass}}=
     tr\left[\left(D_{\mu} \Phi_{min}\right)^{\dagger} \left(D_{\mu} \Phi_{min}\right)\right]
    \end{equation} 
    We remind the reader that the basis of mass eigenstates $A_{\mu}^{i(0)}$ and $A_{\mu}^{i(1)}$ 
    is already diagonal and no mixed terms between the zero mode and the first excited mode in
    (\ref{su7generalmass}) mode occur.

   \subsection{Mass term for the SM weak gauge fields from the off-diagonal 
               part of $\Phi_{min}$}

    In this subsection we calculate mass squared for the the zero mode gauge fields $A_{\mu}^{i(0)}$ 
    associated to the generators $L_{i}$ ($i=1,2$) and weak generators
    $T_{i}$ ($i=1,2$), respectively. First since $\left[ L_{i},\tilde{L}_{j}\right]=0$
    for $j=1,2,3,4$ we see that the diagonal part of $\Phi_{min}$ gives no contribution to the
    mass of zero mode gauge fields $A_{\mu}^{i(0)}$. 
    {\em{This is a consequence of the $\mathcal{S}_{2}$ quasi
    symmetry of the minimum $\Phi_{min}$ of the Higgs potential.}} Therefore the fields 
    $A_{\mu}^{i(0)}$ get their mass only from $\Phi_{min}^{offdiag}$. 
    The covariant derivative  reads
    \begin{eqnarray*}
      D_{\mu}\Phi_{min}^{offdiag} & = & i \frac{g}{\sqrt{2}} \; 
               A_{\mu}^{i(0)} \; \left[ L_{i}, \Phi_{min}^{offdiag} \right] \\
         & = &  i \frac{g}{\sqrt{2}} \; 
         A_{\mu}^{1(0)} \left( \phi_{\hat{43}} 
         \underbrace{\left[ L_{1}, L_{\hat{43}} \right]}_{=\frac{1}{2\sqrt{3}} i L_{38}}
         + \phi_{\hat{45}} \underbrace{\left[ L_{1}, L_{\hat{45}} \right]}_{=\frac{1}{2\sqrt{3}} i L_{40}}
         + \phi_{\hat{47}} \underbrace{\left[ L_{1}, L_{\hat{47}} \right]}_{=\frac{1}{2\sqrt{3}} i L_{42}}
         \right) \\
         & + &  i \frac{g}{\sqrt{2}} \; 
         A_{\mu}^{2(0)} \left( \phi_{\hat{43}} 
         \underbrace{\left[ L_{2}, L_{\hat{43}} \right]}_{=-\frac{1}{2\sqrt{3}} i L_{37}}
         + \phi_{\hat{45}} \underbrace{\left[ L_{2}, L_{\hat{45}} \right]}_{=-\frac{1}{2\sqrt{3}} i L_{39}}
         + \phi_{\hat{47}} \underbrace{\left[ L_{2}, L_{\hat{47}} \right]}_{=-\frac{1}{2\sqrt{3}} i L_{41}} 
         \right) \\
         & = & - i \frac{g}{2\sqrt{6}} \; 
         A_{\mu}^{1(0)} \left( \phi_{\hat{43}}^{\prime} L_{38} + \phi_{\hat{45}}^{\prime} L_{40} 
                        + \phi_{\hat{47}}^{\prime} L_{42} \right) \\
         & + & i \frac{g}{2\sqrt{6}} \; 
         A_{\mu}^{2(0)} \left( \phi_{\hat{43}}^{\prime} L_{37} + \phi_{\hat{45}}^{\prime} L_{39} 
                        + \phi_{\hat{47}}^{\prime} L_{41} \right)  \nonumber \; .
    \end{eqnarray*}        
    Taking the adjoint $(D_{\mu}\Phi_{min}^{offdiag})^{\dagger}$,
    multiplying $(D_{\mu}\Phi_{min}^{offdiag})^{\dagger}$ and $(D_{\mu}\Phi_{min}^{offdiag})$
    and taking the trace we arrive at 
    \begin{equation}  
     \text{tr} \left[\left(D_{\mu} \Phi_{min}^{offdiag}\right)^{\dagger}
       \left(D_{\mu} \Phi_{min}^{offdiag}\right)\right]
     = \frac{1}{24} \; g^{2} \rho_{min}^{2} \; \sum_{a=\hat{43},\hat{45},\hat{47}}  
       \alpha_{a}^{\prime 2} 
       \; \left( A_{\mu}^{i(0)}\right)^{2}
    \end{equation}
    with $i=1,2$. Thus the mass squared of the gauge fields $A_{\mu}^{i(0)}$ reads
    \begin{equation}
     \label{masstermzerowboson}
     m^{2}=\frac{1}{12} \; g^{2} \rho_{min}^{2} \; \sum_{a=\hat{43},\hat{45},\hat{47}}  
       \alpha_{a}^{\prime 2}=\frac{1}{12} \; \frac{1}{R^{2}} \sum_{a=\hat{43},\hat{45},\hat{47}}  
       \alpha_{a}^{\prime 2} \; , 
    \end{equation}
    where we have inserted $g \rho_{min}=1/R$ in the second step.
    For $0 < \alpha_{a}^{\prime} \ll 1$ we deduce that $1/R=\mathcal{O}(1)$ TeV so that 
    $m=m_{W}=80.4$ GeV \cite{Eidelman:2004wy}. Therefore \\
    \\
    {\em{The zero mode gauge fields $A_{\mu}^{1,2(0)}$ are identified with the SM weak gauge fields
    $W_{\mu}^{1,2}$.}} \\ 
    \\

   \subsection{Mass term for the SM $Z$ gauge field from off-diagonal part 
               of $\Phi_{min}$}

    We consider the linear combination of
    the gauge fields $A_{\mu}^{3(0)}$ and $A_{\mu}^{36(0)}$ associated to the generators $L_{3},L_{36}$
    \begin{gather}
     \label{w3hypermixing}
     A_{\mu}^{(0)}=A_{\mu}^{36(0)} \cos \theta_{W}^{SU(7)} + W_{\mu}^{3(0)} \sin \theta_{W}^{SU(7)} \\
     Z_{\mu}^{(0)}=-A_{\mu}^{36(0)} \sin \theta_{W}^{SU(7)} + W_{\mu}^{3(0)} \cos \theta_{W}^{SU(7)} 
     \nonumber
    \end{gather}
    where $\theta_{W}^{SU(7)}$ is the weak mixing angle in the $SU(7)$ model given by equation
    (\ref{weakmixinganglesu7}). The gauge field $Z_{\mu}^{(0)}$ get its mass, like the weak gauge 
    fields $W_{\mu}^{1,2}$, only from $\Phi_{min}^{offdiag}$ and the gauge field $A_{\mu}^{(0)}$ remains
    massless.
    The relation between the mass $m_{Z}$ of the gauge field $Z_{\mu}^{(0)}$ the 
    mass squared $m_{W}$ (\ref{masstermzerowboson}) of the weak gauge
    fields $W_{\mu}^{1,2}$ is given by
    \begin{equation}
     \label{contextzwmass}
     m_{Z}=\frac{m_{W}}{\cos \theta_{W}^{SU(7)}}
    \end{equation}
    Therefore
    \\
    {\em{The zero mode gauge field $Z_{\mu}^{(0)}$ are identified with the SM $Z$ gauge field $Z_{\mu}$
    and the zero mode gauge field $A_{\mu}^{(0)}$ are identified with the SM photon field $A_{\mu}$.}} \\
    \\
    The relation (\ref{contextzwmass}) is familiar in the SM. However according to (\ref{contextzwmass}) 
    the mass of the $Z$ gauge field 
    turns out to be too low which is a consequence of then wrong weak mixing angle $\theta_{W}^{SU(7)}$.

  \subsection{Mass term for the first excited KK mode of the SM weak gauge fields
              from diagonal part of $\Phi_{min}$}

    In this subsection we calculate the mass squared for the first excited KK mode of SM weak gauge
    fields $W_{\mu}^{1,2}$. Since we have identified $W_{\mu}^{1,2}=A_{\mu}^{1,2(0)}$ the first
    excited KK mode gauge fields of the weak gauge fields $W_{\mu}^{1,2}$ are $A_{\mu}^{1,2(1)}$. 
    Therefore in the following we write $W_{\mu}^{1,2(1)}=A_{\mu}^{1,2(1)}$.   
    To calculate the mass squared of $W_{\mu}^{1,2(1)}$, we consider the scenario where $a_{i} \gg 1$ 
    for some $a_{i}$ in $\Phi_{min}^{diag}$. Hence the off-diagonal part $\Phi_{min}^{offdiag}$ will 
    give only a very small contribution to the mass of $W_{\mu}^{1,2(1)}$. For this reason we
    neglect it in the following calculations. Consequently we have to compute  
    \begin{equation}
     \label{covariantderivativephiminmain}
     D_{\mu} \Phi^{diag}_{min}= i \frac{g}{\sqrt{2}} \; W_{\mu}^{i(1)} \{ L_{i}, \Phi^{diag}_{min} \}
    \end{equation}
    with $i=1,2$.
    We calculate the anticommutators
    \begin{eqnarray}
      \label{anticommfirstkksu2}
      \{ L_{1}, \Phi^{diag}_{min} \} & = & \phi_{1} \{L_{1},\tilde{L} _{1} \}
                                          +\phi_{2} \{L_{1},\tilde{L} _{2} \}
                                          +\phi_{3} \{L_{1},\tilde{L} _{3} \} \\
                                     & = & \phi_{1} \; \frac{1}{2\sqrt{3}} \; \begin{pmatrix}
                                           0 & l_{1} \\ l_{1} & 0 \end{pmatrix} +
                                           \phi_{2} \; \frac{1}{2\sqrt{3}} \; \begin{pmatrix}
                                           0 & l_{2} \\ l_{2} & 0 \end{pmatrix} +
                                           \phi_{3} \; \frac{1}{2\sqrt{3}} \; \begin{pmatrix}
                                           0 & l_{3} \\ l_{3} & 0 \end{pmatrix} \nonumber \\
      \{ L_{2}, \Phi^{diag}_{min} \} & = & \phi_{1} \{L_{1},\tilde{L} _{1} \}
                                          +\phi_{2} \{L_{1},\tilde{L} _{2} \}
                                          +\phi_{3} \{L_{1},\tilde{L} _{3} \} \nonumber \\
                                     & = & \phi_{1} \; \frac{1}{2\sqrt{3}} \; \begin{pmatrix}
                                           0 & -il{1} \\ il_{1} & 0 \end{pmatrix} +
                                           \phi_{2} \; \frac{1}{2\sqrt{3}} \; \begin{pmatrix}
                                           0 & -il_{2} \\ il_{2} & 0 \end{pmatrix} +
                                           \phi_{3} \; \frac{1}{2\sqrt{3}} \; \begin{pmatrix}
                                           0 & -il_{3} \\ il_{3} & 0 \end{pmatrix} \; , \nonumber 
    \end{eqnarray}
    where $0$ denotes the $3 \times 3$ zero matrix. Recall  that 
    $l_{1}=\text{diag}(1,0,0)$ and $l_{2}=\text{diag}(0,1,0)$, see (\ref{liquasis2sym}).                    
    Inserting (\ref{anticommfirstkksu2}) in (\ref{covariantderivativephiminmain}), taking the adjoint
    $(D_{\mu} \Phi^{diag}_{min})^{\dagger}$, multiplying the adjoint with $D_{\mu} \Phi^{diag}_{min}$
    and finally taking the trace we get
    \begin{equation}
     \text{tr} \left[(D_{\mu} \Phi^{diag}_{min})^{\dagger} \left( D_{\mu} \Phi^{diag}_{min} \right) \right] =
     \frac{1}{12} \; g^{2} \; \left( \phi_{1}^{2} + \phi_{2}^{2} + \phi_{3}^{2} \right) \; 
     \left(  W_{\mu}^{1,2(1)} \right)^{2} 
    \end{equation}
    Note that $tr \; (\tilde{L}_{i} \tilde{L}_{j})
    =\frac{1}{2} \delta_{ij}$.
    Inserting the expression (\ref{definitionp0p6p11}) for $\phi_{i}$ we finally get
    \begin{equation}
     \text{tr} \left[(D_{\mu} \Phi^{diag}_{min})^{\dagger} \left( D_{\mu} \Phi^{diag}_{min} \right) \right]
      = \frac{1}{6} \; g^{2} \rho_{min}^{2}  \; \left( e^{2a_{1}} + e^{2a_{2}} + 
           e^{3a_{2}} \right) \left(  W_{\mu}^{1,2(1)} \right)^{2} 
    \end{equation}
    Thus the mass squared for the first excited KK mode of the SM weak gauge fields $W_{\mu}^{1,2}$ reads
    \begin{equation}
     \label{masssquaredfemw}
     m^{2}_{W^{1,2(1)}}=\frac{1}{3} \frac{1}{R^{2}}  \left( e^{2a_{1}} + e^{2a_{2}} + 
           e^{3a_{2}} \right)
    \end{equation}
    where we have inserted $g \rho_{min}=1/R$. \\
    \\
    {\em{If $a_{i} \gg 1$ for at least one $i=1,2,3$ the first excited KK mode of the SM
    weak gauge fields $W_{\mu}^{1,2}$ receives very large masses from the Higgs mechanism
    in comparison to the compactification scale $1/R$}}.

   \subsection{Mass term for the first excited KK mode of the SM $Z$ gauge field and SM 
               photon field from diagonal part of $\Phi_{min}$}

    The first excited KK mode of the SM $Z$ gauge field $Z_{\mu}^{(1)}$ and the first excited KK mode
    of the photon field $A_{\mu}^{(1)}$ is a linear combination of
    the fields $A_{\mu}^{3(1)}$ and $A_{\mu}^{36(1)}$ associated to the generators $L_{3}$ and $L_{36}$
    \begin{gather}
     \label{w3hypermixing}
     A_{\mu}^{(1)}=A_{\mu}^{36(1)} \cos \theta_{W}^{SU(7)} + A_{\mu}^{3(1)} \sin \theta_{W}^{SU(7)} \\
     Z_{\mu}^{(1)}=-A_{\mu}^{36(1)} \sin \theta_{W}^{SU(7)} + A_{\mu}^{3(1)} \cos \theta_{W}^{SU(7)} 
     \nonumber
    \end{gather}
    As in the last section we consider the scenario, where $a_{i} \gg 1$ for some $a_{i}$ in 
    $\Phi_{min}^{diag}$. Hence the off-diagonal part $\Phi_{min}^{offdiag}$ will again give only
    a very small contribution to the mass of $Z_{mu}^{(1)}$ and $A_{\mu}^{(1)}$. For this reason we
    neglect it in the following calculations. Consequently we have to compute 
    \begin{equation}
     \label{covariantderivativefempz}
     D_{\mu} \Phi^{diag}_{min}= i \frac{g}{\sqrt{2}} \; W_{\mu}^{3(1)} \{ L_{3}, \Phi^{diag}_{min} \}
                               +i \frac{g}{\sqrt{2}} \; B_{\mu}^{(1)} \{ L_{36}, \Phi^{diag}_{min} \}
    \end{equation}
    We calculate the anticommutators
    \begin{eqnarray}
      \{ L_{3}, \Phi^{diag}_{min} \} & = & \phi_{1} \{L_{3},\tilde{L} _{1} \}
                                          +\phi_{2} \{L_{3},\tilde{L} _{2} \}
                                          +\phi_{3} \{L_{3},\tilde{L} _{3} \} \\
                                     & = & \phi_{1} \; \frac{1}{2\sqrt{3}} \; \begin{pmatrix}
                                           l_{1} & 0 \\ 0 & -l_{1} \end{pmatrix} +
                                           \phi_{2} \; \frac{1}{2\sqrt{3}} \; \begin{pmatrix}
                                           l_{2} & 0 \\ 0 & -l_{2} \end{pmatrix} +
                                           \phi_{3} \; \frac{1}{2\sqrt{3}} \; \begin{pmatrix}
                                           l_{3} & 0 \\ 0 & -l_{3} \end{pmatrix} \nonumber \\
      \{ L_{36}, \Phi^{diag}_{min} \} & = & \phi_{1} \{L_{36},\tilde{L} _{1} \}
                                          +\phi_{2} \{L_{36},\tilde{L} _{2} \}
                                          +\phi_{3} \{L_{36},\tilde{L} _{3} \}
                                          +\phi_{4} \{L_{36},\tilde{L} _{4} \}  \nonumber \\
                                     & = & \phi_{1} \; \frac{1}{2\sqrt{21}} \; \begin{pmatrix}
                                           l_{1} & 0 \\ 0 & l_{1} \end{pmatrix} +
                                           \phi_{2} \; \frac{1}{2\sqrt{21}} \; \begin{pmatrix}
                                           l_{2} & 0 \\ 0 & l_{2} \end{pmatrix} +
                                           \phi_{3} \; \frac{1}{2\sqrt{21}} \; \begin{pmatrix}
                                           l_{3} & 0 \\ 0 & l_{3} \end{pmatrix}  \nonumber \\
                                     & - & \phi_{4} \; \frac{1}{2\sqrt{21}} \; \text{diag}(0,0,0,0,0,0,6)
                                           \nonumber  
    \end{eqnarray}               
    Inserting this equations in (\ref{covariantderivativefempz}), taking the adjoint
    $(D_{\mu} \Phi^{diag}_{min})^{\dagger}$, multiplying the adjoint with $(D_{\mu} \Phi^{diag}_{min})$
    and finally taking the trace we get
    \begin{eqnarray*}
     \text{tr} \left[(D_{\mu} \Phi^{diag}_{min})^{\dagger} \left( D_{\mu} 
      \Phi^{diag}_{min} \right) \right] & = &
      \frac{1}{12} \; g^{2} \left( \phi_{1}^{2} + \phi_{2}^{2} + \phi_{3}^{2} \right) \; 
      \left(  A_{\mu}^{3(1)} \right)^{2}  \\
     & + & \frac{1}{84} \; g^{2} \left( \phi_{1}^{2} + \phi_{2}^{2} 
      + \phi_{3}^{2} + 18 \phi_{4}^{2} \right) \;
      \left(  A_{\mu}^{36(1)} \right)^{2} 
    \end{eqnarray*}
    Inserting the expression (\ref{definitionp0p6p11}) for $\phi_{i}$ we get 
    \begin{eqnarray*}
     \text{tr} \left[(D_{\mu} \Phi^{diag}_{min})^{\dagger}
      \left( D_{\mu} \Phi^{diag}_{min} \right) \right] & = &
      \frac{1}{6} \; g^{2} \rho_{min}^{2} \left( e^{2 a_{1}} + e^{2 a_{2}} + e^{2 a_{3}} \right) \; 
      \left(  A_{\mu}^{3(1)} \right)^{2}  \\
     & + & \frac{1}{42} \; g^{2} \rho_{min}^{2} 
      \left( e^{2 a_{1}} + e^{2 a_{2}} + e^{2 a_{3}} + 18 e^{2 a_{7}}
      \right) \; \left(  A_{\mu}^{36(1)} \right)^{2} 
    \end{eqnarray*}
    If we assume that $a_{7} \ll 1$ we can neglect $e^{2 a_{7}}$ for at least one of the $a_{1}$, $a_{2}$ or
    $a_{3}$ large and thus after transforming
    $A_{\mu}^{36(1)}=A_{\mu}^{(1)} \cos \theta_{W}^{SU(7)} - Z_{\mu}^{(1)} \sin \theta_{W}^{SU(7)}$, $
     A_{\mu}^{3(1)}=A_{\mu}^{(1)} \sin \theta_{W}^{SU(7)} + Z_{\mu}^{(1)} \cos \theta_{W}^{SU(7)}$
    which follows from (\ref{w3hypermixing}) we get the final result
    \begin{eqnarray}
     \text{tr} \left[(D_{\mu} \Phi^{diag}_{min})^{\dagger} 
      \left( D_{\mu} \Phi^{diag}_{min} \right) \right] & = &
      \frac{1}{6} \; g^{2} \rho_{min}^{2} \left( e^{2 a_{1}} + e^{2 a_{2}} + e^{2 a_{3}} \right) \; 
      \left(  Z_{\mu}^{(1)} \right)^{2}  \nonumber \\
     & + & \frac{1}{42} \; g^{2} \rho_{min}^{2} \left( e^{2 a_{1}} + e^{2 a_{2}} + e^{2 a_{3}} 
      \right) \; \left(  A_{\mu}^{(1)} \right)^{2} \nonumber \\
    \end{eqnarray}
    Thus the mass squared for the first excited KK mode of the SM Z gauge field and the
    SM photon field reads
    \begin{equation}
     \label{masssquaredfemzp}
     m^{2}_{Z^{(1)}}=\frac{1}{3} \frac{1}{R^{2}}  \left( e^{2a_{1}} + e^{2a_{2}} + 
           e^{3a_{2}} \right) \quad , \quad m^{2}_{\gamma^{(1)}}=\frac{1}{21} \frac{1}{R^{2}}
           \left( e^{2a_{1}} + e^{2a_{2}} + e^{2a_{2}} \right)
    \end{equation}
    where we have inserted $g \rho_{min}=1/R$.  \\
    \\
    {\em{If $a_{i} \gg 1$ for at least one $i=1,2,3$ the first excited KK mode of the SM
    $Z$ gauge field and the first excited KK mode of the SM
    photon field receive very large masses from the Higgs mechanism 
    in comparison to the compactification scale $1/R$}}.

   \subsection{Mass term for $SO(3)_{F}$ flavour gauge fields from diagonal part of $\Phi_{min}$}

    In this subsection we calculate the mass squared for the zero and the first excited KK mode of the flavour
    gauge fields $A_{\mu}^{j(0,1)}$ associated to the generators $L_{j}$, where $j=5,8,10$. Since the
    corresponding flavour gauge bosons couple to FCNC its mass terms must be very large
    $\mathcal{O}(10^{3})-\mathcal{O}(10^{5})$ TeV. Therefore, in order to be consistent with experiment,
    we consider again the case where $a_{i} \gg 1$ for some $a_{i}$ in $\Phi_{min}^{diag}$. Hence the
    contribution to the mass squared for the flavour gauge fields $A_{\mu}^{j(0,1)}$ from the off-diagonal
    part of $\Phi_{min}$ can be neglected. Consequently we get have to compute
    \begin{equation}
     \label{flavourcov}
     D_{\mu}\Phi^{diag}_{min}=  i \frac{g}{\sqrt{2}} \; A_{\mu}^{j(0)} \left[ A_{\mu}^{j},
                                   \Phi^{diag}_{min} \right]
                               + i \frac{g}{\sqrt{2}} \; A_{\mu}^{j(1)} \{ A_{\mu}^{j}, \Phi^{diag}_{min} \}
    \end{equation}
    where $j=5,8,10$. 
    We calculate the commutators
    \begin{eqnarray}
     \label{flavourcomm}
     \left[ L_{5}, \Phi^{diag}_{min} \right] & = &
                                           \phi_{1} \left[L_{5},\tilde{L} _{1} \right]
                                          +\phi_{2} \left[L_{5},\tilde{L} _{2} \right]
                                          +\phi_{3} \left[L_{5},\tilde{L} _{3} \right] \\
                                     & = &  \frac{i}{4\sqrt{2}} \phi_{1}  \begin{pmatrix}
                                             \lambda_{1} & 0 \\ 0 & \lambda_{1} \end{pmatrix}
                                           - \frac{i}{4\sqrt{2}} \phi_{2}  \begin{pmatrix} 
                                             \lambda_{1} & 0 \\ 0 & \lambda_{1} \end{pmatrix} \nonumber \\ 
     \left[  L_{8}, \Phi^{diag}_{min} \right] & = & 
                                           \phi_{1} \left[ L_{8},\tilde{L} _{1} \right]
                                          +\phi_{2} \left[ L_{8},\tilde{L} _{2} \right]
                                          +\phi_{3} \left[ L_{8},\tilde{L} _{3} \right] \nonumber \\
                                     & = &         \frac{i}{4\sqrt{2}} \phi_{1}  \begin{pmatrix}
                                                   \lambda_{4} & 0 \\ 0 & \lambda_{4} \end{pmatrix}
                                                 - \frac{i}{4\sqrt{2}} \phi_{3}  \begin{pmatrix}
                                                   \lambda_{4} & 0 \\ 0 & \lambda_{4} \end{pmatrix} 
                                                   \nonumber \\
     \left[  L_{10}, \Phi^{diag}_{min} \right] & = & 
                                           \phi_{1} \left[ L_{10},\tilde{L} _{1} \right]
                                          +\phi_{2} \left[ L_{10},\tilde{L} _{2} \right]
                                          +\phi_{3} \left[ L_{10},\tilde{L} _{3} \right] \nonumber \\
                                     & = &     \frac{i}{4\sqrt{2}} \phi_{2}  \begin{pmatrix}
                                                   \lambda_{6} & 0 \\ 0 & \lambda_{6} \end{pmatrix}
                                                  - \frac{i}{4\sqrt{2}} \phi_{3}  \begin{pmatrix}
                                                   \lambda_{6} & 0 \\ 0 & \lambda_{6} \end{pmatrix} \nonumber 
    \end{eqnarray} 
    and anticommutators
    \begin{eqnarray}
     \label{flavouracomm}
     \{  L_{5}, \Phi^{diag}_{min} \} & = &  \phi_{1} \{ L_{5},\tilde{L} _{1} \}
                                          +\phi_{2} \{ L_{5},\tilde{L} _{2} \}
                                          +\phi_{3} \{ L_{5},\tilde{L} _{3} \} \\
                                     & = &     \frac{1}{4\sqrt{2}} \phi_{1}  \begin{pmatrix}
                                           \lambda_{2} & 0 \\ 0 & \lambda_{2} \end{pmatrix}
                                         + \frac{1}{4\sqrt{2}} \phi_{2}  \begin{pmatrix} 
                                           \lambda_{2} & 0 \\ 0 & \lambda_{2} \end{pmatrix} \nonumber \\
     \{  L_{8}, \Phi^{diag}_{min} \} & = &  \phi_{1} \{L_{8},\tilde{L} _{1} \}
                                          +\phi_{2} \{L_{8},\tilde{L} _{2} \}
                                          +\phi_{3} \{L_{8},\tilde{L} _{3} \} \nonumber \\
                                     & = &     \frac{1}{4\sqrt{2}} \phi_{1}  \begin{pmatrix}
                                           \lambda_{5} & 0 \\ 0 & \lambda_{5} \end{pmatrix}
                                         + \frac{1}{4\sqrt{2}} \phi_{3}  \begin{pmatrix} 
                                           \lambda_{5} & 0 \\ 0 & \lambda_{5} \end{pmatrix} \nonumber \\
     \{ L_{10}, \Phi^{diag}_{min} \} & = &  \phi_{1} \{L_{10},\tilde{L} _{1} \}
                                          +\phi_{2} \{L_{10},\tilde{L} _{2} \}
                                          +\phi_{3} \{L_{10},\tilde{L} _{3} \} \nonumber \\
                                     & = &     \frac{1}{4\sqrt{2}} \phi_{2}  \begin{pmatrix}
                                           \lambda_{7} & 0 \\ 0 & \lambda_{7} \end{pmatrix}
                                         + \frac{1}{4\sqrt{2}} \phi_{3}  \begin{pmatrix} 
                                           \lambda_{7} & 0 \\ 0 & \lambda_{7} \end{pmatrix} \nonumber
    \end{eqnarray}
    Let us first calculate the mass squared for the zero mode flavour gauge fields $A_{\mu}^{j(0)}$. 
    Inserting (\ref{flavourcomm}) in (\ref{flavourcov}) we get 
    \begin{eqnarray*}
     D_{\mu} \Phi^{diag}_{min} & = & - \frac{1}{8} \; g \;  \left( A^{5(0)}_{\mu} \left( 
                                     \phi_{1}  \begin{pmatrix}
                                     \lambda_{1} & 0 \\ 0 & \lambda_{1} \end{pmatrix}
                                     - \phi_{2}  \begin{pmatrix} 
                                     \lambda_{1} & 0 \\ 0 & \lambda_{1} \end{pmatrix} \right) \right. \\
                               &&    \left. + A^{8(0)}_{\mu} \left( 
                                     \phi_{1}  \begin{pmatrix}
                                     \lambda_{4} & 0 \\ 0 & \lambda_{4} \end{pmatrix}
                                     - \phi_{3}  \begin{pmatrix} 
                                     \lambda_{4} & 0 \\ 0 & \lambda_{4} \end{pmatrix} \right) \right. \\
                               &&    \left. + A^{10(0)}_{\mu} \left( 
                                     \phi_{1}  \begin{pmatrix}
                                     \lambda_{6} & 0 \\ 0 & \lambda_{6} \end{pmatrix}
                                     - \phi_{3}  \begin{pmatrix} 
                                     \lambda_{6} & 0 \\ 0 & \lambda_{6} \end{pmatrix} \right) 
                                     \right) \nonumber
    \end{eqnarray*}
    Taking the adjoint $(D_{\mu} \Phi^{diag}_{min})^{\dagger}$ multiplying with 
    $(D_{\mu} \Phi^{diag}_{min})$ and taking the trace we obtain
    \begin{eqnarray*}
      && 
      \text{tr} \left[\left( D_{\mu}\Phi^{diag}_{min} \right)^{\dagger} 
         \left( D_{\mu}\Phi^{diag}_{min} \right) \right]  \nonumber \\
      & = & \frac{1}{16} \; g^{2} \;  \left( \left( A^{2(0)}_{\mu} \right)^{2}  \;
          \left( \phi_{1} - \phi_{2} \right)^{2} + \left( A^{5(0)}_{\mu} \right)^{2}  \;
          \left( \phi_{1} - \phi_{3} \right)^{2} + \left( A^{7(0)}_{\mu} \right)^{2}  \;
          \left( \phi_{2} - \phi_{3} \right)^{2} \right)
     \end{eqnarray*}
     Inserting the expressions for $\phi_{1},\phi_{2}$ and $\phi_{3}$ (\ref{definitionp0p6p11})
     we finally get
     \begin{eqnarray*} 
      \text{tr} \left[\left( D_{\mu}\Phi^{diag}_{min} \right)^{\dagger} 
         \left( D_{\mu}\Phi^{diag}_{min} \right) \right]_{\text{zero mode}}
      & = & \frac{1}{8} \; g^{2} \rho_{min}^{2}  \; \left( e^{a1} - e^{a2} \right)^{2}
           \left( A^{5(0)}_{\mu} \right)^{2} \\
      & + & \frac{1}{8} \; g^{2} \rho_{min}^{2}  \; \left( e^{a1} - e^{a3} \right)^{2} 
           \left( A^{8(0)}_{\mu} \right)^{2} \\
      & + & \frac{1}{8} \; g^{2} \rho_{min}^{2}  \; \left( e^{a3} - e^{a2} \right)^{2} 
           \left( A^{10(0)}_{\mu} \right)^{2} 
     \end{eqnarray*} 
     To summarise we have obtained
     \begin{equation}
      \label{masssquaredzmflavour}
      \begin{array}{|c|c|} \hline
       \text{Field} & \text{Mass squared} \\ \hline
       A^{5(0)}_{\mu} & \frac{1}{4} \; \frac{1}{R^{2}}   
       \; \left( e^{a1} - e^{a2} \right)^{2} \\ \hline
       A^{8(0)}_{\mu} & \frac{1}{4} \; \frac{1}{R^{2}}  
       \; \left( e^{a1} - e^{a3} \right)^{2}  \\ \hline
       A^{10(0)}_{\mu} & \frac{1}{4} \; \frac{1}{R^{2}}  
       \; \left( e^{a2} - e^{a3} \right)^{2}  \\ \hline
      \end{array} 
     \end{equation}
     where we have inserted $g \rho_{min}=1/R$. \\
     \\
     {\em{For an appropriate choice of the $a_{i}$, so that in particular $a_{1} \neq a_{2}$,
     $a_{1} \neq a_{3}$, $a_{2} \neq a_{3}$ and
     two $a_{i} \gg 1$, the zero mode flavour gauge fields will receive very large masses from the
     Higgs mechanism in comparison to the compactification scale $g \rho_{min}=1/R$}}. \\
     \\
     Next we calculate the masses for the first excited KK mode flavour gauge fields $A_{\mu}^{j(1)}$. 
     Inserting (\ref{flavouracomm}) in (\ref{flavourcov}) we get
     \begin{eqnarray*}
     D_{\mu} \Phi^{diag}_{min} & = & i \; \frac{1}{8} \; g \; \left( A^{5(1)}_{\mu} \left( 
                                     \phi_{1}  \begin{pmatrix}
                                     \lambda_{2} & 0 \\ 0 & \lambda_{2} \end{pmatrix}
                                     + \phi_{2}  \begin{pmatrix} 
                                     \lambda_{2} & 0 \\ 0 & \lambda_{2} \end{pmatrix} \right) \right. \\
                               &&    \left. + A^{8(1)}_{\mu} \left( 
                                     \phi_{1}  \begin{pmatrix}
                                     \lambda_{5} & 0 \\ 0 & \lambda_{5} \end{pmatrix}
                                     + \phi_{3}  \begin{pmatrix} 
                                     \lambda_{5} & 0 \\ 0 & \lambda_{5} \end{pmatrix} \right) \right. \\
                               &&    \left. + A^{10(1)}_{\mu} \left( 
                                     \phi_{1}  \begin{pmatrix}
                                     \lambda_{7} & 0 \\ 0 & \lambda_{7} \end{pmatrix}
                                     + \phi_{3}  \begin{pmatrix} 
                                     \lambda_{7} & 0 \\ 0 & \lambda_{7} \end{pmatrix} \right) 
                                     \right) \nonumber
    \end{eqnarray*}
    Taking the adjoint $(D_{\mu} \Phi^{diag}_{min})^{\dagger}$ multiplying with 
    $(D_{\mu} \Phi^{diag}_{min})$ and taking the trace we get
    \begin{eqnarray*}
      && 
      \text{tr} \left[\left( D_{\mu}\Phi^{diag}_{min} \right)^{\dagger} 
         \left( D_{\mu}\Phi^{diag}_{min} \right) \right]  \nonumber \\
      & = & \frac{1}{16} \; g \;  \left( \left( A^{5(1)}_{\mu} \right)^{2}  \;
          \left( \phi_{1} + \phi_{2} \right)^{2} + \left( A^{8(1)}_{\mu} \right)^{2}  \;
          \left( \phi_{1} + \phi_{3} \right)^{2} + \left( A^{10(1)}_{\mu} \right)^{2}  \;
          \left( \phi_{2} + \phi_{3} \right)^{2} \right)
    \end{eqnarray*}
    Inserting the expressions for $\phi_{1},\phi_{2}$ and $\phi_{3}$ (\ref{definitionp0p6p11})
    we finally get
    \begin{eqnarray*} 
      \text{tr} \left[\left( D_{\mu}\Phi^{diag}_{min} \right)^{\dagger} 
         \left( D_{\mu}\Phi^{diag}_{min} \right) \right]_{\text{first excited mode}}
      & = & \frac{1}{8} \; g^{2} \rho_{min}^{2}  \; \left( e^{a1} + e^{a2} \right)^{2}
           \left( A^{5(1)}_{\mu} \right)^{2} \\
      & + & \frac{1}{8} \; g^{2} \rho_{min}^{2}  \; \left( e^{a1} + e^{a3} \right)^{2} 
           \left( A^{8(1)}_{\mu} \right)^{2} \\
      & + & \frac{1}{8} \; g^{2} \rho_{min}^{2}  \; \left( e^{a3} + e^{a2} \right)^{2} 
           \left( A^{10(1)}_{\mu} \right)^{2} 
    \end{eqnarray*} 
    To summarise we have obtained
     \begin{equation}
      \label{masssquaredfemflavour}
      \begin{array}{|c|c|} \hline
       \text{Field} & \text{Mass squared} \\ \hline
       A^{5(1)}_{\mu} & \frac{1}{4} \; \frac{1}{R^{2}}  
       \; \left( e^{a1} + e^{a2} \right)^{2} \\ \hline
       A^{8(1)}_{\mu} & \frac{1}{4} \; \frac{1}{R^{2}}  
       \; \left( e^{a1} + e^{a3} \right)^{2}  \\ \hline
       A^{10(1)}_{\mu} & \frac{1}{4} \; \frac{1}{R^{2}}  
       \; \left( e^{a2} + e^{a3} \right)^{2}  \\ \hline
      \end{array} 
    \end{equation} 
    where we have inserted $g \rho_{min}=1/R$.

  \section{Suppression of tree-level FCNC}

    Looking at the different results  (\ref{masstermzerowboson}), (\ref{contextzwmass}), 
    (\ref{masssquaredfemw}), (\ref{masssquaredfemzp}), (\ref{masssquaredfemflavour}) and
    (\ref{masssquaredzmflavour}), 
    we observe that the strongest constraint in the model is the suppression of tree-level FCNC. 
    This suppression leads to the following conditions 
    \begin{enumerate}
     \item 
      \begin{equation}
       a_{1} \neq a_{2} \; , \;
       a_{1} \neq a_{3} \; , \; a_{2} \neq a_{3}
      \end{equation}
     \item
      \begin{equation}
       a_{i} \gg 1
      \end{equation}
      for at least for two $a_{i}$
    \end{enumerate}
    for the diagonal part of $\Phi_{min}$, see (\ref{masssquaredzmflavour}).   
    These conditions are also sufficient to give large masses to the first excited KK mode
    of SM gauge bosons, see (\ref{masssquaredfemw}) and (\ref{masssquaredfemzp}), and to the first 
    excited KK mode of flavour gauge boson, see (\ref{masssquaredfemflavour}).

 \chapter[Fermion masses and CKM mixing matrix]{Fermion masses and CKM mixing matrix in the $SU(7)$ model}

   \label{su7fermionmasses}

   In this chapter we describe how fermion masses are generated by the Higgs mechanism in
   the context of nonunitary parallel transporters $\Phi$.
   In general, fermion masses are given by Yukawa interactions. For example a term
   \begin{equation}
    \bar{q}_{L} \Phi q_{R} \; ,
   \end{equation}
   where $q_{L}$ and $q_{R}$ are given by (\ref{lefthandedquarks}) and (\ref{righthandedquarks}),
   respectively, will lead to a mass term for quarks. However, since quarks and leptons have different 
   masses their mass terms cannot be generated by the same nonunitary parallel transporter $\Phi$.
   In addition we will see that if the Higgs potential $V(\Phi)$ is minimised at a quasi
   $\mathcal{S}_{2}$ symmetric $\Phi_{min}$ it is not possible to obtain the correct quark and
   lepton masses nor the correct CKM mixing matrix. Therefore we make the following  
   \begin{proposal}
    \label{proposalyukawa}
    We introduce a nonunitary parallel transporter $\Phi^{quark}$ which leads to a mass term for
    the quarks via 
    \begin{equation}
     \bar{q}_{L} \; \Phi^{quark} \; q_{R} 
    \end{equation}
    and a nonunitary parallel transporter $\Phi^{lepton}$ which leads to a mass terms for the leptons via
    \begin{equation}
     \bar{l}_{L} \; \Phi^{lepton} \; l_{R} 
    \end{equation}
    and make a clear distinction between $\Phi^{quark}$ and $\Phi^{lepton}$. In particular
    $\Phi^{quark} \neq \Phi^{lepton} \neq \Phi^{gauge}$ where $\Phi^{gauge}=\Phi$
    is the nonunitary parallel transporter which leads to masses for the gauge bosons. 
   \end{proposal}
   It is important that we have three different parallel transporters in the theory. This way we can
   get different quark and lepton masses as it is observed in nature. We may speculate that at the
   GUT scale $\Phi^{quark}=\Phi^{lepton}$. 

   Since we have three different parallel transporters in the theory, we also have three Higgs potentials
   $V(\Phi^{gauge})=V(\Phi)$, $V(\Phi^{quark})$ and $V(\Phi^{lepton})$. The minimum of $V(\Phi^{quark})$ 
   and $V(\Phi^{lepton})$ is denoted by 
   $\Phi_{min}^{quark}$ and $\Phi_{min}^{lepton}$, respectively. According to (\ref{minimumphi})
   they can be parameterised as  
   \begin{equation}
    \label{minimumphiquark}
    \Phi^{quark}_{min}=\rho^{quark}_{min} \begin{pmatrix} e^{a_{1}^{quark}} & 0 & 0 & 0 & 0 & 0 & 0 \\
                   0 & e^{a_{2}^{quark}}  & 0 & 0 & 0 & 0 & 0 \\  0 & 0 & e^{a_{3}^{quark}} & 0 & 0 & 0 & 0 \\
                   0 & 0 & 0 & e^{a_{4}^{quark}} & 0 & 0 & i \alpha_{\hat{43}}^{quark \; \prime} \\
                   0 & 0 & 0 & 0 & e^{a_{5}^{quark}} & 0 & i \alpha_{\hat{45}}^{quark \; \prime} \\
                   0 & 0 & 0 & 0 & 0 & e^{a_{6}^{quark}} & i \alpha_{\hat{47}}^{quark \; \prime} \\
                   0 & 0 & 0 & i \alpha_{\hat{43}}^{quark \; \prime} & i \alpha_{\hat{45}}^{quark \; \prime} 
                   & i \alpha_{\hat{47}}^{quark \; \prime} & e^{a_{7}^{quark}} \end{pmatrix}
    \end{equation} 
   and 
   \begin{equation}
    \label{minimumphilepton}
    \Phi^{lepton}_{min}=\rho^{lepton}_{min} \begin{pmatrix} e^{a_{1}^{lepton}} & 0 & 0 & 0 & 0 & 0 & 0 \\
                0 & e^{a_{2}^{lepton}}  & 0 & 0 & 0 & 0 & 0 \\  0 & 0 & e^{a_{3}^{lepton}} & 0 & 0 & 0 & 0 \\
                0 & 0 & 0 & e^{a_{4}^{lepton}} & 0 & 0 & i \alpha_{\hat{43}}^{lepton \; \prime} \\
                0 & 0 & 0 & 0 & e^{a_{5}^{lepton}} & 0 & i \alpha_{\hat{45}}^{lepton \; \prime} \\
                0 & 0 & 0 & 0 & 0 & e^{a_{6}^{lepton}} & i \alpha_{\hat{47}}^{lepton \; \prime} \\
                0 & 0 & 0 & i \alpha_{\hat{43}}^{lepton \; \prime} & i \alpha_{\hat{45}}^{lepton \; \prime} 
                  & i \alpha_{\hat{47}}^{lepton \; \prime} & e^{a_{7}^{lepton}} \end{pmatrix} \; .
   \end{equation}
   The factor $1/\sqrt{2}$ is absorbed in $\rho_{min}^{quark}$ and $\rho_{min}^{lepton}$, respectively. 
   We assume that $V(\Phi^{quark})$ and $V(\Phi^{lepton})$
   is minimised at non-trivial $\Phi^{quark}_{min}$ and $\Phi^{lepton}_{min}$, respectively,
   i.e. we assume 
   \begin{equation}
    a^{quark}_{i} \neq 0 \quad , \quad \alpha_{\hat{a}}^{quark \prime} \neq 0 \quad , \quad
    a^{lepton}_{i} \neq 0 \quad , \quad \alpha_{\hat{a}}^{lepton \prime} \neq 0
   \end{equation} 
   for $i=1,\dots,7$ and $\hat{a}=\hat{43},\hat{45},\hat{47}$ in (\ref{minimumphiquark}) and
   (\ref{minimumphilepton}). Note that in general $a^{quark}_{i} \neq  a^{lepton}_{i} 
   \neq  a^{gauge}_{i}=a_{i}$ and $\alpha_{\hat{a}}^{quark \prime} \neq \alpha_{\hat{a}}^{lepton \prime} 
   \neq  \alpha_{\hat{a}}^{gauge \prime}=\alpha_{\hat{a}}^{\prime}$. 
   {\em{However, in contrast to $\Phi_{min}^{gauge}$, 
   we do not assume that  $\Phi^{quark}_{min}$ and $\Phi^{lepton}_{min}$
   are  quasi $\mathcal{S}_{2}$ symmetric}}. 
   The minimum $\Phi^{quark}_{min}$ and $\Phi^{lepton}_{min}$ of the Higgs potential 
   $V(\Phi^{quark})$ and $V(\Phi^{lepton})$, respectively, will fix the quark and lepton masses. In addition,
   $\Phi^{quark}_{min}$ also fixes the CKM mixing matrix. In the next section we will review
   the CKM matrix in the SM and in particular some familiar parametrisations of the CKM matrix. This review 
   will also clarify the notation and conventions we use in this 
   chapter.

  \section{Quark masses and the CKM mixing matrix in the SM}

   In the SM the Yukawa interactions for quarks is given by \cite{Hocker:2006xb}
   \begin{equation}
    \label{smyukawawm}
    \mathcal{L}_{Y}=-Y^{d}_{ij} \bar{Q}_{Li} \phi \; d_{Rj}-Y^{u}_{ij} 
                     \bar{Q}_{Li}  \epsilon \phi^{\ast} u_{Rj} + h.c. \; ,
   \end{equation} 
   where $Y^{u,d}$ are  complex $3 \times 3$ matrices, $\phi$ is the 
   SM Higgs field, $i$ and $j$ are generation labels, $\epsilon$ is the $2 \times 2$ antisymmetric tensor, 
   $Q_{L}$ are the left-handed quark 
   doublets and $d_{R}$, $u_{R}$ are the right-handed down- and up-type quark
   singlets, respectively.  When $\phi$ acquires a VEV $\langle \phi \rangle=\left( 0
   \atop v/\sqrt{2} \right)$ the Yukawa interactions yields Dirac mass
   terms for the quarks
   \begin{equation}
    \label{yukawasm}
    M^{u}=\frac{v Y^{u}}{\sqrt{2}} \quad , \quad
    M^{d}=\frac{v Y^{d}}{\sqrt{2}} \; ,
   \end{equation}
   where $M^{u}$ is the $3 \times 3$ up-type quark mass matrix and
   $M^{d}$ is the $3 \times 3$ down-type quark mass matrix. Since $M^{u}$ and $M^{d}$ are given 
   in the basis of flavour eigenstates we must diagonalise $M^{u}$ and
   $M^{d}$ by biunitary transformations
   \begin{gather}
    U_{L}^{u} M^{u} U_{R}^{u \dagger}=M^{u}_{diag} \\
    U_{L}^{d} M^{d} U_{R}^{d \dagger}=M^{d}_{diag}
   \end{gather}
   in order to move to the basis of mass eigenstates. 
   The CKM matrix $V_{CKM}$ is then given by
   \begin{equation}
    V_{CKM}=U_{L}^{u} U_{L}^{d \dagger}=\left( \begin{array}{ccc} 
                                         V_{ud} & V_{us} & V_{ub} \\
                                         V_{cd} & V_{cs} & V_{cb} \\
                                         V_{td} & V_{ts} & V_{tb} \\
                                         \end{array} \right) 
   \end{equation}
   and transforms electroweak eigenstates $(d^{\prime},s^{\prime},b^{\prime})$
   into mass eigenstates $(d,s,b)$. Since $U_{L}^{u},U_{L}^{d \dagger}$ are unitary  $3 \times 3$ matrices 
   $V_{CKM}$ is a unitary $3 \times 3$ matrix too. This feature ensures 
   the absence of tree-level FCNC in the SM.

 \subsection{Standard parametrisation and unitarity of the CKM matrix}

    Any unitary $3 \times 3$ matrix can be parametrised by three angles and
    six phases. Using the freedom to redefine the up- and down-type quarks fields one can remove 
    five unphysical phases. The CKM matrix can therefore be written
    as the product of three Euler matrices 
    \begin{equation} 
     \label{standardparaCKM}
     V_{CKM}=R_{23} U_{13} R_{12}=\left( \begin{array}{ccc} 
             c_{12}c_{13} & s_{12}c_{13} & s_{13}e^{-i\delta_{13}} \\
             -s_{12}c_{23}-c_{12}s_{23}s_{13}e^{i\delta_{13}} &
             c_{12}c_{23}-s_{12}s_{23}s_{13}e^{i\delta_{13}} & s_{23}c_{13} \\
             s_{12}s_{23}-c_{12}c_{23}s_{13}e^{i\delta_{13}} &
             -c_{12}s_{23}-s_{12}c_{23}s_{13}e^{i\delta_{13}} & c_{23}c_{13}
             \end{array} \right)  \;.
    \end{equation}
    where $c_{ij}=\cos{\theta_{ij}}$, $s_{ij}=\sin{\theta_{ij}}$ and
    $\delta_{13}$ is the CP violating phase. This parametrisation is known as the
    standard parametrisation \cite{Chau:1984fp}. The Euler matrices are given by 
    \begin{gather}
     \label{eulermatrices}
     R_{23}=\left( \begin{array}{ccc} 1 & 0 & 0 \\ 0 & c_{23} & s_{23} \\
                   0 & -s_{23} & c_{23} \end{array} \right)  \quad , \quad 
     U_{13}=\left( \begin{array}{ccc} c_{13} & 0 & s_{13}e^{-i\delta_{13}} \\
                   0 & 1 & 0 \\ -s_{13}e^{i\delta_{13}} & 0 & c_{13} 
                   \end{array} \right) \; , \\
     R_{12}=\left( \begin{array}{ccc} c_{12} & s_{12} & 0 \\
                   -s_{12} & c_{12} & 0 \\ 0 & 0 & 1 \end{array} \right) \nonumber
    \end{gather}          
    where $U_{13}$ is a complex Euler matrix and both $R_{23}$ and $R_{12}$ are real Euler matrices.
    The advantage of this parametrisation is that the mixing between
    two generations $i,j$ vanishes if the corresponding mixing angle 
    $\theta_{ij}$ it set to zero. Note that the standard parametrisation 
    satisfies exactly the unitarity relations.

    In the $SU(7)$ model we have tree-level FCNC mediated by the $SO(3)_{F}$
    flavour gauge bosons. This  tree-level FCNC violate the unitarity of the CKM matrix. 
    However this violation is extremely small due to the large flavour gauge bosons masses.
    The current experimental constraints on the unitarity of the CKM matrix are \cite{Hocker:2006xb}
    \begin{gather}
     \mid V_{ud} \mid^{2} +  \mid V_{us} \mid^{2} +  \mid V_{ub} \mid^{2}-1=-0.0008 \pm 0.0011 
     \quad \text{first row} \\
     \mid V_{cd} \mid^{2} +  \mid V_{cs} \mid^{2} +  \mid V_{cb} \mid^{2}-1=-0.03 \pm 0.18
     \quad \text{second row} \\
     \mid V_{ud} \mid^{2} +  \mid V_{cd} \mid^{2} +  \mid V_{td} \mid^{2}-1=-0.001 \pm 0.005
     \quad \text{first column} 
    \end{gather}
    The sum in the second column $\mid V_{us} \mid^{2} +  \mid V_{cs} \mid^{2} +  \mid V_{ts} \mid^{2}$ 
    is practically identical to that in the second row, as the errors in both cases are dominated 
    by $\mid V_{cs} \mid$. These experimental constraints show that the CKM matrix is almost unitary.  
    For simplicity we treat $V_{CKM}$ as a unitary $3 \times 3$ matrix in the following calculations.

  \subsection{Wolfenstein parametrisation of the CKM matrix}

    For later calculations we introduce another familiar parametrisation of the CKM matrix known
    as Wolfenstein parametrisation \cite{Wolfenstein:1983yz}.
    Following the observation that the mixing angles $s_{12},s_{13},s_{23}$ fulfil the
    hierarchy $s_{13} \ll s_{23} \ll s_{12} \ll 1$ Wolfenstein proposed an expansion of the
    CKM matrix in terms of the four parameters $\lambda,A,\rho$ and $\eta$ defined by
    \begin{equation}
     s_{12}:=\lambda \; , \quad s_{23}:=A \lambda^{2} \; , \quad 
     s_{13}:=A \lambda^{3} \left( \rho- i\eta \right)
    \end{equation}
    If we insert these definitions into (\ref{standardparaCKM}) we obtain a
    parametrisation of the CKM matrix as a function of $\lambda,A,\rho$ and $\eta$. If we expand all
    elements of the CKM matrix in powers of the small parameter $\lambda$ we get
    \begin{equation}
      V_{CKM}= \left( \begin{array}{ccc} 1-\frac{1}{2} \lambda^{2} 
               & \lambda & A \lambda^{3} (\rho-i\eta) \\
               -\lambda & 1-\frac{1}{2} \lambda^{2} & A \lambda^{2} \\  
               A \lambda^{3}(1-\rho-i\eta) & -A \lambda^{2} & 1 \\
               \end{array} \right) + \mathcal{O} \left( \lambda^{4}\right)
    \end{equation}
    Neglecting the terms of $\mathcal{O} \left( \lambda^{4}\right)$, this parametrisation is known as 
    Wolfenstein parametrisation. In the following calculations we use the numerical values
    \cite{Hocker:2006xb}
    \begin{gather}
     \label{ckmmixinganglescpphase}
     \lambda=s_{12}=0.22 \; , \quad  A \lambda^{2}=s_{23}=0.042 \\
     A \lambda^{3} (\rho-i\eta)=s_{13} \; e^{-i \delta_{13}} \; , \quad
     s_{13}=0.0039 \; , \quad \delta_{13}=\frac{2 \pi}{3}=60^{\circ}
     \nonumber
    \end{gather}

  \section{CKM mixing matrix and quark masses from $\Phi^{quark}_{min}$}

    \label{ckmandquarksection}
    
    As already mentioned in the introduction the quark masses and the CKM matrix are determined by
    \begin{equation}
     \bar{q}_{L} \Phi^{quark}_{min} q_{R} \; ,
    \end{equation}
    where 
    \begin{equation}
     \Phi^{quark}_{min}=\rho^{quark}_{min} \begin{pmatrix} e^{a_{1}^{quark}} & 0 & 0 & 0 & 0 & 0 & 0 \\
                   0 & e^{a_{2}^{quark}}  & 0 & 0 & 0 & 0 & 0 \\  0 & 0 & e^{a_{3}^{quark}} & 0 & 0 & 0 & 0 \\
                   0 & 0 & 0 & e^{a_{4}^{quark}} & 0 & 0 & i \alpha_{\hat{43}}^{quark \; \prime} \\
                   0 & 0 & 0 & 0 & e^{a_{5}^{quark}} & 0 & i \alpha_{\hat{45}}^{quark \; \prime} \\
                   0 & 0 & 0 & 0 & 0 & e^{a_{6}^{quark}} & i \alpha_{\hat{47}}^{quark \; \prime} \\
                   0 & 0 & 0 & i \alpha_{\hat{43}}^{quark \; \prime} & i \alpha_{\hat{45}}^{quark \; \prime} 
                   & i \alpha_{\hat{47}}^{quark \; \prime} & e^{a_{7}^{quark}} \end{pmatrix} \; ,
    \end{equation} 
    see (\ref{minimumphiquark}). $\Phi^{quark}_{min}$ fixes not only the quark masses but also the 
    CKM mixing matrix. This we will now be explained.

    First, comparing with (\ref{yukawasm}), the upper $3 \times 3$ submatrix
    of $\Phi^{quark}_{min}$ gives the up-type quark mass matrix
    \begin{equation}
     M^{u}=\left( \begin{array}{ccc} m_{u} & 0 & 0 \\ 0 & m_{c} & 0 \\
                  0 & 0 & m_{t} \end{array} \right)
    \end{equation}
    where 
    \begin{equation}
      \label{uptypequarkmassessu7}
      m_{u}=\rho^{quark}_{min} e^{a^{quark}_{1}} \; , \quad
      m_{c}=\rho^{quark}_{min} e^{a^{quark}_{2}} \; , \quad
      m_{t}=\rho^{quark}_{min} e^{a^{quark}_{3}} \; .
    \end{equation}
    \\
    {\bf{Conclusion: }} 
    {\em{The up-type quark masses $m_{u},m_{c},m_{t}$ are given by the
    diagonal part of  $\Phi^{quark}_{min}$ only.}} \\
    \\
    This result is in contrast to the SM where both up- and down-type quark masses are given by the
    (same) SM Higgs doublet (\ref{smyukawawm}).
   
    Second, the lower $4 \times 4$ submatrix of $\Phi^{quark}_{min}$ 
    \begin{equation}
     M=\left( \begin{array}{cccc} 
      \label{4t4mdmatrix}
      \tilde{m}_{d} & 0 & 0 & 
      i \; k_{\hat{43}} \; \langle \mathcal{A}_{y}^{\hat{43}(0)} \rangle^{quark} \\
      0 & \tilde{m}_{s} & 0 & 
      i \; k_{\hat{45}} \; \langle \mathcal{A}_{y}^{\hat{45}(0)} \rangle^{quark} \\
      0 & 0 & \tilde{m}_{b} & 
      i \; k_{\hat{47}} \; \langle \mathcal{A}_{y}^{\hat{47}(0)} \rangle^{quark} \\
      i \; k_{\hat{43}} \; \langle \mathcal{A}_{y}^{\hat{43}(0)} \rangle^{quark} &
      i \; k_{\hat{45}} \; \langle \mathcal{A}_{y}^{\hat{45}(0)} \rangle^{quark} &
      i \; k_{\hat{47}} \; \langle \mathcal{A}_{y}^{\hat{47}(0)} \rangle^{quark} &
      m_{x} \end{array} \right)
    \end{equation}
    where 
    \begin{equation}
      \label{summarymdquarks}
      \tilde{m}_{d}=\rho^{quark}_{min} e^{a^{quark}_{4}} \; , \quad
      \tilde{m}_{s}=\rho^{quark}_{min} e^{a^{quark}_{5}} \; , \quad
      \tilde{m}_{b}=\rho^{quark}_{min} e^{a^{quark}_{6}} \; , \quad
      m_{x}=\rho^{quark}_{min} e^{a^{quark}_{7}}
    \end{equation}
    and 
    \begin{equation}
     i \; k_{\hat{a}} \; \langle \mathcal{A}_{y}^{\hat{a}(0)} \rangle^{quark}=i \; \rho^{quark}_{min}
     \alpha_{\hat{a}}^{quark \; \prime}
    \end{equation}
    for $\hat{a}=\hat{43},\hat{45},\hat{47}$ will lead to the down-type
    quark mass matrix $M^{d}$. In (\ref{4t4mdmatrix}) we have
    introduced the VEVs
    \begin{equation}
     \label{vevsthreehiggsdoubletsquark}
     \langle \mathcal{A}_{y}^{\hat{43}(0)} \rangle^{quark}=\frac{\alpha_{\hat{43}}^{quark}}{g_{4} R}
     \quad , \quad
     \langle \mathcal{A}_{y}^{\hat{45}(0)} \rangle^{quark}=\frac{\alpha_{\hat{45}}^{quark}}{g_{4} R} 
     \quad , \quad
     \langle \mathcal{A}_{y}^{\hat{47}(0)} \rangle^{quark}=\frac{\alpha_{\hat{47}}^{quark}}{g_{4} R}  \; .
    \end{equation}
    in analogy with (\ref{vevsthreehiggsdoublets}). The $k_{\hat{a}}$ in (\ref{4t4mdmatrix}) are given by
    \footnote{This follows with $g \rho_{min}=1/R$, $g=g_{5}/\sqrt{\pi R}$, $g_{4}=g_{5}/\sqrt{2 \pi R}$,
    $\alpha_{\hat{43}}^{quark \prime}=\alpha_{\hat{43}}^{quark} \left(e^{a_{4}} + e^{a_{7}} \right)$
    $\alpha_{\hat{45}}^{quark \prime}=\alpha_{\hat{45}}^{quark} \left(e^{a_{5}} + e^{a_{7}} \right)$
    and $\alpha_{\hat{47}}^{quark \prime}=\alpha_{\hat{47}}^{quark} \left(e^{a_{6}} + e^{a_{7}} \right)$}
    \begin{equation}
     k_{\hat{43}}=\frac{\rho^{quark}_{min}}{\sqrt{2} \; \rho_{min}}
                  \left( e^{a_{4}} + e^{a_{7}} \right) 
     \; , \; 
     k_{\hat{45}}=\frac{\rho^{quark}_{min}}{\sqrt{2} \; \rho_{min}}
                  \left( e^{a_{5}} + e^{a_{7}} \right) 
     \; , \; 
     k_{\hat{47}}=\frac{\rho^{quark}_{min}}{\sqrt{2} \; \rho_{min}}
                  \left( e^{a_{6}} + e^{a_{7}} \right) 
    \end{equation}
    In order to obtain $M^{d}$ we have to bring $M$ on block diagonal form
    by a unitary transformation
    \begin{equation}
     \label{blockdiagonal}
     M \to U \; M \; U^{\dagger}=\left( \begin{array}{cc} M^{d} & 0 \\
                                0 & \tilde{m}_{x} \end{array} \right) \; ,
    \end{equation}
    where $M^{d}$ is the desired $3 \times 3$ down-type quark mass matrix.
    In a second step we diagonalise $M^{d}$ by a second unitary transformation
    \begin{equation}
     \label{mdtrafo}
     U^{d} \; M^{d} \; U^{d \dagger}=M^{d}_{diag} \; .
    \end{equation}
    Since $M^{u}$ is already diagonal the corresponding transformation
    for the up-type mass matrix $M^{u}$ is trivial
    \begin{equation}
     U^{u} \; M^{u} \; U^{u \dagger}=M^{u}_{diag}=M^{u} \; .
     \quad \longrightarrow \quad U^{u}=1
    \end{equation}
    Thus the CKM matrix is given by
    \begin{equation}
     V_{CKM}=U^{u} U^{d \dagger}=U^{d \dagger}  \; .
    \end{equation}  

    In the following calculations we use the numerical values   
    \cite{Eidelman:2004wy}   
    \begin{gather}
     m_{u}=1.5 \; \text{to} \; 4.5 \; \text{MeV} \quad , \quad
     m_{d}=5 \; \text{to} \; 8.5 \; \text{MeV} \\
     m_{c}=1.0 \; \text{to} \; 1.4 \; \text{GeV} \quad , \quad
     m_{s}=80 \; \text{to} \; 155 \; \text{MeV} \nonumber \\
     m_{t}=174.3 \; \pm \; 5.1 \; \text{GeV} \quad , \quad
     m_{b}=4.0 \; \text{to} \; 4.5 \; \text{GeV}  \nonumber 
    \end{gather}
    In particular for later calculations we make the (arbitrary) choice 
    \begin{equation}
     m_{d}=7.4 \; \text{MeV} \quad , \quad
     m_{s}=114.1 \; \text{MeV} \quad , \quad
     m_{b}=4250 \; \text{MeV} \; .
    \end{equation}

 \section{CKM mixing matrix from $M_{d}$}

    \label{ckmfrommd}

    According to (\ref{standardparaCKM}) $V_{CKM}$ can expressed as
    \begin{equation}
     V_{CKM}=R_{23} U_{13} R_{12} 
    \end{equation}
    where $R_{23}$, $U_{13}$ and $R_{12}$ are given by (\ref{eulermatrices}).
    We introduce the phase matrix
    \begin{equation}
     \label{p13}
     P_{13}=\left( \begin{array} {ccc} e^{i \phi_{1}} & 0 & 0 \\
            0 & 1 & 0 \\ 0 & 0 & e^{i \phi_{3}} \end{array} \right)
    \end{equation}
    where the phases $\phi_{1}$ and $\phi_{3}$ have to fulfil the constraint
    \begin{equation}
     \label{phasecon}
     \delta_{13}=\phi_{1}-\phi_{3} \; .
    \end{equation}
    Recall that $\delta_{13}$ is the CP violating phase and occurs in
    \begin{equation}
     U_{13}=\left( \begin{array}{ccc} c_{13} & 0 & s_{13}e^{-i\delta_{13}} \\
                   0 & 1 & 0 \\ -s_{13}e^{i\delta_{13}} & 0 & c_{13} 
                   \end{array} \right) \; .
    \end{equation}
    Without loss of generality let $\phi_{1}$ in (\ref{p13}) be arbitrary. The phase $\phi_{3}$ in 
    (\ref{p13})
    is then fixed by the constraint (\ref{phasecon}). With the help of (\ref{p13}) we can rewrite
    \begin{equation}
     U_{13}=P_{13}^{\ast} R_{13} P_{13}
    \end{equation}
    where $R_{13}$ is the corresponding real Euler matrix given by
    \begin{equation}
     R_{13}=\left( \begin{array}{ccc} c_{13} & 0 & s_{13} \\
                   0 & 1 & 0 \\ -s_{13} & 0 & c_{13} 
                   \end{array} \right) \; .
    \end{equation}
    The CKM matrix $V_{CKM}$ can then be written as the product
    \begin{equation}
     \label{vckmexpansion}
     V_{CKM}=R_{23} P_{13}^{\ast} R_{13} P_{13} R_{12} \; .
    \end{equation}
 
    \begin{lemma}
     Let $M^{d}$ be given by
     \begin{equation}
      M^{d}=\left( \begin{array}{ccc} \tilde{m}^{\prime}_{d} & m_{12} 
           & m_{13} \\ m_{12}^{\ast} & \tilde{m}^{\prime}_{s} & m_{23} \\  
             m_{13}^{\ast} & m_{23} & m_{b} \\ \end{array} \right)
     \end{equation}
     where 
     \begin{eqnarray}
      \label{ckmmatrixoffdiagonalelements}
      & \tilde{m}_{d}^{\prime}=13.4 \; \text{MeV}  
      & m_{12}=\hat{m}_{12} + \hat{m}_{13} s_{23} e^{i \; \delta_{13}}
              \approx 24.46 \; \text{MeV} \nonumber \\
      & \tilde{m}_{s}^{\prime}=119.2 \; \text{MeV} 
      &  m_{13}=\hat{m}_{12} s_{23} - \hat{m}_{13} e^{i \; \delta_{13}}
               \approx 16.65 \; e^{i \; \frac{2 \pi}{3}} \text{MeV} 
               \nonumber \\
      & m_{b}=4250 \; \text{MeV} 
      & m_{23}=173.8 \; \text{MeV} \; ,
     \end{eqnarray}
     and
     \begin{equation}
      \label{valueoffdiagonalelelmentscomplexckm}
      \hat{m}_{12}=24.8 \; \text{MeV}
      \quad , \quad \hat{m}_{13}=16.1 \; \text{MeV} 
      \quad , \quad \delta_{13}=\frac{2 \pi}{3}  \; .
     \end{equation}
     The unitary transformation 
     \begin{equation}
      \label{ckmlemmaunit}
      V_{CKM}^{\dagger} \; M^{d} \; V_{CKM}=M^{d}_{diag} \; .
     \end{equation}
     leads to the CKM matrix
     \begin{equation}
      V_{CKM}= \left( \begin{array}{ccc} 1-\frac{1}{2} \lambda^{2} 
               & \lambda & A \lambda^{3} (\rho-i\eta) \\
               -\lambda & 1-\frac{1}{2} \lambda^{2} & A \lambda^{2} \\  
               A \lambda^{3}(1-\rho-i\eta) & -A \lambda^{2} & 1 \\
               \end{array} \right)
     \end{equation}
     where $\lambda=s_{12}=0.22$, $A \lambda^{2}=s_{23}=0.042$, 
     $A \lambda^{3} (\rho-i\eta)=s_{13} \; e^{-i \delta_{13}}$, 
     $s_{13}=0.0039$ and $\delta_{13}=\frac{2 \pi}{3}$.
   \end{lemma}
   The proof can be found in Appendix \ref{appendixckmfrommd}.

  \section{Obtaining $M^{d}$ from $M$}

   In this section we finally describe how the down-type quark mass matrix $M^{d}$ is obtained
   from (\ref{4t4mdmatrix}) 
   \begin{equation*}
     M=\left( \begin{array}{cccc} 
      \tilde{m}_{d} & 0 & 0 & 
      i \; k_{\hat{43}} \; \langle \mathcal{A}_{y}^{\hat{43}(0)} \rangle^{quark} \\
      0 & \tilde{m}_{s} & 0 & 
      i \; k_{\hat{45}} \; \langle \mathcal{A}_{y}^{\hat{45}(0)} \rangle^{quark} \\
      0 & 0 & \tilde{m}_{b} & 
      i \; k_{\hat{47}} \; \langle \mathcal{A}_{y}^{\hat{47}(0)} \rangle^{quark} \\
      i \; k_{\hat{43}} \; \langle \mathcal{A}_{y}^{\hat{43}(0)} \rangle^{quark} &
      i \; k_{\hat{45}} \; \langle \mathcal{A}_{y}^{\hat{45}(0)} \rangle^{quark} &
      i \; k_{\hat{47}} \; \langle \mathcal{A}_{y}^{\hat{47}(0)} \rangle^{quark} &
      m_{x} \end{array} \right)
   \end{equation*}
   It is convenient to introduce the following abbreviation
   \begin{equation}
     \label{summaryabbre}
     \langle \mathcal{A}_{y}^{\hat{a}} \rangle^{q}_{k}:= 
     k_{\hat{a}} \; \langle \mathcal{A}_{y}^{\hat{a}(0)} \rangle^{quark}
   \end{equation}
   for $\hat{a}=\hat{43},\hat{45},\hat{47}$. With this abbreviation $M$ reads
   \begin{equation}
    \label{matrixMcalculations}
    M=\left( \begin{array}{cccc} 
     \tilde{m}_{d} & 0 & 0 &  \langle \mathcal{A}_{y}^{\hat{43}} \rangle^{q}_{k}
                              \; e^{i \; \frac{\pi}{2}}\\
     0 & \tilde{m}_{s} & 0 &  \langle \mathcal{A}_{y}^{\hat{45}} \rangle^{q}_{k}
                              \; e^{i \; \frac{\pi}{2}} \\
     0 & 0 & \tilde{m}_{b} &  \langle \mathcal{A}_{y}^{\hat{47}} \rangle^{q}_{k}
                              \; e^{i \; \frac{\pi}{2}} \\
     \langle \mathcal{A}_{y}^{\hat{43}} \rangle^{q}_{k} \; e^{i \; \frac{\pi}{2}} &
     \langle \mathcal{A}_{y}^{\hat{45}} \rangle^{q}_{k} \; e^{i \; \frac{\pi}{2}} &
     \langle \mathcal{A}_{y}^{\hat{47}} \rangle^{q}_{k}  \; e^{i \; \frac{\pi}{2}} &
     m_{x} \end{array} \right) 
   \end{equation}
   where we have written all phases explicitly. As already explained in section 
   \ref{ckmandquarksection} in order to obtain the $3 \times 3$ down-type quark mass 
   matrix $M^{d}$ we must transform $M$ to block 
   diagonal form
   \begin{equation}
    \label{unitarytrafo}
    M \to U \; M \; U^{\dagger}=\left( \begin{array}{cc} M^{d} & 0 \\
                                0 & \tilde{m}_{x} \end{array} \right) \; .
   \end{equation}
   We write $U$ as a product of a phase matrix and three Euler matrices
   \begin{equation}
    \label{summarymatrixudagger}
    U^{\dagger}=P_{34} R_{34} R_{24} R_{14}
   \end{equation}
   where
   \begin{gather}
    \label{mixinganglesmmd}
    P_{34}=\left( \begin{array}{cccc} 1 & 0 & 0 & 0 \\ 0 & 1 & 0 & 0 \\ 
           0 & 0 & e^{i \; \phi_{3}} & 0 \\ 0 & 0 & 0 & e^{i \; \phi_{4}}
           \end{array} \right) \; , \quad
    R_{34}=\left( \begin{array}{cccc} 1 & 0 & 0 & 0 \\
                  0 & 1 & 0 & 0 \\ 0 & 0 & c_{34} & s_{34} \\
                  0 & 0 & -s_{34} & c_{34} \end{array} \right) \\
    R_{24}=\left( \begin{array}{cccc} 1 & 0 & 0 & 0 \\ 
                  0 & c_{24} & 0 & s_{24} \\ 0 & 0 & 1 & 0 \\
                  0 & -s_{24} & 0 & c_{24} 
                  \end{array} \right) \; , \quad
    R_{14}=\left( \begin{array}{cccc} c_{14} & 0 & 0 & s_{14} \\
                   0 & 1 & 0 & 0 \\ 0 & 0 & 1 & 0 \\
                  -s_{14} & 0 & 0 & c_{14} \end{array} \right) \nonumber
   \end{gather}   
   Inserting this expansion in (\ref{unitarytrafo}) we get
   \begin{equation}
    \label{mmdcalcu}
    R^{t}_{14} R^{t}_{24} R^{t}_{34} P_{34}^{\ast} \; M \; P_{34} R_{34} R_{24} R_{14}=
                              \left( \begin{array}{cc} M^{d} & 0 \\
                                0 & \tilde{m}_{x} \end{array} \right) \; .
   \end{equation}
   The matrix (\ref{matrixMcalculations}) has seven undetermined parameters: The diagonal elements
   $\tilde{m}_{d},\tilde{m}_{s},\tilde{m}_{b},m_{x}$ and the off-diagonal elements 
   $\langle \mathcal{A}_{y}^{\hat{43}} \rangle^{q}_{k}$, $\langle \mathcal{A}_{y}^{\hat{45}} \rangle^{q}_{k}$,
   $\langle \mathcal{A}_{y}^{\hat{47}} \rangle^{q}_{k}$. The question is now: 
   \begin{itemize}
    \item How do these parameters determine the CKM mixing angles $s_{12},s_{23},s_{13}$ and the CP 
          violating phase $\delta_{13}$ ?
    \item How do  these parameters determine the down-type quark masses?
    \item What is the interpretation of the three off-diagonal elements 
          $\langle \mathcal{A}_{y}^{\hat{43}} \rangle^{q}_{k}$, 
          $\langle \mathcal{A}_{y}^{\hat{45}} \rangle^{q}_{k}$,
          $\langle \mathcal{A}_{y}^{\hat{47}} \rangle^{q}_{k}$ ?
    \item Can we get the SM case?
   \end{itemize}
   The situation is complicated because we do not know anything about the explicit numerical values 
   of these seven parameters. 
   
   In the next subsection we will consider the special but analytically exact solvable case, where we
   assume that
   \begin{enumerate}
    \item $m_{x}$ is the dominant element in $M$, i.e. 
          $m_{x} \gg \tilde{m}_{d},\tilde{m}_{s},\tilde{m}_{b}$ and $m_{x} \gg 
          \langle \mathcal{A}_{y}^{\hat{a}} \rangle^{q}_{k}$ for 
          $\hat{a}=\hat{43},\hat{45},\hat{47}$. We choose in the following calculations
          the (arbitrary) value $m_{x}=10000 \; m_{b}$. This assumption follows from the fact that 
          $m_{x}$ {\em{may}} be interpreted as the mass of a seventh quark. However we remind the
          reader that in the $SU(7)$ model $m_{x}$ is {\em{only}} a parameter. It will turn out
          that the results of the following calculations are independent of the explicit value for
          $m_{x}$ as long as $m_{x}$ is the dominant element is $M$. 
    \item the $d$-quark masses are approximately given by the diagonal part of (\ref{matrixMcalculations})
          \begin{equation}
           m_{d} \approx \tilde{m}_{d} \; , \quad  m_{s} \approx \tilde{m}_{s} \; , \quad 
           m_{b} \approx \tilde{m}_{b}
          \end{equation}
          It will turn out during the following calculations 
          that this assumption is nearly fulfilled for $m_{b}$ and poorly
          fulfilled for $m_{s}$ and $m_{d}$.
   \end{enumerate}
   Since $m_{x}$ is the dominant element in $M$ it will
   turn out that all mixing angles in (\ref{mixinganglesmmd}) are very small.  Therefore we call this
   case {\em{small mixing angle approximation}}.

  \subsection{$M^{d}$ from $M$ in small mixing angle approximation}

   In (\ref{matrixMcalculations}) we choose the following numerical values
   \begin{gather}
    \label{valuesvevso3}
    \langle \mathcal{A}_{y}^{\hat{43}} \rangle^{q}_{k}=63910 \; \text{MeV} \\
    \langle \mathcal{A}_{y}^{\hat{45}} \rangle^{q}_{k}=104240 \; \text{MeV} \nonumber \\
    \langle \mathcal{A}_{y}^{\hat{47}} \rangle^{q}_{k}=70950 \; \text{MeV}  \nonumber 
   \end{gather}
   The reason for this choice will be become clear later.
   In addition, as already mentioned above we choose
   \begin{equation}
    \label{valuemx}
    m_{x}=10000 \; m_{b}=4.25 \times 10^{7} \; \text{MeV}
   \end{equation}
   The diagonal elements $ \tilde{m}_{d}$, $\tilde{m}_{s}$ and $\tilde{m}_{b}$ will be fixed
   during the following computation. According to (\ref{mmdcalcu}) the calculation is divided 
   into four steps.  \\
   \\
   First step: \; We multiply $M$ by the phase matrix
   \begin{equation}
    P_{34}=\left( \begin{array}{cccc} 1 & 0 & 0 & 0 \\ 0 & 1 & 0 & 0 \\ 
           0 & 0 & e^{i \; \phi_{3}} & 0 \\ 0 & 0 & 0 & e^{i \; \phi_{4}}
           \end{array} \right)
   \end{equation}
   The purpose of this rephasing is to get a real angle $\theta_{34}$ in the next step. The
   complete transformation reads
   \begin{eqnarray}
     M_{1} & = & P_{34}^{L \ast} M P_{34}^{R} \\
           & = & \left( \begin{array}{cccc} 
           \tilde{m}_{d} & 0 & 0 & \langle \mathcal{A}_{y}^{\hat{43}} \rangle^{q}_{k} \; 
           e^{i \; (\frac{\pi}{2}+\phi_{4}^{R})}  \\
           0 & \tilde{m}_{s} & 0 & \langle \mathcal{A}_{y}^{\hat{45}} \rangle^{q}_{k} \; e^{i 
           \; (\frac{\pi}{2}+\phi_{4}^{R})} \\
           0 & 0 & \tilde{m}_{b} \; e^{i \; (\phi_{3}^{R}-\phi_{3}^{L})} & 
           \langle \mathcal{A}_{y}^{\hat{47}} \rangle^{q}_{k}
           \; e^{i \; (\frac{\pi}{2}+\phi_{4}^{R}-\phi_{3}^{L})} \\
           \langle \mathcal{A}_{y}^{\hat{43}} \rangle^{q}_{k} \; e^{i \; (\frac{\pi}{2}-\phi_{4}^{L})} &
           \langle \mathcal{A}_{y}^{\hat{45}} \rangle^{q}_{k} \; e^{i \; (\frac{\pi}{2}-\phi_{4}^{L})} & 
           \langle \mathcal{A}_{y}^{\hat{47}} \rangle^{q}_{k} 
           \; e^{i \; (\frac{\pi}{2}+\phi_{3}^{R}-\phi_{4}^{L})} & 
           m_{x}  \; e^{i \; (\phi_{4}^{R}-\phi_{4}^{L})}
           \end{array} \right) \nonumber 
   \end{eqnarray}
   Concerning this transformation we give a comment. In (\ref{unitarytrafo}) we have argued that
   $M$ can be brought to block diagonal form via a unitary transformation. However within the
   approximations we will make during the following calculations we have to
   consider a multiplication of $M$ from right with
   $P_{34}^{R}=diag(1,1,e^{i \; \phi^{R}_{3}},
   e^{i \; \phi^{R}_{4}})$ and from left with $P_{34}^{L \ast}=diag(1,1,e^{-i \; \phi^{L}_{3}},
   e^{-i \; \phi^{L}_{4}})$ and as a start have to make a distinction between $\phi_{3}^{L}$ and $\phi_{3}^{R}$
   and between $\phi_{4}^{L}$ and $\phi_{4}^{R}$, respectively. It will turns out during the following
   calculation that this distinction is necessary in order to get real mixing angles $\theta_{34}$
   and $\theta_{24}$. In addition we will compute explicit values for all phases. The
   resulting transformation (\ref{unitarytrafo}) will turn out to be in fact biunitary within the
   approximations we will make during the following calculations. 
   We introduce the following abbreviations
   \begin{gather}
    \delta_{3}=\phi^{R}_{3}-\phi^{L}_{3} \quad , \quad
    \delta_{4}=\phi^{R}_{4}-\phi^{L}_{4} \\
    \delta=\phi^{R}_{4}-\phi^{L}_{3}+\frac{\pi}{2} \quad , \quad
    \tilde{\delta}=\phi^{R}_{3}-\phi^{L}_{4}+\frac{\pi}{2} \nonumber
   \end{gather}
   \\
   Second step: \; 
   We perform the rotation $R_{34}$ on $M_{1}$. The purpose is to put
   zeroes in the $34,43$ elements.
   The zeroes in the $34,43$ elements are implemented by diagonalising the
   lower $34$ block of $M_{1}$. This block is obtained by striking out the row 
   and the column in which the unit elements of $R_{34}$ appear.
   Thus we get the reduced rotation
   \begin{equation}
    R^{t}_{34} \left( \begin{array}{cc} \tilde{m}_{b} \; e^{i \; \delta_{3}}
             & \langle \mathcal{A}_{y}^{\hat{47}} \rangle^{q}_{k} \; e^{i \; \delta} \\  
             \langle \mathcal{A}_{y}^{\hat{47}} \rangle^{q}_{k} \; e^{i \; \tilde{\delta}} & 
              m_{x}  \; e^{i \; \delta_{4}} \end{array} \right) R_{34} := 
    \left( \begin{array}{cc} m_{b} \; e^{i \; \delta_{b}} & 0 \\  
                0 & \tilde{m}_{x} \; e^{i \; \tilde{\delta}_{x}} \end{array} \right) 
   \end{equation}
   where we have introduced two new phases $\delta_{b}$ and $\delta_{x}$.
   From this matrix equation we obtain the mixing angle $\theta_{34}$ 
   \begin{equation}
    \tan{2 \theta_{34}}= 
    \frac{2\left(m_{x} \; e^{i \; \delta_{4}} \; \langle \mathcal{A}_{y}^{\hat{47}} \rangle^{q}_{k} 
          \; e^{i \; \delta} + \tilde{m}_{b} \; e^{i \; \delta_{3}} 
          \; \langle \mathcal{A}_{y}^{\hat{47}} \rangle^{q}_{k} \; e^{i \; \tilde{\delta}}\right)}
         {m^{2}_{x} \; e^{i \; 2\delta_{4}} - \tilde{m}^{2}_{b} \; e^{i \; 2\delta_{3}} 
          + \left( \langle \mathcal{A}_{y}^{\hat{47}} \rangle^{q}_{k} \right)^{2} \; \left( e^{i \; 2\delta} 
          - e^{i \; 2\tilde{\delta}} \right)}
   \end{equation}
   In order to simplify this equation let us fix
   \begin{equation}
    \label{phasecondition1}
    \delta=\tilde{\delta} \; \longrightarrow \; \phi^{R}_{4}-\phi^{L}_{3}=
    \phi^{R}_{3}-\phi^{L}_{4} \; .
   \end{equation}
   Using this fixation $\theta_{34}$ can be written as
   \begin{equation}
    \tan{2 \theta_{34}}=\frac{2 \; \langle \mathcal{A}_{y}^{\hat{47}} \rangle^{q}_{k}  
                \; e^{i \; \delta}}
                {m_{x} \; e^{i \; \delta_{4}} - \tilde{m}_{b} \; e^{i \; \delta_{3}}} \; .
   \end{equation}
   The requirement that the angle $\theta_{34}$ is real means that the
   numerator and the denominator  
   must have equal phases. This leads to the condition
   \begin{equation}
    \label{phasecondition2}
    m_{x} \sin(\delta-\delta_{4})=\tilde{m}_{b} \sin(\delta-\delta_{3}) \; .
   \end{equation} 
   Now we take into account that $m_{x}=10000 \; m_{b} \gg \tilde{m}_{b}$. This leads to the 
   phase condition $\delta \approx \delta_{4}$ and we obtain 
   \begin{equation}
     \label{phasecondition3}
    \delta=\phi^{R}_{4}-\phi^{L}_{3}+\frac{\pi}{2} \approx \phi^{R}_{4}-\phi^{L}_{4}
    \quad \longrightarrow \quad  \phi_{3}^{L} \approx \phi_{4}^{L}+\frac{\pi}{2} \; .
   \end{equation}
   Thus the mixing angle turns out to be
   \begin{equation}
    \theta_{34} \approx \frac{\langle \mathcal{A}_{y}^{\hat{47}} \rangle^{q}_{k}}{m_{x}}=0.0017 \ll 1 
   \end{equation}
   where we have inserted (\ref{valuesvevso3}) and (\ref{valuemx}). 
   For the diagonal elements one gets in small mixing angle approximation
   \begin{gather} 
    \label{massmb}
    m_{b} \; e^{i \; \delta_{b}}  \approx \tilde{m}_{b} \; e^{i \; \delta_{3}} 
          - \frac{\left( \langle \mathcal{A}_{y}^{\hat{47}} \rangle^{q}_{k} \right)^{2}}{m_{x}}
             \; e^{i \; \delta_{4}} \\
    \tilde{m}_{x} \; e^{i \; \tilde{\delta}_{x}} \approx m_{x} \; e^{i \; \delta_{4}}  
          + \frac{2 \left( \langle \mathcal{A}_{y}^{\hat{47}} \rangle^{q}_{k} \right)^{2}}{m_{x}}
            \; e^{i \; \delta_{4}} 
            \; . 
   \end{gather}
   We calculate from (\ref{valuesvevso3}) and (\ref{valuemx})
   \begin{equation}
    \frac{\left( \langle \mathcal{A}_{y}^{\hat{47}} \rangle^{q}_{k} \right)^{2}}{m_{x}} \approx 118 \; .
   \end{equation}
   Thus $m_{x}=10000 \; m_{b}$ is practically unchanged 
   \begin{equation}
    \tilde{m}_{x}=m_{x}
   \end{equation}
   and hence $\tilde{\delta}_{x}$ is given by 
   \begin{equation}
    \tilde{\delta}_{x}=\delta_{4} \; .
   \end{equation}
   Let us analyse equation (\ref{massmb}). Since 
   $\frac{\left( \langle \mathcal{A}_{y}^{\hat{47}} \rangle^{q}_{k} \right)^{2}}{m_{x}} \approx 118$
   the second  term can only give a small contribution to $m_{b}$. Therefore we 
   approximately get
   \begin{equation}
    \delta_{b} \approx \delta_{3} \; .
   \end{equation}
   The remaining elements are given by the reduced rotation in small mixing angle approximation 
   \begin{equation}
    \left( \begin{array}{cc} \tilde{m}_{13} & \tilde{m}_{14} \\
            \tilde{m}_{23} & \tilde{m}_{24} \end{array} \right)
    =\left( \begin{array}{cc} 0 & \langle \mathcal{A}_{y}^{\hat{43}} \rangle^{q}_{k} 
            \; e^{i \; (\frac{\pi}{2}+\phi_{4}^{R})} \\
            0 & \langle \mathcal{A}_{y}^{\hat{45}} \rangle^{q}_{k}  
            \; e^{i \; (\frac{\pi}{2}+\phi_{4}^{R})} \end{array} \right)    
    \left( \begin{array}{cc} 1 & \theta_{34} \\
            -\theta_{34} & 1 \end{array} \right) \; .
   \end{equation}
   This leads to
   \begin{gather}
    \tilde{m}_{13}=\theta_{34} \;  \langle \mathcal{A}_{y}^{\hat{43}} \rangle^{q}_{k}   
    \; e^{i \; (\frac{3\pi}{2}+\phi_{4}^{R})} \quad , \quad  
    \tilde{m}_{14}=\langle \mathcal{A}_{y}^{\hat{43}} \rangle^{q}_{k} 
    \; e^{i \; (\frac{\pi}{2}+\phi_{4}^{R})}  \\     
    \tilde{m}_{23}=\theta_{34} \; \langle \mathcal{A}_{y}^{\hat{45}} \rangle^{q}_{k}  
    \; e^{i \; (\frac{3\pi}{2}+\phi_{4}^{R})} \quad , \quad  
    \tilde{m}_{24}=\langle \mathcal{A}_{y}^{\hat{45}} \rangle^{q}_{k} 
     \; e^{i \; (\frac{\pi}{2}+\phi_{4}^{R})}  \; .
   \end{gather}
   In addition 
   \begin{equation}
    \left( \begin{array}{cc} \tilde{m}_{31} & \tilde{m}_{32} \\
            \tilde{m}_{41} & \tilde{m}_{42} \end{array} \right)
    = \left( \begin{array}{cc} 1 & -\theta_{34} \\
            \theta_{34} & 1 \end{array} \right)
      \left( \begin{array}{cc} 0 & 0 \\ \langle \mathcal{A}_{y}^{\hat{43}} \rangle^{q}_{k} 
             \; e^{i \; (\frac{\pi}{2}-\phi_{4}^{L})}
            & \langle \mathcal{A}_{y}^{\hat{45}} \rangle^{q}_{k} 
             \; e^{i \; (\frac{\pi}{2}-\phi_{4}^{L})} \end{array} \right)    
   \end{equation}
   leads to
   \begin{gather}
    \tilde{m}_{31}=\theta_{34} \; \langle \mathcal{A}_{y}^{\hat{43}} \rangle^{q}_{k} 
    \; e^{i \; (\frac{3\pi}{2}-\phi_{4}^{L})} \quad , \quad  
    \tilde{m}_{32}=\theta_{34} \; \langle \mathcal{A}_{y}^{\hat{45}} \rangle^{q}_{k}  
    \; e^{i \; (\frac{3\pi}{2}-\phi_{4}^{L})}  \\     
    \tilde{m}_{41}=\langle \mathcal{A}_{y}^{\hat{43}} \rangle^{q}_{k} \; e^{i \; (\frac{\pi}{2}-\phi_{4}^{L})} 
    \quad , \quad  
    \tilde{m}_{42}=\langle \mathcal{A}_{y}^{\hat{45}} \rangle^{q}_{k} 
    \; e^{i \; (\frac{\pi}{2}-\phi_{4}^{L})} \; .
   \end{gather}
   The full transformation reads 
   \begin{eqnarray}
     M_{2} & = & R_{34}^{t} M_{1} R_{34} \\
           & = & \left( \begin{array}{cccc} 
           \tilde{m}_{d} & 0 & \mid \tilde{m}_{13} \mid \; e^{i \; (\frac{3\pi}{2}+\phi_{4}^{R})}  
           &  \langle \mathcal{A}_{y}^{\hat{43}} \rangle^{q}_{k}   \; e^{i \; (\frac{\pi}{2}+\phi_{4}^{R})}  \\
           0 & \tilde{m}_{s} & \mid \tilde{m}_{23} \mid \; e^{i \; (\frac{3\pi}{2}+\phi_{4}^{R})} 
           &  \langle \mathcal{A}_{y}^{\hat{45}} \rangle^{q}_{k}  \; e^{i \; (\frac{\pi}{2}+\phi_{4}^{R})} \\
           \mid \tilde{m}_{13} \mid \; e^{i \; (\frac{3\pi}{2}-\phi_{4}^{L})}  
           & \mid \tilde{m}_{23} \mid \; e^{i \; (\frac{3\pi}{2}+\phi_{4}^{R})}  
           & m_{b} \; e^{i \; \delta_{3}} & 0 \\
           \langle \mathcal{A}_{y}^{\hat{43}} \rangle^{q}_{k}  \; e^{i \; (\frac{\pi}{2}-\phi_{4}^{L})} &
           \langle \mathcal{A}_{y}^{\hat{45}} \rangle^{q}_{k}  \; e^{i \; (\frac{\pi}{2}-\phi_{4}^{L})} & 
           0 & m_{x}  \; e^{i \; \delta_{4}} 
           \end{array} \right) \nonumber
   \end{eqnarray}
   \\
   Third step : \;
   We perform the real rotation $R_{24}$ on $M_{2}$. The purpose is to put
   zeroes in the $24,42$ elements.
   The zeroes in the $24,42$ elements are implemented by diagonalising the
   lower middle block of $M_{2}$. This block is obtained by striking out the 
   row and the column in which the unit elements of $R_{24}$ appear.
   Thus we get the reduced rotation
   \begin{equation}
    R^{t}_{24} \left( \begin{array}{cc} \tilde{m}_{s} 
             &  \langle \mathcal{A}_{y}^{\hat{45}} \rangle^{q}_{k}  
             \; e^{i \; (\frac{\pi}{2}+\phi_{4}^{R})}\\  
             \langle \mathcal{A}_{y}^{\hat{45}} \rangle^{q}_{k}  \; e^{i \; (\frac{\pi}{2}-\phi_{4}^{L})} & 
             m_{x}  \; e^{i \; \delta_{4}}  \end{array} \right) R_{24} := 
    \left( \begin{array}{cc} m_{s} \; e^{i \; \delta_{s}} & 0 \\  
                0 & \tilde{m}_{x} \; e^{i \; \tilde{\delta}_{x}^{\prime}} \end{array} \right) \; .
   \end{equation}
   where we have introduced two new phases $\delta_{s}$ and $\tilde{\delta}_{x}^{\prime}$.
   From this matrix equation we obtain the mixing angle $\theta_{24}$
   \begin{equation}
    \tan{2 \theta_{24}}= 
    \frac{2\left(m_{x} \; e^{i \; \delta_{4}} \; 
                  \langle \mathcal{A}_{y}^{\hat{45}} \rangle^{q}_{k}  \; e^{i \; (\frac{\pi}{2}+\phi_{4}^{R})}+
                  \tilde{m}_{s} \;  \langle \mathcal{A}_{y}^{\hat{45}} \rangle^{q}_{k} 
                  \; e^{i \; (\frac{\pi}{2}-\phi_{4}^{L})} \right)}
         {m^{2}_{x} \; e^{i \; 2\delta_{4}} - \tilde{m}^{2}_{s}  
           +  \left( \langle \mathcal{A}_{y}^{\hat{45}} \rangle^{q}_{k} \right)^{2} 
            \; \left( e^{i \; 2(\frac{\pi}{2}+\phi_{4}^{R})} 
             - e^{i \; 2(\frac{\pi}{2}-\phi_{4}^{L})} \right)}  
   \end{equation}
   Let us fix
   \begin{equation}
    \label{phasecondition4}
    \phi^{R}_{4}=-\phi^{L}_{4} \; .
   \end{equation}
   Using this fixation $\theta_{24}$ can be written as
   \begin{equation}
    \tan{2 \theta_{24}}=\frac{2 \;  \langle \mathcal{A}_{y}^{\hat{45}} \rangle^{q}_{k} 
                 \; e^{i \; \delta^{\prime}}}
                {m_{x} \; e^{i \; \delta_{4}} - \tilde{m}_{s}} 
   \end{equation}
   where $\delta^{\prime}:=\frac{\pi}{2}+\phi_{4}^{R}$.
   The requirement that $\theta_{24}$ is real means that the
   numerator and the denominator of this equation 
   must have equal phases. This leads to the condition
   \begin{equation}
    \label{phasecondition5}
    m_{x} \sin(\delta^{\prime}-\delta_{4})=\tilde{m}_{s} \sin(\delta^{\prime}) \; .
   \end{equation}   
   Now we take into account that $m_{x}=10000 m_{b} \gg \tilde{m}_{s}$. This leads to the 
   phase condition $\delta^{\prime} \approx \delta_{4}$ and the mixing angle $\theta_{24}$ turns out to be
   \begin{equation}
    \theta_{24} \approx \frac{ \langle \mathcal{A}_{y}^{\hat{45}} \rangle^{q}_{k} }{m_{x}}=0.0025 \ll 1
   \end{equation}
   where we have inserted (\ref{valuesvevso3}) and (\ref{valuemx}).
   The phase condition $\delta^{\prime} \approx \delta_{4}$ now fixes the absolute value of
   $\phi_{4}^{R}$ and  $\phi_{4}^{L}$
   \begin{equation}
    \delta^{\prime}=\frac{\pi}{2}+\phi_{4}^{R} \approx \delta_{4}=\phi_{4}^{R}-\phi_{4}^{L}
    \stackrel{!}{=} 2 \phi_{4}^{R}
    \quad \longrightarrow \quad  \phi_{4}^{R}=\frac{\pi}{2}
   \end{equation}
   where we have make use of (\ref{phasecondition4}) in the second step.
   Thus we obtain
   \begin{equation}
    \delta_{4}=2 \phi_{4}^{R}=\pi \; , 
   \end{equation}
   and
   \begin{equation}
    \delta_{3}=\phi_{3}^{R}-\phi_{3}^{L}=\phi_{4}^{R}+\phi_{4}^{L}-2\phi_{3}^{L} \approx 
                -2 (\phi_{4}^{L} + \frac{\pi}{2})=0
   \end{equation}
   where we have used (\ref{phasecondition1}) in the first and 
   (\ref{phasecondition3}) in the second step.
   We see that these values are in accordance with (\ref{phasecondition2}) and
   (\ref{phasecondition5}) for large $m_{x}$. 
   The phases in $P_{34}^{R}$ respectively $P_{34}^{L}$ read  
   \begin{equation}
    \phi_{4}^{R}=-\phi_{4}^{L}=\frac{\pi}{2} \quad , \quad \phi_{3}^{R}=-\phi_{3}^{L}=0 \; .
   \end{equation}
   With $\delta_{3}=0$ and  $\delta_{3}=\pi R$ equation (\ref{massmb}) becomes
   \begin{equation}
     m_{b} \; e^{i \; \delta_{b}}  \approx \tilde{m}_{b} \; e^{i \; \delta_{3}} 
          - \frac{\left( \langle \mathcal{A}_{y}^{\hat{45}} \rangle^{q}_{k} \right)^{2} }
            {m_{x}} \; e^{i \; \delta_{4}}
           =\tilde{m}_{b}+\frac{\left( \langle \mathcal{A}_{y}^{\hat{45}} \rangle^{q}_{k} \right)^{2}}
            {m_{x}} \\
   \end{equation}
   For $m_{b}=4250$ MeV we obtain
   \begin{equation}
    \tilde{m}_{b}=4132 \; \text{MeV} \;
   \end{equation}
   since $\frac{\left( \langle \mathcal{A}_{y}^{\hat{45}} \rangle^{q}_{k} \right)^{2}}{m_{x}}=105$.
   In addition we get
   \begin{equation}
    \delta_{b}=\delta_{3}=0 \; .
   \end{equation}
   
   Next we determine the diagonal elements. Since $\theta_{24} \ll 1$ we get in 
   the small mixing angle approximation
   \begin{equation} 
    m_{s} e^{i \; \delta_{s}} \approx \tilde{m}_{s}
          - \frac{2 \; \left( \langle \mathcal{A}_{y}^{\hat{45}} \rangle^{q}_{k} \right)^{2}}{m_{x}} 
            \; e^{i \; (\frac{\pi}{2}+\phi_{4}^{R})} 
          + \frac{\left( \langle \mathcal{A}_{y}^{\hat{45}} \rangle^{q}_{k} \right)^{2}}{m_{x}} 
            \; e^{i \; \delta_{4}}
          = \tilde{m}_{s} 
          + \frac{\left( \langle \mathcal{A}_{y}^{\hat{45}} \rangle^{q}_{k} \right)^{2}}{m_{x}} \\    
   \end{equation}
   and $m_{x}$ is again practically unchanged. This leads to
   \begin{equation}
    \tilde{\delta}_{x}^{\prime}=\delta_{4}=\pi \; .
   \end{equation}
   We calculate from (\ref{valuesvevso3}) and (\ref{valuemx})
   \begin{equation}
    \frac{\left( \langle \mathcal{A}_{y}^{\hat{45}} \rangle^{q}_{k} \right)^{2}}{m_{x}}=255 \; \text{MeV} \; .
   \end{equation}
   However, since $\tilde{m}_{s}$ cannot be negative we conclude that
   \begin{equation}
     m_{s}=m_{s}^{\star} > 255 \; \text{MeV}
   \end{equation}
   and the phase $\delta_{s}$ turns out to be
   \begin{equation}
    \delta_{s}=0 \; .
   \end{equation}
   The off-diagonal elements are given by the reduced rotation
   \begin{equation}
    \left( \begin{array}{cc} \tilde{m}^{\prime}_{12} 
            & \tilde{m}^{\prime}_{14} \\
            \tilde{m}^{\prime}_{32} 
            & \tilde{m}^{\prime}_{34} \end{array} \right)
    =\left( \begin{array}{cc} 0 &  \langle \mathcal{A}_{y}^{\hat{43}} \rangle^{q}_{k}  e^{i \; \pi}  \\
            \theta_{34} \; \langle \mathcal{A}_{y}^{\hat{45}} \rangle^{q}_{k} & 0 \end{array} \right) 
         \left( \begin{array}{cc} 1 & \theta_{24} \\
            -\theta_{24} & 1 \end{array} \right) \; .
   \end{equation}
   This leads to
   \begin{gather}
    \tilde{m}^{\prime}_{12}=\theta_{24}  \langle \mathcal{A}_{y}^{\hat{43}} \rangle^{q}_{k} 
    \quad , \quad  
    \tilde{m}^{\prime}_{14}=\langle \mathcal{A}_{y}^{\hat{43}} \rangle^{q}_{k} \; e^{i \; \pi}  \\     
    \tilde{m}^{\prime}_{32}=\theta_{34} \; \langle \mathcal{A}_{y}^{\hat{45}} \rangle^{q}_{k}   
    \quad , \quad  
    \tilde{m}^{\prime}_{34} \approx 0 \; .
   \end{gather}
   In addition
   \begin{equation}
    \left( \begin{array}{cc} \tilde{m}^{\prime}_{21} 
            & \tilde{m}^{\prime}_{23} \\
            \tilde{m}^{\prime}_{41} 
            & \tilde{m}^{\prime}_{43} \end{array} \right)
    = \left( \begin{array}{cc} 1 & -\theta_{24} \\
            \theta_{24} & 1 \end{array} \right)
       \left( \begin{array}{cc} 0 & \theta_{34} \; \langle \mathcal{A}_{y}^{\hat{45}} \rangle^{q}_{k} \\
               \langle \mathcal{A}_{y}^{\hat{43}} \rangle^{q}_{k}   \;
               e^{i \; \pi} & 0 \end{array} \right)                    
   \end{equation}
   leads to
   \begin{gather}
    \tilde{m}^{\prime}_{21}=\theta_{24} \; \langle \mathcal{A}_{y}^{\hat{43}} \rangle^{q}_{k} 
    \quad , \quad  
    \tilde{m}^{\prime}_{41}=\langle \mathcal{A}_{y}^{\hat{43}} \rangle^{q}_{k} \;  e^{i \; \pi}  \\     
    \tilde{m}^{\prime}_{23}=\theta_{34} \; \langle \mathcal{A}_{y}^{\hat{45}} \rangle^{q}_{k}   
    \quad , \quad  
    \tilde{m}^{\prime}_{43} \approx 0 \; .
   \end{gather}
   The whole transformation reads
   \begin{equation}
     M_{3}=R^{t}_{24} M_{2} R_{24}=\left( \begin{array}{cccc} 
           \tilde{m}_{d} & \tilde{m}^{\prime}_{12} & \tilde{m}^{\prime}_{13}  
           &  \langle \mathcal{A}_{y}^{\hat{43}} \rangle^{q}_{k}  \; e^{i \; \pi} \\
           \tilde{m}^{\prime}_{12} & m^{\star}_{s} & \tilde{m}^{\prime}_{23} & 0 \\
           \tilde{m}^{\prime}_{13} & \tilde{m}^{\prime}_{23} & m_{b} & 0 \\
            \langle \mathcal{A}_{y}^{\hat{43}} \rangle^{q}_{k}  \; e^{i \; \pi} & 0 & 0 & m_{x} 
            \; e^{i \; \pi}  
           \end{array} \right) 
   \end{equation} 
   where $\tilde{m}^{\prime}_{13}=\theta_{34} \; \langle \mathcal{A}_{y}^{\hat{43}} \rangle^{q}_{k}$.
   {\em{At this state of the calculation it is remarkable that all off-diagonal elements 
   $\tilde{m}_{12}$, $\tilde{m}_{13}$ and $\tilde{m}_{23}$ in the upper left
   $3 \times 3$ matrix are real and positive.}} \\
   \\
   Fourth step: \;
   We perform the real rotation $R_{14}$ on $M_{3}$. 
   The purpose is to put
   zeroes in the $14,41$ elements
   The zeroes in the $14,41$ elements are implemented by diagonalising the
   outer block of $M_{2}$. This block is obtained by striking out the 
   row and the column in which the unit elements of $R_{14}$ appear.
   Thus we get the reduced rotation 
   \begin{equation}
    \label{reducr14}
    R^{t}_{14} \left( \begin{array}{cc} \tilde{m}_{u} 
             & \langle \mathcal{A}_{y}^{\hat{43}} \rangle^{q}_{k} \; e^{i \; \pi} \\  
             \langle \mathcal{A}_{y}^{\hat{43}} \rangle^{q}_{k} \; e^{i \; \pi} & 
             m_{x}  \; e^{i \; \pi}  \end{array} \right) R_{14} := 
    \left( \begin{array}{cc} m_{d} \; e^{i \; \delta_{d}} & 0 \\  
                0 & m_{x} \; e^{i \; \pi} \end{array} \right) 
   \end{equation}
   where we have introduced new phase $\delta_{d}$. Again $m_{x}$ will be practically unchanged and
   therefore we have already written $m_{x} \; e^{i \; \pi}$ in (\ref{reducr14}).
    From this matrix equation we obtain the mixing angle $\theta_{14}$  
   \begin{equation}
     \theta_{14} \approx \frac{\langle \mathcal{A}_{y}^{\hat{43}} \rangle^{q}_{k} \; e^{i \; \pi}}{m_{x} 
                       \; e^{i \; \pi}-\tilde{m}_{d}}
                 \approx \frac{\langle \mathcal{A}_{y}^{\hat{43}} \rangle^{q}_{k}}{\mid m_{x} \mid} 
   \end{equation}
   Thus the mixing angle $\theta_{14}$ turns out to be
   \begin{equation}
    \theta_{14} \approx \frac{\langle \mathcal{A}_{y}^{\hat{43}} \rangle^{q}_{k}}{m_{x}}=0.0002 \ll 1 
   \end{equation}
   where we have inserted (\ref{valuesvevso3}) and (\ref{valuemx}).
   Next we determine the diagonal elements. Since $\theta_{14} \ll 1$ we obtain in 
   the small mixing angle approximation
   \begin{equation} 
    m_{d} \;  e^{i \; \delta_{d}} \approx \tilde{m}_{d}
           - \frac{\left( \langle \mathcal{A}_{y}^{\hat{43}} \rangle^{q}_{k} \right)^{2}}{m_{x}} 
            \; e^{i \; \pi}
           =\tilde{m}_{d} +  
            \frac{\left( \langle \mathcal{A}_{y}^{\hat{43}} \rangle^{q}_{k} \right)^{2}}{m_{x}}  \\    
   \end{equation}
   We calculate from (\ref{valuesvevso3}) and (\ref{valuemx})
   \begin{equation}
    \frac{\left( \langle \mathcal{A}_{y}^{\hat{43}} \rangle^{q}_{k} \right)^{2}}{m_{x}}=2.3 \; \text{MeV}
   \end{equation}
   This yields
   \begin{equation}
    \tilde{m}_{d}=11.0 \; \text{MeV} \quad , \quad \delta_{d}=0 \; .
   \end{equation}
   Finally, when calculating the off-diagonal elements, we observe that since $\theta_{14} \ll 1$ 
   they are approximately unchanged.
   The whole transformation reads
   \begin{equation}
     M_{4}=R^{t}_{14} M_{3} R_{14}=\left( \begin{array}{cccc} 
           m_{d} & \tilde{m}_{12}^{\prime} & \tilde{m}_{13}^{\prime}  & 0 \\
           \tilde{m}^{\prime}_{12} & m^{\star}_{s} & \tilde{m}_{23}^{\prime} & 0 \\
           \tilde{m}^{\prime}_{13} & \tilde{m}_{23}^{\prime} & m_{b} & 0 \\
           0 & 0 & 0 & m_{x} \; e^{i \; \pi}  
           \end{array} \right) \stackrel{!}{=} \left( \begin{array}{cc} M^{d} & 0 \\
                                0 & \tilde{m}_{x} \end{array} \right)  \; .
   \end{equation}
   Thus $M^{d}$ reads 
   \begin{equation}
     \label{mdresult}
     M^{d}=\left( \begin{array}{ccc} 
           m_{d} & m_{12} & m_{13} \\
           m_{12} & m^{\star}_{s} & m_{23} \\
           m_{13} & m_{23} & m_{b} \\ 
           \end{array} \right) 
   \end{equation}
   where 
   \begin{eqnarray}
    \label{msresults}
    m_{d}=13.4 \; \text{MeV} & , & m_{12} 
    =\theta_{24} \langle \mathcal{A}_{y}^{\hat{43}} \rangle^{q}_{k}
    =\frac{\langle \mathcal{A}_{y}^{\hat{43}} \rangle^{q}_{k} 
     \langle \mathcal{A}_{y}^{\hat{45}} \rangle^{q}_{k}}{m_{x}}
    = 24.46  \; \text{MeV} 
    \nonumber \\
    m_{s}^{\star} > 255 \; \text{MeV} & , &
    m_{13}
    =\theta_{34} \langle \mathcal{A}_{y}^{\hat{43}} \rangle^{q}_{k}
    =\frac{\langle \mathcal{A}_{y}^{\hat{43}} \rangle^{q}_{k} 
     \langle \mathcal{A}_{y}^{\hat{47}} \rangle^{q}_{k}}{m_{x}}
    = 16.65  \; \text{MeV} 
    \nonumber \\
    m_{b}=4250 \; \text{MeV} & , & m_{23}
    =\theta_{34} \langle \mathcal{A}_{y}^{\hat{45}} \rangle^{q}_{k}
    =\frac{\langle \mathcal{A}_{y}^{\hat{45}} \rangle^{q}_{k} \langle  
      \mathcal{A}_{y}^{\hat{47}} \rangle^{q}_{k}}{m_{x}}
    = 173.8  \; \text{MeV} \nonumber  \\
   \end{eqnarray}
   Note that in the second step we have inserted 
   \begin{equation}
    \theta_{24}=\frac{\langle \mathcal{A}_{y}^{\hat{45}} \rangle^{q}_{k}}{m_{x}} \quad , \quad
    \theta_{34}=\frac{\langle \mathcal{A}_{y}^{\hat{47}} \rangle^{q}_{k}}{m_{x}} \; .
   \end{equation}

  \section{The meaning of the triplet VEV of $SO(3)_{F}$}

   We compare the result (\ref{msresults})  
   with (\ref{ckmmatrixoffdiagonalelements}). We observe that 
   \begin{itemize}
    \item The off-diagonal elements $m_{12}$, $m_{13}$ and $m_{23}$ are equal.
    \item The diagonal elements $m_{d}$ and $m_{b}$ are equal.
    \item The diagonal element $m_{s}^{\star} > 255$ MeV turns out to be at least two times bigger
          than $\tilde{m}_{s}^{\prime}=119.2$ MeV.
    \item There is no CP violation.
   \end{itemize}
   {\em{We stress that these results are only valid if $m_{x}$ is the dominant element in $M$}}.  
   We will discuss the case where $m_{x}$ is lowered to $\mathcal{O}(\tilde{m}_{s})$ at the end of
   this section. In this case it is possible to obtain the correct value for $\tilde{m}_{s}$
   and a CP violation.

   We calculate the CKM mixing matrix and the down-type quark masses from (\ref{msresults}): 
   \begin{enumerate}
    \item For the CKM mixing angles we obtain
     \begin{equation}
      s_{12}=0.104 \quad , \quad s_{13}=0.0039 \quad , \quad s_{23}=0.043  \; .
     \end{equation}
     We compare these results to the SM case
     \begin{equation}
      s_{12}^{SM}=0.22 \quad , \quad  s_{13}^{SM}=0.0039 \quad , \quad  s_{13}^{SM}=0.042 \; .
     \end{equation}
     We observe that $s_{13}$ and $s_{23}$ are in accordance with the SM case. However the Cabibbo angle
     $s_{12}$ turns out to be small by approximately a factor of two.
     This is a consequence of the too large value for $m_{s}^{\star}$.
    \item For the down-type quark masses we obtain
     \begin{equation}
      m_{d}=10.7 \; \text{MeV} \quad , \quad m_{s}=251 \; \text{MeV}  \quad , \quad
      m_{b}=4250 \; \text{MeV} \; .
     \end{equation}
     We compare these results to the SM case 
     \begin{equation}
      m_{d}=5 \; \text{to} \; 8.5 \; \text{MeV} \quad , \quad m_{s}=80 \; \text{to} \; 155 
      \; \text{MeV} \quad , \quad
      m_{b}=4000 \; \text{to} \; 4500 \; \text{MeV} \; . 
     \end{equation}
     We observe that only $m_{b}$ is in accordance with
     the SM case. 
   \end{enumerate}
   We see that the too large value $m_{s}^{\star}$ in (\ref{msresults}) not only 
   influences $s_{12}$ and $m_{s}$ but also $m_{d}$. As already mentioned above in small mixing angle
   approximation the CP violating phase $\delta_{13}$ is zero.

   The advantage of the small mixing angle approximation case is that ratios for the off-diagonal elements 
   $m_{12}$, $m_{13}$ and $m_{23}$ of $M^{d}$ are given by
   \begin{equation}
    \label{relatinsoffdiagmd}
    \frac{m_{12}}{m_{13}}=\frac{\langle \mathcal{A}_{y}^{\hat{45}} \rangle^{q}_{k}}
                               {\langle\mathcal{A}_{y}^{\hat{47}} \rangle^{q}_{k}}
    \; , \quad
    \frac{m_{12}}{m_{23}}=\frac{\langle \mathcal{A}_{y}^{\hat{43}} \rangle^{q}_{k}}
                               {\langle \mathcal{A}_{y}^{\hat{47}} \rangle^{q}_{k}}
    \; , \quad
    \frac{m_{13}}{m_{23}}=\frac{\langle \mathcal{A}_{y}^{\hat{43}} \rangle^{q}_{k}}
                               {\langle \mathcal{A}_{y}^{\hat{45}} \rangle^{q}_{k}} \; .
   \end{equation}
   and involve only the triplet VEV of $SO(3)_{F}$. \\
   \\
   {\bf{Conclusion:}} 
   {\em{In small mixing angle approximation, i.e. $m_{x} \gg m_{b},m_{s},m_{d}$ and $m_{x} \gg 
   \langle \mathcal{A}_{y}^{\hat{a}} \rangle^{q}_{k}$ for $\hat{a}=\hat{43},\hat{45},\hat{47}$,
   ratios for the off-diagonal elements
   $m_{12}$, $m_{13}$, $m_{23}$ of $M^{d}$ are given by rations   
   of the triplet VEV of $SO(3)_{F}$ according to 
   (\ref{relatinsoffdiagmd})}}. \\
   \\
   Looking at the calculation for the CKM mixing angles $s_{12},s_{13},s_{23}$ in Appendix
   \ref{appendixckmfrommd} we obtain with (\ref{relatinsoffdiagmd}) the ratios
   \begin{equation}
    \label{tripletvevckm}
    \frac{s_{13}}{s_{23}}=\frac{\langle \mathcal{A}_{y}^{\hat{43}} \rangle^{q}_{k}}
                                {\langle \mathcal{A}_{y}^{\hat{45}} \rangle^{q}_{k}}
    \; , \quad
    \frac{s_{12}}{s_{13}}=\frac{\langle \mathcal{A}_{y}^{\hat{45}} \rangle^{q}_{k} \; m_{b}}
                                {\langle \mathcal{A}_{y}^{\hat{47}} \rangle^{q}_{k}  
                                \; \left(m_{s}-m_{d} \right)}
     \; , \quad
     \frac{s_{12}}{s_{23}}=\frac{\langle \mathcal{A}_{y}^{\hat{43}} \rangle^{q}_{k} \; m_{b}}
                                {\langle \mathcal{A}_{y}^{\hat{45}} \rangle^{q}_{k} \; 
                                 \left(m_{s}-m_{d} \right)}
   \end{equation}
   {\bf{Conclusion:}} 
   {\em{In small mixing angle approximation, i.e. $m_{x} \gg m_{b},m_{s},m_{d}$ and $m_{x} \gg 
   \langle \mathcal{A}_{y}^{\hat{a}} \rangle^{q}_{k}$ for $\hat{a}=\hat{43},\hat{45},\hat{47}$,
   the CKM mixing angles $s_{12},s_{13},s_{23}$ are determined by rations
   of the triplet VEV of $SO(3)_{F}$ 
   plus a correction coming from the down-type masses according to (\ref{tripletvevckm})}}. \\
   \\
   We note that these results are independent of the explicit value for $m_{x}$ as long as 
   $m_{x} \gg m_{b},m_{s},m_{d}$ and $m_{x} \gg 
   \langle \mathcal{A}_{y}^{\hat{a}} \rangle^{q}_{k}$ for $\hat{a}=\hat{43},\hat{45},\hat{47}$.
   This can easily be seen as follows. Suppose we replace
   \begin{equation}
    m_{x} \to m_{x} \cdot k^{2}
   \end{equation}
   where $k \in \mathbb{R}$. The elements $m_{12}, m_{13}$ and
   $m_{23}$ are invariant under this replacement if we demand
   \begin{equation} 
    \langle \mathcal{A}_{y}^{\hat{a}} \rangle^{q}_{k} \to 
    \langle \mathcal{A}_{y}^{\hat{a}} \rangle^{q}_{k} \cdot k
   \end{equation}
   $\hat{a}=\hat{43},\hat{45},\hat{47}$. However also
   \begin{equation}
    \frac{\left( \langle \mathcal{A}_{y}^{\hat{a}} \rangle^{q}_{k} \right)^{2}}{m_{x}}
   \end{equation}
   is unchanged by this replacement. 

   Let us discuss the alternative case where $m_{x} \approx \mathcal{O}(\tilde{m}_{s})$. In 
   this case it is possible to get a realistic value for $m_{s}$ and the CP violating 
   phase $\delta_{13}$. However it is not possible to get an exact analytically solution. 
   First from appendix \ref{appendixprocedureM} we see that for general $m_{x}$ the absolute values for the
   off-diagonal elements of $M^{d}$ are given by (\ref{appendixsincosine})
   and (\ref{appendixm23})
   \begin{equation*}
    \mid m_{12} \mid =c_{14}s_{24}c_{34} \langle \mathcal{A}_{y}^{\hat{43}} \rangle^{q}_{k} \; , \quad 
    \mid m_{13} \mid =c_{14}s_{34}       \langle \mathcal{A}_{y}^{\hat{43}} \rangle^{q}_{k} \; , \quad 
    \mid m_{23} \mid =c_{24}s_{34}       \langle \mathcal{A}_{y}^{\hat{45}} \rangle^{q}_{k}
   \end{equation*}
   and involve products of sine and cosine of the mixing angles $\theta_{14}$, $\theta_{24}$,
   and $\theta_{34}$. Thus the rations (\ref{relatinsoffdiagmd}) get modified and depend now also on 
   these sine and cosine. 
   Looking at the equation (\ref{appendixms}) that determines $m_{s}$
   \begin{equation}
    m_{s} \; e^{\delta_{s}}=\tilde{m}_{s}  c_{24}^{2} \; e^{\tilde{\delta}_{s}} -
          2 \langle \mathcal{A}_{y}^{\hat{45}} \rangle^{q}_{k} s_{24} c_{24} \; e^{\delta_{\hat{45}}} 
          + \tilde{m}_{x} s_{24}^{2} \; e^{\delta_{\tilde{m}_{x}}} 
   \end{equation}
   where $\delta_{s}$, $\tilde{\delta}_{s}$, $\delta_{\hat{45}}$ and $\delta_{\tilde{m}_{x}}$ denote 
   phase factors we see that it is possible to obtain a realistic 
   for $m_{s}$ in (\ref{msresults}) for appropriate phase factors and a
   suitable mixing angle $\theta_{24}$.   
   In addition, we will get also non-vanishing phase factors in $M^{d}$ that can lead to a
   CP violating phase $\delta_{13}$ in $V_{CKM}$. 
   However since it is not possible to get an exact analytically solution we suggest to make a 
   numerical analysis. This analysis has to clarify if it is possible to recover the SM case or not.

\chapter{Summary and outlook}

 \label{SummaryOutlook}

 In this thesis we have investigated a Gauge-Higgs-unification model in five dimensions with broken chiral
 $SO(3)_{F}$ flavour symmetry. The model is based on the gauge group $SU(7)$ which
 unifies electroweak-, flavour and Higgs interactions. The Higgs fields are identified with the zero mode
 of some extra components of the higher-dimensional gauge field. Hence the Higgs fields in this model 
 are prevented from obtaining quadratically divergent corrections to their mass by the higher dimensional
 gauge symmetry. Therefore the model provides a solution to the hierarchy problem and will be valid at energy
 scales much above the electroweak breaking scale. The SM Higgs is replaced by three $SU(2)_{L}$
 Higgs doublets $H_{1}$, $H_{2}$ and $H_{3}$ which couple to the first, second and third generation,
 respectively. The electroweak gauge group $SU(2)_{L} \times U(1)_{Y}$ is broken by these $\{ H_{i} \}$
 to $U(1)_{em}$. The model includes an anomaly-free
 chiral $SO(3)_{F}$ flavour symmetry which explains in a natural way why there
 are exactly three generations as it is observed in nature. All fermion masses and mixing
 angles are computable in principle and thus the $SU(7)$ model provides a solution to the flavour problem. 
 The $SO(3)_{F}$ flavour symmetry is broken by additional Higgs fields coming from the selfadjoint
 part of $\Phi$ at energies much above the compactification scale. This way tree-level FCNC are naturally
 suppressed due to the large $SO(3)_{F}$ flavour gauge boson masses. 

 The $SU(7)$ model is an effective bilayered transverse lattice model with nonunitary parallel transporters
 in the extra dimension. 
 In chapter \ref{chaptereffectivetheorie} we have shown explicitly how an effective bilayered transverse
 lattice model can be obtained by starting with an ordinary $S^{1}/\mathbb{Z}_{2}$ model. In a first step
 we have put $S^{1}/\mathbb{Z}_{2}$ on a lattice. In a second step one has to calculate
 the renormalisation group flow.
 The endpoint of the RG-flow is an extra dimension which consist of only two points: the orbifold fixed 
 points. The bulk is completely integrated out. As a result of the blockspin transformations
 the parallel transporters in the extra dimension become nonunitary. In addition, a Higgs potential
 emerges naturally. We have shown 
 that for trivial orbifold projection $P$ and trivial minimum of the Higgs potential the eBTLM equals
 a $S^{1}/\mathbb{Z}_{2}$ continuum orbifold model with trivial orbifold projection $P$ 
 and KK-mode expansion truncated for all fields at the first KK mode. 
 We have seen that the truncated $S^{1}/\mathbb{Z}_{2}$ model and 
 consequently also the eBTLM is renormalisable. 

 We have analysed orbifold conditions for nonunitary parallel transporters. As an important result we
 have seen that for complex holonomy groups $H$ it is always possible to choose a maximal noncompact
 Abelian subgroup $A$ with Lie algebra $\mathfrak{a}$ such that $ P \eta P^{-1}=\eta$ for $\eta \in 
 \mathfrak{a}$. Using this, the nonunitary parallel
 transporter $\Phi$ can be written as
 \begin{equation*}
  \Phi = e^{A_{y}} \; e^{\eta} \; e^{A_{y}} \; ,
 \end{equation*}
 with $A_{y}$ such that $P A_{y} P^{-1}=-A_{y}$. It is
 essential that $P$ is involutive. On $S^{1}/\mathbb{Z}_{2}$ this is always the case.
 We have seen that when spontaneous symmetry breaking occurs and the orbifold projection $P$ is
 non-trivial the Higgs potential $V(\Phi)$ does not only depend on the selfadjoint factor $e^{\eta}$
 but also on the unitary factor $e^{A_{y}}$, i.e.
 \begin{equation*}
  V(\Phi)=V(e^{A_{y}} \; e^{\eta} \; e^{A_{y}})=\mathcal{V}(\eta,A_{y}) \; .
 \end{equation*}

 In chapter \ref{chaptereffectivetheorie} we have furthermore analysed in detail an eBTLM based on the
 gauge group $SU(2)$. In particular we have studied two important cases, which provide a basis for the
 $SU(7)$ model:
 \begin{enumerate}
  \item For trivial orbifold projection, i.e $P=\text{diag}(1,1)$, and non-trivial minimum $\Phi_{min}$ 
        of the Higgs potential $V(\Phi)$ at 
        \begin{equation}
         \Phi_{min}=\rho_{min} \frac{1}{\sqrt{2}} \left( \begin{array}{cc} e^{a_{1}} & 0 \\ 
                      0 & e^{a_{2}} \end{array} \right) \; ,
        \end{equation}
        see (\ref{minimumofphi1}),  we have obtained 
        (see Proposition 
        \ref{propositionlargegaugebosonmasses}) that in the limit of large $a_{2}$ (i.e.
        $a_{2} \gg 1$) the eBTLM allows gauge boson masses for some zero mode  
        and first excited KK mode gauge fields, which are much larger than the compactification scale
        (see \ref{masstermsu2finalresultabularlargeai}).
        This behaviour has no counterpart in an ordinary $S^{1}/\mathbb{Z}_{2}$ model.
  \item For non-trivial orbifold projection (\ref{nontrivialpsuu1}), i.e. $P=\text{diag}(1,-1)$ , we have seen
        that the Higgs potential $V(\Phi)$ depends also on the unitary factors
        $e^{A_{y}}$, and if the latter assume a VEV the gauge group $G_{0}=U(1)$, which is left unbroken
        by the orbifold projection $P$, is spontaneously broken. This 
        spontaneous symmetry breaking is however completely different from the spontaneous symmetry
        breaking by VEVs for the selfadjoint factor $e^{\eta}$. Indeed we have seen that a spontaneous 
        symmetry breaking by a VEV for the unitary factors $e^{A_{y}}$ equals a continuous Wilson line 
        breaking or Hosotani breaking. This symmetry breaking allows the reduction of the $rank$ of the
        underlying gauge group $G_{0}=U(1)$, i.e. $G_{0}$ is completely broken. 
 \end{enumerate}
 
 Based on these two results, in chapter \ref{su7model} we have formulated 
 a realistic Gauge-Higgs unification model: The $SU(7)$ model. The group $SU(7)$ unifies 
 electroweak-, flavour- and Higgs interactions. 
 Colour was ignored.  As an intermediate step the model also unifies weak- and flavour
 interactions in the gauge group $SU(6)_{L} \subset SU(7)$. We have shown that zero modes of the 
 extra-dimensional component of the five-dimensional gauge fields 
 transform according to the fundamental representation of 
 $SU(2)_{L}$ and carry the hypercharge $\frac{1}{2}$. They serve as a substitute for the
 SM Higgs. The theory includes three $SU(2)_{L}$ Higgs doublets $H_{1}$ (\ref{higgsdoubleth1}), $H_{2}$
 (\ref{higgsdoubleth2}) and $H_{3}$ (\ref{higgsdoubleth3}). $H_{1}$ couples to the first, $H_{2}$ couples to
 the second and $H_{3}$ couples to the third generation.

 We have calculated all zero- and first KK mode gauge boson masses associated to the gauge group 
 $SU(2)_{L} \times U(1)_{Y} \times SO(3)_{F}$ in terms of the minimum $\Phi_{min}$ (\ref{minimumphi})
 of the Higgs potential. We have identified the SM gauge bosons, i.e. the $W$ bosons, the $Z$ boson and 
 the photon, as zero mode gauge bosons of the electroweak gauge group
 $SU(2)_{L} \times U(1)_{Y} \subset SU(7)$. We have seen
 that the $W$ and the $Z$ boson get masses only from the off-diagonal part of (\ref{minimumphi}). 
 Thus their masses are $\mathcal{O}(246)$ GeV. The photon remains massless.
 All other gauge bosons receive masses mainly from the diagonal part of 
 (\ref{minimumphi}). The diagonal part of (\ref{minimumphi}) reads
 \begin{equation}
  \label{summarydiagphi}
  \Phi^{diag}_{min}=\rho_{min}
  \frac{1}{\sqrt{2}} \text{diag}(e^{a_{1}},e^{a_{2}},e^{a_{3}},e^{a_{4}},e^{a_{5}},e^{a_{6}},e^{a_{7}}) \; .
 \end{equation}
 We stress that in order to obtain a realistic model, it is important that: 
 \begin{itemize}
  \item The minimum (\ref{minimumphi}) of the Higgs potential is quasi $\mathcal{S}_{2}$ symmetric, 
        i.e. we have $a_{4}=a_{1}$, $a_{5}=a_{2}$ and $a_{6}=a_{3}$ in 
        (\ref{minimumphi}) and (\ref{summarydiagphi}), respectively. Otherwise the
        $W$ bosons will get masses from the diagonal part of $\Phi_{min}$ and thus will
        be very heavy.  
  \item For a quasi $\mathcal{S}_{2}$ symmetric (\ref{minimumphi}) and (\ref{summarydiagphi}), respectively, 
        we need $a_{1} \neq a_{2} \neq a_{3}$ and 
        $a_{i} \gg 1$ for at
        least two $a_{i}$, $i=1,2,3$. If these
        conditions are fulfilled, all zero mode- and first excited KK mode flavour gauge boson receive
        masses from spontaneous symmetry breaking much above the compactification scale
        and thus tree-level FCNC are naturally suppressed.
        If $a_{i} \gg 1$ already for one $a_{i}$, $i=1,2,3$, all first excited KK modes of the SM gauge
        bosons, i.e. $W^{(1)}$, $Z^{(1)}$ and $\gamma^{(1)}$, get masses from spontaneous symmetry 
        breaking much above the compactification scale and thus will be very heavy. 
 \end{itemize}

 Furthermore we have calculated the weak mixing angle in the $SU(7)$ model which unfortunately turns out
 to be too small. Following Antoniadis, Benakli  and Quiros \cite{Antoniadis:2001cv} this problem could  
 be solved as follows: We start with the larger gauge group $SU(7) \times
 U(1)^{\prime\prime}$. The larger gauge group $SU(7) \times
 U(1)^{\prime\prime}$ is broken by orbifolding and the
 additional boundary conditions to  
 \begin{equation*}
  SU(7) \times U(1)^{\prime\prime} \stackrel{\text{P} + \text{b.c.}}{\longrightarrow}
  SU(2)_{L} \times U(1) \times SO(3)_{F}
  \times U(1)^{\prime\prime} \; ,
 \end{equation*}
 i.e. the extra $U(1)^{\prime\prime}$ is unaffected by the orbifold projection and the 
 additional boundary conditions.
 The hypercharge $U(1)_{Y}$ is identified as the sum of the $U(1)$
 and $U(1)^{\prime\prime}$ charges. We then have a gauge field $B_{\mu}$ associated to the hypercharge and
 a gauge field $A^{X}_{\mu}$ associated to its orthonormal combination \cite{Antoniadis:2001cv}. 
 The additional $U(1)^{\prime\prime}$ comes equipped with
 an additional coupling constant $g^{\prime\prime}_{5D}$. Since $g^{\prime\prime}_{5D}$ is undetermined we
 can set it to any desired value. This way we can restore the weak mixing angle of the SM 
 \cite{Antoniadis:2001cv}. Note that the additional $U(1)^{\prime\prime}$ is anomalous.
 However,
 these anomalies can be cancelled by a generalised Green-Schwarz mechanism \cite{Antoniadis:2001cv}.

 In chapter \ref{su7fermionmasses} we have analysed how fermion masses and CKM mixing angles 
 are generated by the Higgs mechanism in the context of nonunitary parallel transporters. 
 We have seen that in order to
 explain why quarks and leptons have different masses, we need to introduce two additional nonunitary
 parallel transporters $\Phi^{Quark}$ and $\Phi^{Lepton}$, and we have to make a clear distinction
 between $\Phi^{Quark}$, $\Phi^{Lepton}$ and $\Phi^{Gauge}=\Phi$. Hence we have also three different 
 Higgs potentials $V(\Phi^{Quark})$, $V(\Phi^{Lepton})$ and $V(\Phi^{Gauge})$ in the model. When
 spontaneous symmetry breaking occurs quark and lepton masses are given by the Yukawa interactions
 \begin{equation*}
  \bar{q}_{L} \Phi^{Quark}_{min} q_{R} \quad , \quad \bar{q}_{L} \Phi^{Lepton}_{min} q_{R} \; .
 \end{equation*} 

 The model has a large Higgs sector. The reason is that we have three different 
 nonunitary parallel transporters in the model and therefore also 
 three different Higgs potentials. The minima $\Phi_{min}^{Gauge}$,
 $\Phi^{Quark}_{min}$ and $\Phi^{Lepton}_{min}$ are parametrised by altogether
 thirty parameters.
 For instance ten parameters $a_{1}^{quark},\dots,a_{7}^{quark}$ and $\alpha_{\hat{43}}^{quark},
 \alpha_{\hat{45}}^{quark}, \alpha_{\hat{47}}^{quark}$ parametrise $\Phi_{min}^{Quark}$. 
 Fluctuations of $a_{1}^{quark},\dots,a_{7}^{quark}$ and $\alpha_{\hat{43}}^{quark},
 \alpha_{\hat{45}}^{quark}, \alpha_{\hat{47}}^{quark}$
 around the minimum $\Phi_{min}^{Quark}$ of $V(\Phi^{Quark})$
 give rise to $10$ Higgs particles:
 \begin{itemize}
  \item $3$ Higgs particles which are associated to fluctuations of  
        $\alpha_{\hat{43}}^{quark}, \alpha_{\hat{45}}^{quark}, \alpha_{\hat{47}}^{quark}$. 
        Their mass squared is given by
        \begin{equation*}
         \label{massesunitaryhiggs}
         m_{A_{y}^{\hat{a}(0)}}^{2} \sim \left(\frac{g_{4}}{R} \right)^{2} \; 
         \frac{\partial^{2} V(\Phi^{quark})}{\partial \alpha_{a}}
        \end{equation*}
        at the minimum $\Phi_{min}^{Quark}$ of $V(\Phi^{Quark})$.
  \item $7$ Higgs particles which are associated to fluctuations of   
        $a_{1}^{quark},\dots,a_{7}^{quark}$.
        Their mass squared is given by
        \begin{equation*}
         \label{massesselfadjointhiggs}
         m_{a_{i}}^{2} \sim \left(\frac{g_{4}}{R} \right)^{2} \; \frac{\partial^{2} 
         V(\Phi^{quark})}{\partial a^{quark}_{i}}
        \end{equation*}
        at the minimum $\Phi_{min}^{Quark}$ of $V(\Phi^{Quark})$.
 \end{itemize}
 For $\Phi^{Lepton}_{min}$ and  $\Phi^{Gauge}_{min}$ the situation is analogous.
 Thus we conclude that the $SU(7)$ model predicts $30$ Higgs particles 
 which may be found at the LHC/ILC.

 In chapter \ref{su7fermionmasses} we have also seen that the up-type quark masses $m_{u},m_{c},m_{t}$
 are given by the diagonal part of $\Phi_{min}^{Quark}$ only. This is an important result and stands
 in contrast to the SM case where up- and down-type masses are given by the same Higgs doublet.
 In particular, this means that we can produce Higgs particles associated to the fluctuations of 
 $a_{1}^{quark},\dots,a_{7}^{quark}$ and
 originating from the selfadjoint part of $\Phi^{quark}$ exclusively by e.g. $t\bar{t}$ scattering. 
  
 In chapter \ref{su7fermionmasses} we have also investigated how the CKM mixing angles and the
 down-type quark masses are determined
 by the parameters $a_{1}^{quark},\dots,a_{7}^{quark}$ and $\alpha_{\hat{43}}^{quark},
 \alpha_{\hat{45}}^{quark}, \alpha_{\hat{47}}^{quark}$. It turned out
 that for 
 \begin{equation} 
  \label{summarysmallmixing}
  \tilde{m}_{d}, \tilde{m}_{s}, \tilde{m}_{b}, 
  \langle \mathcal{A}_{y}^{\hat{43}} \rangle^{q}_{k}, 
  \langle \mathcal{A}_{y}^{\hat{45}} \rangle^{q}_{k},
  \langle \mathcal{A}_{y}^{\hat{47}} \rangle^{q}_{k} \ll m_{x} \; ,
 \end{equation}
 where $\tilde{m}_{d}$, $\tilde{m}_{s}$, $\tilde{m}_{b}$ and $m_{x}$ are given by (\ref{summarymdquarks})
 and $ \langle \mathcal{A}_{y}^{\hat{a}} \rangle^{q}_{k}$ for $\hat{a}=\hat{43},\hat{45},\hat{47}$
 are given by (\ref{summaryabbre}), it is possible to get an
 exact analytical solution for all down-type quark masses and CKM mixing angles. We have called this
 scenario small mixing angle approximation. As a result
 we have obtained the correct hierarchy for the CKM mixing angles, i.e. 
 \begin{equation*}
  s_{13} \ll s_{23} \ll s_{12} \; .
 \end{equation*}
 However the Cabibbo angles turn out to be too small by a factor of two. In addition, also the s-quark
 mass turned out to be too large by approximately a factor of two. We we have found that the
 off-diagonal entries in the down-type quark mass matrix (\ref{mdresult})
 \begin{equation}
    \label{summarymdmatrix}
    M^{d}=\left( \begin{array}{ccc} 
           \tilde{m}_{d} & m_{12} & m_{13} \\
           m_{12} & \tilde{m}_{s} & m_{23} \\
           m_{13} & m_{23} & \tilde{m}_{b} \\ 
           \end{array} \right) \;  
 \end{equation}
 i.e. $m_{12}$, $m_{13}$ and $m_{23}$, are given by the ratios
 \begin{equation*}
    \frac{m_{12}}{m_{13}}=\frac{\langle \mathcal{A}_{y}^{\hat{45}} \rangle^{q}_{k}}
                               {\langle\mathcal{A}_{y}^{\hat{47}} \rangle^{q}_{k}}
    \; , \quad
    \frac{m_{12}}{m_{23}}=\frac{\langle \mathcal{A}_{y}^{\hat{43}} \rangle^{q}_{k}}
                               {\langle \mathcal{A}_{y}^{\hat{47}} \rangle^{q}_{k}}
    \; , \quad
    \frac{m_{13}}{m_{23}}=\frac{\langle \mathcal{A}_{y}^{\hat{43}} \rangle^{q}_{k}}
                               {\langle \mathcal{A}_{y}^{\hat{45}} \rangle^{q}_{k}} \; .
 \end{equation*} 
 These ratios give an interpretation for the triplet VEV
 of $SO(3)_{F}$. 
 The triplet VEV of $SO(3)_{F}$ plus a correction coming from the down-type quark masses 
 determine the CKM mixing angles via
 \begin{equation*}
    \frac{s_{13}}{s_{23}}=\frac{\langle \mathcal{A}_{y}^{\hat{43}} \rangle^{q}_{k}}
                                {\langle \mathcal{A}_{y}^{\hat{45}} \rangle^{q}_{k}}
    \; , \quad
    \frac{s_{12}}{s_{13}}=\frac{\langle \mathcal{A}_{y}^{\hat{45}} \rangle^{q}_{k} \; m_{b}}
                                {\langle \mathcal{A}_{y}^{\hat{47}} \rangle^{q}_{k}  
                                \; \left(m_{s}-m_{d} \right)}
     \; , \quad
     \frac{s_{12}}{s_{23}}=\frac{\langle \mathcal{A}_{y}^{\hat{43}} \rangle^{q}_{k} \; m_{b}}
                                {\langle \mathcal{A}_{y}^{\hat{45}} \rangle^{q}_{k} \; 
                                 \left(m_{s}-m_{d} \right)}
 \end{equation*}
 in small mixing angle approximation (\ref{summarysmallmixing}). 
 In addition, in this case we have no CP violation. 

 We have argued that it is possible to cure the problem that $\tilde{m}_{s}$ is too large  
 while $s_{12}$ is too small by
 lowering $m_{x}$ to $\mathcal{O}(\tilde{m_{s}})$. In addition, in this case  
 we get non-vanishing phase factors in $M^{d}$ that can lead to CP violation. However it is not clear
 if we can recover the SM case or not.
 
 The latter problem is a motivation for further investigations and extensions. First, since the case 
 $m_{x} \sim \mathcal{O}(\tilde{m_{s}})$ is analytically not exactly solvable, we suggest to make a 
 detailed numerical analysis. This analysis has to clarify whether it is possible to recover the SM
 case or not. Suppose that this is not possible. Then of course the $SU(7)$ model in its simplest
 version is ruled out. However there are several
 possible extensions to the $SU(7)$ model.  
 One possible extension is to consider the $SU(7)$ model in two extra
 dimensions. In this case, the underlying orbifold can be e.g. $T^{2}/\mathbb{Z}_{2}$. It is known
 \cite{Antoniadis:2001cv,Scrucca:2003ut} that on the orbifold $T^{2}/\mathbb{Z}_{2}$ one has
 two independent Higgs doublets. In the case of the $SU(7)$ model this means that 
 (\ref{4t4mdmatrix}) can be non-symmetric, i.e. 
 \begin{equation}
     \label{nonsymmM}
     M=\left( \begin{array}{cccc} 
      \tilde{m}_{d} & 0 & 0 & 
      i \; k_{\hat{43}} \; \langle \mathcal{A}_{y}^{\hat{43}(0)} \rangle^{quark \; 1} \\
      0 & \tilde{m}_{s} & 0 & 
      i \; k_{\hat{45}} \; \langle \mathcal{A}_{y}^{\hat{45}(0)} \rangle^{quark \; 1} \\
      0 & 0 & \tilde{m}_{b} & 
      i \; k_{\hat{47}} \; \langle \mathcal{A}_{y}^{\hat{47}(0)} \rangle^{quark \; 1} \\
      i \; k_{\hat{43}} \; \langle \mathcal{A}_{y}^{\hat{43}(0)} \rangle^{quark \; 2} &
      i \; k_{\hat{45}} \; \langle \mathcal{A}_{y}^{\hat{45}(0)} \rangle^{quark \; 2} &
      i \; k_{\hat{47}} \; \langle \mathcal{A}_{y}^{\hat{47}(0)} \rangle^{quark \; 2} &
      m_{x} \end{array} \right)
 \end{equation} 
 where $\langle \mathcal{A}_{y}^{\hat{43}(0)} \rangle^{quark \; 1}$, 
 $\langle \mathcal{A}_{y}^{\hat{45}(0)} \rangle^{quark \; 1}$ and 
 $\langle \mathcal{A}_{y}^{\hat{47}(0)} \rangle^{quark \; 1}$
 denote VEVs for the neutral components of 
 the first three $SU(2)_{L}$ Higgs doublets while 
 $\langle \mathcal{A}_{y}^{\hat{43}(0)} \rangle^{quark \; 2}$, 
 $\langle \mathcal{A}_{y}^{\hat{45}(0)} \rangle^{quark \; 2}$ and 
 $\langle \mathcal{A}_{y}^{\hat{47}(0)} \rangle^{quark \; 2}$
 denote VEVs for the neutral components of  
 the second three $SU(2)_{L}$ Higgs doublets. In contrast to the $SU(7)$ model in five dimensions,
 $\langle \mathcal{A}_{y}^{\hat{a}(0)} \rangle^{quark \; 1}$ and 
 $\langle \mathcal{A}_{y}^{\hat{a}(0)} \rangle^{quark \; 2}$ for $\hat{a}=\hat{43},\hat{45},\hat{47}$ 
 are independent and can in particular get different VEVs, i.e. we could have
 \begin{equation*}
  \langle \mathcal{A}_{y}^{\hat{a}(0)} \rangle^{quark \; 1} \neq 
  \langle \mathcal{A}_{y}^{\hat{a}(0)} \rangle^{quark \; 2} 
 \end{equation*}
 for $\hat{a}=\hat{43},\hat{45},\hat{47}$. Thus $M$ will  
 in general be non-symmetric.
 The down-type quark mass matrix $M^{d}$ (\ref{summarymdmatrix}) 
 in this case is obtained by first bringing $M$ to block
 diagonal form by a biunitary transformation
 \begin{equation*}
     M \to U^{L} \; M \; U^{R \dagger}=\left( \begin{array}{cc} M^{d} & 0 \\
                                0 & \tilde{m}_{x} \end{array} \right) \; ,
 \end{equation*}
 where $U^{L}$ and $U^{R}$ are now different matrices that can be written as
 \begin{equation*}
  U^{L \dagger}=P^{L}_{34} R^{L}_{34} R^{L}_{24} R^{L}_{14} \; , \quad
  U^{R \dagger}=P^{R}_{34} R^{R}_{34} R^{R}_{24} R^{R}_{14} \; ,
 \end{equation*}
 compare with (\ref{summarymatrixudagger}).
 In contrast to the symmetric case we now have to determine four phases (two in the symmetric case)
 and six mixing angles (three in the symmetric case). For instance the mixing angles $\theta_{34}^{L}$ 
 respectively $\theta_{34}^{R}$ are, ignoring phases, given by  
 \footnote{Note that $\langle \mathcal{A}_{y}^{\hat{a}} \rangle^{q \; 1}_{k}:= 
 k_{\hat{a}} \; \langle \mathcal{A}_{y}^{\hat{a}(0)} \rangle^{quark \; 1}$ and
 $\langle \mathcal{A}_{y}^{\hat{a}} \rangle^{q \; 2}_{k}:= 
 k_{\hat{a}} \; \langle \mathcal{A}_{y}^{\hat{a}(0)} \rangle^{quark \; 2}$ 
 for $\hat{a}=\hat{43},\hat{45},\hat{47}$.}
 \begin{gather*}
   \tan{2 \theta_{34}^{L}}= 
     \frac{2 \left( m_{x} \langle \mathcal{A}_{y}^{\hat{47}} \rangle^{q \; 1}_{k} 
           + \tilde{m}_{b} \langle \mathcal{A}_{y}^{\hat{47}} \rangle^{q \; 2}_{k} \right)}
     {m^{2}_{x} - \tilde{m}_{b}^{2} + \left( \langle \mathcal{A}_{y}^{\hat{47}} \rangle^{q \; 2}_{k} 
           \right)^{2}
           -\left( \langle \mathcal{A}_{y}^{\hat{47}} \rangle^{q \; 1}_{k} 
           \right)^{2}}  \\
     \tan{2 \theta_{34}^{R}}=\frac{2 \left( m_{x} \langle \mathcal{A}_{y}^{\hat{47}} \rangle^{q \; 2}_{k} 
           + \tilde{m}_{b} \langle \mathcal{A}_{y}^{\hat{47}} \rangle^{q \; 1}_{k} \right)}
     {m^{2}_{x} - \tilde{m}_{b}^{2} + \left( \langle \mathcal{A}_{y}^{\hat{47}} \rangle^{q \; 1}_{k} 
           \right)^{2}
           -\left( \langle \mathcal{A}_{y}^{\hat{47}} \rangle^{q \; 2}_{k} 
           \right)^{2}}      \; .   
 \end{gather*} 
 Thus $\theta_{34}^{L}$ and $\theta_{34}^{R}$ are different if  
 $\langle \mathcal{A}_{y}^{\hat{47}(0)} \rangle^{quark \; 1}$ 
 and 
 $\langle \mathcal{A}_{y}^{\hat{47}(0)} \rangle^{quark \; 2}$ get different VEVs.
 In a second step we diagonalise $M^{d}$ by a second biunitary transformation
 \begin{equation*}
  U^{d}_{L} \; M^{d} \; U^{d \dagger}_{R}=M^{d}_{diag} \; ,
 \end{equation*}
 and the CKM matrix is given by
 \begin{equation*}
  V_{CKM}=U^{d \dagger}_{R} \; .
 \end{equation*}    
 The matrix $U^{d}_{L}$ involves three additional mixing angles $\theta_{12}^{L}$, $\theta_{13}^{L}$
 and $\theta_{23}^{L}$ and in principle two additional phases.  
 However, only the CP violating phase and three mixing angles of the CKM matrix
 $V_{CKM}=U^{d \dagger}_{R}$ are of importance. 
 We immediately see that in the non-symmetric case we have vastly more freedom 
 to recover the correct CP violation of the SM and the correct down-type quark masses and
 CKM mixing angles. However such a model may reintroduce unsuppressed tree-level FCNC
 and thus a carful analysis is needed. 
        
 Second, we would like to calculate the RG-flow for the $SU(7)$ model. The result of this RG-flow will 
 determine definite numerical values for all thirty parameters $a_{i},\dots,a_{7}$,
 $a_{i}^{quark},\dots,a_{7}^{quark}$,  $a_{i}^{lepton},\dots,a_{7}^{lepton}$, 
 $\alpha_{\hat{\hat{43}}}, \alpha_{\hat{\hat{45}}}, \alpha_{\hat{\hat{47}}}$, 
 $\alpha^{quark}_{\hat{\hat{43}}}, \alpha^{quark}_{\hat{\hat{45}}}, \alpha^{quark}_{\hat{\hat{47}}}$ and
 $\alpha^{lepton}_{\hat{\hat{43}}}, \alpha^{lepton}_{\hat{\hat{45}}}, \alpha^{lepton}_{\hat{\hat{47}}}$
 parametrising $\Phi^{gauge}_{min}$, $\Phi_{min}^{quark}$ and $\Phi_{min}^{lepton}$, respectively.
 This means that we can calculate definite numerical values for 
 \begin{itemize}
  \item all gauge boson masses: $W^{1,2}_{\mu}$, $Z_{\mu}$, $W^{1,2(1)}_{\mu}$, $Z^{(1)}_{\mu}$,
         $\gamma^{(1)}$, $H^{j(0)}_{\mu}$
        and $H_{\mu}^{j(1)}$
  \item the fermion masses: $m_{u},m_{c},m_{t},m_{d},m_{s},m_{b}$ and $e,\mu,\tau$
  \item the mixing angles: $s_{12},s_{13}$ and $s_{23}$ \; .
 \end{itemize}
 This is a remarkable feature of our model. Note that the neutrino masses are not determined
 by the results of the RG-flow.
  
 In addition, the RG-flow will also determine the shape of
 the Higgs potentials $V(\Phi^{qauge})$,  $V(\Phi^{quark})$ and  $V(\Phi^{lepton})$. Therefore it
 would be possible to calculate the masses of all thirty Higgs particles in the $SU(7)$ model.
 This is also a remarkable feature of our model. However there is one question which has to been
 answered: What is the origin of the quasi $\mathcal{S}_{2}$ symmetry? Is it an accidental symmetry
 or is there any principle behind this?

 Third, as already indicated above, we want to be able to extend our model to higher-dimensional 
 orbifolds, e.g. two dimensional orbifolds like $T^{2}/\mathbb{Z}_{2}$. In principle the same
 steps which have lead to an eBTLM can be repeated for the orbifold  $T^{2}/\mathbb{Z}_{2}$.
 As in the five-dimensional case, we can put  $T^{2}/\mathbb{Z}_{2}$ on a
 lattice which must be two-dimensional and then calculate the renormalisation group flow. Also
 we should be able to determine orbifold conditions for nonunitary parallel transporters. Note
 that on  $T^{2}/\mathbb{Z}_{2}$ the orbifold projection $P$ is still involutive. However 
 for two-dimensional orbifolds this not true in general. In particular,
 on the orbifold $T^{2}/\mathbb{Z}_{3}$, $T^{2}/\mathbb{Z}_{4}$
 and $T^{2}/\mathbb{Z}_{6}$ the orbifold projection $P$ has to fulfil $P^{3}=1$, $P^{4}=1$ and $P^{6}=1$,
 respectively. Therefore we have to reinvestigate orbifold conditions for nonunitary parallel transporters
 in the case of noninvolutive $P$. 

 Fourth we remind the reader that in the $SU(7)$ model we have completely ignored colour. If we
 include strong interactions we could ask whether it is possible to find a GUT extention for the
 $SU(7)$ model. A possible GUT group which is compatible with the assignments of fermions to the 
 different fixed points of the orbifold is the Pati-Salam group
 \begin{equation}     
  G_{PS}=SU(2)_{L} \times SU(2)_{R} \times SU(4)_{c} \; .
 \end{equation}
 We can extend $G_{PS}$ to
 \begin{gather}
  G_{extendedPS}=SU(6)_{L} \times SU(6)_{R} \times SU(4)_{c} \; .
 \end{gather}
 The extension from $SU(2)_{L}$ to $SU(6)_{L}$ is exactly the same
 as in the $SU(7)$ model. In addition, the analogue extension can be made for $SU(2)_{R}$. 
 It would be interesting to investigate how a Gauge-Higgs unification model can be built from
 $G_{extendedPS}$. In addition we have to determine the symmetry breaking pattern. 
 We expect that such a GUT breaking is possible only on higher-dimensional
 orbifolds.

\begin{appendix}

  \chapter{CKM mixing matrix from $M^{d}$}
  
    \label{appendixckmfrommd}

    \begin{lemma}
     Let $M^{d}$ be given by
     \begin{equation}
      M^{d}=\left( \begin{array}{ccc} \tilde{m}^{\prime}_{d} & m_{12} 
           & m_{13} \\ m_{12}^{\ast} & \tilde{m}^{\prime}_{s} & m_{23} \\  
             m_{13}^{\ast} & m_{23} & m_{b} \\ \end{array} \right)
     \end{equation}
     where 
     \begin{eqnarray}
      \label{ckmmatrixoffdiagonalelements1}
      & \tilde{m}_{d}^{\prime}=13.4 \; \text{MeV}  
      & m_{12}=\hat{m}_{12} + \hat{m}_{13} s_{23} e^{i \; \delta_{13}}
              \approx 24.46 \; \text{MeV} \nonumber \\
      & \tilde{m}_{s}^{\prime}=119.2 \; \text{MeV} 
      &  m_{13}=\hat{m}_{12} s_{23} - \hat{m}_{13} e^{i \; \delta_{13}}
               \approx 16.65 \; e^{i \; \frac{2 \pi}{3}} \text{MeV} 
               \nonumber \\
      & m_{b}=4250 \; \text{MeV} 
      & m_{23}=173.8 \; \text{MeV} \; ,
     \end{eqnarray}
     and
     \begin{equation}
      \label{valueoffdiagonalelelmentscomplexckm1}
      \hat{m}_{12}=24.8 \; \text{MeV}
      \quad , \quad \hat{m}_{13}=16.1 \; \text{MeV} 
      \quad , \quad \delta_{13}=\frac{2 \pi}{3}  \; .
     \end{equation}
     The unitary transformation 
     \begin{equation}
      \label{ckmlemmaunit1}
      V_{CKM}^{\dagger} \; M^{d} \; V_{CKM}=M^{d}_{diag} \; .
     \end{equation}
     leads to the CKM matrix
     \begin{equation}
      V_{CKM}= \left( \begin{array}{ccc} 1-\frac{1}{2} \lambda^{2} 
               & \lambda & A \lambda^{3} (\rho-i\eta) \\
               -\lambda & 1-\frac{1}{2} \lambda^{2} & A \lambda^{2} \\  
               A \lambda^{3}(1-\rho-i\eta) & -A \lambda^{2} & 1 \\
               \end{array} \right)
     \end{equation}
     where $\lambda=s_{12}=0.22$, $A \lambda^{2}=s_{23}=0.042$, 
     $A \lambda^{3} (\rho-i\eta)=s_{13} \; e^{-i \delta_{13}}$, 
     $s_{13}=0.0039$ and $\delta_{13}=\frac{2 \pi}{3}$.
   \end{lemma}

   \begin{proof}
    We write $V_{CKM}$ as a product of three Euler matrices and a phase matrix 
    \begin{equation}
     V_{CKM}=R_{23} P_{13}^{\ast} R_{13} P_{13} R_{12}
    \end{equation}
    where
    \begin{equation}
     P_{13}=\left( \begin{array} {ccc} e^{i \phi_{1}} & 0 & 0 \\
            0 & 1 & 0 \\ 0 & 0 & e^{i \phi_{3}} \end{array} \right) \; ,
    \end{equation}
    $\phi_{1}$ arbitrary and $\phi_{3}$ such that $\phi_{1}-\phi_{3}=\delta_{13}=\frac{2 \pi}{3}$.   
    Inserting this expansion in (\ref{ckmlemmaunit1}) we obtain
    \begin{equation}
     R^{t}_{12} P_{13}^{\ast} R^{t}_{13} P_{13} R^{t}_{23} \; M_{d} \; 
     R_{23}  P^{\ast}_{13} R_{13} P_{13} R_{12}=M^{d}_{diag}
    \end{equation}
    The calculation will therefore be divided into five steps.\\
    \\
    First step : \;
    We perform the real rotation $R_{23}$ on $M^{d}$. The purpose is to put
    zeroes in the $23,32$ elements
    \begin{equation}
     M^{d}_{1}=R^{t}_{23} M^{d} R_{23}=\left( \begin{array}{ccc} 
              \tilde{m}^{\prime}_{d} 
              & \tilde{m}_{12} & \tilde{m}_{13} \\ \tilde{m}_{12} & 
              \tilde{m}_{s} & 0 \\  
              \tilde{m}^{\ast}_{13} & 0 & \tilde{m}_{b} \\ \end{array} \right) 
    \end{equation} 
    The zeroes in the $23,32$ elements are implemented by diagonalising the
    lower 23 block of $M^{d}$. This block is obtained by striking out the row 
    and the column in which the unit element of $R_{23}$ appears.
    Thus we get the reduced rotation
    \begin{equation}
     \label{reduced23}
     R^{t}_{23} \left( \begin{array}{cc} \tilde{m}^{\prime}_{s} & m_{23} \\  
                m_{23} & m_{b} \end{array} \right) R_{23} := 
     \left( \begin{array}{cc} \tilde{m}_{s} & 0 \\  
                0 & \tilde{m}_{b} \end{array} \right) \; .
    \end{equation}
    This matrix equation fixes the mixing angle $\theta_{23}$, i.e.
    \begin{equation}
     \tan{2 \theta_{23}}
     =\frac{2 m_{23}}{m_{b}-\tilde{m}^{\prime}_{s}}  \; .
    \end{equation}
    Inserting (\ref{ckmmatrixoffdiagonalelements1})  we obtain 
    \begin{equation}
     \label{proofse23}
     s_{23}=0.042=A \lambda^{2}
    \end{equation}
    For the matrix elements $\tilde{m}_{s}$ and $\tilde{m}_{b}$ we get from (\ref{reduced23})  
    \begin{gather}
     \tilde{m}_{s}=c_{23}^{2} \tilde{m}^{\prime}_{s}-2 s_{23} c_{23} m_{23}
                   + s_{23}^{2} m_{b} \\
     \tilde{m}_{b}=s_{23}^{2} \tilde{m}^{\prime}_{s}+2 s_{23} c_{23} m_{23}
                   + c_{23}^{2} m_{b} \nonumber \; .
    \end{gather}
    Inserting (\ref{ckmmatrixoffdiagonalelements1}) and (\ref{proofse23}) we
    obtain
    \begin{equation}
     \label{mstilde}
     \tilde{m}_{s}=111.9 \; \text{MeV} 
     \quad , \quad \tilde{m}_{b} \approx m_{b} \; .
    \end{equation}
    The remaining elements of $M_{1}^{d}$ are given by
    \begin{equation}
     \left( \tilde{m}_{12}, \tilde{m}_{13} \right)=
     \left( m_{12}, m_{13} \right) R_{23} \; , 
    \end{equation}
    and consequently, since 
    \begin{equation}
     R_{23}=\left( \begin{array}{cc} c_{23} & s_{23} \\ -s_{23}
             & c_{23} \end{array} \right) \; ,
    \end{equation}
    we obtain
    \begin{gather} 
     \tilde{m}_{12}=m_{12}-s_{23} m_{13} =(=:) \hat{m}_{12} \\ 
     \tilde{m}_{13}=s_{23} m_{12}+ m_{13} =(=:) \mid \hat{m}_{13} \mid
     e^{i \; \delta_{13}} \nonumber
    \end{gather}
    where we have ignored terms of the order $\mathcal{O}(s_{23}^{2})$.
    Indeed, as indicated in the brackets, this equations define $m_{12}$ and
    $m_{13}$ in (\ref{ckmmatrixoffdiagonalelements1}). 
    In addition we  get
    \begin{gather} 
     \tilde{m}_{21}=m_{12}^{\ast}-s_{23} m_{13}^{\ast}=\hat{m}_{12}
                   =\tilde{m}_{12}  \\ 
     \tilde{m}_{31}=s_{23} m_{12}^{\ast}+ m_{13}^{\ast}=\mid \hat{m}_{13} \mid 
                    e^{-i \; \delta_{13}}=\tilde{m}_{13}^{\ast} \nonumber \; . 
    \end{gather}
    In the following let us write all phases explicitly. Thus $M^{d}_{1}$ reads
    \begin{equation}
     \label{m1dphases}
     M^{d}_{1}=\left( \begin{array}{ccc} 
               \tilde{m}^{\prime}_{d} & \tilde{m}_{12} 
               & \mid \tilde{m}_{13} \mid e^{-i \; \delta_{13}} \\
               \tilde{m}_{12} & \tilde{m}_{s} & 0 \\  
               \mid \tilde{m}_{13} \mid e^{i \; \delta_{13}} & 0 
               & \tilde{m}_{b} \\ \end{array} \right) 
    \end{equation} 
    where $\delta_{13}=\phi_{1}-\phi_{3}$. \\
    \\
    Second step: \;
    We multiply $M_{1}^{d}$ by the phase matrix $P_{13}$. The purpose of this
    rephasing is to put real values in the $13,31$ elements
    in order to get a real angle $\theta_{13}$ in step three. Indeed
    \begin{equation}
     \label{phasetrafomd}
     M^{d}_{2}=P_{13} M^{d}_{1} P^{\ast}_{13}=\left( \begin{array}{ccc} 
             \tilde{m}^{\prime}_{d} & \mid \tilde{m}_{12} \mid e^{i \; \phi_{1}} 
             &  \tilde{m}_{13}   \\
             \mid \tilde{m}_{12} \mid e^{-i \; \phi_{1}} & \tilde{m}_{s} & 0 \\  
             \tilde{m}_{13} & 0 & \tilde{m}_{b} \\ \end{array} \right) 
    \end{equation}
    transforms away $\delta_{13}$ in (\ref{m1dphases}). \\
    \\
    Third step: \;
    We perform the real rotation $R_{13}$ on $M_{2}^{d}$. The purpose is to put
    zeroes in the $13,31$ elements
    \begin{equation}
     \label{md3}
     M^{d}_{3}=R^{t}_{13} M^{d}_{2} R_{13}=\left( \begin{array}{ccc} 
               \tilde{m}_{d} 
               & \mid \tilde{m}_{12}^{\prime} \mid e^{i \; \phi_{1}} & 0 \\ 
               \mid \tilde{m}_{12}^{\prime} \mid  e^{-i \; \phi_{1}} & 
               \tilde{m}_{s} & 0 \\  
               0 & 0 & \tilde{m}_{b} \\ \end{array} \right) \; .
    \end{equation}
    The zeroes in the $13,31$ elements are implemented by diagonalising the
    outer 13 block of $M^{d}_{2}$.   This block is obtained by striking out the row 
    and the column in which the unit element of $R_{13}$ appears.
    Thus we get the reduced rotation reads 
    \begin{equation}
     \label{reducedrot13}
     R^{t}_{13} \left( \begin{array}{cc} \tilde{m}_{d}^{\prime} & 
                \tilde{m}_{13} \\  \tilde{m}_{13} & m_{b} 
                \end{array} \right) R_{13} := 
     \left( \begin{array}{cc} \tilde{m}_{d} & 0 \\  
                0 & \tilde{m}_{b} \end{array} \right) \; .
    \end{equation}
    Note that all quantities in this matrix equation are real due to the phase
    transformation (\ref{phasetrafomd})in step two. This leads to
    the real mixing angle
    \begin{equation}
     s_{13} \approx \frac{\tilde{m}_{13}}{m_{b}-\tilde{m}_{d}^{\prime}} \; .
    \end{equation}
    Inserting (\ref{ckmmatrixoffdiagonalelements1}) and (\ref{valueoffdiagonalelelmentscomplexckm1}) we
    obtain
    \begin{equation}
     \label{anglea13}
     s_{13}=0.0039=A \lambda^{3} \rho \; .
    \end{equation}
    For the matrix elements $\tilde{m}_{d}$ and $\tilde{m}_{b}$ we obtain from (\ref{reducedrot13}) 
    \begin{gather}
     \tilde{m}_{d}=\tilde{m}^{\prime}_{d}-2 s_{13} m_{13} \\
     \tilde{m}_{b}=m_{b}+2 s_{13} m_{13} \nonumber \; .
    \end{gather}
    Inserting (\ref{ckmmatrixoffdiagonalelements1}), (\ref{valueoffdiagonalelelmentscomplexckm1})
    and (\ref{anglea13}) we obtain
    \begin{equation}
     \label{mdtilde}
     \tilde{m}_{d}=13.3 \; \text{MeV}
     \quad , \quad \tilde{m}_{b} \approx m_{b} \; .
    \end{equation}    
    The remaining elements of $M^{d}_{2}$ read
    \begin{equation} 
     \tilde{m}_{12}^{\prime}=\tilde{m}_{12} \quad , \quad
     \tilde{m}_{23}^{\prime}=s_{13} m_{12}=0.097 \approx 0 \; .
    \end{equation}
    
    Fourth step: \;
    We multiply $M_{3}^{d}$ by the phase matrix $P_{13}$. The purpose of this
    rephasing is to put real values in the $12,21$ elements of (\ref{md3}) 
    in order to get a real angle $\theta_{12}$ in step five. Indeed
    \begin{equation}
     M^{d}_{4}=P^{\ast}_{13} M^{d}_{3} P_{13}=\left( \begin{array}{ccc} 
               \tilde{m}_{d} & \tilde{m}_{12}^{\prime}  & 0 \\ 
               \tilde{m}_{12}^{\prime} & \tilde{m}_{s} & 0 \\  
               0 & 0 & \tilde{m}_{b} \\ \end{array} \right)
    \end{equation}
    transforms away $\phi_{1}$ in the $12,21$ elements  of (\ref{md3}). Thus we are left with
    a complete real matrix $M^{d}_{4}$. \\
    \\
    Fifth step: \;
    We perform the real rotation $R_{12}$ on $M_{4}^{d}$. The purpose is to put
    zeroes in the $12,21$ elements
    \begin{equation}
     M^{d}_{diag}=R^{t}_{12} M^{d}_{4} R_{12}=\left( \begin{array}{ccc} 
               m_{d} & 0 & 0 \\ 0 & m_{s} & 0 \\  
               0 & 0 & m_{b} \\ \end{array} \right) 
    \end{equation}
    The zeroes in the $12,21$ elements are implemented by diagonalising the
    upper 12 block of $M^{d}_{4}$.  This block is obtained by striking out the row 
    and the column in which the unit element of $R_{12}$ appears.
    Thus we get the reduced rotation 
    \begin{equation}
     R^{t}_{12} \left( \begin{array}{cc} \tilde{m}_{d} & 
                \tilde{m}^{\prime}_{12} \\  \tilde{m}^{\prime}_{12} & 
                \tilde{m}_{s} 
                \end{array} \right) R_{12} := 
     \left( \begin{array}{cc} m_{d} & 0 \\  
                0 & m_{s} \end{array} \right) \; .
    \end{equation}
    This matrix equation fixes the mixing angle $s_{12}$, i.e. 
    \begin{equation}
     \tan{2 \theta_{12}} = \frac{2 \tilde{m}^{\prime}_{12}}{\tilde{m}_{s}
                         -\tilde{m}_{d}} \; .
    \end{equation}
    Inserting (\ref{mstilde}), (\ref{mdtilde}) and (\ref{valueoffdiagonalelelmentscomplexckm1}) we obtain
    \begin{equation}
     \label{angles12}
     s_{12}=0.22=\lambda
    \end{equation}
    Using the approximation $c_{12}=0.973 \approx 1-\frac{\lambda^{2}}{2}$
    the reduced rotation $R_{12}$ reads
    \begin{equation}
     R_{12} \approx \left( \begin{array}{cc} 1-\frac{\lambda^{2}}{2} 
             & s_{12} \\ -s_{12}
             & 1-\frac{\lambda^{2}}{2} \end{array} \right) \; .
    \end{equation}
    For the matrix elements $\tilde{m}_{d}$ and $m_{b}$ one obtain
    \begin{gather}
     m_{d}=\tilde{m}_{d}(1-\frac{\lambda^{2}}{2})^2
           -2 \lambda (1-\frac{\lambda^{2}}{2}) \tilde{m}^{\prime}_{12}
           +\tilde{m}_{s}  \lambda^{2} \\
     m_{s}=\tilde{m}_{d} \lambda^{2}
           +2 \lambda (1-\frac{\lambda^{2}}{2}) \tilde{m}^{\prime}_{12}
           +\tilde{m}_{s} (1-\frac{\lambda^{2}}{2})^2  \nonumber \; . 
    \end{gather}
    Inserting (\ref{mstilde}), (\ref{mdtilde}),(\ref{valueoffdiagonalelelmentscomplexckm1}) 
    and (\ref{angles12}) we obtain 
    \begin{equation}
     m_{d}=7.4 \; \text{MeV} \quad , \quad m_{s}=114.1 \; \text{MeV}
    \end{equation}    
    Note that $m_{b}=4250$ is approximately unchanged by all transformations.   
   \end{proof}

  \chapter{General procedure for obtaining $M^{d}$ and $V_{CKM}$ from $M$}

   \label{appendixprocedureM}
   We start with the $4 \times 4$ matrix
   \begin{equation}
    M=\left( \begin{array}{cccc} 
      \tilde{m}_{d} & 0 & 0 & m_{14}  \\
      0 & \tilde{m}_{s} & 0 & m_{24} \\
      0 & 0 & \tilde{m}_{b} & m_{34} \\
      m_{14} & m_{24} & m_{34} & m_{x} \end{array} \right)
   \end{equation}
   where $m_{14}$, $m_{24}$ and $m_{34}$ denote the off-diagonal elements of $M$. {\em{In the following we
   ignore phases and treat $M$ as real}}.
   We must transform $M$ on block diagonal form
   \begin{equation}
    \label{appendixmdtrafo}
    M \to U \; M \; U^{\dagger}=\left( \begin{array}{cc} M^{d} & 0 \\
                                0 & \tilde{m}_{x} \end{array} \right) \; ,
   \end{equation}
   where $M^{d}$ is the $3 \times 3$ down-quark mass matrix.
   We write $U$ as a product of the real three Euler matrices
   \begin{equation}
    U^{\dagger}=R_{34} R_{24} R_{14}
   \end{equation}
   where
   \begin{gather}
    R_{34}=\left( \begin{array}{cccc} 1 & 0 & 0 & 0 \\
                  0 & 1 & 0 & 0 \\ 0 & 0 & c_{34} & s_{34} \\
                  0 & 0 & -s_{34} & c_{34} \end{array} \right) \\
    R_{24}=\left( \begin{array}{cccc} 1 & 0 & 0 & 0 \\ 
                  0 & c_{24} & 0 & s_{24} \\ 0 & 0 & 1 & 0 \\
                  0 & -s_{24} & 0 & c_{24} 
                  \end{array} \right) \\
    R_{14}=\left( \begin{array}{cccc} c_{14} & 0 & 0 & s_{14} \\
                   0 & 1 & 0 & 0 \\ 0 & 0 & 1 & 0 \\
                  -s_{14} & 0 & 0 & c_{14} \end{array} \right) \\
   \end{gather}   
   Inserting this expansion (\ref{appendixmdtrafo}) reads
   \begin{equation}
    R^{t}_{14} R^{t}_{24} R^{t}_{34} M R_{34} R_{24} R_{14}=
                              \left( \begin{array}{cc} M_{d} & 0 \\
                                0 & \tilde{m}_{x} \end{array} \right)
   \end{equation}
   The calculatation will be divided into three steps. \\
   \\
   First step : \;
   We perform the real rotation $R_{34}$ on $M$. The purpose is to put
   zeroes in the $34,43$ elements
   \begin{equation}
     M_{1}=R^{t}_{34} M R_{34}=\left( \begin{array}{cccc} 
               \tilde{m}_{d} & 0 & \tilde{m}_{13} & \tilde{m}_{14} \\ 
               0 & \tilde{m}_{s} & \tilde{m}_{23} & \tilde{m}_{24} \\
               \tilde{m}_{13} & \tilde{m}_{23} & m_{b} & 0 \\ 
               \tilde{m}_{14} & \tilde{m}_{24} & 0 & \tilde{m}_{x} 
               \end{array} \right) 
   \end{equation}
   The zeroes in the $34,43$ elements are implemented by diagonalising the
   lower 34 block of $M$. This block is obtained by striking out the row 
   and the column in which the unit elements of $R_{34}$ appear.
   Thus we get the reduced rotation
   \begin{equation}
    R^{t}_{34} \left( \begin{array}{cc} \tilde{m}_{b} & m_{34} \\  
                m_{43} & m_{x} \end{array} \right) R_{34} := 
    \left( \begin{array}{cc} m_{b} & 0 \\  
                0 & \tilde{m}_{x} \end{array} \right) \; .
   \end{equation}
   From this matrix we get the mixing angle $\theta_{34}$ 
   \begin{equation}
    \tan{2 \theta_{34}}=\frac{2 m_{34}}{m_{x}-\tilde{m}_{b}} 
   \end{equation} 
   and the diagonal elements 
   \begin{gather} 
    m_{b}=\tilde{m}_{b}c_{34}^{2}-2 m_{34} s_{34} c_{34}+m_{x} s_{34}^{2} \\
    \tilde{m}_{x}=\tilde{m}_{b}s_{34}^{2}+2 m_{34} s_{34} c_{34}+m_{x} c_{34}^{2} \; .
   \end{gather}
   The remaing elements are given by the reduced rotation
   \begin{equation}
    \left( \begin{array}{cc} \tilde{m}_{13} & \tilde{m}_{14} \\
            \tilde{m}_{23} & \tilde{m}_{24} \end{array} \right)
    =\left( \begin{array}{cc} 0 & m_{14} \\
            0 & m_{24} \end{array} \right)                    
    \left( \begin{array}{cc} c_{34} & s_{34} \\
            -s_{34} & c_{34} \end{array} \right) \; .
   \end{equation}
   This leads to
   \begin{gather}
    \tilde{m}_{13}=-s_{34} m_{14} \quad , \quad  
    \tilde{m}_{14}=c_{34} m_{14} \\     
    \tilde{m}_{23}=-s_{34} m_{24} \quad , \quad  
    \tilde{m}_{24}=c_{34} m_{24}
   \end{gather}
   \\
   Second step : \;
   We perform the real rotation $R_{24}$ on $M_{1}$. The purpose is to put
   zeroes in the $24,42$ elements
   \begin{equation}
     M_{2}=R^{t}_{24} M_{1} R_{24}=\left( \begin{array}{cccc} 
          \tilde{m}_{d} & \tilde{m}^{\prime}_{12} & \tilde{m}_{13} 
          & \tilde{m}^{\prime}_{14} \\ 
          \tilde{m}^{\prime}_{12} & m_{s} 
          & \tilde{m}^{\prime}_{23} & 0 \\
          \tilde{m}_{13} & \tilde{m}^{\prime}_{23} & m_{b} 
          & \approx 0 \\ 
          \tilde{m}^{\prime}_{14} & 0 & \approx 0 & \tilde{m}^{\prime}_{x} 
          \end{array} \right) 
   \end{equation}
   The zeroes in the $24,42$ elements are implemented by diagonalising the
   lower middle block of $M_{1}$. This block is obtained by striking out the 
   row and the column in which the unit elements of $R_{24}$ appear.
   Thus we get the reduced rotation
   \begin{equation}
    R^{t}_{24} \left( \begin{array}{cc} \tilde{m}_{s} & \tilde{m}_{24} \\  
                \tilde{m}_{24} & \tilde{m}_{x} \end{array} \right) R_{24} := 
    \left( \begin{array}{cc} m_{s} & 0 \\  
                0 & \tilde{m}^{\prime}_{x} \end{array} \right) \; .
   \end{equation}
   From this matrix we get the mixing angle $\theta_{24}$ 
   \begin{equation}
    \label{teta24}
    \tan{2 \theta_{24}}=\frac{2 \tilde{m}_{24}}{\tilde{m}_{x}-\tilde{m}_{s}}=
                        \frac{2 c_{34} m_{24}}{\tilde{m}_{x}-\tilde{m}_{s}}
   \end{equation} 
   and the diagonal elements 
   \begin{gather} 
    \label{appendixms}
    m_{s}=\tilde{m}_{s}c_{24}^{2}-2 \tilde{m}_{24} s_{24} c_{24}
                 +\tilde{m}_{x} s_{24}^{2} \\
    \tilde{m}^{\prime}_{x}=\tilde{m}_{s}s_{24}^{2}+2 \tilde{m}_{24} s_{24} c_{24}
                          +\tilde{m}_{x} c_{24}^{2} \; .
   \end{gather}
   The remaing elements are given by the reduced rotation
   \begin{equation}
    \left( \begin{array}{cc} \tilde{m}^{\prime}_{12} 
            & \tilde{m}^{\prime}_{14} \\
            \tilde{m}^{\prime}_{23} 
            & \tilde{m}^{\prime}_{34} \end{array} \right)
    =\left( \begin{array}{cc} 0 & \tilde{m}_{14} \\
            \tilde{m}_{23} & 0 \end{array} \right)                    
    \left( \begin{array}{cc} c_{24} & s_{24} \\
            -s_{24} & c_{24} \end{array} \right) \; .
   \end{equation}
   This leads to
   \begin{gather}
    \tilde{m}^{\prime}_{12}=-s_{24} \tilde{m}_{14}=-s_{24}c_{34} m_{14}
    \quad , \quad  
    \tilde{m}^{\prime}_{14}=c_{24} \tilde{m}_{14}=c_{24}c_{34} m_{14}  \\ 
    \label{appendixm23}    
    \tilde{m}^{\prime}_{23}=c_{24} \tilde{m}_{23}=-c_{24}s_{34} m_{24}
    \quad , \quad  
    \tilde{m}^{\prime}_{34}=s_{24} \tilde{m}_{32}=-s_{24}s_{34} m_{24}
    \approx 0
   \end{gather}
   \\
   Third step : \;
   We perform the real rotation $R_{14}$ on $M_{2}$. The purpose is to put
   zeroes in the $14,41$ elements
   \begin{equation}
     M_{3}=R^{t}_{14} M_{2} R_{14}=\left( \begin{array}{cccc} 
          m_{d} & \tilde{m}^{\prime\prime}_{12} 
          & \tilde{m}^{\prime}_{13} & 0 \\ 
          \tilde{m}^{\prime\prime}_{12} & m_{s} 
          & \tilde{m}^{\prime}_{23} & \approx 0 \\
          \tilde{m}^{\prime}_{13} & \tilde{m}^{\prime}_{23} & m_{b} 
          & \approx 0 \\ 
          0 & \approx 0 & \approx 0 & \tilde{m}^{\prime}_{x} 
          \end{array} \right) 
   \end{equation}
   The zeroes in the $14,41$ elements are implemented by diagonalising the
   outer block of $M_{2}$. This block is obtained by striking out the 
   row and the column in which the unit elements of $R_{14}$ appear.
   Thus we get the reduced rotation
   \begin{equation}
    R^{t}_{14} \left( \begin{array}{cc} \tilde{m}_{d} & \tilde{m}^{\prime}_{14} \\  
    \tilde{m}^{\prime}_{14} & \tilde{m}^{\prime}_{x} \end{array} \right) 
    R_{14} := 
    \left( \begin{array}{cc} m_{d} & 0 \\  
                0 & \tilde{m}^{\prime\prime}_{x} \end{array} \right) \; .
   \end{equation}
   From this matrix we get the mixing angle $\theta_{14}$ 
   \begin{equation}
    \label{teta14}
    \tan{2 \theta_{14}}
    =\frac{2 \tilde{m}^{\prime}_{14}}{\tilde{m}^{\prime}_{x}-\tilde{m}_{d}}=
    \frac{2 c_{24}c_{34} m_{14}}{\tilde{m}^{\prime}_{x}-\tilde{m}_{d}}
   \end{equation} 
   and the diagonal elements 
   \begin{gather} 
    m_{d}=\tilde{m}_{d}c_{14}^{2}-2 \tilde{m}^{\prime}_{14} s_{14} c_{14}
                 +\tilde{m}_{x} s_{14}^{2} \\
    \tilde{m}^{\prime\prime}_{x}=\tilde{m}_{d}s_{14}^{2}+2 \tilde{m}_{24} s_{14} c_{14}
                          +\tilde{m}_{x} c_{14}^{2} \; .
   \end{gather}
   The remaing elements are given by the reduced rotation
   \begin{equation}
    \left( \begin{array}{cc} \tilde{m}^{\prime\prime}_{12} 
            & \tilde{m}^{\prime\prime}_{24} \\
            \tilde{m}^{\prime}_{13} 
            & \tilde{m}^{\prime\prime}_{34} \end{array} \right)
    =\left( \begin{array}{cc} \tilde{m}^{\prime}_{12}  & 0 \\
            \tilde{m}^{\prime}_{13} & 0 \end{array} \right)                    
    \left( \begin{array}{cc} c_{14} & s_{14} \\
            -s_{14} & c_{14} \end{array} \right) \; .
   \end{equation}
   This leads to
   \begin{gather}
    \label{appendixsincosine}
    \tilde{m}^{\prime\prime}_{12}=c_{14} \tilde{m}^{\prime}_{12}
    =-c_{14}s_{24}c_{34} m_{14}
    \quad , \quad  
    \tilde{m}^{\prime\prime}_{24}=s_{14} \tilde{m}^{\prime}_{12}
    =s_{14}s_{24}c_{34} m_{14} \approx 0  \\     
    \tilde{m}^{\prime}_{13}=c_{14} \tilde{m}_{13}
    =-c_{14} s_{34} m_{14}
    \quad , \quad  
    \tilde{m}^{\prime\prime}_{34}=s_{14} \tilde{m}_{13}
    =s_{14}s_{34} m_{14} \approx 0 \nonumber
   \end{gather}

\end{appendix}

\nocite{Nilse:2006jv}
\nocite{Hebecker:2003jt}
\nocite{Hebecker:2001jb}
\nocite{Quiros:2003gg}
\nocite{Irges:2004gy}
\nocite{Dienes:1998vg}
\nocite{Gaitan-Lozano:1995sm}
\nocite{Cotti:1998de}
\nocite{vonGersdorff:2002rg}

\addtocontents{toc}{\bigskip \textbf{Bibliography}} 
\bibliography{Literatur.bib}

\newpage
\thispagestyle{empty}
\vspace*{5ex}
{\LARGE
\noindent
\textbf{Acknowledgements}}

\vspace{10ex}

First of all, I would like to thank my supervisor Prof. Gerhard Mack for
his guidance and support. 

Im a also very grateful to the members of my group Thorsten Pr\"ustel,
Falk Neugebohrn and Michael R\"ohrs for their friendship and many fruitful
discussions about physics. 

Diverse conversations broadened and deepened my physical knowledge. Hence
I am greatly indebted to a number of colleagues, in particular 
Thorsten Pr\"ustel, Falk Neugebohrn, Michael  R\"ohrs, Florian 
Schwennsen, Martin Hentschinski, Frank Fugel and Thorben Kneesch.

It is also a pleasure to thank Sven Grosskreutz for his friendship and help.

Moreover, I am grateful to the members of the II. Institut f\"ur Theoretische
Physik and the DESY theory group for creating a very pleasant and stimulating
working atmosphere.

This work was supported by the Graduiertenkolleg ``Zuk\"unfitige Entwicklungen
in der Teilchenphysik''.

Finally, I would like to thank my parents for their love, support and trust.

\vspace{2ex}
\normalsize
\noindent

\end{document}